\documentclass[12pt]{article}

\usepackage[utf8]{inputenc}
\usepackage{amsfonts}
\usepackage{amsmath}
\usepackage{mathtools}
\usepackage{mathrsfs}
\usepackage{physics}
\usepackage{csquotes}
\usepackage{amssymb}
\usepackage{dsfont}
\usepackage[left=2cm, right=2cm, top=2cm, bottom=3cm]{geometry}
\usepackage{xcolor}
\usepackage[T1]{fontenc}
\usepackage{hyperref}
\hypersetup{
    colorlinks=true,
    linkcolor=blue!70!,
    filecolor=magenta,      
    urlcolor=cyan,
    }
\usepackage{framed}
\usepackage[framed]{ntheorem}
\usepackage{authblk}
\usepackage{tikz,tikz-cd}
\usepackage{adjustbox}

\let\oldemptyset\emptyset
\let\emptyset\varnothing

\newtheorem{definition}{Definition}

\newtheorem{lemma}{Lemma}
\newframedtheorem{postulate}{Postulate}
\newtheorem{proposition}{Proposition}
\newtheorem{theorem}{Theorem}
\newtheorem{corollary}{Corollary}

\newcommand{\af}{\alpha}
\newcommand{\ri}{\rightarrow}
\newcommand{\id}{\text{id}}
\newcommand{\ci}{\circ}

\newcommand{\Tra}[1]{\text{tr}\left( #1 \right)}
\newcommand{\Exp}[1]{\left\langle #1 \right\rangle_{\rho}}

\newcommand{\conj}[1]{\left\lfloor #1 \right\rfloor}

\newenvironment{proof}{\paragraph{Proof:}}{\hfill$\square$}

\title{Step-by-step derivation of the algebraic structure of quantum mechanics
\\
(or from nondisturbing to quantum correlations by connecting incompatible observables)
}
\author{Alisson Tezzin\footnote{alisson.tezzin@usp.br}}
\affil{\small\textit{Department of Mathematical Physics, Institute of Physics,
University of São Paulo
\\
R. do Matão 1371, São Paulo 05508-090, SP, Brazil}}
\date{\today}

\begin{document}

\maketitle
\begin{abstract}
Recently there has been much interest in deriving the quantum formalism and the set of quantum correlations from simple axioms. In this paper, we provide a step-by-step derivation of the quantum formalism that tackles both these problems and helps us to understand why this formalism is as it is. We begin with a structureless system that only includes real-valued observables, states and a (not specified) state update, and we gradually identify theory-independent conditions that make  the algebraic structure of quantum mechanics be assimilated by it. In the first part of the paper (sections~\ref{sec: observableEvents}-\ref{sec: commutativePart}), we derive essentially all the ``commutative part'' of the quantum formalism, i.e., all definitions and theorems that do not involve algebraic operations between incompatible observables, such as projections, Specker's principle, and the spectral theorem; at the statistical level, the system is nondisturbing and satisfies the exclusivity principle at this stage. In the second part of the paper (sections~\ref{sec: connectingIncompatible}-\ref{sec: quantumTheory}), we connect incompatible observables by taking transition probabilities between pure states into account. This connection is the final step needed to embed our system in a Hilbert space and to go from nondisturbing to quantum correlations. 
\end{abstract}
\tableofcontents

\section{Introduction}

In 1932, John von Neumann published his monumental \textit{Mathematische Grundlagen der Quantenmechanik}  \cite{von2018mathematical} (mathematical foundations of quantum mechanics), where, for the first time, quantum mechanics was presented as a unified and mathematically sound theory \cite{landsman2017foundations}; in general terms, von Neumann's formulation of quantum mechanics is the theory we use nowadays \cite{von2018mathematical, landsman2017foundations}. Right at the beginning of his book, von Neumann makes it clear that the emphasis of his work was placed on investigating the conceptual and interpretative questions that emerged from the then newly developed physical theory; in his own words (translated from German by Robert T. Beyer):
\begin{displayquote}
    ``...\textit{the principal emphasis shall be placed on the general and fundamental questions which have arisen in connection with this theory. In particular, the difficult problems of interpretation, many of which are even now not fully resolved, will be investigated in detail}.'' (\cite{von2018mathematical}, p.1).
\end{displayquote}
As we now know, von Neumann's efforts in this direction weren't enough, and the ``difficult problems of interpretation'' he was aiming at are still alive \cite{landsman2017foundations, norsen2017foundations, deRonde2020noInterpretation}. As a result, the range of existing interpretations of quantum mechanics is enormous, going from extreme realist views (on the wave function) such as the many-worlds interpretation \cite{everett1957relative, norsen2017foundations, halvorson2019realist}, to antirealist or ``pragmatist'' ones like quantum Bayesianism \cite{fuchs2010qbism, healey2022pragmatist}, a situation described by Hans Halvorson as a continuum of different views \cite{halvorson2019realist} and by Adán Cabello as a  map of madness \cite{cabello2017madness}. 

Although the conceptual problems addressed by von Neumann and many of his contemporaries  were left aside by most twentieth-century physicists, the rise of quantum information theory in the last two decades put these problems back in evidence. It was recognized that an adequate understanding of fundamental aspects of quantum systems was crucial to explaining the potential advantages of quantum information processes, and it was demonstrated that many characteristic features of quantum systems, such as contextuality, nonlocality, and entanglement \cite{howard2014contextuality,jozsa2003role, Soerensen2000entanglement}, provided possible explanations. The necessity of a proper understanding of these features and of the quantum formalism as a whole thus imposed itself, and remarkable progress has been made  in this direction since then. Nonlocality and contextuality were promoted to the position of research fields \cite{budroni2021review, brunner2014Bell}, going so far as to be studied outside the realm of physics \cite{dzhafarov2017contextuality, wang2021analysing}, and concepts such as quantum Darwinism \cite{zurek2009quantum, baldijao2021emergence}, decoherence~\cite{schlosshauser2005decoherence}, generalized contextuality~\cite{spekkens2005contextuality}, and the exclusivity principle \cite{chiribella2020exclusivity, amaral2014exclusivity} were explored at length. Operational (or ``generalized'') probability theories were developed to analyze, among other things, what distinguishes quantum theory from other conceivable physical theories \cite{janotta2014structure, plavala2021generalized}, and many operational frameworks for the study of classical, quantum, non-signaling, and contextual ``correlations'' were designed \cite{brunner2014Bell, amaral2018graph, csw2014graph}. Interesting proposals of first principles from which quantum theory can be derived were made \cite{hardy2001reasonable, chiribella2011derivation, masanes2011derivation, hardy2013reconstructing, cabello2019simple}, and new interpretations of quantum mechanics were put forward \cite{fuchs2010qbism, doring2010thing}. All these topics are open fields of study, in which much research is still being developed.

In this paper, we try to contribute to this debate by analyzing how a physical system (which we define as a purely mathematical structure that has systems in both quantum and classical mechanics as particular instances) takes the shape of a quantum system. We begin with a very broad definition of system (definition~\ref{def: system}), which consists essentially of a set of real-valued observables, a set of states that assign probabilities to the values of observables (as made precise by Kochen and Specker in Ref.~\cite{kochen1967problem}), and a mapping accounting for the update that takes place in the state of the system when an observable is measured, and we start to impose constraints on it via postulates. All constraints are key theory-independent features of quantum mechanics, as we explain below, and we argue that each one of them simplifies or refines some aspect of the general definition of system. We impose these constraints one by one, and after each imposition, some facet of the quantum formalism emerges. We conclude by showing that any system satisfying all conditions we pose (postulates~\ref{ax: separability}-\ref{post: pureState}) can be embedded in a quantum system, which naturally completes the derivation. Since all aspects of quantum mechanics relevant to quantum foundations have their roots in the mathematical formalism, shedding light on this formalism is  important, and this is exactly what most contemporary research on quantum foundations consists in. More specifically, our work is part of a longstanding line of research that aims to derive the quantum formalism from simple axioms \cite{hardy2001reasonable, chiribella2011derivation, masanes2011derivation, hardy2013reconstructing}, and it also tackles the problem of recovering the so-called set of quantum correlations \cite{amaral2014exclusivity, cabello2019simple, navascues2015almost, gonda2018almost, elie2018geometry}. 

The conditions that we single out are presented as constraints that simplify or refine the notion of system just to indicate that they are not counterintuitive or purely \textit{ad hoc} requirements. Our intention is to emphasize that, in general terms, we end up with quantum systems by making the concept of system  easier to work with, and not by introducing counterintuitive features to it. Nevertheless, this appeal to simplification is secondary in our work, and what is really important to us  is to provide a mathematically precise derivation of the quantum formalism that, due to its ``step-by-step'' or, say, ``tomographic'' character, can help us understand why this formalism is the way it is, i.e., why events are associated with certain observables that we call projections, why projections form an orthocomplemented lattice, why observables have traces, why the Specker's principle is satisfied, why expectations are given by the Born rule, and so on. To do so, we gradually identify key theory-independent features of quantum systems and derive the rest of the formalism from them. Having a mathematically precise derivation of these pivotal elements of quantum mechanics is fundamental to shedding  light on conceptually important aspects of the theory, such as contextuality, the exclusivity principle and the principle of local orthogonality \cite{fritz2013local}, and consequently to addressing the aforementioned problems of deriving the quantum formalism and the set of quantum correlations from simple axioms.  

By ``algebraic structure of quantum mechanics'' we essentially mean the  C*-algebra of operators  \cite{landsman2017foundations},  the orthocomplemented lattice of projections, the convex set of states, and all important definitions and theorems that follow from these structures, such as traces of observables, the spectral theorem, the functional calculus, and the Born rule. By ``quantum formalism'' we mean this algebraic structure, and therefore, as is commonplace in quantum foundations nowadays, we exclude dynamics from our derivation. Another common practice in quantum foundations that we follow is restricting the  discussion to finite-dimensional systems. Finally, with a slight abuse of terminology, we use the expressions ``quantum theory'' and ``quantum mechanics'' interchangeably. 

Kochen and Specker's definition of system (which is just a mathematically rigorous way of saying that the system includes real-valued observables and states that assign probabilities to the values of observables) enables us to define \textbf{functional relations} between observables, as they point out~\cite{kochen1967problem}, and it also enables us to introduce the notion of \textbf{statistically equivalent} observable events, where by \textbf{observable event} we mean a pair $(\Delta;A)$ in which $A$ is an observable and $\Delta$ is a Borel set. On the other hand, using the state update we can assign probabilities to \textit{sequences of observable events} (definition~\ref{def: sequentialMeasure}). Our entire discussion, and consequently our constraints, revolves around these concepts. 

In any system containing states and real-valued observables, each observable has a measurable space naturally associated with it, which corresponds to its spectrum, and each state of the system defines a probability measure in each one of these measurable spaces. Observable events associated with distinct observables thus lie in different spaces, so in principle we cannot reason about them in statistical terms, i.e.,  as if they were events in a single probability space whose measure is determined by the state of the system. However, as pointed out by Kochen and Specker~\cite{kochen1967problem}, the existence of functional relations between observables naturally implies that a \textit{compatibility relation} (definition~\ref{def: compatibility}) for observables is well defined, which in turn ensures that all observable events associated with a pair of compatible observables can be embedded (via statistical equivalence) in the spectrum of a single observable. For this reason, we say that there is an \textbf{informational link} connecting compatible observables, and the first simplifying constraint that a system must satisfy for the quantum formalism to emerge is the assumption that this link can be used to predict the outcomes of sequential measurements of pairs of compatible observables. Put differently, it means that a sequential measurement of two compatible observables is statistically equivalent to a single measurement of a fine-graining shared by them. This constraint appears in postulates~\ref{post: selfCompatibility} and \ref{post: compatibility} (postulate~\ref{post: selfCompatibility} is a particular case of postulate~\ref{post: compatibility}, but we present it separately for the sake of argument). The existence of an information link between two compatible observables $A$ and $B$, i.e., the existence of a single spectrum accounting for all events associated with  $A$ and $B$ (up to statistical equivalence), implies that these events  satisfy relations which are not universally valid among observable events. For instance, events that result from a sequential measurement of two compatible observables $A$ and $B$ obey the Bayes rule  (definition~\ref{def: commutativity}), and we can predict with certainty the outcome of a measurement of $B$ whenever we prepare the state of the system using an observable $A$ that is compatible with $B$  (definition~\ref{def: eigenstate}). The question thus naturally arises of whether, for some unknown reason, these rules can eventually be valid when $A$ and $B$ are incompatible, i.e., when the informational link that justifies them is not available. In quantum mechanics, these natural consequences of the existence of informational links cannot be achieved in any other way, so postulates \ref{post: commutativity} and \ref{post: eigenstates} cut these exceptions out. As William of Ockham puts it, ``plurality is not to be posited without necessity'' \cite{Spade1999Ockham, adamson2019medieval}, so it is fair to say that these postulates simplify the definition of system. All postulates we mentioned thus far have a strong information appeal, and so is postulate~\ref{post: pureState}, which asserts that a measurement can turn a pure state (which is naturally defined in any system satisfying postulates \ref{post: selfCompatibility} and \ref{post: observables}) into a non-pure one only if the experimentalist fails to acquire all the information that this measurement provides. Other key aspects of quantum systems simplify the general definition of system by making some of its features ``less complex'', in some sense. For instance, in finite-dimensional systems (definition~\ref{def: finiteSystem}), a distinction between nondegenerate and degenerate observables naturally arises, and it follows from the existence of functional relations that some degenerate observables are coarse-grainings of nondegenerate ones. Postulate~\ref{post: observables} thus asserts that there is no exception to this rule, i.e., that degenerate observables are always coarse-grainings of nondegenerate observables. The same postulate also asserts that  the set of states is the minimal convex set containing all experimentally accessible states (definition~\ref{def: accessibleState}) and that nondegenerate observables are as simple as ``$n$-sided dices'', by which we mean that, if we measure a nondegenerate observable many times without specifying the state of the system, the relative frequencies of all its possible values will be the same. The first part of postulate~\ref{post: transitionProbability} follows a similar line of thought: we show that ``transition probabilities'' between pure states are well defined, and we impose that the probability of transitioning between two pure states does not depend on the order of the transition. On the other hand, we say that the second part of postulate~\ref{post: transitionProbability} ``refine'' the definition of system, by which we mean, despite introducing a reasonable condition which is necessary for the emergence of \textit{interference terms} \cite{tausk2018foundations}, it does not have a clear explanation from our point of view. Finally, postulate~\ref{ax: subjectiveUpdate} does not depend on any physical system or theory in particular and must be satisfied due to the simple fact that we reason in statistical terms, whereas postulate~\ref{ax: separability}, which could be part of the definition of system, just explains what we mean when we say that two observables (or states) are equal. 
 
This paper is divided into two parts. In the first part, which ranges from section~\ref{sec: observableEvents} to section~\ref{sec: commutativePart}, we show that basically all the ``commutative part'' of the quantum formalism, i.e., all definitions and theorems that do not involve algebraic operations between incompatible observables, follow from postulates \ref{ax: separability}-\ref{post: commutativity} (with the exception of the second part of postulate~\ref{post: observables}), which essentially consist of the aforementioned considerations about informational links connecting compatible observables,  the assumption that degenerate observables are always coarse-grainings of nondegenerate ones, and the idea that nondegenerate observables are as simple as $n$-sided dices. At the statistical level, we have nondisturbance \cite{amaral2018graph} and the exclusivity principle (corollary~\ref{cor: exclusivityPrinciple}) in this part of the work, but we do not have quantum correlations yet. In the second part of the paper, which goes from section~\ref{sec: connectingIncompatible} to section~\ref{sec: quantumTheory}, we  connect incompatible observables by considering transition probabilities between pure states, and to do so we introduce postulates \ref{post: eigenstates}-\ref{post: pureState}. Postulate~\ref{post: eigenstates} basically says that the capacity to predict with certainty to outcome of some measurement is always a matter of gathering the right amount of information, whereas, as mentioned above, postulate~\ref{post: pureState} asserts that a measurement can turn a pure state into a non-pure one  only if the experimentalist fails to acquire all the information that this measurement provides. Postulate~\ref{post: transitionProbability}, on the other hand, says that the probability of transitioning between two pure states does not depend on the order of the transition, and it also singles out a ``marginalization condition'' that is satisfied by quantum systems and that is necessary for the emergence of the important \textit{interference terms} that appear in quantum mechanics when we evaluate the expectation of an observable with respect to some pure state by decomposing this state as a linear combination of vectors in a certain orthonormal basis \cite{tausk2018foundations, norsen2017foundations}. This second part of the work is the final step needed to embed our system in a Hilbert space and, consequently, to go from nondisturbing to quantum correlations \cite{csw2014graph, amaral2018graph}.

We adopt throughout the paper the convention that $0/0=0$. Given any set $O$, we denote  its power set by $\mathscr{P}(0)$, and the symbol $\mathfrak{B}(\mathbb{R})$ always denote the Borel $\sigma$-algebra on $\mathbb{R}$ \cite{folland1999real}. As usual, if $f:A \ri B$ is a function and $\beta$ is an element of $B$, we write $f^{-1}(\beta)$ rather than $f^{-1}(\{\beta\})$ to denote the pre-image of the singleton $\{\beta\}$ under $f$. 
\section{Basic framework}\label{sec: basicFramework}

As we mentioned in the introduction, a physical system in this paper consists in a mathematical structure that has systems in classical and quantum mechanics as particular cases. Roughly speaking, we use the label ``physical system'' as a synonym for ``system in some physical theory'', and we talk about systems rather than theories just because it is simpler --- even precisely defining physical theories is a very difficult task \cite{halvorson2012notBe, halvorson2016scientificTheories}. The minimal structure that a  system needs to have in order to be a candidate for the position of a quantum system is a set of real-valued observables, a set of states that assign probabilities to the values of observables, and a mapping that accounts for the update that takes place in the state of the system when an observable is measured. Inspired by Kochen and Specker \cite{kochen1967problem}, our definition of system goes as follows. 

\begin{definition}[Physical system]\label{def: system}
A physical system $\mathfrak{S}$ consists of nonempty sets $\mathcal{O}$ and $\mathcal{S}$ whose elements represent, respectively, observable features and states of the system under description, and of mappings $P,T$  defined as follows.
\begin{itemize}
    \item[(a)] $T$ associates, for each \textbf{observable event} $(\Delta;A) \in \mathfrak{B}(\mathbb{R}) \times \mathcal{O}$, a mapping $T_{(\Delta;A)}: \mathcal{S} \rightarrow \mathcal{S}$ corresponding to the state update associated with this event. For convenience, we assume that there is an element $0 \in \mathcal{S}$ (with no physical meaning) such that, given any observable $A$, $T_{(\Delta;A)}(\rho) = 0$ if $\Delta = \oldemptyset$ or $\rho = 0$. 
    \item[(b)] $P$ associates, for each pair $(\rho,A) \in \mathcal{S} \times \mathcal{O}$, a Borel measure $P_{\rho}^{A} \equiv P_{\rho}( \  \cdot \ ; A)$ on $\mathbb{R}$. If $\rho \neq 0$, $P_{\rho}^{A}$ is a probability measure, and it is the null measure otherwise. For any state $\rho \neq 0$ and any Borel set $\Delta$, $P_{\rho}^{A}(\Delta) \equiv P_{\rho}(\Delta;A)$ is the probability that a measurement of $A$ yields an outcome lying in $\Delta$ \cite{kochen1967problem}.
\end{itemize}
\end{definition}
To be precise, the system $\mathfrak{S}$ is the quadruple $(\mathcal{O},\mathcal{S}, P,T)$. As we mentioned in the introduction, $\mathfrak{B}(\mathbb{R})$ denotes the Borel $\sigma$-algebra. The element $0 \in \mathcal{S}$ is said to be the  \textbf{null state}, and, unless explicitly stated otherwise, by a state we  mean a non-null state. Similarly, unless explicitly stated otherwise, by an observable event we mean a pair $(\Delta;A)$ with $\Delta \neq \oldemptyset$.  For obvious reasons, we assume that $\mathcal{S}$ contains non-null states. 

We say that the elements of $\mathcal{O}$ are observable \textit{features}, and not observable properties, to emphasize that we are not committing ourselves to any realist view on observables. For simplicity, we will write ``observable'' rather than ``observable feature'' from now on.

As we know, systems in both classical and quantum mechanics satisfy definition~\ref{def: system}. In fact, for a quantum system associated with a (finite-dimensional, for simplicity) Hilbert space $H$, $\mathcal{O}$ is the set of all selfadjoint operators on $H$, $\mathcal{S}$ is the convex set of density operators, $P_{\rho}^{A}$ is the mapping $\mathfrak{B}(\mathbb{R}) \ni \Delta \mapsto \Tra{\rho\chi_{\Delta}(A)}$, where $\chi_{\Delta}$ is the characteristic function of $\Delta$ and $\chi_{\Delta}(A)$ is given by the (Borel) functional calculus \cite{kadison1997fundamentalsI}, whereas, for any state $\rho$,
\begin{align*}
    T_{(\Delta;A)}(\rho) &= \sum_{\af \in \sigma(A)} \frac{P_{\rho}^{A}(\{\alpha\})}{P_{\rho}^{A}(\Delta)}T_{(\alpha,A)}(\rho) =\sum_{\af \in \sigma(A)} \frac{\Tra{\rho E_{\af}}}{\Tra{\rho E_{\Delta}}} \frac{E_{\af}\rho E_{\af}}{\Tra{E_{\af}}\rho} = \sum_{\af \in \sigma(A)} \frac{E_{\af}\rho E_{\af}}{\Tra{\rho E_{\Delta}}} , 
\end{align*}
where $\sigma(A)$ denotes the spectrum of $A$, $E_{\af} \equiv \chi_{\{\af\}}(A)$ and $E_{\Delta} \equiv \chi_{\Delta}(A)$ \cite{nielsen2002quantum}. In classical mechanics, a system $\mathfrak{S}$ is represented by a certain phase space $\Lambda$ (usually a topological manifold), observables of $\mathfrak{S}$ are identified as real measurable functions on $\Lambda$, states are probability measures on $\Lambda$, $P_{\rho}^{A}$ is the pushforward of the measure $\rho$ along $A: \Lambda \rightarrow \mathbb{R}$, and $T_{(\Delta;A)}$ is determined by the conditional probability, i.e., 
\begin{align*}
    T_{(\Delta;A)}(\rho) &\doteq \rho( \ \cdot \ \vert A^{\Delta}) \equiv \frac{\rho( \ \cdot \ \cap A^{\Delta})}{\rho(A^{\Delta})},
\end{align*}
where $A^{\Delta} \equiv A^{-1}(\Delta)$ (recall that we are adopting the convention that $0/0=0$ ). 

Let $\mathfrak{F}$ be the class of all systems satisfying definition~\ref{def: system}. Adding constraints to the definition of system is equivalent to introducing postulates that single out a certain subclass of $\mathfrak{F}$  --- the class of all systems satisfying these postulates. It seems to us more natural to talk about postulates than constraints, so we will use this terminology from now on.

Definition~\ref{def: system} is too vague: although it includes each essential ingredient of a physical system, it makes no demands of them. The mapping $T$ is there to represent the state update, but definition~\ref{def: system} does not say that it must represent a physically meaningful update; $T$ is any mapping, so in some particular cases it can be, for instance, a constant function. Similarly, $P$ can be any random assignment of Borel measures, whereas the sets $\mathcal{S}$ and $\mathcal{O}$ are nothing more than nonempty sets. As a result, we can construct quadruples $(\mathcal{O},\mathcal{S},P,T)$ that satisfy definition~\ref{def: system} but that are not worthy of being called physical systems. Systems in classical and quantum mechanics are instances of definition~\ref{def: system} constructed within appropriate mathematical structures, and these structures naturally impose rules for $T,P$, $\mathcal{S}$ and $\mathcal{O}$ that make  $\mathfrak{S} \equiv (\mathcal{O},\mathcal{S},P,T)$ meaningful in these particular cases. This is an imposition from inside, let's say. Here we work from outside: we dictate rules for $\mathfrak{S}$ to make the definition  reasonable by itself, i.e.,  regardless of any particular case. By doing so, the initially contentless definition starts to incorporate characteristic features of the quantum formalism, until we get to the point where any system satisfying the (updated) definition gives rise to a quantum one, as theorem~\ref{thm: quantumEmbedding} shows.

Our first postulate simply explains what we mean by ``equality'' for states and observables. Our assumption is that the set of states shapes the sets of observables, i.e., two observables are equal iff we cannot distinguish them using states, and \textit{vice-versa}. It is important not to confuse this postulate with Spekkens's notion of ``ontological identity of empirical indiscernibles'' \cite{spekkens2005contextuality, spekkens2019ontologicalIdentity}: the identity  we impose here is not at the ontological level, to use Spekkens' terminology \cite{spekkens2005contextuality}, but at the operational one. This is not one of the constraints mentioned in the introduction, and this requirement could be part of the definition of system (equation~\ref{eq: KSseparability}, for instance, is part of Kochen and Specker's definition \cite{kochen1967problem}). For this reason, we enumerate it as postulate ``zero''. 

\begin{postulate}[Separability]\label{ax: separability} Two observables $A,B$ in a system $\mathfrak{S}$ are equal if and only if, for every state $\rho$,
\begin{align}
    P_{\rho}^{A} =\label{eq: KSseparability} P_{\rho}^{B}.
\end{align}
Similarly, two states $\rho$, $\rho'$ in $\mathfrak{S}$ are equal if and only if, for every observable $A$,
\begin{align}
    P_{\rho}^{A} = P_{\rho'}^{A}.
\end{align}
\end{postulate}

Let $A$ be an observable in a physical system $\mathfrak{S}$. We define the \textbf{spectrum} of $A$, denoted $\sigma(A)$, as the smallest set containing all its possible values, i.e.,
    \begin{align}
        \sigma(A) \doteq \cap\{\Delta \in \mathfrak{B}(H): \forall_{\rho \in \mathcal{S}}: P_{\rho}^{A}(\Delta) = 1\}.
    \end{align} 
$\mathcal{S}$ contains non null states, therefore $\sigma(A) \neq \oldemptyset$. A real number $\alpha$ is said to be an \textbf{eigenvalue} of an observable $A$ if $P_{\rho}^{A}(\{\af\})=1$ for some state $\rho$, and the \textbf{point spectrum} of $A$ consists of the collection of all its eigenvalues. Clearly, the point spectrum of $A$ lies inside its spectrum. It is important to note that, by construction, the spectrum of an observable depends on the set of states of the system.

As usual, we will focus on finite-dimensional systems:
\begin{definition}[Finite-dimensional system]\label{def: finiteSystem} We say that a system $\mathfrak{S}$ (definition~\ref{def: system}) is finite-dimensional if the following conditions are satisfied.
\begin{itemize}
    \item[(a)] The spectrum of every observable in $\mathfrak{S}$ is a finite set.
    \item[(b)] $\text{dim}(\mathfrak{S}) < \infty$, where
    \begin{align}
        \text{dim}(\mathfrak{S}) \doteq \sup\{\vert \sigma(A)\vert: A \in \mathcal{O}\}.    
    \end{align}
\end{itemize}
\end{definition}
The \textbf{dimension} of a finite-dimensional system $\mathfrak{S}$ is the number $\text{dim}(\mathfrak{S})$. We saw that, in any system, the spectrum of an observable is never empty. Therefore, in any finite-dimensional system $\mathfrak{S}$ we have $\text{dim}(\mathfrak{S}) \geq 1$, and consequently $\text{dim}(\mathfrak{S}) \in \mathbb{N}$. It is also important to note that, in a finite-dimensional system, the spectrum of an observable $A$ satisfies
\begin{align}
    \sigma(A) &=\label{eq: preEigenvector} \{\af \in \mathbb{R}: \exists_{\rho \in \mathcal{S}} (P_{\rho}^{A}(\{\af\}) \neq 0)\}.
\end{align}

Unless explicitly stated otherwise, we will assume that $\mathfrak{S}$ is finite-dimensional from now on. For convenience, we will  restrict the measure $P_{\rho}^{A}$ to $\sigma(A)$, and, unless explicitly stated otherwise, by an \textbf{observable event} we now mean a pair $(\Delta;A)$ where $A$ is an observable and $\Delta \subset \sigma(A)$. We denote by $p_{\rho}^{A}$ (or $p_{\rho}( \ \cdot \ , A)$) the probability distribution on $\sigma(A)$ induced by $P^{A}_{\rho} \equiv P_{\rho}( \ \cdot \ ,A)$, i.e., for any $\af \in \sigma(A)$,
\begin{align*}
    p_{\rho}^{A}(\af) \doteq P_{\rho}^{A}(\{\af\}).
\end{align*}
\section{Observable events}\label{sec: observableEvents}

Let $A$ be an observable in a (finite-dimensional) system $\mathfrak{S}$ (definition~\ref{def: finiteSystem}). An \textbf{$A$-event} is an observable event of the form $A^{\Delta} \equiv (\Delta;A)$, where $\Delta \subseteq \sigma(A)$, and we denote by $\mathbb{E}_{A}$ the collection of all $A$-events. If $\Delta$ is a singleton $\{\af\}$, we write $A^{\af}$ and $(\af;A)$ rather than $A^{\{\af\}}$ and $(\{\af\},A)$ respectively. Clearly, $\mathbb{E}_{A}$ is isomorphic to the power set of $\sigma(A)$ --- denoted $\mathscr{P}(\sigma(A))$, as discussed at the end of the introduction---, and this isomorphism canonically induces a Boolean algebra structure in $\mathbb{E}_{A}$. We have $A^{\Delta} \vee A^{\Sigma} = A^{\Delta \cup \Sigma}$, and so on. We denote by $\mathbb{E}$ the collection of all events of the system, i.e., $\mathbb{E} \doteq \cup_{A \in \mathcal{O}} \mathbb{E}_{A}$, and any state $\rho$ induces a mapping $P_{\rho}: \mathbb{E} \rightarrow [0,1]$ given by $P_{\rho}(A^{\Delta}) \doteq P_{\rho}^{A}(\Delta)$ (see definition~\ref{def: system}). By construction, $\mathbb{E}_{A} \cap \mathbb{E}_{B} = \emptyset$ whenever $A \neq B$, so $\mathbb{E}$ is defined as a union of pairwise disjoint Boolean algebras. It also follows by construction that, for any observable $A$, $P_{\rho}^{A}$ is the composition $P_{\rho} \circ \iota_{A}$, where $\iota_{A}$ is the embedding $\mathscr{P}(\sigma(A)) \ni \Delta \mapsto A^{\Delta} \in \mathbb{E}$. For the sake of illustration:
\begin{center}
    \begin{tikzcd}
        \mathscr{P}(\sigma(A)) \arrow{r}{\iota_{A}} \arrow{dr}[swap]{P_{\rho}^{A}} & \mathbb{E}\arrow{d}{P_{\rho}}
        \\
        & A
    \end{tikzcd}
\end{center}
Let's emphasize this definition.

\begin{definition}[Observable events and probability assignments]\label{def: eventsAndAssignments} Let $\mathfrak{S}$ be a system. Given any observable $A$ of $\mathfrak{S}$, we denote by $\mathbb{E}_{A}$ the set of all $A$-events, i.e., $\mathbb{E}_{A} \doteq \{A^{\Delta} \equiv (\Delta;A): \Delta \subset \sigma(A)\}$, and the set of all observable events of the system is denoted by $\mathbb{E}$, that is, $\mathbb{E}\doteq\cup_{A \in \mathcal{O}}\mathbb{E}_{A}$. For each state $\rho$, we associate a function $P_{\rho}: \mathbb{E} \rightarrow [0,1]$ such that, for any $A^{\Delta} \in \mathbb{E}$,
\begin{align}
    P_{\rho}(A^{\Delta}) \doteq P_{\rho}^{A}(\Delta),
\end{align}
where $P_{\rho}^{A}$ denotes the probability measure on $\sigma(A)$ introduced in definition \ref{def: system}. Finally, if $\Delta\subset \sigma(A)$ is a singleton $\{\af\}$, we write $A^{\af}$ and $(\af;A)$ instead of $A^{\{\af\}}$ and $(\{\af\};A)$, for simplicity. 
\end{definition}

For all practical purposes, $P_{\rho}$ is the collection $\{P_{\rho}^{A}: A \in \mathcal{O}\}$ of probability measures, and, for each observable $A$, $P_{\rho}^{A} \equiv P_{\rho}( \ \cdot \ ,A)$ can be seen as the ``component $A$'' of $P_{\rho}$. Similarly, $P_{\rho}$ is equivalent to the mapping $A \mapsto P_{\rho}^{A}$, which assigns  for each observable $A$ a probability measure on $\mathscr{P}(\sigma(A))$. It immediately follows from postulate~\ref{ax: separability} that, for any pair of states $\rho,\rho'$,
\begin{align}
    \rho = \rho'\Leftrightarrow\label{eq: stateAndP} P_{\rho} = P_{\rho'},
\end{align}
so a state is uniquely defined by the probabilities it assigns to observable events. Put differently, the state $\rho$ is, for all practical purposes, the mapping $P_{\rho}$, and for this reason we will sometimes say that $P_{\rho}$ is a state.  

When a measurement of an observable $A$ is performed, a single outcome $\af \in \sigma(A)$ is obtained. We thus distinguish between the objective event $A^{\af}$ that happens when $A$ is measured, and the subjective event $A^{\Delta}$ corresponding to the information the experimentalist extracts from this procedure. Let's emphasize these definitions.

\begin{definition}[Subjective and objective events]\label{def: events} An observable event in a system $\mathfrak{S}$ consists of a pair $A^{\Delta} \equiv (\Delta;A)$, where $A$ is an observable of $\mathfrak{S}$ and $\Delta$ is a subset of the spectrum $\sigma(A)$ of $A$. An observable event $(\Delta;A)$ is said to be objective if  $\Delta=\{\af\}$ for some $\af \in \sigma(A)$, and it is said to be subjective otherwise.
\end{definition}

We will assume that an objective event $A^{\Delta}$ satisfies $\Delta \neq \oldemptyset$, unless explicitly stated differently, because the subjective event $A^{\oldemptyset}$, despite being mathematically important, has no physical meaning. We say that an event $A^{\Delta}$ ``has happened'' or ``has occurred'' to indicate that $A$ has been measured and that the only information acquired by the experimentalist from this procedure is that some value lying in $\Delta$ was obtained. We say that $P_{\rho}^{A}(\Delta) \equiv P_{\rho}(A^{\Delta})$ is the probability of this event for a system in the state $\rho$, and $T_{(\Delta;A)}(\rho)$ is the state of the system immediately after its occurrence (see definition~\ref{def: system}). The fact that the state is updated by the subjective event implies that some degree of subjectivity must be allowed for states. We will discuss it in more detail in section~\ref{sec: categoryOfObservables}.

Events associated with different observables are, by construction, different --- they are different ordered pairs ---, and $\mathbb{E}$ does not establish any connection between them, since $\mathbb{E}$ is simply a union of pairwise disjoint sets. However, despite being necessarily different as ordered pairs, events associated with different observables can be statistically equivalent, i.e., they can be equally probable with respect to all states. It suggests the following definition.

\begin{definition}[Statistical equivalence]\label{def: statisticalEquivalence} Two events $A^{\Delta}$, $B^{\Sigma}$ are said to be statistically equivalent if $P_{\rho}(A^{\Delta}) = P_{\rho}(B^{\Sigma})$ for each state $\rho$.
\end{definition}

Statistical equivalence canonically induces an equivalence relation $\sim_{\mathcal{S}}$ in $\mathbb{E}$, so we naturally have equivalence classes of events. These classes are constructed in the spirit postulate~\ref{ax: separability} --- we group together all events that cannot be distinguished by states ---, so the coset $\mathbb{E}_{\mathcal{S}} \equiv \mathbb{E}/_{\sim_{\mathcal{S}}}$ is certainly important to our discussion.  $\mathbb{E}_{\mathcal{S}}$ connects the non-intersecting Boolean algebras that construct $\mathbb{E}$, enabling us to go from one algebra to another, and we will see later that these Boolean algebras induce a structure of partial Boolean algebra in $\mathbb{E}_{\mathcal{S}}$, which in turn can be embedded into the set of observables $\mathcal{O}$. Observables associated with equivalent events will be called projections, as in quantum mechanics \cite{kochen1967problem}.

Now let's take a look at the state update. To begin with, let's explain the notation we will use throughout the paper:
\begin{definition}[Updated state]\label{def: updatedState} Given a state $\rho$ and an event $A^{\Delta}$, we denote by $P_{\rho}( \ \cdot \ \vert A^{\Delta})$ the mapping $P_{T_{(\Delta;A)}(\rho)}: \mathbb{E} \rightarrow [0,1]$. Similarly, the component $B \in \mathcal{O}$ of $P_{T_{(\Delta;A)}(\rho)}$, namely the probability measure $P_{T_{(\Delta;A)}(\rho)}^{B} \equiv P_{T_{(\Delta;A)}(\rho)}( \ \cdot \ ;B)$, is denoted by $P_{\rho}(B^{( \ \cdot \ )}\vert A^{\Delta})$.    
\end{definition}

Clearly, $P_{\rho}( \ \cdot \ \vert A^{\Delta})$ (or equivalently $T_{(\Delta;A)}(\rho)$) is the state of the system when all we know about it is both that the event $A^{\Delta}$ has occurred and that its state immediately before this event was $\rho$. It is easy to justify that, for any component $B \in \mathcal{O}$ of $P_{\rho}( \ \cdot \ \vert A^{\Delta})$ and any $\Sigma \subset \sigma(B)$, we must have
\begin{align}
    P_{\rho}(B^{\Sigma}\vert A^{\Delta}) &=\label{eq: subjective} \sum_{\af \in \Delta} P_{\rho}^{A}(\{\af\} \vert \Delta)P_{\rho}(B^{\Sigma}\vert A^{\af}),
\end{align}
where $P_{\rho}^{A}(\{\af\} \vert \Delta) \doteq  \frac{P_{\rho}^{A}(\{\af\} \cap \Delta)}{P_{\rho}^{A}(\Delta)}$.  In fact, subjective events are, well, subjective, so they are not subject to the physical laws governing the system in question. There is absolutely no distinction between a subjective event in a system $\mathfrak{S}$ and any other probabilistic event we consider in our daily lives, such as a dice roll. The ``logic'' of subjective events is determined by our metalanguage, where  probability theory gives the rules. If all the information the experimentalist has is that, in a measurement of $A$, some outcome lying in $\Delta$ was obtained, then she knows that the state $\rho$ of the system has been updated to $T_{(\af;A)}(\rho)$ for one, and only one, $\af \in \Delta$. She also knows that the probability of occurrence of $A^{\af}$ (and consequently of the update $T_{(\af;A)}(\rho)$) for a system in the state $\rho$ is $P_{\rho}^{A}(\{\af\})$. Reasoning in statistical terms (and using  the probability space $(\sigma(A),\mathscr{P}(\sigma(A)),P_{\rho})$ to do it), she concludes that, under the evidence that $A^{\Delta}$ has happened, the probability that $A^{\af}$ has occurred (and consequently that the updated state is $T_{(\af;A)}(\rho)$) is given by the marginal probability $P_{\rho}^{A}(\{\af\} \vert \Delta) \doteq \frac{P_{\rho}(\{\af\} \cap \Delta)}{P_{\rho}^{A}(\Delta)}$. Hence, being in the state $T_{(\Delta;A)}(\rho)$ is equivalent to being in the state $T_{(\af;A)}(\rho)$ with probability $P_{\rho}^{A}(\{\af\} \vert \Delta)$, and therefore, for a system in the state $T_{(\Delta;A)}(\rho)$, the probability that a measurement of $B$ returns some outcome lying in $\Sigma \subset \sigma(B)$ must be given by equation~\ref{eq: subjective}. This equation has to be valid for each event $B^{\Sigma}$, which means that the mapping $P_{\rho}( \ \cdot \ \vert A^{\Delta}): \mathbb{E} \rightarrow [0,1]$ has to be a convex combination of the mappings $P_{\rho}( \ \cdot \ \vert A^{\af}): \mathbb{E} \rightarrow [0,1]$,  $\af \in \Delta$, i.e., 
\begin{align}
    P_{\rho}( \ \cdot \ \vert  A^{\Delta}) &= \sum_{\af \in \Delta} P_{\rho}^{A}( \{\af\}  \vert \Delta)P_{\rho}(\ \cdot \ \vert A^{\af}).
\end{align}
We know that two states $\rho$, $\rho'$ are equal if and only if  $P_{\rho} = P_{\rho'}$, so it is convenient to internalize this equation, i.e., to express it as a correspondence between states. To do so, we need to assume that $\mathcal{S}$ is at least a convex set with zero element. Furthermore, to connect convex combinations in $\mathcal{S}$ with convex combinations in the set of functions $\mathbb{E} \rightarrow [0,1]$ we need to require that the assignment $\rho \mapsto P_{\rho}$ is convex. Note that these requirements are not constraints to simplify the definition of system; equation~\ref{eq: subjective} follows from the very idea of subjective event --- it must be valid in $\mathfrak{S}$ for the same reason it is valid in a roll dice ---, and the other requirements are made for theoretical convenience. 

The validity of equation~\ref{eq: subjective} cannot depend on how many observables we are going to measure after preparing the state of the system. That is, if, instead of a single event $B^{\Sigma}$, we are evaluating the probability of obtaining a sequence $(B_{i}^{\Sigma_{i}})_{i=1}^{m}$ of observable events in a sequential measurement of $B_{1},\dots, B_{m}$ respectively, we must have
\begin{align}
    P_{\rho}((B_{i}^{\Sigma_{i}})_{i=1}^{m}\vert A^{\Delta}) &=\label{eq: subjectiveSequential} \sum_{\af \in \Delta} P_{\rho}^{A}(\{\af\} \vert \Delta)P_{\rho}((B_{i}^{\Sigma_{i}})_{i=1}^{m}\vert A^{\af})
\end{align}
for exactly the same reason why equation~\ref{eq: subjective} has to be satisfied. Note that the argument we presented to justify equation~\ref{eq: subjective} has nothing to do with the event $B^{\Sigma}$, only with the subjective update $T_{\Delta,A}(\rho)$, and this is why we can apply the same reasoning here. So, to continue our discussion, let's make the idea of a sequence of observable events precise.

\begin{definition}[Sequential observable event]\label{def: sequentialEvent} A sequential observable event consists of  a sequence $(A_{i}^{\Delta_{i}})_{i=1}^{m}$ of observable events. For each sequence $(A_{i}^{\Delta_{i}})_{i=1}^{m}$ we associate the mapping
\begin{align}
    T_{(\underline{\Delta};\underline{A})} \equiv T_{(\Delta_{1}, \dots ,\Delta_{m}; A_{1},\dots, A_{m})} \doteq T_{(\Delta_{m};A_{m})} \circ \dots \circ T_{(\Delta_{1};A_{1})},
\end{align}
which encodes state update determined by $(A_{i}^{\Delta_{i}})_{i=1}^{m}$. Given a state $\rho$ and a sequence $(A_{i}^{\Delta_{i}})_{i=1}^{m}$, we denote by $P_{\rho}( \ \cdot \ \vert (A_{i}^{\Delta_{i}})_{i=1}^{m})$ the mapping $\mathbb{E} \rightarrow [0,1]$ associated with the state $ T_{(\Delta_{1}, \dots ,\Delta_{m}; A_{1},\dots, A_{m})}(\rho)$, i.e.,
\begin{align}
    P_{\rho}( \ \cdot \ \vert (A_{i}^{\Delta_{i}})_{i=1}^{m}) \equiv P_{T_{(\Delta_{1}, \dots ,\Delta_{m}; A_{1},\dots, A_{m})}(\rho)}( \ \cdot \ ).
\end{align}
\end{definition}
Note that, if $m=1$, definition~\ref{def: sequentialEvent} reduces to our previous definition. 

Consider a sequence of three events  $(A_{1}^{\Delta_{1}},A_{2}^{\Delta_{2}},A_{3}^{\Delta_{3}})$. For a system in the state $\rho$, the probability of occurrence of $A_{1}^{\Delta_{1}}$, i.e., the probability that a measurement of $A_{1}$ returns some outcome lying in $\Delta_{1}$, is $P_{\rho}(A_{1}^{\Delta_{1}})$, and the occurrence of this event leads us to the state $\rho_{1} \equiv T_{(\Delta_{1};A_{1})}(\rho)$. Similarly, for a system in the state $\rho_{1}$, the probability of occurrence of $A_{2}^{\Delta_{2}}$ is $P_{\rho_{1}}(A_{2}^{\Delta_{2}})=P_{\rho}(A_{2}^{\Delta_{2}} \vert A_{1}^{\Delta_{1}})$, and this event updates the state to $\rho_{2} \equiv T_{(\Delta_{2},A_{2})}(\rho_{1}) = T_{(\Delta_{1},\Delta_{2}; A_{1},A_{2})}(\rho)$. Finally, the probability assigned to $A_{3}^{\Delta_{3}}$ by the state $\rho_{2}$ is $P_{\rho_{2}}(A_{3}^{\Delta_{3}}) = P_{\rho}(A_{3}^{\Delta_{3}} \vert (A_{k}^{\Delta_{k}})_{k=1}^{2})$. According to definition \ref{def: system}, the probability of an event is completely determined by the state of the system immediately before its occurrence, so the probability of obtaining the sequence of events $(A_{i}^{\Delta_{i}})_{i=1}^{3}$ in a sequential measurement of $A_{1},A_{2},A_{3}$ respectively must be defined as
\begin{align*}
    P_{\rho}((A_{i}^{\Delta_{i}})_{i=1}^{3}) &\doteq \prod_{i=1}^{3} P_{\rho}(A_{i}^{\Delta_{i}} \vert (A_{k}^{\Delta_{k}})_{k=1}^{i-1})
    \\
    &= P_{\rho}(A_{1}^{\Delta_{1}})P_{\rho}(A_{2}^{\Delta_{2}} \vert A_{1}^{\Delta_{1}})P_{\rho}(A_{3}^{\Delta_{3}} \vert (A_{k}^{\Delta_{k}})_{k=1}^{2}),
\end{align*}
where $P_{\rho}(A_{i}^{\Delta_{i}} \vert (A_{k}^{\Delta_{k}})_{k=1}^{i-1}) \equiv P_{\rho}(A_{i}^{\Delta_{i}})$ if $i=1$. The general definition goes as follows.

\begin{definition}[Probability of a sequential event]\label{def: sequentialProbability} Let $(A_{i}^{\Delta_{i}})_{i=1}^{m}$ be a sequence of observable events. Given any state $\rho$, we define, with a slight abuse of notation,
\begin{align}
    P_{\rho}((A_{i}^{\Delta_{i}})_{i=1}^{m}) &\doteq \prod_{i=1}^{m} P_{\rho}(A_{i}^{\Delta_{i}} \vert (A_{k}^{\Delta_{k}})_{k=1}^{i-1}).
\end{align}
For a system in the state $\rho$, this number corresponds to the probability of obtaining the sequence of events $(A_{i}^{\Delta_{i}})_{i=1}^{m}$ in a sequential measurement of $A_{1},\dots,A_{m}$ respectively.
\end{definition}

We can finally make equation~\ref{eq: subjectiveSequential} precise:
\begin{postulate}[Subjective update]\label{ax: subjectiveUpdate} $\mathcal{S}$ is a convex set, and the mapping $\mathcal{S} \ni \rho \mapsto P_{\rho}$ is convex. Furthermore, for any event $A^{\Delta}$, any state $\rho$ and any sequence of events $(B_{i}^{\Sigma_{i}})_{i=1}^{m}$,
\begin{align}
    P_{\rho}((B_{i}^{\Sigma_{i}})_{i=1}^{m}\vert A^{\Delta}) &=\label{eq: subjectiveSequentialPostulate} \sum_{\af \in \Delta} P_{\rho}^{A}(\{\af\} \vert \Delta)P_{\rho}((B_{i}^{\Sigma_{i}})_{i=1}^{m}\vert A^{\af})
\end{align}
where $P_{\rho}^{A}(\{\af\}\vert \Delta) =  \frac{P_{\rho}^{A}(\{\af\} \cap \Delta)}{P_{\rho}^{A}(\Delta)}$.
\end{postulate}

As we discussed, this postulate does not aim to simplify the notion of system, nor is it one of the key features of quantum systems. Subjective events are subject only to our metalanguage, not to the laws governing the hypothetical physical system we are trying to describe, and equation~\ref{eq: subjectiveSequentialPostulate} reflects the simple fact that we reason in statistical terms. The other constraints are imposed  for theoretical convenience: the set of mappings $\mathbb{E} \rightarrow [0,1]$ has an algebraic structure that enables us to establish relations like equation~\ref{eq: subjective}, and by assuming that $\mathcal{S}$ and $P$ are convex we are naturally incorporating these relations into $\mathcal{S}$, as we do in proposition~\ref{prop: subjectiveUpdate}.

As we said, we internalize equation~\ref{eq: subjective} as follows:

\begin{proposition}[Subjective Update]\label{prop: subjectiveUpdate} Let  $T_{(\Delta;A)}$ be the update determined by an event $(\Delta;A)$. Then, for any state $\rho$,
\begin{align}
    T_{(\Delta;A)}(\rho) &=\label{eq: subjectiveUpdateProp} \sum_{\af \in \Delta}  P_{\rho}^{A}(\{\af\}\vert \Delta) T_{(\af;A)}(\rho),
\end{align}
or equivalently
\begin{align}
    P_{\rho}( \ \cdot \ \vert  A^{\Delta}) &= \sum_{\af \in \Delta} P_{\rho}^{A}( \{\af\}  \vert \Delta)P_{\rho}(\ \cdot \ \vert A^{\af}),
\end{align}
where $P_{\rho}^{A}(\{\af\}\vert \Delta) =  \frac{P_{\rho}^{A}(\{\af\} \cap \Delta)}{P_{\rho}^{A}(\Delta)}$.
\end{proposition}
\begin{proof}
According to postulate~\ref{ax: subjectiveUpdate}, the set of states $\mathcal{S}$ is convex, so the state in the right-hand side of equation~\ref{eq: subjectiveUpdateProp} is well defined. Denote by $\rho_{\Delta}$ this state. According to the same postulate, the function $\rho \mapsto P_{\rho}$ is  convex, therefore
\begin{align*}
    P_{\rho_{\Delta}} = \sum_{\af \in \Delta}  P_{\rho}^{A}(\{\af\}\vert \Delta) P_{T_{(\af;A)}(\rho)}.
\end{align*}
Postulate~\ref{ax: subjectiveUpdate} now implies that, for any event $B^{\Sigma}$, $P_{\rho_{\Delta}}(B^{\Sigma})= P_{T_{(\Delta;A)}(\rho)}(B^{\Sigma})$, which in turn is satisfied if and only if $\rho_{\Delta}=T_{(\Delta;A)}(\rho)$. It completes the proof.
\end{proof}

For the same reason why equation~\ref{eq: subjectiveSequentialPostulate} must be satisfied, definition~\ref{def: sequentialProbability} should induce a probability measure in the power set of $\sigma(A_{1}) \times \dots \times \sigma(A_{m})$. In fact, the conditions defining $\sigma$-algebras and probability measures, namely the Kolmogorov axioms \cite{klenke2020probability}, are precisely the conditions we need to ensure that our (sequential) subjective events behave as they should do: as psychological entities obeying probability theory. It follows from postulate~\ref{ax: subjectiveUpdate} that this is the case:

\begin{proposition}[Sequential measure]\label{prop: sequentialMeasure} Let $A_{1},\dots,A_{m}$ be observables. Given a state $\rho$, define a mapping $p_{\rho}( \ \cdot \ ;A_{1},\dots,A_{m}): \prod_{i=1}^{m}\sigma(A_{i}) \rightarrow [0,1]$ by
\begin{align}
    p_{\rho}(\af_{1},\dots,\af_{m};A_{1},\dots A_{m}) \doteq\label{eq: definingSequentialDistribution} P_{\rho}((A_{i}^{\af_{i}})_{i=1}^{m}),
\end{align}
where $(\af_{1},\dots,\af_{m})$ is any element of $\prod_{i=1}^{m}\sigma(A_{i})$ and the right hand side of equation is given by definition~\ref{def: sequentialProbability}. Then $p_{\rho}( \ \cdot \ ;A_{1},\dots,A_{m})$ is a probability distribution, and the probability measure $P_{\rho}( \ \cdot \ ;A_{1},\dots,A_{m})$ that it induces in the power set of $\prod_{i=1}^{m}\sigma(A_{i})$ satisfies, for any $\Delta_{1} \times \dots \times \Delta_{m} \subset \prod_{i=1}^{m}\sigma(A_{i})$,
\begin{align}
    P_{\rho}(\Delta_{1} \times \dots \times \Delta_{m};A_{1},\dots A_{m}) =\label{eq: definingSequentialMeasure} P_{\rho}((A_{i}^{\Delta_{i}})_{i=1}^{m}).
\end{align}
\end{proposition}
\begin{proof}
    Let $A_{1},\dots,A_{m}$ be observables, and let $\mathfrak{B} \subset \prod_{i=1}^{m} \sigma(A_{i})$ be the set of all boxes in $\prod_{i=1}^{m} \sigma(A_{i})$, that is, $\Sigma \in \mathfrak{B}$ iff $\Sigma = \Delta_{1}\times \dots \Delta_{m}$ for some sequence $\Delta_{i} \subset \sigma(A_{i})$, $i=1,\dots,m$. Given any state $\rho$, define a function $P_{\rho}( \ \cdot \ ;A_{1},\dots,A_{m}): \mathfrak{B} \rightarrow [0,1]$ via equation~\ref{eq: definingSequentialMeasure}. This function extends the mapping $p_{\rho}( \ \cdot \ ;A_{1},\dots,A_{m})$ defined in the statement of the proposition. Now fix a one-dimensional (or zero-dimensional) box $\prod_{i=1}^{m}\Delta_{i}$ in $\prod_{i=1}^{m}\sigma(A_{i})$, i.e., fix some $k \in \{1,\dots,m\}$ and set $\Delta_{i}$ as a singleton $\{\af_{i}\}$ whenever $i \neq k$, whereas $\Delta_{k}$ is any nonempty set. For simplicity, write  $\underline{\af} \equiv (\af_{1},\dots,\af_{m})$ and $\underline{A} \equiv (A_{1},\dots,A_{m})$. Also, denote by $\rho_{k-1}$ the state $(T_{(\af_{k-1},A_{k-1})} \circ \dots \circ T_{(\af_{1},A_{1})})(\rho)$. Then
    \begin{align*}
        \sum_{\underline{\af} \in \underline{\Delta}} p_{\rho}(\underline{\af};\underline{A}) &= \sum_{\af_{k} \in \Delta_{k}} p_{\rho}(\af_{1},\dots,\af_{m};A_{1},\dots,A_{m})
        \\
        &= \sum_{\af_{k} \in \Delta_{k}} p_{\rho}((A_{i}^{\af_{i}})_{i=1}^{k-1})P_{\rho_{k-1}}(A_{k}^{\af_{k}})P_{T_{(\af_{k},A_{k})}(\rho)}((A_{j}^{\af_{j}})_{j=k+1}^{m})
        \\
        &=  p_{\rho}((A_{i}^{\af_{i}})_{i=1}^{k-1}) P_{\rho_{k-1}}(A_{k}^{\Delta_{k}})\sum_{\af_{k} \in \Delta_{k}}P_{\rho_{k-1}}^{A_{k}}(\{\af_{k}\}\vert \Delta_{k})P_{T_{(\af_{k},A_{k})}(\rho)}((A_{j}^{\af_{j}})_{j=k+1}^{m})
        \\
        &=  p_{\rho}((A_{i}^{\af_{i}})_{i=1}^{k-1}) P_{\rho_{k-1}}(A_{k}^{\Delta_{k}})P_{T_{(\Delta_{k},A_{k})}(\rho)}((A_{j}^{\af_{j}})_{j=k+1}^{m})
        \\
        &= P_{\rho}(\{\af_{1}\}\times \dots\times\Delta_{k}\times \dots \times \{\af_{m}\};A_{1},\dots,A_{m}),
        \\
        &= P_{\rho}(\Delta_{1} \times \dots \times \Delta_{m};\underline{A}).
    \end{align*}
    Similarly, one can show that, given any box $\Delta_{1} \times \dots \times \Delta_{m} \subset \prod_{i=1}^{m}\sigma(A_{i})$, 
    \begin{align*}
        \sum_{\underline{\af} \in \underline{\Delta}} p_{\rho}(\af_{1},\dots,\af_{m};A_{1},\dots,A_{m}) =\label{eq: inducedMeasureInBox} P_{\rho}(\Delta_{1} \times \dots \times \Delta_{m};A_{1},\dots,A_{m}).
    \end{align*}
    If $\Delta_{1}\times\dots\times\Delta_{m} = \sigma(A_{1}) \times \dots \times \sigma(A_{m})$, the right hand side of equation~\ref{eq: inducedMeasureInBox} is equals one, therefore $p_{\rho}( \ \cdot \ ;A_{1},\dots,A_{m})$ is a probability distribution. Furthermore, the same equation says that $P_{\rho}( \ \cdot \ ;A_{1},\dots,A_{m})$ coincides with the measure induced by $p_{\rho}( \ \cdot \ ;A_{1},\dots,A_{m})$ in all boxes, which concludes the proof.
\end{proof}

We can now introduce the following definition.
\begin{definition}[Sequential probability measure]\label{def: sequentialMeasure} Let $\underline{A} \equiv (A_{1}, \dots, A_{m})$ be a sequence of observables. Given any state $\rho$, we denote by $P_{\rho}( \ \cdot \ ;\underline{A})  \equiv P_{\rho}( \ \cdot \ ;A_{1},\dots ,A_{m})$, or $P_{\rho}^{\underline{A}} \equiv P_{\rho}^{(A_{1},\dots,A_{m})}$, the (necessarily unique) probability measure in the power set of $\sigma(\underline{A}) \equiv \prod_{i=1}^{m}\sigma(A_{i})$ satisfying
\begin{align}
    P_{\rho}(\underline{\Delta};\underline{A}) \equiv P_{\rho}(\Delta_{1} \times \dots \times \Delta_{m}; A_{1}, \dots A_{m}) =  \prod_{i=1}^{m}P_{\rho}(A_{i}^{\Delta_{i}} \vert (A_{k}^{\Delta_{k}})_{i=1}^{i-1})
\end{align}
for any $\underline{\Delta} \equiv \Delta_{1} \times \dots \times \Delta_{m} \subset \prod_{i=1}^{m}\sigma(A_{i})$. Similarly, we denote by $p_{\rho}( \ \cdot \ ;\underline{A})  \equiv p_{\rho}( \ \cdot \ ;A_{1},\dots ,A_{m})$ (or $p_{\rho}^{\underline{A}} \equiv p_{\rho}^{(A_{1},\dots,A_{m})}$) the probability distribution associated with this measure.
\end{definition}

The following generalization of proposition~\ref{prop: subjectiveUpdate} is important.

\begin{proposition}[Sequential subjective update]\label{prop: subjectiveSequentialUpdate} Let $\underline{{A}} \equiv (A_{1},\dots,A_{m})$ be a sequence of observables, and fix some $\underline{\Delta} \equiv \Delta_{1}\times \dots\times\Delta_{m} \subset \sigma(A_{1}) \times \dots  \times \sigma(A_{m}) \equiv \sigma(\underline{A})$. Then, for any state $\rho$,
\begin{align}
    T_{(\underline{\Delta};\underline{A})}(\rho)&=\label{eq: subjectiveSequentialUpdate}\sum_{\underline{\af} \in \underline{\Delta}} P_{\rho}^{\underline{A}}(\{\underline{\af}\} \vert \underline{\Delta}) T_{(\underline{\af},\underline{A})}(\rho),
\end{align}
where $P_{\rho}^{\underline{A}}(\{\underline{\af}\} \vert \underline{\Delta})= \frac{P_{\rho}^{\underline{A}}(\{\underline{\af} \} \cap \underline{\Delta})}{P_{\rho}^{\underline{A}}( \underline{\Delta})}$, and where $T_{(\underline{\Delta};\underline{A})} \equiv T_{(\Delta_{m};A_{m})} \circ \dots T_{(\Delta_{1};A_{1})}$ (see definition \ref{def: sequentialEvent}).
\end{proposition}
\begin{proof}
    Let $B^{\Sigma}$ be any event. Given any state $\rho$, write $P_{\rho}(\underline{\Delta} \times \Sigma;\underline{A},B) \equiv P_{\rho}(\Delta_{1} \times \dots \times \Delta_{m} \times \Sigma;A_{1},\dots,A_{m},B)$. Then, according to definition~\ref{def: sequentialMeasure} and proposition~\ref{prop: sequentialMeasure},
    \begin{align*}
        P_{T_{(\underline{\Delta};\underline{A})}(\rho)}(B^{\Sigma}) &=  \frac{P_{\rho}(\underline{\Delta} \times \Sigma; \underline{A}, B)}{P_{\rho}^{\underline{A}}(\underline{\Delta})} = \frac{1}{P_{\rho}^{\underline{A}}(\underline{\Delta})}\sum_{\underline{\af} \in \underline{\Delta}}P_{\rho}(\{\underline{\af}\} \times \Sigma; \underline{A}, B)
        \\
        &= \frac{1}{P_{\rho}^{\underline{A}}(\underline{\Delta})}\sum_{\underline{\af} \in \underline{\Delta}}P_{\rho}^{\underline{A}}(\{\underline{\af}\}) P_{T_{(\underline{\af},\underline{A})(\rho)}}(B^{\Sigma})
        \\
        &=\sum_{\underline{\af} \in \underline{\Delta}}P_{\rho}^{\underline{A}}(\{\underline{\af}\}\vert \underline{\Delta}) P_{T_{(\underline{\af},\underline{A})(\rho)}}(B^{\Sigma})
    \end{align*}
    This is valid for any event $B^{\Sigma}$, therefore equation~\ref{eq: subjectiveSequentialUpdate} is satisfied.
\end{proof}

Proposition~\ref{prop: subjectiveUpdate} tells us that the update caused by a subjective event is completely determined by the update due to its corresponding objective events, so we can focus on objective updates from now on.  Let $\rho$ be any state, $A^{\af}$ be any objective event, and consider the state $P_{\rho}( \ \cdot \ \vert A^{\af})$. Let's analyze the component $A$ of $P_{\rho}( \ \cdot \ \vert A^{\af})$. Recall that $(\sigma(A),\mathscr{P}(\sigma(A)),P_{\rho})$ is a probability space, and that each $A$-event $A^{\Delta}$ corresponds to an event $\Delta$ in this space.  As we have already discussed, according to probability theory, the probability of an event $\Delta \subset \sigma(A)$ occurring, under the evidence that another event $\Delta' \subset \sigma(A)$ has already occurred, is given by the conditional probability
\begin{align}
    P_{\rho}^{A}(\Delta \vert \Delta') =\label{eq: conditionalProbability} \frac{P_{\rho}^{A}(\Delta \cap \Delta')}{P_{\rho}^{A}(\Delta')}.
\end{align}
Hence, from a purely mathematical perspective, the most straightforward way of defining component $A$ of $P_{\rho}( \ \cdot \ \vert A^{\af})$ consists in defining it as the probability measure  $P_{\rho}^{A}( \ \cdot \  \vert \{\af\})$, that is, consists in requiring that, for any $\Delta \subset \sigma(A)$,
\begin{align}
    P_{\rho}(A^{\Delta} \vert A^{\af}) &=\label{eq: selfupdate} P_{\rho}^{A}(\Delta \vert \{\af\}) = \frac{P_{\rho}^{A}(\Delta \cap \{\af\})}{P_{\rho}^{A}(\{\af\})}.
\end{align}
Note that, together with proposition~\ref{prop: subjectiveUpdate}, equation~\ref{eq: selfupdate} implies that, for any pair $\Delta,\Delta' \subset \sigma(A)$,
\begin{align}
    P_{\rho}(A^{\Delta} \vert A^{\Delta'}) &=\label{eq: subjectiveSelfupdate}P_{\rho}^{A}(\Delta \vert \Delta') = \frac{P_{\rho}^{A}(\Delta \cap \Delta')}{P_{\rho}^{A}(\Delta')},
\end{align}
and consequently
\begin{align}
    P_{\rho}(\Delta' \times \Delta; A,A) &=\label{eq: sequentialSelf} P_{\rho}(A^{\Delta'}) \frac{P_{\rho}^{A}(\Delta \cap \Delta')}{P_{\rho}^{A}(\Delta')} = P_{\rho}^{A}(\Delta \cap \Delta') \equiv P_{\rho}(\Delta \cap\Delta';A).
\end{align}
Equation~\ref{eq: sequentialSelf} tells us that, if equation~\ref{eq: selfupdate} is valid, then the sequential event $(A^{\Delta'},A^{\Delta})$ and the event $A^{\Delta'} \wedge A^{\Delta} \equiv A^{\Delta' \cap \Delta}$ are equally probable with respect to all states. Recall that, according to probability theory \cite{klenke2020probability}, the sequence of events  $(A^{\Delta'},A^{\Delta})$, seen as a sequence of events in $(\sigma(A),\mathscr{P}(\sigma(A)),P_{\rho}^{A})$, updates $P_{\rho}^{A}$ in the same way $A^{\Delta} \wedge A^{\Delta'}$ does. In fact, the event $A^{\Delta'}$ updates $P_{\rho}^{A}$ to $P_{\rho_{\Delta'}}^{A} \equiv P_{\rho}^{A}( \ \cdot \ \vert \Delta')$, which in turn is lead by $A^{\Delta}$ to the marginal $P_{\rho_{\Delta'}}^{A}( \ \cdot \ \vert \Delta) \equiv P_{\rho_{(\Delta';\Delta)}}$. Therefore, for any $\Sigma \subset \sigma(A)$,
\begin{align*}
    P^{A}_{\rho_{(\Delta';\Delta)}}(\Sigma) &= \frac{P^{A}_{\rho_{\Delta'}}( \Sigma \cap \Delta )}{P_{\rho_{\Delta'}}^{A}(\Delta)} = \frac{P^{A}_{\rho}( \Sigma \cap \Delta \vert \Delta')}{P_{\rho}^{A}(\Delta \vert \Delta')} =  \frac{P^{A}_{\rho}(\Sigma \cap \Delta \cap \Delta')}{P^{A}_{\rho}(\Delta \cap \Delta')}
    \\
    &= P^{A}_{\rho}(\Sigma \vert \Delta \cap \Delta').
\end{align*}
Hence, If we accept equation~\ref{eq: sequentialSelf}, we should require that $T_{(\Delta;A)} \circ T_{(\Delta';A)} = T_{(\Delta \cap \Delta',A)}$. Just as proposition~\ref{prop: subjectiveUpdate} tells us that the subjective update $T_{(\Delta;A)}$ is determined by the objective updates $T_{(\af;A)}$, $\af \in \Delta$, proposition~\ref{prop: subjectiveSequentialUpdate} tells us that the sequential subjective update $T_{(\Delta',\Delta;A,A)}$ is determined by the sequential objective updates $T_{(\af',\af;A,A)}$, $(\af',\af) \in \Delta' \times \Delta$, thus, in order to obtain the equality $T_{(\Delta;A)} \circ T_{(\Delta';A)} = T_{(\Delta \cap \Delta',A)}$, it is sufficient to impose $T_{(\af;A)} \circ T_{(\af',A)} = T_{(\{\af\} \cap \{\af'\},A)}$, which in turn means that $T_{(\af;A)} \circ T_{(\af',A)} = \delta_{\af,\af'} T_{(\af;A)}$. Recall that we have a ``null state'' $0 \in \mathcal{S}$, and that, for any state $\rho$ and any observable $A$, $T_{(\oldemptyset,A)}(\rho)=0$, so by $T_{(\af;A)} \circ T_{(\af',A)} =0$ we mean that $T_{(\af;A)} \circ T_{(\af',A)} $ is the constant  function $\mathcal{S} \ni \rho \mapsto 0 \in \mathcal{S}$.

The physical meaning of equation~\ref{eq: selfupdate} is essentially that, whether or not some physically real disturbance is involved in the state update, this disturbance does not prevent us from  reasoning about sequential measurements of the same observable in statistical terms. To put it another way, it says that the influence that a measurement of an observable $A$ exerts on a future measurement of the same observable is purely informational, that is, the information  provided by a state $\rho$  about $A$, namely the probability measure $P_{\rho}( \ \cdot \ ;A)$, is updated according to the standard rule of marginal probability when $A$ is measured. This is the informational link mentioned in the introduction, but restricted to the particular case where the compatible observables are actually the same observable, and this is the first simplifying constraint that  $\mathfrak{S}$ must satisfy  to be a quantum system:

\begin{postulate}[Self-compatibility]\label{post: selfCompatibility} 
Let $A$ be any observable, and $\af$ be any element of its spectrum $\sigma(A)$. Then, given any state $\rho$, the component $A$ of $P_{\rho}( \ \cdot \ \vert A^{\af})$ is the marginal probability measure $P_{\rho}^{A}( \ \cdot \ \vert \{\af\})$, i.e., for any $\Delta \subset \sigma(A)$,
\begin{align}
    P_{\rho}(A^{\Delta}  \vert A^{\af}) = P_{\rho}^{A}(\Delta \vert \{\af\}) = \frac{P_{\rho}^{A}(\Delta \cap \{\af\})}{P_{\rho}^{A}(\Delta)}.
\end{align}
Furthermore, given any pair $\af,\af' \in \sigma(A)$, we have
\begin{align}
    T_{(\af;A)} \circ T_{(\af',A)} = \delta_{\af,\af'}T_{(\af;A)}.
\end{align}
\end{postulate}

As we mentioned in the introduction, this postulate is a particular case of postulate~\ref{post: compatibility}, so we could have proved it instead of postulating it. We pose it as a separate postulate just for the sake of argument.

\begin{lemma}[Self-compatibility]\label{lemma: selfCompatibility} Let $A$ be an observable and $\rho$ be a state. Then, for any pair $\Delta,\Delta' \subset \sigma(A)$,
\begin{align}
    P_{\rho}(A^{\Delta} \vert A^{\Delta'}) &=\label{eq: selfCompatibilitylemma}P_{\rho}^{A}(\Delta \vert \Delta'),
\end{align}
which in turn is equivalent to saying that
\begin{align}
    P_{\rho}(\Delta \times \Delta';A,A) =\label{eq: jointprobability} P_{\rho}(\Delta \cap \Delta';A).
\end{align}
Furthermore, for any pair of $A$-events $(\Delta;A)$, $(\Delta';A)$ we have $T_{(\Delta;A)} \circ T_{(\Delta';A)} = T_{(\Delta \cap \Delta';A)}$, which means that the following diagram commutes.
\begin{center}
    \begin{tikzcd}
        \mathcal{S} \arrow{dr}{T_{(\Delta \cap \Delta',A)}} \arrow{d}[swap]{T_{(\Delta';A)}}
        \\
        \mathcal{S} \arrow{r}[swap]{T_{(\Delta;A)}} & \mathcal{S}
    \end{tikzcd}
\end{center}
\end{lemma}
\begin{proof}
    We already proved equation~\ref{eq: selfCompatibilitylemma}, and showing that it is equivalent to equation~\ref{eq: jointprobability} is trivial. Finally, according to proposition~\ref{prop: subjectiveSequentialUpdate} and postulate~\ref{post: selfCompatibility}, for any state $\rho$ we have
    \begin{align*}
        (T_{(\Delta;A)} \circ T_{(\Delta';A)})(\rho) &= \sum_{\af \in \Delta}\sum_{\af' \in \Delta'} P_{\rho}^{(A,A)}(\{(\af,\af')\}\vert \Delta \times \Delta')(T_{(\af;A)} \circ T_{(\af',A)})(\rho)
        \\
        &= \sum_{\af \in \Delta \cap \Delta'} \frac{P_{\rho}^{(A,A)}(\{\af\} \times \{\af\})}{P_{\rho}^{(A,A)}(\Delta \times \Delta')}T_{(\af;A)}(\rho) =  \sum_{\af \in \Delta \cap \Delta'} \frac{P_{\rho}^{A}(\{\af\})}{P_{\rho}^{A}(\Delta \cap \Delta')}T_{(\af;A)}(\rho) 
        \\
        &= \sum_{\af \in \Delta}P^{A}_{\rho}(\{\af\}\vert \Delta \cap \Delta')T_{(\af;A)}(\rho) = T_{(\Delta \cap \Delta',A)}(\rho).
    \end{align*}
    This is true for any state $\rho$, therefore $T_{(\Delta;A)} \circ T_{(\Delta';A)} = T_{(\Delta \cap \Delta',A)}$, which completes the proof.
\end{proof}

If $\Delta=\{\af\}$ and $\Delta'=\{\af'\}$, $\Delta \cap \Delta = \{\af\}$ if $\af = \af'$ and $\Delta \cap \Delta =\oldemptyset$ otherwise, therefore:

\begin{corollary}[Repeatability of outcomes]\label{cor: repeatability} For any observable $A$, and any $\af,\af' \in \sigma(A)$ and any state $\rho$,
\begin{align}
    P_{\rho}(A^{\af}\vert A^{\af'}) &= \delta_{\af,\af'}.
\end{align}
\end{corollary}

In section~\ref{sec: basicFramework} we defined an eigenvalue of an observable $A$ as a real number $\af$ satisfying $P^{A}_{\rho}(\{\af\})=1$ for some state $\rho$, and we defined the point spectrum of $A$ as the collection of all its eigenvalues. It follows from the definition of spectrum (see section~\ref{sec: basicFramework}) that the point spectrum of an observable $A$ is included in its spectrum $\sigma(A)$. Now that we have postulate~\ref{post: selfCompatibility}, we can easily show that, in finite-dimensional systems, the point spectrum and the spectrum of $A$ coincide. In fact, let $A$ be any observable, and let $\af$ be any element of $\sigma(A)$. It follows from the definition of spectrum that there is a state $\rho$ satisfying $p_{\rho}(\af;A) \neq 0$, so let $\rho$ be such a state. Then, according to postulate~\ref{post: selfCompatibility}, we have
\begin{align*}
    P_{T_{(\af;A)}(\rho)}(A^{\af}) = P_{\rho}(A^{\af} \vert A^{\af}) = P_{\rho}^{A}(\{\af\}\vert \{\af\})= \frac{p_{\rho}(\af)}{p_{\rho}(\af)}=1.
\end{align*}
It proves the following lemma.
\begin{lemma}[Eigenvalues]\label{lemma: eigenvalues} Let $A$ be an observable in some finite-dimensional system. Then any element of its spectrum $\sigma(A)$ is an eigenvalue of $A$, i.e., $\af \in \sigma(A)$ if and only if $p_{\rho}(\af;A)=1$ for some state $\rho$.
\end{lemma}

\section{The category of observables and the completely mixed state}\label{sec: categoryOfObservables}
 As we mentioned in section~\ref{sec: observableEvents}, states must be subjective (or ``epistemic'' \cite{leifer2014real, spekkens2007toy}) to some extent because some of them result from subjective updates. For the sake of argument, we will imagine  that, at any given time, the physical system described by $\mathfrak{S}$ has a well-defined \textit{state of affairs}, and that $\rho \in \mathcal{S}$ describes the degree of knowledge of the experimentalist about it. Note that there is no contradiction between this assumption (which we make just for the sake of argument and that has no theoretical importance in our work) and the Kochen-Specker theorem (theorem~\ref{thm: kochenSpecker}) \cite{kochen1967problem}, since, as we mentioned in section~\ref{sec: basicFramework}, we do not commit ourselves to the realist view according to which observables represent properties  that are fully specified by the state of affairs of the system. The element in $\mathcal{S}$ representing the epistemological situation where the experimentalist has no information at all about the state of affairs --- or, for those who prefer not to talk about states of affairs, information about ``the system'' --- is called the \textbf{completely mixed state}, and it is denoted as $\emptyset$. As soon as we introduce functional relations we will define the completely mixed state properly --- all we have thus far is an idea and a symbol with no mathematical content. Saying that the state of $\mathfrak{S}$ is $\emptyset$ is equivalent to saying that the experimentalist has no information at all about the state of affairs, which in turn is equivalent to saying that this state of affairs can be anything. Finally, note that the general form of a state that the experimentalist is able to access by performing measurements on the system (which is the only process that definition~\ref{def: system} allows us to describe) is given by
\begin{align}
    \rho =\label{eq: accessibleStatesPre} (T_{(\Delta_{m};A_{m})} \circ \dots \circ T_{(\Delta_{1};A_{1})})(\emptyset)
\end{align}
for some sequence $(\Delta_{1};A_{1}),\dots,(\Delta_{m};A_{m})$ of observable events. We will assume that the set of states $\mathcal{S}$ (which, according to postulate~\ref{ax: subjectiveUpdate}, is convex) is the smallest convex set containing all these states. According to our epistemic approach to states, the completely mixed state must have a counterpart in $\mathcal{S}$, and it follows from the very definition of system (definition~\ref{def: system}) that, if $\emptyset\in\mathcal{S}$, then all states given by equation~\ref{eq: accessibleStatesPre} belong to $\mathcal{S}$. Hence, our assumption is essentially that the set of states $\mathcal{S}$ contains as few states as possible, by which we mean that it is the minimal convex set that includes the completely mixed state and all its possible updates. We introduce this simplifying condition, which is satisfied by quantum systems, in postulate~\ref{post: observables}.

Now, let's discuss functional relations between observables. As Kochen and Specker point out \cite{kochen1967problem}, functional relations are naturally defined in any physical system, and it is thanks to these relations that we can construct observables that are functions of other observables, as potential energy or ``position squared''. The definition of functional relation goes as follows.
\begin{definition}[Functional relation \cite{kochen1967problem}]\label{def: functionalRelation} Let $A$ and $B$ be observables in a system $\mathfrak{S}$, and let $f:  \sigma(A) \rightarrow \sigma(B)$ be a function. We say that $B$ is a function of $A$ via $f$, denoted $B=f(A)$, if $P_{\rho}( \ \cdot \ ;B)$ is the pushforward of $P_{\rho}( \ \cdot \ ;A)$ along $f$ for each state $\rho$. It means that, for any $\Sigma \subset \sigma(B)$,
    \begin{align}
        P_{\rho}(\Sigma; B) =\label{eq: KSdefinition} P_{\rho}(f^{-1}(\Sigma);A).
    \end{align}
\end{definition}

Recall that, for any observable $C$, $P_{\rho}( \ \cdot \ ;C)$ is just an alternative notation for $P^{C}_{\rho}$, so $P_{\rho}(\Delta;B) \equiv P^{B}_{\rho}(\Delta) \equiv P_{\rho}(B^{\Delta})$ (see definition \ref{def: eventsAndAssignments}). We say that there is a \textbf{functional relation} between two observables $A$ and $B$ if $B$ is a function of $A$ or \textit{vice-versa}. Also, we say that $B$ is a function of $A$ if $B$  is a function of $A$ via some function.

From the point of view of physics, it only makes sense to take definition \ref{def: functionalRelation} seriously if the function connecting two observables is unique, as in the particular cases of classical and quantum mechanics. The following lemma shows that, fortunately, this is the case. Note that we use postulate~\ref{post: selfCompatibility} to prove this result, which reinforces the plausibility of this requirement. 
\begin{lemma}[Functional relation]\label{lemma: functionalRelation} Let $A,B$ be observables in a system $\mathfrak{S}$, and suppose that $B$ is a function of $A$. Then the mapping $f:\sigma(A) \rightarrow \sigma(B)$ satisfying $P_{\rho}^{B} = P_{\rho}^{A} \circ f^{-1}$ for every state $\rho$ is unique and surjective.
\end{lemma}
\begin{proof}
Let's begin by showing that, if $B$ is a function of $A$ via $f$, then $f$ is surjective. According to equation~\ref{eq: preEigenvector}, for any $\beta \in \sigma(B)$ there is a state $\rho_{\beta}$ satisfying $P_{\rho_{\beta}}^{B}(\{\beta\}) \neq 0$. Thus, if $B=f(A)$, we have $0 \neq P^{A}_{\rho}(f^{-1}(\{\beta\}))$, which means that $f^{-1}(\{\beta\}) \neq \oldemptyset$, or equivalently $\beta \in f(A)$. This is true for any $\beta \in \sigma(B)$, therefore $\sigma(B)= f(\sigma(A))$. Now let $f,g: \sigma(A) \rightarrow \sigma(B)$ be functions satisfying, for every state $\rho$, $P_{\rho}^{A} \circ f^{-1} = P_{\rho}^{A} \circ g^{-1}$. Then, for any $\af \in \sigma(A)$ and any state $\rho$, we have $P_{T_{(\af;A)}(\rho)}^{A} \circ f^{-1} = P_{T_{(\af;A)}(\rho)}^{A} \circ g^{-1}$, which, according to lemma~\ref{lemma: selfCompatibility}, is satisfied if and only if, for every $\beta \in \sigma(B)$,
\begin{align}
    \frac{P_{\rho}^{A}(f^{-1}(\beta) \cap \{\af\})}{P_{\rho}^{A}(\{\af\})}=\label{eq: lemmaRelation} \frac{P_{\rho}^{A}(g^{-1}(\beta) \cap \{\af\})}{P_{\rho}^{A}(\{\af\})}.
\end{align}
For any $\af \in \sigma(A)$, let $\rho_{\af}$ be a state such that $P_{\rho}^{A}(\{\af\}) \neq 0$ --- equation~\ref{eq: preEigenvector} ensures that this state exists. Clearly, $P_{\rho_{\af}}^{A}(f^{-1}(\beta) \cap \{\af\}) \neq 0$ if and only if $\af \in f^{-1}(\beta)$, whereas $P_{\rho_{\af}}^{A}(g^{-1}(\beta) \cap \{\af\}) \neq 0$ if and only if $\af \in g^{-1}(\beta)$. Hence, equation~\ref{eq: lemmaRelation} implies that, for every $\beta \in \sigma(B)$, $\af \in f^{-1}(\beta)$ iff $\af \in g^{-1}(\beta)$, which in turn is equivalent to saying that $f(\af) = g(\af)$. This is true for every $\af \in \sigma(A)$, therefore $f=g$. 
\end{proof}

It is worth emphasizing the following corollary, which immediately follows from the fact that a function $f:\sigma(A) \rightarrow \sigma(B)$ satisfying $B=f(A)$ must be surjective (lemma~\ref{lemma: functionalRelation}).
\begin{corollary} If $B$ is a function of $A$, $\vert \sigma(B)\vert \leq \vert \sigma(A)\vert$.   
\end{corollary}
Let $B$ be a function of $A$ via $f$. We say that $B$ is a \textbf{coarse-graining} of $A$ if $f$ is non-injective, whereas $A$ is a \textbf{fine-graining} of $B$ if $B$ is a coarse-graining of $A$. Equivalently, $B=f(A)$ is a coarse graining of $A$ if $\vert \sigma(B)\vert < \vert \sigma(A)\vert$. It is worth to emphasize that, if $B=f(A)$, then, for any $\Delta \subset \sigma(B)$ and any state $\rho$,
\begin{align}
    P_{\rho}(\Delta; f(A)) =\label{eq: realKSdefinition} P_{\rho}(f^{-1}(\Delta);A).
\end{align}
It is important to note that, by construction, the spectral mapping theorem \cite{kadison1997fundamentalsI} holds in $\mathfrak{S}$. It means that, if $B=f(A)$, then the spectrum of $B$ is the set $f(\sigma(A)) \doteq \{f(\af): \af \in \sigma(A)\}$, i.e., 
\begin{align*}
    \sigma(f(A)) = f(\sigma(A)).
\end{align*}

The very definition of functional relation suggests that $\mathcal{O}$ must be closed under them. That is, if $A$ is an observable and $f$ is a real function on its spectrum, then there must be an observable $B \in \mathcal{O}$ satisfying $B=f(A)$. As Kochen and Specker say, given an observable $A$ and a function $f$, we define the observable $f(A)$ using equation~\ref{eq: realKSdefinition}, and one way of measuring $f(A)$ consists in measuring $A$ and evaluating $f$ in the resulting value \cite{kochen1967problem}. One who accepts this interpretation has no reason to suppose that $\mathcal{O}$ is not closed under functional relations. Furthermore, we know that this condition is satisfied by quantum systems. Hence, for theoretical convenience, we include this requirement in postulate~\ref{post: observables}, and we assume its validity from now on.

Let $n$ be the dimension of $\mathfrak{S}$ (see definition \ref{def: finiteSystem}). We say that an observable $A$ is \textbf{nondegenerate} if $\vert \sigma(A) \vert = n$, and $A$ is said to be \textbf{degenerate} otherwise. It immediately follows from definition~\ref{def: finiteSystem} that there is at least one nondegenerate observable in a system $\mathfrak{S}$. A nondegenerate observable $A$ cannot be a coarse-graining of any other observable of the system, since a fine-graining of $A$ would have more than $n$ outcomes, contradicting definition~\ref{def: finiteSystem}. Put differently, a nondegenerate observable $A$ cannot be refined: if $A = f(C)$ for some other observable $C$, then $f$ is injective. Nondegenerate observables are the most refined observables we can have, and, according to equation~\ref{eq: realKSdefinition}, they dictate how all their coarse-grainings behave: given any state $\rho$, if we know $P_{\rho}( \ \cdot \ ; A)$ then we also know $P_{\rho}( \ \cdot \ ;f(A))$ for any $f$ on $\sigma(A)$. Hence, any restriction upon nondegenerate observables will have a considerable impact on the system $\mathfrak{S}$. For a system $\mathfrak{S}$ to be a quantum system,  nondegenerate observables in $\mathfrak{S}$ must be as simple as a n-sided dice, by which we mean that, if we measure a nondegenerate observable many times without specifying the state of the system, the relative frequencies of all its possible values will be the same. Not specifying the state of the system is the same as being in the completely mixed state, so the condition is that, if $A$ is nondegenerate,  for any $\af\in \sigma(A)$ we have $p_{\emptyset}(\af;A) = \frac{1}{\vert \sigma(A)\vert} = \frac{1}{n}$, and consequently $P_{\emptyset}(\Delta;A) = \frac{\vert \Delta \vert}{n}$ for any $\Delta \subset \sigma(A)$. Equation~\ref{eq: KSdefinition} thus implies that, for any function $f(A)$ of $A$, where $A$ is nondegenerate, and any $\Delta \subset \sigma(f(A))$, we have $P_{\emptyset}(\Delta;f(A)) = \frac{\vert f^{-1}(\Delta) \vert }{n}$. What about observables that are not coarse grainings of nondegenerate ones? Fortunately, they do not exist in quantum systems, so we can assume that degenerate observables are simply coarse grainings of nondegenerate ones, which is equivalent to saying that any observable can be refined until it becomes nondegenerate. Under this assumption, any state is completely specified by its action on nondegenerate observable, and, in particular, the completely mixed state can be defined as follows:

\begin{definition}[Completely mixed state]\label{def: completelyMixed} Let $\mathfrak{S}$ be a $n$-dimensional system. A state $\emptyset$ in $\mathfrak{S}$ is said to be the completely mixed state if, for each nondegenerate observable $A$, $p_{\emptyset}^{A}$ is the uniform probability distribution on $\sigma(A)$, i.e., for any $\af \in \sigma(A)$,
\begin{align}
    p_{\emptyset}(\af;A)&= \frac{1}{\vert \sigma(A)\vert} = \frac{1}{n}.
\end{align}
\end{definition}

Note that, if a completely mixed state exists, it is unique. As we discussed above, the following definition is important.

\begin{definition}[Experimentally accessible state]\label{def: accessibleState} We say that a state $\rho\in \mathcal{S}$ is experimentally accessible if $\rho = \emptyset$ or 
\begin{align}
    \rho =\label{eq: accessibleStates} (T_{(\Delta_{m};A_{m})} \circ \dots \circ T_{(\Delta_{1};A_{1})})(\emptyset)
\end{align}
for some sequence $(\Delta_{1};A_{1}),\dots,(\Delta_{m};A_{m})$ of observable events. We denote by $\mathcal{S}_{\emptyset}$ the set of all experimentally accessible states.    
\end{definition}
We can finally introduce the following postulate.

\begin{postulate}[Observables and states]\label{post: observables} Any degenerate observable is a coarse-graining of a nondegenerate one, and every conceivable function of a nondegenerate observable has a counterpart in the theory, i.e., if $A$ is a nondegenerate observable and $f$ is a real function on its spectrum, then there is an observable $B \in \mathcal{O}$ satisfying $B = f(A)$ (definition~\ref{def: functionalRelation}). Furthermore, the completely mixed state (definition~\ref{def: completelyMixed}), denoted  $\emptyset$, belongs to $\mathcal{S}$, and every state is a convex combination of experimentally accessible states, which means that $\mathcal{S}$ is the convex hull of $\mathcal{S}_{\emptyset}$ (definition~\ref{def: accessibleState}).
\end{postulate}

Note that we can take all conceivable functions of an observable into account without ending up with an ill-defined set of observables just because the values of observables are assumed to be real numbers. 

Postulate~\ref{post: observables} implies that, if $A$ is any observable and $g$ is a function on $\sigma(A)$, there is an observable $B \in \mathcal{O}$ satisfying $B = g(A)$ --- note that we postulate it only for nondegenerate observables. In fact, let $A$ be any observable and $g$ any function on $\sigma(A)$. According to postulate~\ref{post: observables}, $A = h(C)$ for some nondegenerate observable $C$ and some function $h$, and the same postulate ensures that the observable $B \equiv (g \circ h)(C)$ exists. We have $\sigma(B) = (g \ci h)(\sigma(C)) = g(h(\sigma(C))) = g(\sigma(A))$, and, for any $\Sigma \subset \sigma(B)$ and $\rho \in \mathcal{S}$,
\begin{align*}
    P_{\rho}(\Sigma;B)&= P_{\rho}(\Sigma;(g\circ h)(C)) = P_{\rho}((g \circ h)^{-1}(\Sigma);C) = P_{\rho}(h^{-1}(g^{-1}(\Sigma));C)
    \\
    &= P_{\rho}(g^{-1}(\Sigma);h(C)) = P_{\rho}(g^{-1}(\Sigma);A),
\end{align*}
which shows that $B = g(A)$. It is also important to note that the definition of $B \doteq g(A)$ does not depend on the ``choice of basis'', by which we mean that, if $A = h'(D)$ for some other observable $D$, then $(f \circ h')(D) = B$. In fact, let $B'$ be the observable $(g \circ h')(D)$. Then $\sigma(B') = g(\sigma(A)) = \sigma(B)$, and, for any  $\Sigma \subset \sigma(B)$ and $\rho \in \mathcal{S}$, 
\begin{align*}
    P_{\rho}(\Sigma;B')&= P_{\rho}(g^{-1}(\Sigma);A) =  P_{\rho}(\Sigma;B),
\end{align*}
which implies that $B' = B$.

The properties of functional relations can be nicely illustrated from the perspective of category theory. Roughly speaking, a category consists of a class of objects, together with a class of composable arrows connecting some of them. Each arrow $f$ has a domain (or source) and a codomain (or target), which are objects of the category, and we write $A \xrightarrow{f} B$ (or $f: A \rightarrow B$) to indicate that $f$ is an arrow whose domain and codomain are $A$ and $B$ respectively. If the domain of $g$ and the codomain of $f$ are equal, i.e., if $A \xrightarrow{f} B$ and $B \xrightarrow{g} C$, the associative composition $A \xrightarrow{f} B \xrightarrow{g} C$ is well defined, and every object $A$ of the category has an identity arrow $A \xrightarrow{\text{id}_{A}} A$, which acts as a left and right unit w.r.t. the composition of arrows.  The prototypical example of a category is the category \textbf{Set}, whose objects are sets and whose arrows are functions between them. As this example shows, the collections of objects and arrows of a category are not necessarily sets (recall that, as the famous Russell's paradox shows, the collection of all sets cannot be a set \cite{halmos1960naive}), so the meaning we attribute to the term ``class'' here is the same it has in set theory \cite{levy2002basic}. The definition of category goes as follows \cite{awodey2010category, mac2013categories}.

\begin{definition}[Category]\label{def: category} Let $\mathcal{C}$ be a $6$-tuple $\mathcal{C} \equiv (\mathcal{C}_{0},\mathcal{C}_{1}, \circ, \text{dom}, \text{cod}, \text{id})$, where
\begin{itemize}
    \item[(a)] $\mathcal{C}_{0},\mathcal{C}_{1}$ are classes whose elements we call objects and arrows respectively.
    \item[(b)]  $\text{dom}$ and $\text{cod}$ are functions $\mathcal{C}_{1} \rightarrow \mathcal{C}_{0}$. Given an arrow $f$, $\text{dom}(f)$ and $\text{cod}(f)$ are said to be the domain and codomain of $f$ respectively, and we write $A \xrightarrow{f} B$ (equivalently, $f:A \rightarrow B$) to indicate that $f$ is an arrow whose domain and codomain are $A$ and $B$ respectively.
    \item[(c)] $\text{id}$ is a mapping $\mathcal{C}_{0} \rightarrow \mathcal{C}_{1}$. Given an object $A$, the arrow $\text{id}_{A} \equiv \text{id}(A)$ is said to be the identity arrow on $A$, and both its domain and codomain are the object $A$.
    \item[(d)] $\circ$ is a partial function which assigns, for each pair $(f,g)$ of arrows satisfying $\text{cod}(f)=\text{dom}(g)$, an arrow $g \circ f$ whose domain is $\text{dom}(f)$ and whose codomain is $\text{cod}(g)$. That is, if $A \xrightarrow{f} B$ and $B \xrightarrow{g} C$, then $A \xrightarrow{g \circ f} C$.
\end{itemize}
We say that $\mathcal{C}$ is a category if the following conditions are satisfied.
\begin{itemize}
    \item[(i)] The composition of arrows is associative, i.e., if $A \xrightarrow{f} B$, $B \xrightarrow{g} C$ and $C \xrightarrow{h} B$, then
    \begin{align*}
        h \circ (g \circ h) = (h \circ g) \circ f.
    \end{align*}
    \item[(ii)] Identity arrows are left and right units, i.e., for any arrow $A \xrightarrow{f} B$ we have
    \begin{align*}
         f \circ \text{id}_{A} = f = \text{id}_{B} \circ f.
    \end{align*}
\end{itemize}
\end{definition}
Diagrams are powerful tools in category theory. If we want to say, for instance, that an arrow $A \xrightarrow{h} C$ is the composition of $A \xrightarrow{f} B$ and $B \xrightarrow{g} C$, we simply say that the following diagram is commutative.
\begin{center}
    \begin{tikzcd}
        A\arrow[r,"f"]\arrow[dr,"h"] & B\arrow[d,"g"]\\
        &C
    \end{tikzcd}
\end{center}
Diagrams can be precisely defined as functors between categories \cite{awodey2010category, mac2013categories}, but this level of precision is not necessary here. For us, it is sufficient to understand a diagram as a formal representation of a certain collection of arrows, and a commutative diagram as one in which every path (i.e., every composition of arrows, which in turn is itself an arrow) with the same start and end points (namely domain and codomain) coincide. We say that a diagram commutes to indicate that it is a commutative diagram.

The set $\mathcal{O}$ of observables of a system $\mathfrak{S}$ can be seen as a category, also denoted by $\mathcal{O}$ --- with a slight abuse of notation ---, whose  objects are observables and whose  arrows are the functional relations between them, i.e., there is an arrow from $A$ to $B$ if and only if $B = f(A)$ for some (necessarily surjective) function $\sigma(A) \xrightarrow{f} \sigma(B)$. According to lemma~\ref{lemma: functionalRelation}, there is at most one arrow $A \rightarrow B$, which means that $\mathcal{O}$ is a \textbf{thin category} \cite{nlab2022thin}. We will denote the arrow $A \rightarrow B$ using the same symbol we use to denote the function $\sigma(A) \rightarrow \sigma(B)$ making $B$ a function of $A$, that is to say, if $B=f(A)$, then the arrow $A \rightarrow B$ will also be denoted by $f$.  If $A \xrightarrow{f} B$ and $B \xrightarrow{g} C$ are arrows (i.e., if $B = f(A)$ and $C = g(B)$), then their composition is the arrow $A \rightarrow C$ determined by the mapping $\sigma(A) \xrightarrow{g \ci f} \sigma(C)$, and the composition of arrows in $\mathcal{O}$ is associative because the composition of functions is associative. Finally, the identity arrow $\text{id}_{A}$ of an object (observable) $A$ is the arrow determined by the identity function $\sigma(A) \ni \af \xmapsto{\text{id}_{A}} \af \in \sigma(A)$. In fact, given any arrow $f$ whose codomain is $A$, namely  $B \xrightarrow{f} A$, the composition of functions  $\id_{A} \ci f$ is simply $f$, thus the composition of arrows $B \xrightarrow{f} A \xrightarrow{\id_{A}} A$ is simply $A \xrightarrow{f}B$; analogously, for any arrow $A \xrightarrow{g} C$, we have $ g \ci \id_{A} = g$. It enables us to introduce the following definition.

\begin{definition}[Category of observables]\label{def: categoryOfObservables} Let $\mathfrak{S}$ be a system. The category of observables of $\mathfrak{S}$ is the thin category whose objects are the observables of $\mathfrak{S}$ and whose arrows are the functional relations between them. It means that, if $A$, $B$ are observables in $\mathfrak{S}$, then there is an arrow $A \rightarrow B$ if and only if $B$ is a function of $A$.  
\end{definition}
As we mentioned before, if $\mathcal{O}$ denotes the set of observables of $\mathfrak{S}$, we denote the category of observables also by $\mathcal{O}$.

The category of observables enables us to treat different fine grainings of the same observable as different mathematical entities, something that is not possible at the level of observables. In fact, let $A,C,D$ be observables satisfying $f(C) = A = g(D)$ for some pair of functions $f,g$. Although $f(C)$, $A$, $g(D)$ are the same observable, each one of them is associated with a distinct arrow in the category of observables, as represented in the following diagram.
\begin{center}
    \begin{tikzcd}
        C \arrow[r, "f", swap] & A \arrow[loop above, "\id_{A}"] & D \arrow[l, "g"]
    \end{tikzcd}
\end{center}
This is interesting because  preserving functional relations when assigning values to observables is the main assumption behind Kochen-Specker theorem \cite{kochen1967problem}. Seeing $\mathcal{O}$ as a category can be convenient to discuss contextuality (the theoretical edifice that was constructed upon the Kochen-Specker theorem \cite{budroni2021review, amaral2018graph}) not only mathematically but also conceptually, insofar it can shed light on aspects of physical systems that are usually hidden. Exploring contextuality and the Kochen-Specker theorem is out of the scope of this paper, but we briefly discuss them in the appendix.

At first glance, postulate~\ref{post: observables} seems to reduce degenerate observables to a marginal condition: just as position and momentum are the real observable features of classical mechanics systems, whereas all other observables, functions of position and momentum by definition, are convenient theoretical constructions with no necessary physical significance, it seems that nondegenerate observables are the real features of $\mathfrak{S}$, whereas all degenerate ones are simply theoretical representations of experimental post-processings of them. To put it differently, it seems that nondegenerate observables are more fundamental than degenerate ones. However, this reading of postulate~\ref{post: observables} is incorrect. Consider, for instance, a finite-dimensional quantum system containing two spacelike separated parties, and let $H_{1},H_{2}$ be the Hilbert spaces associated with them. As we know, the entire system is represented by the Hilbert space $H \equiv H_{1} \otimes H_{2}$. An observable $A$ in $H$ is said be the a local observable of $H_{1}$ if $A = A' \otimes \mathds{1}$ for some $A' \in \mathcal{B}(H_{1})_{\text{sa}}$. Similarly, $B$ is said to be an local observable of $H_{2}$  if $B=\mathds{1} \otimes B'$ for some $B' \in \mathcal{B}(H)_{\text{sa}}$. Let $A$,$B$ be local observables of $H_{1}$ and $H_{2}$ respectively, and let $A=\sum_{\af \in \sigma(A)}P_{\af}$, $B=\sum_{\beta \in \sigma(B)} \beta Q_{\beta}$ be their spectral decompositions \cite{landsman2017foundations}. Given any injective mapping $\sigma(A) \times\sigma(B) \ni (\af,\beta) \mapsto \gamma_{(\af,\beta)} \in \mathbb{R}$, define
\begin{align}
    A \ast B \doteq \sum_{\af \in \sigma(A)}\sum_{\beta \in \sigma(B)} \gamma_{(\af,\beta)} P_{\af} \otimes Q_{\beta}.
\end{align}
By construction, both $A,B$ are coarse grainings of $A \ast B$, and, if $\vert \sigma(A)\vert = \text{dim}(H_{1})$ and $\vert \sigma(B)\vert = \text{dim}(H_{2})$,  $A \ast B$ is nondegenerate. We thus have a clear example of a system where promoting nondegenerate observables to the position of ``fundamental features'' is a mistake. In this example,  the fine-graining is a convenient theoretical construction derived from the coarse grainings, and not the reverse. Note that it does not make postulate~\ref{post: observables} untenable, nor contradict Kochen and Specker's claim that one way of measuring a coarse-graining $f(C)$ consists in measuring $C$ and evaluating $f$ in the resulting value \cite{kochen1967problem}. This quantum system obeys postulate~\ref{post: observables} and agrees with Kochen and Specker's proposal of measurement procedure. However, the possibility of measuring $A$ as a coarse-graining of $A \ast B$ does not reduce $A$ to a mere mental or theoretical construction derived from $A \ast B$.

The way we interpret functional relations forces us to take these relations into account in the state update. In fact, if we say that we are measuring an observable $B=f(A)$ by measuring $A$ and evaluating $f$ in the resulting value, then, in this procedure, the occurrence of an event $(\beta;B)$ simply means the occurrence of the event $(\Delta_{\beta};A)$, where $\Delta_{\af} \doteq f^{-1}(\af)$. Therefore, in this procedure, the state has to be updated by the event $(\Delta_{\af};A)$, and not by $(\beta;B)$. Furthermore, the state update cannot depend only on the observable $B$ but also on the measurement procedure: if $B=f(A)=g(C)$, measuring $B$ as a post-processing of $A$ is in general different than measuring it as a post-processing of $C$. Hence, the state update must be defined not by an observable event $(\beta;B)$ but actually by a measurement event $(\beta;B \leftarrow A)$, where $\beta \in \sigma(B)$ and $A \rightarrow B$ is an arrow in the category of observables --- we invert the order of the arrow in the measurement event to highlight the  observable $B$. The measurement event $(\beta;B \leftarrow A)$ represents the experimental situation where  $B$ is measured as a post-processing of $A$ and the outcome $\beta \in \sigma(B)$ is obtained, and, as we argued, this measurement event corresponds to the observable event $(\Delta_{\beta};A)$. In general, a \textbf{measurement event} is a pair $(\Sigma;B \xleftarrow{f} A)$, where $A \xrightarrow{f} B$ is an arrow in the category of observables and $\Sigma \subset \sigma(B)$. This event corresponds to the experimental situation where the observable $B$ is measured as a post-processing of $A$ and some value lying in $\Sigma \subset \sigma(B)$ is obtained, which in turn corresponds to the observable event $(f^{-1}(\Delta);A)$. Therefore, for any state $\rho$ we must define
\begin{align}
    P_{\rho}(\Sigma;B \xleftarrow{f}A) &\doteq P_{\rho}(f^{-1}(\Sigma);A) = P_{\rho}(\Sigma;f(A)) \\
    T_{(\Sigma;B \xleftarrow{f} A)}(\rho) &\doteq T_{(f^{-1}(\Sigma);A)}.
\end{align}
In particular, if $B=A$ and, consequently, if $f$ is the identity arrow $A \rightarrow A$, we have
\begin{align}
    P_{\rho}(\Sigma;A \leftarrow A) &=\label{eq: probMeasurementEvent}P_{\rho}(\Sigma;A) \\
    T_{(\Sigma;A \leftarrow A)}(\rho) &=\label{eq: updateMeasurementEvent} T_{(\Sigma;A)}.
\end{align}
We  see that, although our view on functional relations forces us to redefine both $P$ and $T$ (definition~\ref{def: system}) in terms of measurement events, the previous definitions, based on observable events, are good enough for all practical purposes, since the measurement event $(\Sigma;B \xleftarrow{f}A)$ is  equivalent, w.r.t. both $P$ and $T$, to the observable event $(f^{-1}(\Sigma);A)$. More precisely, an observable event $(\Delta;A)$ defines an equivalence class $[\Delta,A]$ of measurement observables, where $(\Sigma;B \xleftarrow{f}A') \in [\Delta,A]$ if and only if $A'=A$ and $f^{-1}(\Sigma) =\Delta$, thus the set of observable events is, up to isomorphism, a coset of the set of measurement events. Given a state $\rho$, $P_{\rho}$ and $T_{( \ \cdot \ )}(\rho)$ are mappings in this coset, and equations~\ref{eq: probMeasurementEvent},~\ref{eq: updateMeasurementEvent} are simply introducing the natural extensions of these mappings to the set of measurement events. We access events only using $P$ and $T$, so it is usually unnecessary to work with measurement events. Instead of referring to a measurement event $(\Sigma;B \xleftarrow{f} A)$, we can simply refer to the observable event associated with it, namely $(f^{-1}(\Sigma);A)$. Unless explicitly stated otherwise, whenever we refer to an observable event $(\Delta;A)$ from now on we will assume that the associated measurement event is $(\Delta;A \leftarrow A)$, i.e., by a  ``measurement of $A$'' we mean a ``direct'' measurement of $A$, not a post-processing of some other observable.

Now let $A,B$ be observables, and suppose that there exists an observable $C$ such that $A=f(C)$ and $B=g(C)$. According to definition~\ref{def: functionalRelation}, the events $A^{\af}$ and $C^{\Delta_{\af}}$, where $\Delta_{\af} \equiv f^{-1}(\af)$, are statistically equivalent, which means that, for any state $\rho$, $P_{\rho}(A^{\af}) = P_{\rho}(C^{\Delta_{\af}})$. Similarly, the events $B^{\beta}$ and $C^{\Delta_{\beta}}$, where $\Delta_{\beta} \equiv g^{-1}(\beta)$, are statistically equivalent. Therefore,
\begin{align}
    p^{(A,B)}_{\rho}(\af,\beta)= P_{\rho}(\Delta_{\af};C)P_{T_{(\af;A)}(\rho)}(\Delta_{\beta};C).
\end{align}

Now assume that we are measuring both $A$ and $B$ as post-processings of $C$. In this experimental situation, the probability of obtaining the sequence $(\af,\beta)$ of outcomes in a sequential measurement $(A,B)$ is given by
\begin{align}
    p_{\rho}^{(A,B \vert C)}(\af,\beta) \doteq p_{\rho}(\af;A \leftarrow C)P_{T_{(\beta;B \leftarrow C)}(\rho)}(\beta;B \leftarrow C),
\end{align}
whereas equations~\ref{eq: probMeasurementEvent},~\ref{eq: updateMeasurementEvent} and postulate~\ref{post: selfCompatibility} ensure that
\begin{align}
    p_{\rho}^{(A,B \vert C)}(\af,\beta) = P_{\rho}(\Delta_{\af};C)P_{T_{(\Delta_{\af};C)}(\rho)}(\Delta_{\beta};C) = P_{\rho}^{C}(\Delta_{\af} \cap \Delta_{\beta}).
\end{align}

In a system where $p^{(A,B)}_{\rho}(\af,\beta) \neq p^{(A,B\vert C)}_{\rho}(\af,\beta)$, the probability distribution that a state $\rho$ assigns to a sequence of two compatible observables (i.e., two functions of the same observable) will depend not only on the observables we are measuring but also on the procedures we use to measure them. For instance, if $A=f'(D)$ and $B=g'(D)$ for some other observable $D$, then, in principle, $p^{(A,B\vert D)}_{\rho} \neq p^{(A,B \vert C)}_{\rho}$ and $p^{(A,B \vert D)} \neq p^{(A,B)}_{\rho}$. It is thus clear that, for both mathematical and physical reasons, equality $p^{(A,B)}_{\rho}=p^{(A,B\vert C)}_{\rho}$ is highly desirable. This equality is equivalent to the assumption that the sequential measure $P_{\rho}( \ \cdot \ ;A,B)$ (see definition~\ref{def: sequentialMeasure}) is the pushforward of $P_{\rho}( \ \cdot \ ; C)$ along $(f,g): \sigma(C) \rightarrow \sigma(A) \times \sigma(B)$, where $(f,g)(\gamma) \doteq (f(\gamma),g(\gamma))$ for all $\gamma \in \sigma(C)$, which in turn is the mathematical formulation of our considerations about the ``informational link'' that exists between compatible observables, which we discussed in the introduction. We will introduce this reasonable necessary condition in section~\ref{sec: compatibility}. Before diving into this discussion, however, it is worth defining projections.
\section{Projections and their traces}\label{sec: projections}

As we mentioned in section~\ref{sec: observableEvents}, a projection is an observable that singles out an equivalence class of observable events. The definition goes as follows.
\begin{definition}[Projection]\label{def: projection} We say that an observable $E$ is a projection if there is an observable $A$ and a set $\Delta \subset \sigma(A)$ such that $E=\chi_{\Delta}(A)$, where $\chi_{\Delta}$ denotes the characteristic function of $A$. We denote by $\mathcal{P}$ the collection of all projections of $\mathfrak{S}$.
\end{definition}

Note that, if $E=\chi_{\Delta}(A)$, then $\sigma(E)=\chi_{\Delta}(\sigma(A)) \subset \{0,1\}$. We have $\sigma(E) = \{1\}$ iff $\Delta = \sigma(A)$, whereas $\sigma(E) = \{0\}$ iff $\Delta = \oldemptyset$. Otherwise, $\sigma(E)=\{0,1\}$. If a projection $E = \chi_{\Delta}(A)$ is a function of another observable $B$, i.e., if $E = g(B)$ for some function $g$, then $g$ is the characteristic function of some subset $\Sigma \subseteq \sigma(B)$. In fact, $\chi_{\Delta}(A) = E = g(B)$ implies that $\{0,1\} \supset \sigma(E) = \sigma(g(B)) = g(\sigma(B))$, hence $g = \chi_{\Sigma}$, where $\Sigma \equiv g^{-1}(1)$. Let's emphasize this result.
\begin{lemma}\label{lemma: projection}
    Let $E$ be a projection, and let $A$ be any observable such that $E=f(A)$. Then $f$ is the characteristic function of some $\Delta \subset \sigma(A)$ i.e., $f=\chi_{\Delta}$ and, consequently, $E = \chi_{\Delta}(A)$.
\end{lemma}

Let $E=\chi_{\Delta}(A)$ be a projection, and let $\rho$ be any state. Then, according to definition~\ref{def: functionalRelation},
    \begin{align*}
        p_{\rho}(1;\chi_{\Delta}(A)) &= P_{\rho}(\chi_{\Delta}^{-1}(\{1\});A) = P_{\rho}(\Delta;A) \equiv P_{\rho}(A^{\Delta}),\\
        p_{\rho}(0; \chi_{\Delta}(A)) &= 1- p_{\rho}(1;\chi_{\Delta}(A)) = 1 - P_{\rho}(A^{\Delta}).
\end{align*}
Therefore, if $\chi_{\Delta}(A) = \chi_{\Sigma}(B)$, for any state $\rho$ we have $P_{\rho}(A^{\Delta}) = P_{\rho}(1;  \chi_{\Delta}(A)) = P_{\rho}(1; \chi_{\Sigma}(B)) = P_{\rho}(B^{\Sigma})$, which means that $A^{\Delta}$ and $B^{\Sigma}$ are statistically equivalent (definition \ref{def: statisticalEquivalence}). On the other hand, suppose that the events $A^{\Delta}$, $B^{\Sigma}$ are statistically equivalent, which means that, for every state $\rho$, $P_{\rho}(A^{\Delta}) = P_{\rho}(B^{\Sigma})$. Then, for any state $\rho$, $p_{\rho}(1;\chi_{\Delta}(A)) = P_{\rho}(A^{\Delta}) = P_{\rho}(B^{\Sigma}) = p_{\rho}(1;\chi_{\Sigma}(B))$, and consequently $p_{\rho}(0;\chi_{\Delta}(A)) = p_{\rho}(0;\chi_{\Sigma}(B))$. This is equivalent to saying that, for every state $\rho$, $P_{\rho}( \ \cdot \ ; \chi_{\Delta}(A)) = P_{\rho}( \ \cdot \ ; \chi_{\Sigma}(B))$, which in turn means, according to postulate~\ref{ax: separability}, that $\chi_{\Delta}(A) = \chi_{\Sigma}(B)$. It shows that, as we asserted in section \ref{sec: observableEvents}, projections single out equivalent classes of observable events. Let's emphasize this result.

\begin{lemma}[Statistically equivalent events]\label{lemma: statisticalEquivalence} Let $A^{\Delta}$, $B^{\Sigma}$ be observable events (definition \ref{def: events}). Then the following claims are equivalent.
        \begin{itemize}
            \item[(a)] $A^{\Delta}$ and $B^{\Sigma}$ are statistically equivalent. That is, for every state $\rho$,
            \begin{align}
                P_{\rho}(A^{\Delta}) = P_{\rho}(B^{\Sigma}).
            \end{align}
            \item[(b)] $A^{\Delta}$ and $B^{\Sigma}$ are associated with the same projection, that is to say,
            \begin{align}
                \chi_{\Delta}(A) = \chi_{\Sigma}(B).
            \end{align}
    \end{itemize}
\end{lemma}

Let's explore the connection between projections and equivalence classes of events in more depth. Let $\mathbb{E}$ and $\mathcal{P}$ be, respectively, the sets of observable events and projections of a system $\mathfrak{S}$. Let $\chi: \mathbb{E} \rightarrow \mathcal{P}$ be the canonical association between events and projections, i.e., $\chi(A^{\Delta}) \doteq \chi_{\Delta}(A)$ for any event $A^{\Delta}$. According to item $(b)$ of lemma~\ref{lemma: statisticalEquivalence}, two events $A^{\Delta},B^{\Sigma}$ are statistically equivalent iff they belong to the same fiber of $\chi$, i.e., iff $\chi(A^{\Delta}) = \chi(B^{\Sigma})$, thus the equivalence relation induced by $\chi$ on $\mathbb{E}$ is precisely the equivalence relation we denoted by $\sim_{\mathcal{S}}$ in section~\ref{sec: observableEvents}, namely statistical equivalence. It means that the coset $\mathbb{E}/_{\sim_{\mathcal{S}}}$ is, up to isomorphism, the set of projections $\mathcal{P}$, by which we mean that
\begin{align*}
    \mathbb{E}/_{\sim_{\mathcal{S}}} \cong \mathcal{P}.
\end{align*}

More importantly, it immediately follows from  lemma~\ref{lemma: statisticalEquivalence} that any state $\rho$ defines a mapping $\Exp{\cdot}:\mathcal{P} \rightarrow [0,1]$ by
    \begin{align*}
        \Exp{E} \doteq P_{\rho}(A^{\Delta}),
    \end{align*}
where $A^{\Delta}$ is any event associated with $E$ (i.e., $\chi_{\Delta}(A) = E$). We can thus introduce the following definition.
\begin{definition}[Expectation of projections]\label{def: expectationProjection} Let $\rho$ be a state in some system $\mathfrak{S}$. The expectation defined by $\rho$ is the mapping $\Exp{ \ \cdot \ }: \mathcal{P} \rightarrow [0,1]$ given by  
\begin{align*}
    \Exp{E} \doteq P_{\rho}(A^{\Delta}),
\end{align*}
where $E$ is any projection and $A^{\Delta}$ is any event associated with $E$, i.e., $\chi_{\Delta}(A) = E$. It means that $\Exp{ \ \cdot \ }$ is the unique mapping $\mathcal{P} \rightarrow [0,1]$ for which the following diagram is commutative
\begin{center}
    \begin{tikzcd}
        \mathbb{E}\arrow{r}{\chi}\arrow{dr}[swap]{P_{\rho}} & \mathcal{P}\arrow{d}{\Exp{\cdot}}
        \\
        &\left[0,1\right]
    \end{tikzcd}
\end{center}
\end{definition}
The number $\Exp{E}$ is said to be the \textbf{expectation (or expected value) of $E$ with respect to $\rho$}. This terminology is justified by the fact that $\Exp{E}$ is the  expected value of the identity function in $\sigma(E)$ w.r.t. to the probability measure $P_{\rho}^{E}$. We will discuss it in more detail in section~\ref{sec: StatesAsFunctionals}.

It immediately follows from definitions~\ref{def: projection} and~\ref{def: functionalRelation} that, given any state $\rho$ and any projection $E$,
\begin{align*}
    P_{\rho}(1;E) &= \Exp{E},\\
    P_{\rho}(0;E) &= 1-\Exp{E}.
\end{align*}

We have seen in section~\ref{sec: observableEvents} that there is a one-to-one correspondence between states and mappings $\mathbb{E} \rightarrow [0,1]$, i.e., for any pair of states $\rho_{1},\rho_{2}$, we have $P_{\rho_{1}}=P_{\rho_{2}}$ if and only if $\rho_{1}=\rho_{2}$. This injectivity is also satisfied by the mapping $\rho \mapsto \Exp{ \ \cdot \ }$, as the following lemma shows.
\begin{lemma}\label{lemma: expectationAsState} Let $\rho_{1},\rho_{2}$ be states. Then $\rho_{1} = \rho_{2}$ if and only if
\begin{align}
    \langle \ \cdot \ \rangle_{\rho_{1}} = \langle \ \cdot \ \rangle_{\rho_{2}}
\end{align}
\end{lemma}
\begin{proof}
Clearly, $\rho_{1}=\rho_{2}$ implies $\langle \ \cdot \ \rangle_{\rho_{1}} = \langle \ \cdot \ \rangle_{\rho_{2}}$. On the other hand, suppose that $\langle \ \cdot \ \rangle_{\rho_{1}}= \langle \ \cdot \ \rangle_{\rho_{2}}$. It means that, for every projection $E$, we have $\langle E \rangle_{\rho_{1}} = \langle E \rangle_{\rho_{2}}$. Let $A$ be any observable, and $\af \in \sigma(A)$. Then
\begin{align*}
    P_{\rho_{1}}(\{\af\};A) = \langle \chi_{\{\af\}}(A) \rangle_{\rho_{1}} = \langle \chi_{\{\af\}}(A) \rangle_{\rho_{2}}  = P_{\rho_{2}}(\{\af\};A),
\end{align*}
which implies that $P_{\rho_{1}}( \ \cdot \ ; A) =  P_{\rho_{2}}( \ \cdot \ ; A)$. It is true for any $A$, thus, according to postulate~\ref{ax: separability}, $\rho_{1} = \rho_{2}$.
\end{proof}    

Lemma~\ref{lemma: expectationAsState} enables us to identify the state $\rho$ with the expectation $\Exp{ \ \cdot \ }$ it defines. For this reason, we will eventually call $\Exp{ \ \cdot \ }$  a  ``state''.

Since $\mathbb{E}/_{\sim_{\mathcal{S}}} \cong \mathcal{P}$, states separate projections:
\begin{lemma}\label{lemma: separatingProjections} Two projections $E,F$ are equal if and only if, for every state $\rho$,
\begin{align}
    \Exp{E} =\label{eq: separatingProjections} \Exp{F}.
\end{align}
\end{lemma}
\begin{proof}
    Let $E,F$ be projections, and let $A,B$ be observables such that $E=\chi_{\Delta}(A)$ and $F=\chi_{\Sigma}(B)$. Then equation~\ref{eq: separatingProjections} is equivalent to $P_{\rho}(A^{\Delta}) = P_{\rho}(B^{\Sigma})$, which in turn is satisfied for every state $\rho$ if and only if $A^{\Delta}$ and $B^{\Sigma}$ are statistically equivalent. According to lemma~\ref{lemma: statisticalEquivalence}, this is equivalent to saying that $E=F$.
\end{proof}

The following condition translates equation~\ref{eq: realKSdefinition} to the set of projections:

\begin{lemma}[Projections and functional relations]\label{lemma: KSdefinitionProjections} Let $A$ be an observable and $f$ a real function on its spectrum. Then, for any  $\Sigma \subset \sigma(f(A))$,
\begin{align}
    \chi_{\Sigma}(f(A)) = \chi_{f^{-1}(\Sigma)}(A),
\end{align}
which is equivalent to saying that the following diagram commutes in the category of observables
\begin{center}
    \begin{tikzcd}
        A\arrow[r,"f"]\arrow[dr,swap, "\chi_{f^{-1}(\Sigma)}"] &B\arrow[d,"\chi_{\Sigma}"]\\
        & E
    \end{tikzcd}
\end{center}
where $B \equiv f(A)$ and $E \equiv \chi_{\Sigma}(B)$.
\end{lemma}
\begin{proof}
    Let $\rho$ be any state. Then 
    \begin{align*}
        \Exp{\chi_{\Sigma}(f(A))} = P_{\rho}(\Sigma;f(A)) = P_{\rho}(f^{-1}(\Sigma);A) = \Exp{\chi_{f^{-1}(\Sigma)}(A)}.
    \end{align*}
    According to lemma~\ref{lemma: separatingProjections}, it means that $\chi_{\Sigma}(f(A)) = \chi_{f^{-1}(\Sigma)}(A)$.
\end{proof}

The completely mixed state (definition~\ref{def: completelyMixed}) plays a crucial role in our description of projections and observable events. To begin with, we introduce the following definition.
\begin{definition}[Intrinsic probability]\label{def: intrinsicProbability} The intrinsic probability of an observable event $A^{\Delta} \equiv (\Delta;A)$ consists in its probability with respect to the completely mixed state, namely
\begin{align}
    P_{\emptyset}(A^{\Delta}) \equiv P_{\emptyset}(\Delta;A) =\sum_{\af \in \Delta} P_{\emptyset}(A^{\af}).
\end{align} 
Similarly, the intrinsic expectation of a projection $E$ consists in its expected value w.r.t. the completely mixed state, i.e., in the number $\langle E\rangle_{\emptyset}$.
\end{definition}
Note that the intrinsic expectation of a projection $E$ is necessarily equal to the intrinsic probability of any event associated with it. That is, the intrinsic expectation is the natural translation of the intrinsic probability to the coset $\mathbb{E}/\sim_{\mathcal{S}}$.  

According to definition~\ref{def: completelyMixed}, if $A$ is nondegenerate, for any $\af \in \sigma(A)$ we have $P_{\emptyset}(A^{\af})=\frac{1}{n}$, where $n$ denotes the dimension of the system. Consequently, for any $\Delta \subset \sigma(A)$, $P_{\emptyset}(A^{\Delta}) = \frac{\vert \Delta\vert}{n}$. It means that, if $A$ is nondegenerate, the intrinsic probability of an event $A^{\Delta} \equiv (\Delta;A)$ is completely determined by the number of objective events associated with it, namely $\vert \Delta \vert$. Now let $A$ be any observable and $\Delta$ any subset of $\sigma(A)$. According to postulate~\ref{post: observables}, there is a nondegenerate observable $C$ such that $A=f(C)$ for some function $f$. Therefore,
\begin{align*}
    P_{\emptyset}(\Delta;A) = P_{\emptyset}(f^{-1}(\Delta);C) = \frac{\vert f^{-1}(\Delta) \vert}{n}.
\end{align*}
Consequently, if $A=g(D)$ for some other nondegenerate observable $D$, we have, for any $\Delta \subset \sigma(A)$,
\begin{align*}
    \frac{\vert g^{-1}(\Delta) \vert}{n}=P_{\emptyset}(\Delta;A) = \frac{\vert f^{-1}(\Delta) \vert}{n},
\end{align*}
which implies that $\vert g^{-1}(\Delta)\vert = \vert f^{-1}(\Delta)\vert$. Hence, the intrinsic probability of any event $(\Delta;A)$ is determined by the number of objective events associated with $(f^{-1}(\Delta);C)$, where $C$ is a nondegenerate fine graining of $A$ and $f$ is the function satisfying $A=f(C)$. Statistically equivalent events have the same intrinsic probability, thus we can introduce the following definition.
\begin{definition}[Trace of a projection]\label{def: traceProjection} Let $E$ be a projection, and let $A$ be any nondegenerate observable such that $E = \chi_{\Delta}(A)$.The trace of $E$ is the number
\begin{align}
    \Tra{E} \doteq \vert \Delta \vert.
\end{align}
\end{definition}
We define the \textbf{rank} of a projection as its trace, so by a rank-$k$ projection, where $k \in \mathbb{N}$, we mean a projection whose trace is equal to $k$.

As we said, the trace of a projection $E$ is well defined because, if $A,B$ are nondegenerate observables satisfying $\chi_{\Delta}(A)=E=\chi_{\Sigma}(B)$, we obtain $\frac{\vert \Delta \vert}{n} = P_{\rho}(\Delta;A)=\langle E\rangle_{\emptyset}=P_{\emptyset}(\Sigma,B) = \frac{\vert \Sigma \vert}{n}$, which means that $\vert \Delta \vert = \vert \Sigma \vert$. Note also that, in a $n$-dimensional system, the trace of a projection $E$ is a natural number between $1$ and $n$.

It is worth emphasizing the following trivial result.
\begin{lemma}\label{lemma: traceSimple} Let $E$ be a projection in a $n$-dimensional system. Then the intrinsic expectation of $E$ is its trace normalized, i.e.,
    \begin{align}
       \langle E\rangle_{\emptyset} = \frac{\Tra{E}}{n}. 
    \end{align}
\end{lemma}
\begin{proof}
    Let $A$ be any nondegenerate observable satisfying $E=\chi_{\Delta}(A)$ for some $\Delta \subset \sigma(A)$. According to definition~\ref{def: expectationProjection}, $\langle E \rangle_{\emptyset} = P_{\emptyset}(\Delta;A)$. Also, we have seen that $P_{\emptyset}(A^{\Delta})=\frac{\vert \Delta \vert}{n}$, thus, according to definition~\ref{def: traceProjection}, $P_{\emptyset}(A^{\Delta})=\frac{\Tra{E}}{n}$, hence $\langle E \rangle_{\emptyset} = \frac{\Tra{E}}{n}$. 
\end{proof}

The following corollary immediately follows.
\begin{corollary} The intrinsic probability of an event $A^{\Delta}$ satisfies
\begin{align}
    P_{\emptyset}(A^{\Delta}) = \frac{\Tra{\chi_{\Delta}(A)}}{n},
\end{align}
where $n$ denotes the dimension of the system.
\end{corollary}

Let $A$ and $B$ be observables, and consider the projections $\mathds{1}_{A} \equiv \chi_{\sigma(A)}(A)$ and $\mathds{1}_{B} \equiv \chi_{\sigma(B)}(B)$. As we have discussed at the beginning of this section, $\sigma(\mathds{1}_{A}) = \{1\} = \sigma(\mathds{1}_{B})$. Furthermore, for any state $\rho$ we have $\Exp{\mathds{1}_{A}} = p_{\rho}(1;\chi_{\sigma(A)}(A)) = 1 = \Exp{\mathds{1}_{B}}$, thus, according to lemma~\ref{lemma: separatingProjections}, $\mathds{1}_{A}=\mathds{1}_{B}$. It shows that, if we define a projection $\mathds{1}$ by $\mathds{1} \doteq \chi_{\sigma(A)}(A)$, where $A$ is any observable, this definition does not depend on the choice of $A$. Similarly, it is easy to see that $\chi_{\oldemptyset}(A) = \chi_{\oldemptyset}(B)$ for any pair of observables $A,B$, and that $\sigma(\chi_{\oldemptyset}(A))=\{0\}$, thus the projection $0\doteq \chi_{\oldemptyset}(A)$, where $A$ is any observable, is well defined. It justifies the following definition.

\begin{definition}[Zero and unit]\label{def: zeroAndUnit} The \textbf{unit} of a system $\mathfrak{S}$ is the (necessarily unique) projection $\mathds{1}$ satisfying $\mathds{1} = \chi_{\sigma(A)}(A)$ for every observable $A$. Similarly, the \text{zero operator}, or simply the zero, is the (necessarily unique) projection $0$ satisfying $0=\chi_{\oldemptyset}(A)$ for every observable $A$. 
\end{definition}

Note that, in a $n$-dimensional system, $\Tra{\mathds{1}}=n$ and $\Tra{0}=0$. Also, it immediately follows from definition~\ref{def: zeroAndUnit} that, for any state $\rho$, $\Exp{0}=0$ and $\Exp{\mathds{1}}=1$. 

The following lemma, which has already been proved, characterizes the spectrum of any projection.
\begin{lemma}\label{lemma: projectionSpectrum}
    Let $E$ be a projection such that $0 \neq E \neq \mathds{1}$. Then
    \begin{align}
        \sigma(E) = \{0,1\}.
    \end{align}
    Furthermore,
    \begin{align}
        \sigma(\mathds{1}) &= \{1\},\\
        \sigma(0) &= \{0\}
    \end{align}
\end{lemma}

Let $A$ be an observable, and let $\af$ be an eigenvalue of $A$ (recall that, according to lemma~\ref{lemma: eigenvalues}, the spectrum of an observable in a finite-dimensional system is the collection of all its eigenvalues). Let $E_{\af}$ be the projection associated with the event $A^{\af}$, i.e., $E_{\af} \doteq \chi_{\{\af\}}(A)$, which is the unique projection satisfying, for every state $\rho$,
\begin{align*}
    p_{\rho}(\af;A) = \Exp{E_{\af}}.
\end{align*}
It follows from the definition of spectrum that $p_{\rho}(\af;A)\neq 0$ for some state $\rho$, thus, according to lemma~\ref{lemma: separatingProjections}, $E_{\af} \neq 0$. We will refer to the projection associated with an eigenvalue many times throughout the paper, so it is worth emphasizing this definition:
\begin{definition}[Projection associated with an eigenvalue]\label{def: projectionOfEigenvalue} Let $A$ be an observable in some finite-dimensional system, and let $\af$ be an eigenvalue of $A$. We say that the projection associated with the event $A^{\af}$, namely $E_{\af} \doteq \chi_{\{\af\}}(A)$, is the ``projection associated with the eigenvalue $\af$ of $A$''.
\end{definition}

\section{Compatibility (or binary products in the category of observables)}\label{sec: compatibility}

A fundamental aspect of quantum systems is the existence of ``\textit{incompatible}'', ``\textit{incommensurable}'' or ``\textit{non-commuting}'' observables  \cite{landsman2017foundations, bongaarts2015quantum}. As Kochen and Specker point out \cite{kochen1967problem}, the existence of functional relations ensures that a compatibility relation (and consequently the eventual existence of incompatible observables) is well defined in any physical system:
\begin{definition}[Compatibility \cite{kochen1967problem}]\label{def: compatibility} Two observables $A,B$ in a system $\mathfrak{S}$ are said to be compatible if they are both functions of the same observable, i.e., if there is an observable $C$ and  functions $f,g$ on $\sigma(C)$ such that $A=f(C)$ and $B=g(C)$. Equivalently, $A$ and $B$ are compatible if there is \textbf{cone} for them in the category of observables, which consists of an object (i.e., observable) $C$ and a pair of arrows
\begin{center}
    \begin{tikzcd}
        A &\arrow[l]  C \arrow[r] & B
    \end{tikzcd}
\end{center}
Two observables $A,B$ are said to be incompatible if they are not compatible.
\end{definition}
By definition, compatibility is a symmetric relation on $\mathcal{O}$, i.e., $A$ is compatible with $B$ if and only if $B$ is compatible with $A$. It is also a reflexive relation, that is, any observable is compatible with itself, since each observable has an identity arrow associated with it (see definition \ref{def: categoryOfObservables}).

A cone for a diagram is a fundamental concept in category theory, and it is discussed at length in any standard book on the subject (see, for instance, Refs. \cite{mac2013categories, awodey2010category}). We will explore the connection between compatibility and cones in more detail at the end of this section.

It is important to note that two observables $A,B$ are compatible if and only if they are functions of the same \textit{nondegenerate} observable. In fact, suppose that $A$ and $B$ satisfy $A=f(C)$ and $B=g(C)$ for some observable $C$. According to postulate~\ref{post: observables}, we have $C=h(D)$ for some nondegenerate observable $D$, and therefore $A = f'(D)$, $B = g'(D)$, where $f' \doteq f \circ h$ and $g' \doteq g \circ h$. We can illustrate it in the category of observables  using the following commutative diagram.
\begin{center}
    \begin{tikzcd}
        & D\arrow[dl,swap, "f'"]\arrow[d, "h"]\arrow[dr, "g'"] &\\
        A & C\arrow[l,"f"]\arrow[r, swap, "g"] & B
    \end{tikzcd}
\end{center}  
Let $A$, $B$ be compatible observables, and let $A \xleftarrow{f} C \xrightarrow{g} B$ a cone for them. At the end of section~\ref{sec: categoryOfObservables} we asserted that a desirable condition that is necessary for $\mathfrak{S}$ to be a quantum system is the assumption that, for any state $\rho$, the sequential measure $P_{\rho}( \ \cdot \ ; A,B)$ (definition \ref{def: sequentialMeasure}) is the pushforward of $P_{\rho}( \ \cdot \ ; C)$ along $(f,g): \sigma(C) \rightarrow \sigma(A) \times \sigma(B)$, where $(f,g)$ is the product of the arrows $\sigma(C) \xrightarrow{f} \sigma(A)$, $\sigma(C) \xrightarrow{g} \sigma(B)$ in the category \textbf{Set}, i.e., $(f,g)(\gamma) = (f(\gamma),g(\gamma))$ for all $\gamma \in \sigma(C)$. It means that, for any state $\rho$,
\begin{align}\label{eq: sequentialKSdefinitionPre}
    P_{\rho}( \ \cdot \ ;f(C),g(C)) = P_{\rho}((f,g)^{-1}( \ \cdot \ ); C).
\end{align}
This is the natural extension of equation~\ref{eq: realKSdefinition} to compatible observables and the precise formulation of the ``information link'' we discussed in the introduction. Note that, for any $\Delta \times \Sigma \subset \sigma(A) \times \sigma(B)$,
\begin{align}
    p_{\rho}(\Delta \times \Sigma;A,B) &= P_{\rho}(f^{-1}(\Delta)\cap g^{-1}(\Sigma);C),
\end{align}
and in particular
\begin{align}
    p_{\rho}(\af,\beta;A,B) &= P_{\rho}(\Delta_{\af}\cap \Delta_{\beta};C)
\end{align}
for any pair $(\af,\beta) \in \sigma(A) \times \sigma(B)$, where $\Delta_{\af} \equiv f^{-1}(\af)$ and $\Delta_{\beta} \equiv g^{-1}(\beta)$.

In section~\ref{sec: projections} we saw that observable events are associated with projections, and equation~\ref{eq: sequentialKSdefinitionPre} enables us to do the same with sequential events associated with compatible observables. In fact, let $E,F$ be compatible projections, and let $E \xleftarrow{\chi_{\Delta}} C \xrightarrow{\chi_{\Sigma}} F$ be a cone for them. If equation~\ref{eq: sequentialKSdefinitionPre} is satisfied, for any state $\rho$ we obtain
\begin{align*}
    P_{\rho}(1,1;E,F) &= P_{\rho}(\Delta \cap \Sigma;C) = P_{\rho}(1;\chi_{\Delta \cap \Sigma}(C)) = \Exp{\chi_{\Delta \cap \Sigma}(C)}
\end{align*}
(see definition~\ref{def: expectationProjection}). This is true for any cone of $E$ and $F$, so, given any other cone $E \xleftarrow{\chi_{\Delta'}} D \xrightarrow{\chi_{\Sigma'}} F$, we have, for every state $\rho$,
\begin{align*}
     \Exp{\chi_{\Delta \cap \Sigma}(C)}=P_{\rho}(1,1;E,F) = \Exp{\chi_{\Delta' \cap \Sigma'}(D)},
\end{align*}
which implies that $\chi_{\Delta \cap \Sigma}(C) = \chi_{\Delta' \cap \Sigma'}(D)$, according to lemma \ref{lemma: statisticalEquivalence}. Hence, equation~\ref{eq: sequentialKSdefinitionPre} enables us to introduce the following definition.
\begin{definition}[Product of projections]\label{def: productProjections} Let $\mathfrak{S}$ be a system where equation~\ref{eq: sequentialKSdefinitionPre} is valid for any diagram $A \xleftarrow{f} C \xrightarrow{g} B$ and any state $\rho$. Let $E,F$ be compatible projections in this system. The product of $E$ and $F$ is the unique projection  $E \circ F$ satisfying $E \circ F = (\chi_{\Delta} \cdot \chi_{\Sigma})(C)$ for each cone $E \xleftarrow{\chi_{\Delta}} C \xrightarrow{\chi_{\Sigma}} F$. That is, if $E=\chi_{\Delta}(C)$ and $F = \chi_{\Sigma}(C)$,
    \begin{align}
        \chi_{\Delta}(C) \circ \chi_{\Sigma}(C) = (\chi_{\Delta} \cdot \chi_{\Sigma})(C) = \chi_{\Delta \cap \Sigma}(C).
    \end{align}
\end{definition}

It is important to note that, for any projection $E$, we have $E \circ E = E$. In fact, given any $C$ such that $E=\chi_{\Delta}(C)$, we obtain $E \circ E = \chi_{\Delta \cap \Delta}(C) = \chi_{\Delta}(C) = E$. Also, it follows by definition that $E \circ F = F \circ E$ for any pair $E,F$ of compatible projections. 

Now let $A,B$ be compatible observables, and let $E_{\af} \equiv \chi_{\{\af\}}(A)$,  $F_{\beta} \equiv \chi_{\{\beta\}}(B)$ be the projections associated with the eigenvalues $\af$ and $\beta$ of $A$ and $B$ respectively. Since $A$ and $B$ are compatible, $E_{\af}$ and $F_{\beta}$ are compatible. Furthermore, according to equation \ref{eq: sequentialKSdefinitionPre} and definition \ref{def: productProjections}, for any state $\rho$ we have
\begin{align}
    P_{\rho}(\af,\beta;A,B)&=\label{eq: sequentialUpdatePre} \Exp{E_{\af} \circ F_{\beta}} = P_{\rho}(1;E_{\af} \circ F_{\beta}).
\end{align}

Equation \ref{eq: sequentialUpdatePre} shows that, if  $A$ and $B$ are compatible observables, the sequence of objective events $(A^{\af}, B^{\beta})$ and the objective event $(E_{\af} \circ F_{\beta})^{1}$ are equally probable with respect to all states. Consequently, the projection $E_{\af} \circ F_{\beta}$ can be seen as the theoretical counterpart of the sequence $(A^{\af}, B^{\beta})$, just as it is the counterpart of the event $(E_{\af} \circ F_{\beta})^{1}$. For the correspondence between $(A^{\af}, B^{\beta})$ and  $(E_{\af} \circ F_{\beta})^{1}$ to be complete, we have to require that the state update due to the sequence $(A^{\af}, B^{\beta})$, namely the composition $T_{(\beta;B)} \circ T_{(\af;A)}$, is equals to the state update  $T_{(1;E_{\af} \circ F_{\beta})}$ due to the event $(E_{\af} \circ F_{\beta})^{1}$. Note that we are implicitly assuming here that statistically equivalent objective events must update the state of the system in the same way. In principle, it might not be the case, but in the best-case scenario, namely quantum theory, it is. Note also that we cannot make a similar assumption over statistically equivalent subjective events because, according to proposition \ref{prop: subjectiveUpdate}, subjective events are completely determined by the objective events associated with them, so imposing conditions over them could lead us to an inconsistent set of postulates. More importantly, as we will see at the end of this section, it follows from equation \ref{eq: sequentialKSdefinitionPre} that, if $A$ and $B$ are compatible, the sequence $(A,B)$ itself has a counterpart in the theory, which is an observable $\conj{A,B}$ whose spectrum is given by $\sigma(\conj{A,B}) = \{\conj{\af,\beta} \in \mathbb{R}: (\af,\beta) \in \sigma(A)\times \sigma(B), E_{\af} \circ F_{\beta} \neq 0\}$ for some injective function $\conj{ \ , \ }: \sigma(A) \times \sigma(B) \rightarrow \mathbb{R}$, and which satisfies $p_{\rho}(\conj{\af,\beta};\conj{A,B})=p_{\rho}(\af,\beta;A,B)$ for each state $\rho$ and each pair $(\af,\beta) \in \sigma(A) \times \sigma(B)$. The projection associated with the eigenvalue $\conj{\af,\beta}$ of $\conj{A,B}$ (see definition \ref{def: projectionOfEigenvalue}) is  $E_{\af} \circ F_{\beta}$, so, under the assumption that statistically equivalent objective events update the state in the same way, requiring that $T_{(\af;A)} \circ T_{(\beta;B)} = T_{(1,E_{\af} \circ F_{\beta})}$ is equivalent to imposing $T_{(\af;A)} \circ T_{(\beta;B)} = T_{(\conj{\af,\beta};\conj{A,B})}$, which is necessary for the analogy between the sequence $(A,B)$ and the observable $\conj{A,B}$ to be complete. For all these reasons, equation \ref{eq: sequentialUpdate}, which is a necessary condition for $\mathfrak{S}$ to be embedded in a quantum system, is as reasonable as equation~\ref{eq: sequentialKSdefinition}, so we can finally single out the following key feature of quantum mechanics:

\begin{postulate}[Compatibility]\label{post: compatibility} Let $A$, $B$ be compatible observables, and let $A \xleftarrow{f} C \xrightarrow{g} B$ be a cone for them. For any state $\rho$, $P_{\rho}( \ \cdot \ ; A,B)$ (definition~\ref{def: sequentialMeasure}) is the pushforward of $P_{\rho}( \ \cdot \ ; C)$ along $(f,g): \sigma(C) \rightarrow \sigma(A) \times \sigma(B)$, i.e., 
\begin{align}\label{eq: sequentialKSdefinition}
    P_{\rho}( \ \cdot \ ;A,B) = P_{\rho}((f,g)^{-1}( \ \cdot \ ); C).
\end{align}
Furthermore, for any pair $(\af,\beta) \in \sigma(A) \times \sigma(B)$, we have
\begin{align}\label{eq: sequentialUpdate}
    T_{(\af,A)} \circ T_{(\beta,B)} = T_{(1,E_{\af} \circ F_{\beta})},
\end{align}
where $E_{\af} \equiv \chi_{\{\af\}}(A)$ and $F_{\beta} \equiv \chi_{\{\beta\}}(B)$, whereas $E_{\af} \circ F_{\beta}$ is determined by definition~\ref{def: productProjections}.
\end{postulate}

According to postulate~\ref{post: selfCompatibility}, for any observable $A$ and any $\af \in \sigma(A)$, we have $T_{(\af,A)} \circ T_{(\af,A)} = T_{(\af,A)}$. It thus follows from equation~\ref{eq: sequentialUpdate} that
\begin{align*}
    T_{(\af,A)} &= T_{(1,E_{\af} \circ E_{\af})} = T_{(1,E_{\af})},
\end{align*}
where $E \equiv \chi_{\{\af\}}(A)$. Consequently, if $(\af,A)$ and $(\beta,B)$ are statistically equivalent objective events, i.e., if $\chi_{\{\af\}}(A) = \chi_{\{\beta\}}(B)$ (see lemma~\ref{lemma: statisticalEquivalence}), then
\begin{align*}
    T_{(\af,A)} = T_{(1,E)} = T_{(\beta,B)},
\end{align*}
where $E \equiv \chi_{\{\af\}}(A) = \chi_{\{\beta\}}(B)$. It shows that statistically equivalent objective events update the state of the system in the same way, and that this update is completely determined by the projection associated with them. We have thus proved the following lemma.
\begin{lemma}\label{lemma: objectiveProjectiveUpdate} For any objective event $(\af,A)$, we have
\begin{align}
    T_{(\af;A)} = T_{(1;\chi_{\{\af\}}(A))}.
\end{align}
\end{lemma}

This result justifies the following definition.

\begin{definition}[Projection update]\label{def: projectionUpdate} Let $E$ be a projection. The state update determined by $E$, denoted $T_{E}$, is the update $T_{(1,E)}$.
\end{definition}
We have shown in lemma~\ref{lemma: projectionSpectrum} that, if $E=0$, then $\sigma(E) = \{0\}$, and consequently $1 \notin \sigma(E)$. In definition~\ref{def: system} we introduced the ``null state'' $0$, which acts as the null measure in the spectrum of each observable, i.e., for any event $(\Delta,A)$, we have $P_{0}( \Delta;A) = 0$. Also, remember that we restricted the definition of observable event to pairs $(\Sigma;B)$, where $B$ is an observable and $\Sigma \subset \sigma(B)$, for convenience only; in definition~\ref{def: system}, we defined the update $T_{(\Sigma;B)}$ determined by any pair $(\Sigma,B) \in \mathfrak{B}(\mathbb{R}) \times \mathcal{O}$, where $ \mathfrak{B}(\mathbb{R})$ denotes the Borel $\sigma$-algebra, so, even if $E=0$, the update $T_{(1,0)} \equiv T_{(\{1\};E)}$ is well defined. However, since $1 \notin \sigma(O)$, the event $(1;0)$ is impossible, i.e., if we prepare a state $\rho$ and measure the observable $0$, we will never obtain the outcome $1$, given that $P_{\rho}(\{1\};0) = 0$ (recall that, in definition~\ref{def: system}, we defined $P_{\rho}( \ \cdot \ ;B)$ as a Borel measure, so $P_{\rho}(\{1\};0)0$ is well defined either). For this reason, we set the state update determined by the projection $E=0$, namely $T_{0} \equiv T_{(1,0)}$, as the trivial update $\mathcal{S} \ni \rho \mapsto 0 \in \mathcal{S}$ that assigns to null state to all states.

It is important to note that the unit (definition~\ref{def: zeroAndUnit}) does not update the state of the system:
\begin{lemma}\label{lemma: unitUpdate} Let $\mathds{1}$ be the unit (definition~\ref{def: zeroAndUnit}). For any state $\rho$, we have
\begin{align}
    T_{\mathds{1}}(\rho) &= \rho,
\end{align}
which is equivalent to saying that $T_{\mathds{1}}$ is the constant function $\mathcal{S} \ni \rho \mapsto \rho \in \mathcal{S}$.
\end{lemma}
\begin{proof}
    We know that, for any observable $A$, we can write $\mathds{1} = \chi_{\sigma(A)}(A)$. Therefore, given any state $\rho$ and any event $(\Delta,A)$, definition~\ref{def: sequentialMeasure} and postulate~\ref{post: compatibility} imply that
    \begin{align*}
        P_{T_{(1,\mathds{1})}(\rho)}(\Delta;A) &= \frac{P_{\rho}(\{1\} \times \Delta ;\mathds{1},A)}{P_{\rho}(1;\mathds{1})}  = \frac{P_{\rho}(\sigma(A) \times \Delta;A)}{P_{\rho}(\sigma(A); A)} =  P_{\rho}^{A}(\Delta\vert \sigma(A)) = P_{\rho}(\Delta;A),
    \end{align*}
    and therefore $T_{\mathds{1}}(\rho) \equiv T_{(1,\mathds{1})}(\rho)=\rho$.
\end{proof}

As we show in the appendix (lemma~\ref{lemma: nondisturbance}), any system satisfying postulates~\ref{ax: separability}-\ref{post: compatibility} satisfies the so-called nondisturbance condition \cite{amaral2018graph}. As we mentioned in the introduction, to go from nondisturbing to quantum ``correlations'' \cite{csw2014graph, amaral2018graph} we need to connect incompatible observables by taking transition probabilities into account. This is done in section~\ref{sec: transitionProbabilities}. 

Let $A$ and $B$ be compatible observables and let $A \xleftarrow{f} C \xrightarrow{g} B$ be a cone for them. Since, for any state $\rho$, both $P_{\rho}( \ \cdot \ ; A,B)$ and $P_{\rho}( \ \cdot \ ; B,A)$ are simply the pushforward of $P_{\rho}( \ \cdot \ ; C)$ along $(f,g)$, we have  $P_{\rho}( \ \cdot \ ; A,B) = P_{\rho}( \ \cdot \ ; B,A)$. This is equivalent to saying that, for any $\Delta \subset \sigma(A)$ and $\Sigma \subset \sigma(B)$,
\begin{align*}
    P_{\rho}(A^{\Delta})P_{\rho}(B^{\Sigma} \vert A^{\Delta}) = P_{\rho}(B^{\Sigma})P_{\rho}(A^{\Delta} \vert B^{\Sigma}).
\end{align*}
It proves that compatible observables satisfy the following definition.
\begin{definition}[Commutativity]\label{def: commutativity} We say that two observables $A$ and $B$ commute, denoted $A \leftrightarrow B$, if they satisfy the \textbf{Bayes rule}, by which we mean that, for any state $\rho$ and any pair  $\Delta\subset \sigma(A)$, $\Sigma \subset \sigma(B)$,
    \begin{align}
        P_{\rho}(A^{\Delta})P_{\rho}(B^{\Sigma} \vert A^{\Delta}) =\label{eq: bayes} P_{\rho}(B^{\Sigma})P_{\rho}(A^{\Delta} \vert B^{\Sigma}).
    \end{align}
\end{definition}

The following lemma helps us to simplify the definition:
\begin{lemma}\label{lemma: commutativityObjective} Two observables $A,B$ commute iff, for any state $\rho$ and any pair  $\af \in  \sigma(A)$, $\beta \in \sigma(B)$,
    \begin{align}
        P_{\rho}(A^{\af})P_{\rho}(B^{\beta} \vert A^{\af}) =\label{eq: bayesObjective} P_{\rho}(B^{\beta})P_{\rho}(A^{\af} \vert B^{\beta}).
    \end{align}
\end{lemma}
\begin{proof}
As we have discussed, $A$ and $B$ commute iff the probability measures $P_{\rho}( \ \cdot \ ; A,B) = P_{\rho}( \ \cdot \ ; B,A)$ are equal, which in turn is equivalent to saying that $p_{\rho}( \ \cdot \ ; A,B) = p_{\rho}( \ \cdot \ ; B,A)$ (see definition~\ref{def: sequentialMeasure}). This is exactly what equation~\ref{eq: bayesObjective} says, so the proof is complete.
\end{proof}

The reason why compatible observables commute is clear: if $A$ and $B$ are compatible, equation~\ref{eq: bayes} turn out to be equivalent to the  Bayes rule in the probability space $(\sigma(C),\mathcal{P}(\sigma(C)),P_{\rho}^{C})$, where $A \xleftarrow{f} C \xrightarrow{g} B$ is any cone for $A$ and $B$, because $P_{\rho}( \ \cdot \ ;A,B)$ is the pushforward of $P_{\rho}( \cdot \ ;C)$ along $(f,g)$. In fact, for any cone $A \xleftarrow{f} C \xrightarrow{g} B$, equation~\ref{eq: bayes} is equivalent to
\begin{align*}
        P_{\rho}(C^{\Delta'})P_{\rho}(C^{\Sigma'} \vert C^{\Delta'}) = P_{\rho}(C^{\Sigma'})P_{\rho}(C^{\Delta'} \vert C^{\Sigma'}),
\end{align*}
where $\Delta' \equiv f^{-1}(\Delta)$ and $\Sigma' \equiv g^{-1}(\Sigma)$, which in turn, thanks to postulate~\ref{post: selfCompatibility}, is equivalent to the Bayes rule in $(\sigma(C),\mathcal{P}(\sigma(C)),P_{\rho}^{C})$, namely
\begin{align*}
        P_{\rho}^{C}(\Delta')P_{\rho}^{C}(\Sigma' \vert \Delta') =P_{\rho}^{C}(\Sigma')P_{\rho}^{C}(\Delta' \vert \Sigma').
\end{align*}
Since $E\circ F = F \circ E$ for any pair of compatible projections $E,F$, the state update associated with compatible observables also commute:
\begin{lemma}\label{lemma: updateCommute} Let $A$ and $B$ be compatible observables. Then, for any pair $\Delta \subset \sigma(A)$ and $\Sigma \subset \sigma(B)$,
\begin{align}
    T_{(\Sigma,B)} \circ T_{(\Delta,A)} = T_{(\Delta,A)} \circ T_{(\Sigma,B)}.
\end{align}
\end{lemma}
\begin{proof}
    It immediately follows from postulate~\ref{post: compatibility} that $T_{(\beta,B)} \circ T_{(\af,A)} = T_{(\af,A)} \circ T_{(\beta,B)}$ for any $(\af,\beta) \in \sigma(A) \times \sigma(B)$.  Now let $A \xleftarrow{f} C \xrightarrow{g} B$ be any cone for $A$ and $B$. According to proposition~\ref{prop: subjectiveSequentialUpdate} we have, for any $\Delta \subset \sigma(A)$, $\Sigma \subset \sigma(B)$ and any state $\rho$,
    \begin{align*}
        (T_{(\Sigma,B)} \circ T_{(\Delta,A)})(\rho) &= \sum_{(\af,\beta) \in \Delta \times \Sigma} P_{\rho}^{(A,B)}(\{\af,\beta\}\vert \Delta \times \Sigma) (T_{(\af,A)} \circ T_{(\beta,B)}) (\rho)
        \\
        &=  \sum_{(\af,\beta) \in \Delta \times \Sigma} \frac{P_{\rho}^{(A,B)}(\{\af,\beta\})}{P_{\rho}^{(A,B)}(\Delta \times \Sigma)}(T_{(\af,A)} \circ T_{(\beta,B)}) (\rho)
        \\
         &=  \sum_{(\af,\beta) \in \Delta \times \Sigma} \frac{P_{\rho}^{(B,A)}(\{\beta,\af\})}{P_{\rho}^{(B,A)}(\Sigma \times \Delta)}(T_{(\beta,A)} \circ T_{(\af,B)}) (\rho)
         \\
         &=(T_{(\Delta,A)} \circ T_{(\Sigma,B)})(\rho).
    \end{align*}
\end{proof}

Commutativity for compatible observables is a reflection of the mathematically desirable and physically consistent assumption that, for any state $\rho$, $P_{\rho}( \ \cdot \ ; f(C),g(C)) = P_{\rho}((f,g)^{-1}( \ \cdot \ ) ; C)$, which is no more than a natural extension of  equation~\ref{eq: realKSdefinition} imposed in postulate~\ref{post: compatibility}. This condition prevents sequential measurements of compatible observables from depending on the choice of measurement procedure, as we discussed at the end of section~\ref{sec: categoryOfObservables}, and it enables us to use cones connecting compatible observables as informational links to reason about sequential measurements of them in statistical terms. If $A$ and $B$ are incompatible, there is no cone $A \leftarrow C \rightarrow B$ in the category of observables, so $P_{\rho}( \ \cdot \ ; A,B)$ cannot be the pushforward of a function $(f,g): \sigma(C) \rightarrow \sigma(A) \times \sigma(B)$ for some observable $C$. What is missing between two incompatible observables $A$ and $B$ is exactly what justifies commutativity for compatible observables, namely a cone $A \leftarrow C \rightarrow B$. In quantum mechanics, commutativity and compatibility are equivalent concepts, which means that  commutativity can be seen as a mere informational property in this case. From this perspective, a system where incompatible observables can commute is a system whose set of observables is failing to include all possible observables. As we will see in the next sections, having a ``complete'' set of observables, i.e., a set in which a cone exists for each pair of commuting observables, is essential for the emergence of the quantum formalism, since it enables us to prove the Specker's principle and consequently to show that algebraic operations are well defined in any set of pairwise compatible observables. The equivalent between compatibility and commutativity is a key feature of quantum mechanics, and, together with our previous postulates, it enables us to derive essentially all the ``commutative part'' of the quantum formalism, i.e., all theorems and definitions that do not involve algebraic operations between incompatible observables and transition probabilities (definition~\ref{def: transitionProbability}). Hence:

\begin{postulate}[Commutativity]\label{post: commutativity} Only compatible observables are able to commute, i.e., compatibility is a necessary condition for commutativity.    
\end{postulate}

Since all compatible observables commute, we have the following result.
\begin{proposition}[Compatibility and commutativity]\label{prop: compatibilityAndCommutativity} Two observables $A$ and $B$ commute if and only if they are compatible.
\end{proposition}

Postulates~\ref{post: selfCompatibility} and \ref{post: compatibility} make our definition of observable very similar to what is usually called ideal measurement in the literature \cite{budroni2021review, chiribella2020exclusivity, cabello2019simple}. However, it does not mean that we can assume the validity of everything that has been proved about ideal measurements by other authors. Theorems about ideal measurements are usually based on ``operational'' approaches to physics in which some operationally justifiable assumptions that we are not making here are made. For instance, it is proved in Ref.~\cite{cabello2019simple} that the exclusivity principle (corollary~\ref{cor: exclusivityPrinciple} of our work) is satisfied in any KS-scenario  (which is basically a collection of ideal measurements connected by compatibility relations) if we have ideal measurements. However, their operational framework allows them to assume in their proof that any sequence of pairwise compatible binary measurements has a counterpart in the theory, something that we haven't postulated or proved yet. Hence, we cannot simply take for granted that these results will be valid in our work; we need to derive everything from our postulates and nothing more.

To conclude this section, we will show that there is a binary product (definition~\ref{def: binaryProduct}) for any pair of compatible observables in the category of observables. This product is the observable $\conj{A,B}$ that  ``internalizes'' the sequential measurement $(A,B)$ in the system $\mathfrak{S}$, as we mentioned above. The existence of binary products for pairs of compatible observables ensures that the important Specker's principle \cite{cabello2012specker, gonda2018almost} follows from postulates \ref{post: selfCompatibility}-\ref{post: commutativity}, as we will show in section \ref{sec: specker}. Let's begin with some technical results.

\begin{lemma}\label{lemma: sequentialExpectation} Let $A,B$ be compatible observables. For any $\Delta \subset \sigma(A)$, $\Sigma \subset \sigma(B)$, and any state $\rho$,
\begin{align}
    P_{\rho}(\Delta \times \Sigma;A,B) = \Exp{E_{\Delta} \circ F_{\Sigma}},
\end{align}
where $E_{\Delta} \equiv \chi_{\Delta}(A)$ and $F_{\Sigma} \equiv \chi_{\Sigma}(B)$.
\end{lemma}
\begin{proof}
    Let $A \xleftarrow{f} C \xrightarrow{g} B$ be any cone for $A$ and $B$. According to definitions~\ref{def: projection}, \ref{def: expectationProjection}, \ref{def: productProjections} and postulate~\ref{post: compatibility}, for any state $\rho$ we have
    \begin{align*}
        P_{\rho}(\Delta \times \Sigma; A,B) &= P_{\rho}((f,g)^{-1}(\Delta \times \Sigma); C) = P_{\rho}(f^{-1}(\Delta) \cap g^{-1}(\Sigma);C)
        \\
        &= P_{\rho}(1;\chi_{f^{-1}(\Delta) \cap g^{-1}(\Sigma)}(C)) = \Exp{\chi_{f^{-1}(\Delta) \cap g^{-1}(\Sigma)}(C)}
        \\
        &= \Exp{(\chi_{f^{-1}(\Delta)}\cdot \chi_{g^{-1}(\Sigma)})(C)} =  \Exp{\chi_{f^{-1}(\Delta)}(C)\circ \chi_{g^{-1}(\Sigma)}(C)}
        \\
        &= \Exp{\chi_{\Delta}(f(C))\circ \chi_{\Sigma}(g(C))} = \Exp{E_{\Delta} \circ F_{\Sigma}},
    \end{align*}
    where $E_{\Delta} \equiv \chi_{\Delta}(A)$ and $F_{\Sigma} \equiv \chi_{\Sigma}(B)$.
\end{proof}

\begin{lemma}\label{lemma: preservingDimension} Let $A,B$ be compatible observables, and let $\af,\beta$ be eigenvalues of $A$ and $B$ respectively. The following claims are equivalent.
\begin{itemize}
    \item[(a)] For every state $\rho$, $p_{\rho}(\af,\beta;A,B) = 0$.
    \item[(b)] $E_{\af}\circ F_{\beta} = 0$, where $E_{\af} \equiv \chi_{\{\af\}}(A)$ and $F_{\beta} \equiv \chi_{\{\beta\}}(B)$.
    \item[(c)] Given any cone  $A \xleftarrow{f} C \xrightarrow{g} B$, we have $f^{-1}(\af) \cap g^{-1}(\beta) = 0$.
\end{itemize}
\end{lemma}
\begin{proof}
    According to lemma~\ref{lemma: sequentialExpectation}, $p_{\rho}(\af,\beta;A,B) = \Exp{E_{\af} \circ F_{\beta}}$, thus it follows from lemma~\ref{lemma: separatingProjections} and definition \ref{def: zeroAndUnit} that item $(a)$ and $(b)$ are equivalent. Now let $A \xleftarrow{f} C \xrightarrow{g} B$ be any cone for $A$ and $B$. Lemma~\ref{lemma: KSdefinitionProjections} ensures that $E_{\af} = \chi_{f^{-1}(\af)}(C)$ and $F_{\beta} = \chi_{g^{-1}(\beta)}(C)$, thus 
    \begin{align*}
        E_{\af} \circ F_{\beta} = \chi_{f^{-1}(\af)}(C) \circ \chi_{g^{-1}(\beta)}(C) = \chi_{f^{-1}(\Delta) \cap g^{-1}(\beta)}(C).
    \end{align*}
    We know that $ \chi_{f^{-1}(\Delta) \cap g^{-1}(\beta)}(C) = 0$ if and only if $f^{-1}(\Delta) \cap g^{-1}(\beta) = \oldemptyset$, thus items $(b)$ and $(c)$ are equivalent.
\end{proof}
\begin{corollary}\label{cor: rangeOf(f,g)} Let $A,B$ be compatible observables, and let $A \xleftarrow{f} C \xrightarrow{g} B$ be any cone for them. Then the range of $(f,g): \sigma(C) \rightarrow \sigma(A) \times \sigma(B)$, denoted $ (f,g)(\sigma(C))$, satisfies
\begin{align}
    (f,g)(\sigma(C)) = \{(\af,\beta) \in \sigma(A)\times\sigma(B): E_{\af} \circ F_{\beta} \neq 0\},
\end{align}
    where $E_{\af} \equiv \chi_{\{\af\}}(A)$ and $F_{\beta} \equiv \chi_{\{\beta\}}(B)$.
\end{corollary}
\begin{proof}
     Let $A \xleftarrow{f} C \xrightarrow{g} B$ be a cone for pair $A,B$ of compatible observables. For any $(\af,\beta) \in  \sigma(A) \times \sigma(B)$ we have $(f,g)^{-1}(\af,\beta) = f^{-1}(\af) \cap g^{-1}(\beta)$, thus $(\af,\beta) \in (f,g)(\sigma(C))$ iff $f^{-1}(\af) \cap g^{-1}(\beta) \neq \oldemptyset$. According to lemma~\ref{lemma: preservingDimension}, this is equivalent to saying that $E_{\af} \circ F_{\beta} \neq 0$, so the proof is complete.
\end{proof}

Let $A$ and $B$ be compatible observables, and let $\conj{ \ , \ }: \sigma(A) \times \sigma(B) \rightarrow \mathbb{R}$ be any injective function. For any cone $A \xleftarrow{f} C \xrightarrow{g} B$, denote by $\conj{f,g}$ the composition $\conj{ \ , \ } \circ (f,g): \sigma(C) \rightarrow \mathbb{R}$. According to postulate~\ref{post: observables}, the observable $\conj{A,B} \doteq \conj{f,g}(C)$ is well defined, and corollary~\ref{cor: rangeOf(f,g)} implies that 
\begin{align*}
    \sigma(\conj{A,B}) &= \conj{f,g}(\sigma(C)) = \{\conj{\af,\beta}: (\af,\beta) \in (f,g)(\sigma(C))\}
    \\
    &=\{\conj{\af,\beta}: (\af,\beta) \in \sigma(A) \times \sigma(B), E_{\af} \circ F_{\beta} \neq 0\},
\end{align*}
where $E_{\af} \equiv \chi_{\{\af\}}(A)$ and $F_{\beta} \equiv \chi_{\{\beta\}}(B)$. Also, for any $\conj{\af,\beta} \in \sigma(\conj{A,B})$ and any state $\rho$,
\begin{align*}
    p_{\rho}(\conj{\af,\beta};\conj{A,B})&= P_{\rho}(\conj{f,g}^{-1}(\conj{\af,\beta});C) = P_{\rho}((f,g)^{-1}(\af,\beta);C)
    \\
    &= p_{\rho}(\af,\beta;A,B) = \Exp{E_{\af} \circ F_{\beta}}.
\end{align*}
In particular, the projection associated with the event $(\conj{\af,\beta};\conj{A,B})$ is $E_{\af} \circ F_{\beta}$, and consequently
\begin{align*}
    T_{(\conj{\af,\beta};\conj{A,B})}&= T_{E_{\af}\circ F_{\beta}} = T_{(\beta,B)} \circ T_{(\af,A)}.
\end{align*}
Finally, note that the definition of $\conj{A,B}$ does not depend on the cone we chose. That is, let $A \xleftarrow{f'} D \xrightarrow{g'} B$ be any other cone for $A$ and $B$, and define $\conj{A,B}' \doteq \conj{f',g'}(D)$. It immediately follows from corollary~\ref{cor: rangeOf(f,g)} that $\sigma(\conj{A,B}')=\sigma(\conj{A,B})$, and, for any $\xi \in \sigma(\conj{A,B}')$,  the projection associated with the event $(\xi;\conj{A,B}')$ is $E_{\af} \circ F_{\beta}$. It implies that, for any state $\rho$ and any $\xi \in \sigma(\conj{A,B}')$, $p_{\rho}(\xi;\conj{A,B}')=p_{\rho}(\xi;\conj{A,B})$, which in turn implies that $\conj{A,B}'=\conj{A,B}$. It justifies the following definition.

\begin{definition}[Binary conjunction]\label{def: binaryConjunction} Let $A$, $B$ be compatible observables, and let $\conj{ \ , \ }: \sigma(A) \times \sigma(B) \rightarrow \mathbb{R}$ be any injective function. The conjunction of $A$ and $B$ induced by $\conj{ \ , \ }$ is the unique observable $\conj{A,B}$ satisfying $\conj{A,B} = \conj{f,g}(C)$ for every cone $A \xleftarrow{f} C \xrightarrow{g} B$, where $\conj{f,g} \equiv \conj{ \ , \ } \circ (f,g)$.    
\end{definition}

The following lemma has been proved.
\begin{lemma}\label{lemma: conjunction} Let $\conj{A,B}$ be a conjunction of $A$ and $B$. Then, for any $\conj{\af,\beta} \in \sigma(\conj{A,B})$, the projection $E_{\conj{\af,\beta}}$ associated with the event $(\conj{\af,\beta},\conj{A,B})$ is the product $E_{\af} \circ F_{\beta}$, where $E_{\af} \equiv \chi_{\{\af\}}(A)$ and $F_{\beta} \equiv \chi_{\{\beta\}}(B)$. Consequently, 
\begin{align}
    T_{(\conj{\af,\beta};\conj{A,B})} = T_{E_{\af} \circ F_{\beta}} = T_{(\beta,B)}\circ T_{(\af,A)},
\end{align}
and for any state $\rho$,
\begin{align}
    p_{\rho}(\conj{\af,\beta};\conj{A,B})= \Exp{E_{\conj{\af,\beta}}} = p_{\rho}(\af,\beta;A,B).
\end{align}
Finally, we have $\sigma(\conj{A,B}) = \{\conj{\af,\beta}: (\af,\beta) \in \sigma(A) \times \sigma(B), E_{\af} \circ F_{\beta} \neq 0\}$.
\end{lemma}

Compatibility is defined by a cone, so the question naturally arises of whether there is a \textit{limit} for any pair of compatible observables \cite{awodey2010category, mac2013categories}. This limit  is called a \textit{binary product} in category theory \cite{awodey2010category, mac2013categories}, and its definition goes as follows. 

\begin{definition}[Binary product]\label{def: binaryProduct} Let $\mathcal{C}$ be a category (definition~\ref{def: category}), and let $A$,$B$ be objects of this category. A cone for $A,B$ consists of an object $C$ and arrows $A \xleftarrow{f} C \xrightarrow{g} B$. A cone $A \xleftarrow{\theta_{A}} P \xrightarrow{\theta_{B}} B$ is said to be a product diagram if it is a limit cone, i.e., if and only if, given any cone $A \xleftarrow{f} C \xrightarrow{g} B$, there is one, and only one, arrow $C \xrightarrow{h} P$  making the following diagram commute
\begin{center}
    \begin{tikzcd}
        & C\arrow[dl, swap, "f"]\arrow[d, dashed, "h"]\arrow[dr, "g"] &\\
        A & P\arrow[l,"\theta_{A}"]\arrow[r, swap, "\theta_{B}"] & B
    \end{tikzcd}
\end{center}  
\end{definition}

The arrow $C \rightarrow P$ in the diagram above is dashed to indicate that it is the unique arrow making this diagram commute. In the category of observables, it makes no difference because there can be at most one arrow from one observable to another.

In the category \textbf{Set}, the canonical example of a product for the sets $U_{1}$, $U_{2}$ consists in the diagram $U_{1} \xleftarrow{\pi_{1}} U_{1} \times U_{2} \xrightarrow{\pi_{2}} U_{2}$, where $U_{1} \times U_{2}$ denotes the Cartesian product and $\pi_{1}$, $\pi_{2}$ denote the projections $(u_{1},u_{2}) \mapsto u_{1}$ and $(u_{1},u_{2}) \mapsto u_{2}$ respectively. Furthermore,  given a pair of arrows $U_{1} \xleftarrow{f} V\xrightarrow{g} U_{2}$, it is easy to see that the unique arrow $V\rightarrow U_{1} \times U_{2}$ making the diagram in definition~\ref{def: binaryProduct} commute  is the function $(f,g): V\rightarrow U_{1}\times U_{2}$ given by $\forall_{v \in V}: \ (f,g)(v) \doteq (f(v),g(v))$.

An arrow $A \xrightarrow{f} B$ in a category $\mathcal{C}$ is said to be an \textbf{isomorphism} (or simply an \textit{iso}) if there exists an arrow $B \xrightarrow{g} A$ (called the inverse of $f$ and usually denoted $f^{-1}$) satisfying $g \circ f = \text{id}_{A}$ and $f \circ g = \id_{B}$. Clearly, $f=(f^{-1})^{-1}$. By construction, an arrow $A \xrightarrow{f} B$ in the category of observables is an isomorphism if and only if the function $\sigma(A) \xrightarrow \sigma(B)$ associated with it is a bijection (equivalently, an isomorphism in \textbf{Set}). In any category, limits are unique up to isomorphism \cite{awodey2010category}, so, if both cones $A \leftarrow P \rightarrow B$ and $A \leftarrow Q \rightarrow B$ are product diagrams for $A$ and $B$, then $P$ and $Q$ are \textit{isomorphic objects} \cite{awodey2010category}. For this reason, it is commonplace in category theory to talk about \textit{the} product of two objects. We will show that conjunctions  are binary products in the category of observables, so, even though definition \ref{def: binaryConjunction} depends on an isomorphism $\conj{ \ ,\ }$, we can talk about \textit{the} conjunction of $A$ and $B$ without risk of confusion. 

Now let's show  that the conjunction $\conj{A,B}$, endowed with appropriate arrows, is a product for $A$ and $B$. To begin with, note that, if $\af \in \sigma(A)$, then $\conj{\af,\beta} \in \sigma(\conj{A,B})$ for some $\beta \in \sigma(B)$, which is equivalent to saying that $E_{\af} \circ F_{\beta} \neq 0$, where $E_{\af} \equiv \chi_{\{\af\}}(A)$ and $F_{\beta} \equiv \chi_{\{\beta\}}(B)$ (see lemma~\ref{lemma: conjunction}). In fact, assume, for a proof by contradiction, that $\af \in \sigma(A)$ is such that $\conj{\af,\beta} \notin \sigma(\conj{A,B})$ for all $\beta \in \sigma(B)$. According to lemmas~\ref{lemma: preservingDimension} and~\ref{lemma: conjunction}, it implies that
\begin{align*}
    0 = \sum_{\beta \in \sigma(B)} p_{\emptyset}(\af,\beta;A,B)  =  p_{\emptyset}(\af;A)\sum_{\beta \in \sigma(B)} p_{T_{(\af,A)}(\emptyset)}(\beta;B) = p_{\emptyset}(\af;A) = \langle E_{\af} \rangle_{\emptyset},
\end{align*}
where, as usual, $\emptyset$ denotes the completely mixed state (definition~\ref{def: completelyMixed}). $E_{\af}$ is the projection associated with the eigenvalue $\af$ of $A$ (definition~\ref{def: projectionOfEigenvalue}), so, as we have discussed at the end of section~\ref{sec: projections}, it is non zero. Hence, its trace is non zero (see definition~\ref{def: traceProjection}), and consequently $\langle E_{\af} \rangle_{\emptyset} \neq 0$, which completes the proof. Since $A$ and $B$ commute, it is analogous to show that, for any $\beta \in \sigma(B)$, we have $\conj{\af,\beta} \in \sigma(\conj{A,B})$ for some $\af \in \sigma(A)$. These results ensure that the mappings $\sigma(\conj{A,B}) \ni \conj{\af,\beta} \xmapsto{\theta_{A}} \af \in \sigma(A)$ and  $\sigma(\conj{A,B}) \ni \conj{\af,\beta} \xmapsto{\theta_{B}} \beta \in \sigma(B)$ are surjective (note that these functions are well defined because $\conj{ \ , \ }$ is injective). Finally, for any $\af \in \sigma(A)$ and any state $\rho$ we have
\begin{align*}
    P_{\rho}(\theta_{A}^{-1}(\af);\conj{A,B})&= \sum_{\substack{\beta \in \sigma(B) \\ E_{\af} \circ F_{\beta} \neq 0}}p_{\rho}(\conj{\af,\beta};\conj{A,B}) = \sum_{\beta \in \sigma(B)} p_{\rho}(\af,\beta;A,B) = p_{\rho}(\af;A),
\end{align*}
which implies that $P_{\rho}( \ \cdot \  ;A)$ is the pushforward of  $P_{\rho}( \ \cdot \ ;\conj{A,B})$ along $\theta_{A}$. It is analogous to show that  $P_{\rho}( \ \cdot \ ;B)$ is the pushforward of  $P_{\rho}( \ \cdot \ ;\conj{A,B})$ along $\theta_{B}$. Hence, $\theta_{A}$ and $\theta_{B}$ induce arrows $A \xleftarrow{\theta_{A}} \conj{A,B} \xrightarrow{\theta_{B}} B$ in the category of observables, which means that $A=\theta_{A}(\conj{A,B})$ and $B=\theta_{B}(\conj{A,B})$. Finally:

\begin{proposition}[Conjunctions and binary products]\label{prop: conjunctionAndProduct} Let $A$ and $B$ be compatible observables and $\conj{A,B}$ a conjunction of them. Let $\theta_{A}: \conj{A,B} \ri A$ and $\theta_{B}: \conj{A,B} \ri B$ be the arrows induced by the (surjective) functions $\sigma(\conj{A,B}) \ni \conj{\af,\beta} \mapsto \af \in \sigma(A)$ and  $\sigma(\conj{A,B}) \ni \conj{\af,\beta} \mapsto \beta \in \sigma(B)$ respectively. Then the diagram
\begin{align}
    A \xleftarrow{\theta_{A}} \conj{A,B} \xrightarrow{\theta_{B}} B
\end{align}
is a binary product for $A$ and $B$.
\end{proposition}
\begin{proof}
    Let $A,B$ be compatible observables, and let $A \xleftarrow{f} C \xrightarrow{g} B$ be any cone for them. It is easy to see that the following diagram commutes in \textbf{Set}.
    \begin{center}
        \begin{tikzcd}
            & \sigma(C)\arrow[dl, swap, "f"]\arrow[d, "(f\text{,}g)"]\arrow[dr, "g"] &\\
            \sigma(A) & \sigma(A) \times \sigma(B)\arrow[d,"\conj{ \ \text{,} \ }"] & \sigma(B)\\
            & \sigma(\conj{ A\text{,} B })\arrow[ul,"\theta_{A}"]\arrow[ur,swap, "\theta_{B}"] &
        \end{tikzcd}
    \end{center}
    In fact, for any $\gamma \in \sigma(C)$ we have $(\theta_{A}\circ \conj{ \ , \ } \circ (f,g))(\gamma) = \theta_{A}(\conj{f(\gamma),g(\gamma)}) = f(\gamma)$, thus $f=\theta_{A} \circ \conj{ \ , \ } \circ (f,g)$, and it is analogous to show that $g=\theta_{B} \circ \conj{ \ , \ } \circ (f,g)$. Recall that $\conj{f,g} \equiv \conj{ \ , \ } \circ (f,g)$. Therefore, the following diagram commutes in the category of observables.    
    \begin{center}
    \begin{tikzcd}[row sep=large, column sep=large]
        & C\arrow[dl, swap, "f"]\arrow[d, dashed, "\conj{f,g}"]\arrow[dr, "g"] &\\
        A& \lfloor A\text{,}B \rfloor \arrow[l,"\theta_{A}"]\arrow[r, swap, "\theta_{B}"] & B
    \end{tikzcd}
\end{center}  
The category of observables is a thin category, so there is no other arrow $C \rightarrow \conj{A,B}$ making this diagram commute, therefore the proof is complete.
\end{proof}

As we have discussed, products (more generally, limits) are unique up to isomorphism, and it easily follows from this fact that any product diagram of $A$ and $B$ is given by a conjunction. 
\section{Specker's principle (or finite products in the category of observables)}\label{sec: specker}

In 1932, von Neumann proved that any Abelian von Neumann algebra is generated by a single selfadjoint operator \cite{von2018mathematical, doring2005kochen}.  It implies that, if $A_{1},\dots,A_{m}$ are pairwise compatible observables in a quantum system, then they are all functions of a single observable $C$, and a simple corollary of this result is that, if $E_{1},\dots,E_{m}$ are pairwise orthogonal projections and $\rho$ is a state of the (finite-dimensional, for simplicity) quantum system, then $\sum_{i=1}^{m} \Exp{E_{i}} \leq 1$, where $\Exp{ \ \cdot \ } \equiv \Tra{ \rho \ \cdot \ }$. As notably recognized by Ernst Specker \cite{cabello2012specker, gonda2018almost}, von Neumann's theorem (or, more specifically, the fact that finitely many pairwise compatible observables are always functions of a single observable) and its corollary are very important features of quantum systems. Their translations to ``theory-independent approaches to physics'' are usually called Specker's principle and exclusivity principle in the literature, and they have received much attention in recent years \cite{cabello2012specker, gonda2018almost, amaral2014exclusivity, amaral2018graph}. Specker's principle plays a crucial role in our work, especially in section~\ref{sec: algebraicOperations}, and the aim of this section is to prove it. 

Specker's principle follows from the fact that any finite set of pairwise compatible observables has a finite product  \cite{awodey2010category, goldblatt2006topoi, tezzin2020estados} in the category of observables. The definition of product goes as follows.

\begin{definition}[Finite product] Let $A_{1},\dots,A_{m}$ be objects in some category $\mathcal{C}$. A \textbf{cone} for $A_{1},\dots,A_{m}$ consists of an object $C$ and arrows $C\xrightarrow{f_{i}} A_{i}$, $i=1,\dots,m$. A cone  $P\xrightarrow{\theta_{i}} A_{i}$, $i=1,\dots,m$, for $A_{1},\dots,A_{m}$ is said to be a product diagram if it is a limit cone, which means that, given any other cone $C \xrightarrow{f_{i}} A_{i}$, $i=1,\dots,m$, there is one, and only one, arrow $C \xrightarrow{h} P$ making the following diagram commute for all $i$
\begin{center}
    \begin{tikzcd}
        C\arrow[d, dashed,"h"]\arrow[dr,"f_{i}"] & \\
        P\arrow[r,"\theta_{i}"] & A_{i}
    \end{tikzcd}
\end{center}
\end{definition}
If $P\xrightarrow{\theta_{i}} A_{i}$, $i=1,\dots,m$ is a product diagram for $A_{1},\dots,A_{m}$, we will usually denote $P$ by $\conj{A_{1},\dots,A_{m}}$. Similarly, given a cone $C \xrightarrow{f_{i}} A_{i}$, $i=1,\dots,m$, we will usually denote the arrow $C \xrightarrow{h} P$ by $\conj{f_{1},\dots,f_{m}}$. In the category of observables, we will call $\conj{A_{1},\dots,A_{m}}$ a \textbf{conjunction} of $A_{1},\dots,A_{m}$, as in definition~\ref{def: binaryConjunction}.

For the same reason why binary products are unique up to isomorphism, finite products are unique up to isomorphism (more broadly, any limit in a category is unique up to isomorphism \cite{awodey2010category, mac2013categories}). Furthermore, it is well known that, in any category, products are ``associative'' (insofar they are well defined, which is not always the case), by which we mean that
\begin{align*}
    \conj{A,B,C} &\cong \conj{\conj{A,B},C} \cong \conj{A,\conj{B,C}} \cong \conj{\conj{A,C},B}.
\end{align*}
It implies that
\begin{align}
    \conj{A_{1},\dots,A_{m}} \cong\label{eq: SpeckerConjunction} \conj{\dots\conj{\conj{A_{1},A_{2}},A_{3}},\dots, A_{m}}.
\end{align}

If $A_{1},\dots,A_{m}$ are pairwise compatible observables, all conjunctions (or, to use the standard terminology of category theory, products) at the right-hand side of equation~\ref{eq: SpeckerConjunction} are well defined. To prove this, we need the following lemma.

\begin{lemma}\label{lemma: compatibleWithConjunction} Let $A,B$ be compatible observables, and let $C$ be any observable. Then the following claims are equivalent.
\begin{itemize}
    \item[(a)] $C$ is compatible with both $A$ and $B$.
    \item[(b)] $C$ is compatible with the  conjunction $\conj{A,B}$.
\end{itemize}
\end{lemma}
\begin{proof}
    If $C$ and $\conj{A,B}$ are compatible, there is a cone $ \conj{A,B} \leftarrow D \rightarrow C$, which immediately implies that there is a cone for $A$,$B$ and $C$, as the following diagram shows
    
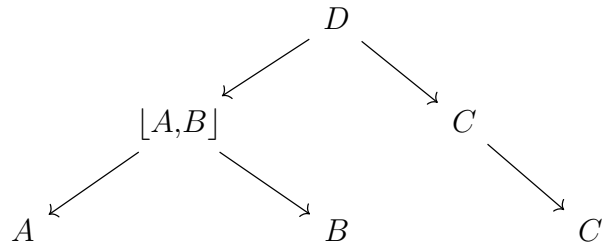
\begin{center}
    \begin{tikzcd}
        & &  D\arrow[dl]\arrow[dr] & \\
        & \conj{A\text{,}B}\arrow[dl]\arrow[dr] & &  C\arrow[dr]\\
        A & & B & & C
    \end{tikzcd}
    \end{center}
    (we introduce the identity arrow $C \rightarrow C$ for purely aesthetic reasons). In particular, $A,B$, and $C$ are pairwise compatible. Now, suppose that $A$, $B$ and $C$ are pairwise compatible, and let $\af,\beta,\gamma$ be eigenvalues of $A,B,C$ respectively. According to lemma~\ref{lemma: updateCommute}, for any state $\rho$ we have
    \begin{align*}
        p_{\rho}(\gamma,\conj{\af,\beta};C,\conj{A,B}) &= p_{\rho}(\gamma;C) P_{T_{(\gamma,C)}(\rho)}(\conj{\af,\beta};\conj{A,B})
        \\
        &= p_{\rho}(\gamma;C) P_{T_{(\gamma,C)}(\rho)}(\af,\beta;A,B)
        \\
        &=P_{\rho}(\gamma,\af;C,A)P_{T_{(\gamma,\af;C,A)}(\rho)}(\beta,B)
        \\
        &= P_{\rho}(\af,\gamma;A,C)P_{T_{(\af,\gamma;A,C)}(\rho)}(\beta,B)
        \\
        &=P_{\rho}(\af,\gamma,\beta;A,C,B) = P_{\rho}(\af,\beta,\gamma;A,B,C)
        \\
        &= P_{\rho}(\conj{\af,\beta},\gamma;\conj{A,B},C).
    \end{align*}
    According to lemma~\ref{lemma: commutativityObjective}, it proves that $C$ and $\conj{A,B}$ commute, and, according to postulate~\ref{post: commutativity}, it means that $C$ and $\conj{A,B}$ are compatible. The proof is thus complete.
\end{proof}

\begin{theorem}[Product diagram and compatibility]\label{thm: productCompatibility} Let $A_{1},\dots,A_{m}$ be observables. There is a product diagram for $A_{1},\dots,A_{m}$ if and only if they are pairwise compatible.    
\end{theorem}
\begin{proof}
   Having a product is clearly sufficient for being pairwise compatible, then we just need to prove that it is also necessary. So let $A_{1},\dots,A_{m}$ be pairwise compatible observables. Let's prove, by induction, that the chain of conjunctions $\conj{\conj{\conj{A_{1},A_{2}},A_{3}},\dots, A_{m}}$ is well defined.  According to lemma~\ref{lemma: compatibleWithConjunction}, $\conj{\conj{A_{1},A_{2}},A_{3}}$ is well defined, so all we need to do is to prove the induction step. Suppose thus that, for some $k \in \{1,\dots,m\}$, the chain of conjunctions $\conj{\conj{\conj{A_{1},A_{2}},A_{3}},\dots, A_{k}}$ is well defined. For each $i \in \{2,\dots,k\}$, define $P_{i} \doteq \conj{\conj{\conj{A_{1},A_{2}},A_{3}},\dots, A_{i}}$, and define $P_{1} \doteq A_{1}$. We have, by definition, $P_{i}=\conj{P_{i-1},A_{i}}$ for every $i\in \{i=2,\dots,k\}$, and $P_{1}=A_{1}$. According to lemma~\ref{lemma: compatibleWithConjunction}, if an observable $C$ is compatible with $P_{i-1}$ and $A_{i}$  for some $i \in \{2,\dots,k\}$, it is also compatible with $P_{i}$. We know that $A_{k+1}$ is compatible with $A_{i}$ for every $i=1,\dots,m$, thus, if $A_{k+1}$ is compatible with $P_{i-1}$, it is also compatible with $P_{i}$. We know that $A_{k+1}$ is compatible with $P_{1}$ and $P_{2}$, so, by induction, it is compatible with $P_{i}$ for every $i \in \{1,\dots,k\}$. Consequently, $\conj{\conj{\conj{A_{1},A_{2}},A_{3}},\dots, A_{k+1}}$ is well defined, which completes the proof of the induction step. To illustrate:
   \begin{center}
    \begin{tikzcd}
    \conj{\conj{\conj{\conj{A_{1},A_{2}},A_{3}},A_{4}},A_{5}}\arrow[d]\arrow[dr]   & & & &\\
    \conj{\conj{\conj{A_{1},A_{2}},A_{3}},A_{4}}\arrow[d]\arrow[dr]  & A_{5}\arrow[dr] & & & \\
    \conj{\conj{A_{1},A_{2}},A_{3}}\arrow[d]\arrow[dr]  & A_{4}\arrow[dr] & A_{5}\arrow[dr] & &  \\
    \conj{A_{1},A_{2}}\arrow[d]\arrow[dr]  & A_{3}\arrow[dr] & A_{4}\arrow[dr] & A_{5}\arrow[dr] & \\
    A_{1}  &  A_{2} & A_{3} & A_{4} &  A_{5}
    \end{tikzcd}
\end{center}
   (as before, we introduce  identity arrows for purely aesthetic reasons). To conclude, we need to show that $\conj{\conj{\conj{A_{1},A_{2}},A_{3}},\dots, A_{m}}$ is a product (or, as we say, conjunction) for $A_{1},\dots,A_{m}$. However, it immediately follows the fact that the product is associative in any category\cite{awodey2010category, mac2013categories}, so the proof is complete.
   
\end{proof}

\begin{corollary}[Specker's principle]\label{cor: specker} If $A_{1}, \dots,A_{m}$ are pairwise compatible observables, there is a cone $C \xrightarrow{f_{i}}A_{i}$, $i=1,\dots,m$, for them, i.e., there is an observable $C$ and real functions $f_{1},\dots,f_{m}$ on $\sigma(C)$ such that  $A_{i}=f_{i}(C)$ for all $i$.    
\end{corollary}

Although the Exclusivity principle immediately follows from corollary~\ref{cor: specker}, it is more convenient to discuss it after introducing algebraic operations between observables. For this reason, we will prove it only in section~\ref{sec: spectralTheory}, as corollary~\ref{cor: exclusivityPrinciple}.

\begin{corollary}\label{cor: nondegenerateSpecker} Let $A_{1},\dots,A_{m}$ be pairwise compatible observables. Then there exists a \textbf{nondegenerate cone} for them, i.e., a cone $D \xrightarrow{f_{i}} A_{i}$, $i=1,\dots,m$, such that $D$ is a nondegenerate observable. 
\end{corollary}
\begin{proof}
    Let $A_{1},\dots,A_{m}$ be pairwise compatible observables, and let $C \xrightarrow{f_{i}}A_{i}$, $i=1,\dots,m$, be any cone for them. Let $D$ be a nondegenerate observable such that $C=g(D)$ for some function $g$, which exists according to postulate~\ref{post: observables}, and define $g_{i} \doteq f_{i} \circ g$ for each $i$. Then $D \xrightarrow{g_{i}} A_{i}$, $i=1,\dots,m$, is a cone for $A_{1},\dots,A_{m}$, which completes the proof.
\end{proof}

We can now prove the following generalized version of postulate~\ref{post: compatibility}.

\begin{proposition}\label{prop: jointPushforward} Let $A_{1},\dots,A_{m}$ be pairwise compatible observables, and let $C \xrightarrow{f_{i}} A_{i}$, $i=1,\dots,m$ be any cone for them. For any state $\rho$, the sequential measure $P_{\rho}( \ \cdot \ ;f_{1}(C),\dots,f_{m}(C))$ (see definition~\ref{def: sequentialMeasure}) is the pushforward of $P_{\rho}( \ \cdot \ ; C)$ along $(f_{1},\dots,f_{m}): \sigma(C) \rightarrow \prod_{i=1}^{m}\sigma(A_{i})$, i.e.,
\begin{align}
    P_{\rho}( \ \cdot \ ;f_{1}(C),\dots,f_{m}(C)) = P_{\rho}((f_{1},\dots,f_{m})^{-1}( \ \cdot \ ); C).
\end{align}
Furthermore, for any $\underline{\af} \equiv (\af_{1},\dots,\af_{m}) \in \prod_{i=1}^{m}\sigma(A_{i})$ we have
\begin{align}
    T_{(\af_{m},A_{m})} \circ \dots \circ T_{(\af_{1},A_{1})} &= T_{(1,\chi_{\cap_{i=1}^{m}}\Delta_{i})}(C),
\end{align}
where $\Delta_{i} \equiv f_{i}^{-1}(\af_{i})$ for each $i=1,\dots,m$.
\end{proposition}
\begin{proof}
The sample space $\sigma(\underline{A}) \equiv \prod_{i=1}^{m}\sigma(A_{i})$ is finite, so it is sufficient to show that, for any $\underline{\af} \in \sigma(\underline{A})$, where $\underline{\af} \equiv (\af_{1},\dots,\af_{m})$, we have $p_{\rho}(\underline{\af};f_{1}(C),\dots,f_{m}(C)) = p_{\rho}((f_{1},\dots,f_{m})^{-1}(\underline{\af}); C)$. So let $\underline{\af}$ be any element of  $\sigma(\underline{A})$, and fix a state $\rho$. For each $ i=1,\dots,m$, write $\Delta_{i} \equiv f_{i}^{-1}(\af_{i})$ and $ \rho_{i} \equiv T_{(\af_{1},\dots,\af_{i};A_{1},\dots,A_{i})}(\rho)$. Then
\begin{align*}
        p_{\rho}(\underline{\af}; f_{1}(C),\dots,f_{m}(C))  &= p_{\rho}((\af_{1},\dots, \af_{m-2}); A_{1},\dots,A_{m-2}) 
        \\
        &\times p_{\rho_{m-2}}((\af_{m-1},\af_{m}); A_{m-1},A_{m})
        \\
        &= P_{\rho}(\{\af_{1} \dots \af_{m-2}\};A_{1},\dots,A_{m-2})
        \\
        &\times P_{\rho_{m-2}}(\cap_{j=m-1}^{m}\Delta_{j};C)
        \\
        &=P_{\rho}(\{\af_{1}\} \times \dots \times \{\af_{m-2}\} \times  \cap_{j=m-1}^{m}\Delta_{j}; A_{1}, \dots, A_{m-2},C)
        \\
        &= p_{\rho}((\af_{1},\dots, \af_{m-3};A_{1}, \dots A_{m-3})
        \\
        &\times P_{\rho_{m-3}}(\{\af_{m-2}\}\times \cap_{j=m-1}^{m}\Sigma_{j};A_{m-2},C)
        \\
        &=  p_{\rho}((\af_{1},\dots, \af_{m-3};A_{1}, \dots A_{m-3})
        \\
        &\times P_{\rho_{m-3}}( \cap_{j=m-2}^{m}\Delta_{j};C)
        \\
        &= P_{\rho}(\{\af_{1}\} \times \dots \times \{\af_{m-3}\} \times \cap_{j=m-2}^{m}\Delta_{j} ; A_{1}, \dots A_{m-3},C)
        \\
        &=P_{\rho}(\{\af_{1}\} \times \dots \times \{\af_{m-4}\} \times \cap_{j=m-3}^{m}\Delta_{j} ; A_{1}, \dots A_{m-4},C)
        \\
        &=\dots
        \\
        &=P_{\rho}(\{\af_{1}\} \times \cap_{j=2}^{m}\Delta_{j} ; A_{1},C)
        \\
        &= P_{\rho} (\cap_{i=1}^{m} \Delta_{i},C).
\end{align*}
(the step ``$=\dots$'' can be made precise with induction). It proves the first part of the lemma. Now note that, for any $i =1,\dots,m$  and any $\Sigma \subset \sigma(C)$, $A_{i}$ and $\chi_{\Sigma}(C)$ are functions of $C$, whereas $(\af_{i},A_{i})$ and $(1,\chi_{\Sigma}(C))$ are objective events, thus, according to postulate~\ref{post: compatibility},
\begin{align*}
        T_{(\af_{i},A_{i})} \circ T_{(1,\chi_{\Sigma}(C))} = T_{(1,\chi_{\Delta_{i} \cap \Sigma}(C))},
\end{align*}
where $\Delta_{i} \equiv f_{i}^{-1}(\af_{i})$. Therefore,
\begin{align*}
    T_{(\af_{m},A_{m})} \circ \dots \circ T_{(\af_{1},A_{1})} &= (T_{(\af_{m},A_{m})} \circ \dots \circ T_{(\af_{2},A_{2})}) \circ (T_{(\af_{1},A_{1})} \circ T_{(\af_{1},A_{1})})
    \\
    &= (T_{(\af_{m},A_{m})} \circ \dots \circ T_{(\af_{2},A_{2})}) \circ T_{(1,\chi_{\Delta_{1}}(C))}
    \\
    &= (T_{(\af_{m},A_{m})} \circ \dots \circ T_{(\af_{3},A_{3})}) \circ T_{(1,\chi_{\cap_{i=1}^{2}\Delta_{i}}(C))}
    \\
    &=\dots
    \\
    &= T_{(1,\chi_{\cap_{i=1}^{m}\Delta_{i}}(C))},
\end{align*}
(again, the argument can be made precise with induction). It completes the proof.
\end{proof}

It is worth making the following corollary explicit.

\begin{corollary}\label{cor: permutingProbability} Let $A_{1},\dots,A_{m}$ be pairwise compatible observables, and let $\pi$ be any permutation of $\{1,\dots,m\}$. Let $\rho$ be any state. Then, for any $\Delta_{1}\times \dots \times \Delta_{m} \subset \prod_{i=1}^{m}\sigma(A_{i})$,
\begin{align}
    p_{\rho}(\Delta_{1} \times \dots \times\Delta_{m};A_{1},\dots,A_{m}) &= \label{eq: permutingProbability}p_{\rho}(\Delta_{\pi(1)}\times\dots\times\Delta_{\pi(m)};A_{\pi(1)},\dots,A_{\pi(m)}),\\
    T_{(\Delta_{1},\dots,\Delta_{m};A_{1},\dots,A_{m})} &=T_{(\Delta_{\pi(1)},\dots,\Delta_{\pi(m)};A_{\pi(1)},\dots,A_{\pi(m)})}.
\end{align}
\end{corollary}
\section{Algebraic aspects of the ``commutative part'' of quantum mechanics}\label{sec: commutativePart}

Postulates~\ref{ax: separability}-\ref{post: commutativity} enable us to introduce  basically every definition and theorem of quantum mechanics that do not depend on the particular way in which quantum systems connect incompatible observables. We have already seen some important examples, such as projections (which actually follow from postulates \ref{ax: separability}-\ref{post: observables}) and Specker's principle. We now focus on results that depend on algebraic operations between compatible observables. To begin with, let's define these operations.
\subsection{Algebraic operations for compatible observables}\label{sec: algebraicOperations}

Let $A,B$ be compatible observables, and let $\conj{A,B}$ be its conjunction (definition~\ref{def: binaryConjunction}). Recall that $\conj{A,B}$ is defined by an injective function $\conj{ \ , \ }: \sigma(A) \times \sigma(B) \rightarrow \mathbb{R}$, and that
\begin{align*}
    \sigma(\conj{A,B})&= \{\conj{\af,\beta} \in \mathbb{R}: (\af,\beta) \in \sigma(A) \times \sigma(B), E_{\af} \circ F_{\beta} \neq 0\},
\end{align*}
where $E_{\af} \equiv \chi_{\{\af\}}(A)$ and $F_{\beta} \equiv \chi_{\{\beta\}}(B)$. Let $h$ be any real function on $\sigma(A) \times \sigma(B)$, and let $h'$ be its ``translation'' to $\sigma(\conj{A,B})$, i.e., $h'(\conj{\af,\beta}) \doteq h(\af,\beta)$ for every $\conj{\af,\beta} \in \sigma(\conj{A,B})$. It means that $h'$ is the only function from $\sigma(\conj{A,B})$ to $\mathbb{R}$ making the following diagram commute.
\begin{center}
    \begin{tikzcd}
        \sigma(A) \times \sigma(B) \arrow[r,"\conj{ \ \text{,} \ }"]\arrow[dr,"h"] & \sigma(\conj{ A, B })\arrow[d,"h'"]\\
        & \mathbb{R}
    \end{tikzcd}
\end{center}

According to postulate~\ref{post: observables}, given any function $h: \sigma(A) \times \sigma(B) \rightarrow \mathbb{R}$ we can define an observable $h(A,B)$ by
\begin{align}
    h(A,B) \doteq h'(\conj{A,B}).
\end{align}
If an observable $C$ is compatible with $A$ and $B$, then it is also compatible with $h(A,B)$ for any function $h$. In fact, lemma~\ref{lemma: compatibleWithConjunction} tells us that $C$ is compatible with $A$ and $B$ iff it is compatible with $\conj{A,B}$, which in turn implies that $C$ is compatible with any function of $\conj{A,B}$.

We are interested in the particular case where $h$ is a binary operation. First, let $h$ be the  addition $\sigma(A) \times \sigma(B) \ni(\af,\beta) \xmapsto{+} \af + \beta \in \mathbb{R}$, and define $A + B \equiv +(A,B) \doteq +'(\conj{A,B})$. For any $\conj{\af,\beta} \in \sigma(\conj{A,B})$, we have $+'(\conj{\af,\beta}) = \af + \beta$. Furthermore, since $A\xleftarrow{\theta_{A}} \conj{A,B}\xrightarrow{\theta_{B}} B$ is a product diagram for $A$ and $B$ (see proposition~\ref{prop: conjunctionAndProduct}), given any cone  $A \xleftarrow{f} C \xrightarrow{g} B$ for $A,B$ we have
\begin{align*}
    A+B &= +'(\conj{A,B}) = +'(\conj{f,g}(C)) = (+' \circ \conj{ \ , \ } \circ (f,g))(C) = (+ \circ (f,g))(C)
    \\
    &= (f+g)(C)
\end{align*}
(recall that $\conj{f,g} = \conj{ \ , \ } \circ (f,g)$, and that the range of $(f,g)$ is $\sigma(\conj{A,B})$). This is illustrated in the following commutative diagram.
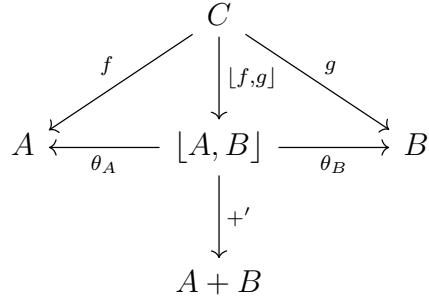
\begin{center}
    \begin{tikzcd}[row sep=large, column sep=large]
        & C\arrow[dl, swap, "f"]\arrow[d, "\conj{f\text{,}g}"]\arrow[dr, "g"] &\\
        A & \conj{A,B}\arrow[l,"\theta_{A}"]\arrow[r,swap, "\theta_{B}"]\arrow[d,"+'"]& B\\
        & A+B& 
    \end{tikzcd}
\end{center}
We just proved that, for any observable $C$ such that $A=f(C)$ and $B=g(C)$, we have $A+B = f(C) + g(C) = (f+g)(C)$, and it is easy to see that $A + B$ is the unique observable satisfying this condition. In fact, suppose that an observable $D$ satisfies $D=(f+g)(C)$ for every cone $A \xleftarrow{f}C\xrightarrow{g} B$. In particular, we have $D=(\theta_{A}+\theta_{B})(\conj{A,B})$. It is straightforward to show that $(\theta_{A}+\theta_{B}) = +'$, thus $D=A + B$. It is important to note that, although $\conj{A,B}$ is unique only up to isomorphism, the observable $A + B$ is unique.  Finally, note that
\begin{align*}
    \sigma(A+B) &= \{\af + \beta: (\af,\beta) \in \sigma(A) \times \sigma(B), E_{\af} \circ F_{\beta} \neq 0\},
\end{align*}
where $E_{\af} \equiv \chi_{\{\af\}}(A)$ and $F_{\beta} \equiv \chi_{\{\beta\}}(B)$. 

Similarly, we can define the (algebraic) product $A \circ B$ of two compatible observables $A,B$ using the mapping $\sigma(A) \times \sigma(B) \ni (\af,\beta) \mapsto \af \cdot \beta \in \mathbb{R}$, and it is analogous to show that $A \circ B$ is the unique observable satisfying $A \circ B = (f \cdot g)(C)$ whenever $A=f(C)$ and $B=g(C)$ (equivalently, for any cone $A \xleftarrow{f} C \xrightarrow{g} B$). The spectrum of $A \circ B$ satisfies
\begin{align*}
    \sigma(A\circ B) &= \{\af \cdot \beta: (\af,\beta) \in \sigma(A) \times \sigma(B), E_{\af} \circ F_{\beta} \neq 0\},
\end{align*}
and the product of projections we defined in section \ref{sec: compatibility} (see definition~\ref{def: productProjections}) is just a particular case of this definition.

To conclude, let $a$ be any real number, and let $A$ be any observable. Let $h_{a}$ be the mapping $\sigma(A) \ni \af \mapsto a\cdot \af$, and define $aA \equiv \af \cdot A \doteq h_{a}(A)$. If $A=f(C)$, then $a \cdot A = (h_{a} \circ f)(C) = (a\cdot f)(C)$. Now let $B_{a}$ be any observable satisfying $B_{a} = (a \cdot f)(C)$ whenever $f(C)=A$. Then $B_{a} = (a \cdot \id_{A})(A) = (h_{a} \circ \id_{A})(A) = h_{a}(A) = a\cdot A$. Also,
\begin{align*}
    \sigma(a\cdot A) &= \{a\cdot \af: \af \in\ \sigma(A)\} \equiv \af \cdot \sigma(A).
\end{align*}

This discussion shows that the following definition is consistent.

\begin{definition}[Algebraic operations]\label{def: algebraicOperations} Let $A,B$ be compatible observables. We denote by $A + B$ the unique observable satisfying $A+B = f(C)+g(C) = (f+g)(C)$ for every cone $A \xleftarrow{f} C \xrightarrow{g} B$. Analogously, $A \circ B$ denotes the unique observable satisfying $A\circ B = f(C)\circ g(C) = (f\cdot g)(C)$ for every cone $A \xleftarrow{f} C \xrightarrow{g} B$. Finally, given any real number $a$ and any observable $A$, we denote by $aA \equiv a \cdot A$ the unique observable such that $a\cdot A= (a\cdot f)(C)$ whenever $A=f(C)$.
\end{definition}
As usual, we say that $A+B$ and $A \circ B$ are the sum and the product respectively of $A$ and $B$, whereas $a \cdot A$ is said to be the scalar multiplication of $a$ and $A$. For simplicity, we eventually write $AB$ rather than $A \circ B$, and $aA$ instead of $a \cdot A$. Finally, note that, by definition, $A+B = B+A$ and $A \circ B = B \circ A$. 

It is important to note that $+$ and $\circ$ do not define binary operations on $\mathcal{O}$, since they act only on compatible observables. They define \textbf{partial operations}, as  the composition of arrows in a category does (see definition~\ref{def: category}).

As we said above, it follows from lemma~\ref{lemma: compatibleWithConjunction} that any observable $C$ which is compatible with $A$ and $B$ is also compatible with $A+B$ and $A \circ B$, which in turn implies that the observables $(A+B) + C$, $(A\circ B) \circ C$, $(A + B ) \circ C$, and so on are well defined. Saying that $A$ and $B$ are compatible and that an observable $C$ is compatible with them is equivalent to saying that $A,B$ and $C$ are pairwise compatible, which in turn, according to theorem \ref{thm: productCompatibility}, is equivalent to saying that there exists an observable $D$ and arrows $f,g,h$ such that $A=f(D)$, $B=g(D)$, $C=h(D)$. Hence, it follows from definition~\ref{def: algebraicOperations} that
\begin{align*}
    (A+B)+C &= (f(D)+g(D))+h(D) = (f+g)(D)+h(D)
    \\
    &= (f+g+h)(D)
    \\
    &=(f+(g+h))(D) = f(D)+(g+h)(D) = A + (B + C).
\end{align*}
More generally, Specker's principle  (corollary~\ref{cor: specker}) ensures that, if $A_{1},\dots,A_{m}$ are pairwise compatible, the finite sum $\sum_{i=1}^{m} A_{i}$ is well defined and, given any cone $C\xrightarrow{f_{i}}A_{i}$, $i=1,\dots,m$, for $A_{1},\dots,A_{m}$, we have
\begin{align}
    \sum_{i=1}A_{i} = \sum_{i=1}^{m} f_{i}(C) = \left(\sum_{i=1}^{m}f_{i} \right)(C).
\end{align}
It is analogous to show that the product $\circ$ is associative, and that, for any cone $C\xrightarrow{f_{i}}A_{i}$, $i=1,\dots,m$, of $A_{1},\dots,A_{m}$, we have
\begin{align}
    \prod_{i=1}^{m}A_{i} = \prod_{i=1}^{m}f_{i}(C) = \left(\prod_{i=1}^{m}f_{i}\right)(C).
\end{align}

Let $C$ be an observable, and let $\mathfrak{C}^{\ast}(C)$ be the set of all functions of $C$. $\mathfrak{C}^{\ast}(C)$ canonically inherits the algebraic structure of the collection of real functions on $\sigma(C)$, which is a commutative algebra, so, thanks to Specker's principle (corollary~\ref{cor: specker}) the  operations defined above satisfy all properties that operations between functions satisfy, i.e., they are associative, commutative, the product distributes over addition, and so on. We summarize some of these results in the following proposition.
\begin{proposition}\label{prop: algebraicProperties} The partial operations $+,\ci$ and the scalar multiplication given by definition \ref{def: algebraicOperations} satisfy the following conditions.
\begin{itemize}
    \item[(a)] The addition $+$ and the product $\ci$ are commutative and associative. Furthermore, the zero  and the unit, namely the projections $0$ and $\mathds{1}$ (definition~\ref{def: zeroAndUnit}), satisfy
    \begin{align}
        0 + A &= A,\\
        0 \ci A &= 0,\\
        \mathds{1} \ci A &= A.\\
    \end{align}
    for every observable $A$.
    \item[(b)] The product and the scalar product distribute over addition, i.e., if $A,B,C$ are pairwise compatible observables and $a$ is a real number,
    \begin{align}
        A \ci (B + C) &= (A \ci B) + (A \ci C),\\
        a \cdot (B + C) &= (a \cdot B) + (a \cdot C). 
    \end{align}
    \item[(c)] If $A,B$ are compatible observables and $a,b$ are real numbers,
    \begin{align}
        a \cdot (A \ci B) &= (a \cdot A) \ci B = A \ci (a \cdot B),\\
        (ab)\cdot A &= a\cdot (bA),\\
        (a+b) \ci A &= (a \cdot A) + (b \cdot A). 
    \end{align}
    \item[(d)] Given any observable $A$,
    \begin{align}
        0 \cdot A &= 0,\\
        1 \cdot A &= A
    \end{align}
\end{itemize}
\end{proposition}

To conclude this section, let's make it clear that definitions \ref{def: functionalRelation} and \ref{def: algebraicOperations} are in agreement.
\begin{proposition}\label{prop: operationsAndFunctions} Let $A$ be an observable, and let $p$ is a polynomial on its spectrum $\sigma(A)$, i.e., there are real numbers $a_{0},\dots,a_{m}$ such that  $p(\af) = \sum_{k=0}^{m}a_{k} \af^{k}$ for each $\af \in \sigma(A)$. Let $p(A) \in \mathcal{O}$ be the function of $A$ via $p$ given by definition \ref{def: functionalRelation}. Then
\begin{align}
    p(A) = \sum_{k=0}^{m} a_{k}A^{k},
\end{align}
where $A^{k} \equiv \prod_{i=1}^{k} A$ for each $k=1,\dots,m$ and $A^{0} \doteq \mathds{1}$.
\end{proposition}
\begin{proof}
    Let $\text{id}\equiv \text{id}_{A}$ be the identity function $\sigma(A) \ni \af \mapsto \af \in \sigma(A)$. For each $k=1,\dots,m$, define  $\text{id}^{k} \doteq\prod_{i=1}^{k} \text{id}$, and let $\text{id}^{0}$ be the constant function $\sigma(A) \ni \af \mapsto 1 \in \mathbb{R}$, i.e., $\text{id}^{0} = \chi_{\sigma(A)}$.  We have $p = \sum_{k=0}^{m} a_{k} \id^{k}$, and therefore
    \begin{align}
        p(A) &= \left(\sum_{k=0}^{m} a_{k} \id^{k}\right)(A) = \sum_{k=0}^{m}a_{k}\id^{k}(A) = a_{0}\chi_{\sigma(A)}(A) + \sum_{k=1}^{m}a_{k}\left\{\left(\prod_{i=1}^{k}\id\right)(A)\right\}
        \\
        &= a_{0}\mathds{1} +  \sum_{k=1}^{m}a_{k}A^{k} = \sum_{k=0}^{m}a_{k}A^{k}.
    \end{align}
\end{proof}

Proposition~\ref{prop: operationsAndFunctions} shows that there is no ambiguity in writing, for instance, $A^{2}$. It does not matter if by $A^{2}$ we mean $A \circ A$ or $g(A)$, where $g(\af) \doteq \af^{2}$ for every $\af \in \sigma(A)$; in both cases, the observable is the same.
\subsection{Spectral theory, spectral theorem and the functional calculus}\label{sec: spectralTheory}

Recall that a projection is an observable $E$ satisfying $E=\chi_{\Delta}(C)$ for some observable $C$ and some $\Delta \subset \sigma(C)$ (see definition~\ref{def: projection}). Recall also that, as we have shown in lemma \ref{lemma: projection}, there is no ambiguity in this definition, by which we mean that, if a projection $E$ satisfies $E=g(D)$ for some observable $D$, then $g$ is the characteristic function of some $\Sigma \subset \sigma(A)$, and consequently $E=\chi_{\Sigma}(D)$. We saw in section~\ref{sec: projections} that a projection $E$ represents the equivalence class of all observable events $(\Delta,A)$ satisfying $\chi_{\Delta}(A)=E$, and we proved that each state $\rho$ defines a mapping $\Exp{ \ \cdot \ }$ on the set of projections $\mathcal{P}$, which we called the expectation defined by $\rho$ (definition~\ref{def: expectationProjection}). Now that we have algebraic operations between compatible observables, we can explore projections in more depth. The very definition of projection, for instance, can be given in  purely algebraic terms:
\begin{proposition}[Idempotence of projections]\label{prop: projectionAndSquare}
    Let $E$ be an observable. The following claims are equivalent.
    \begin{itemize}
        \item[(a)] $E$ is a projection
        \item[(b)] $\sigma(E) \subseteq \{0,1\}$
        \item[(c)] $E$ is \textbf{idempotent}, that is,  $E^{2} = E$ (see proposition~\ref{prop: operationsAndFunctions}).
\end{itemize}
\end{proposition}
\begin{proof}
    In lemma~\ref{lemma: projectionSpectrum} we proved that $(a)$ implies $(b)$. Item $(a)$ implies item $(c)$ because, for any projection $E = \chi_{\Delta}(C)$, we have $E^{2} = (\chi_{\Delta} \cdot \chi_{\Delta})(C) = \chi_{\Delta \cap \Delta}(C) = \chi_{\Delta}(C) = E$ (see proposition~\ref{prop: operationsAndFunctions}). Suppose now that $E$ is an observable satisfying $\sigma(E) \subseteq \{0,1\}$, and let $C \xrightarrow{f} E$ be any arrow whose codomain is $E$. Then $ f(\sigma(C)) = \sigma(f(C)) = \sigma(E) \subseteq \{0,1\}$, which means that $f = \chi_{\Delta}(C)$ for $\Delta \equiv f^{-1}(\{1\})$. It shows that $E$ is a projection, thus $(b)$ implies $(a)$ and, consequently, $(c)$. Finally, suppose that $E^{2} = E$, and let $C \xrightarrow{f} E$ be any arrow whose codomain is $E$. Then $f(C) = E = E^{2} = (f \cdot f)(C) = f^{2}(C)$, and, according to lemma~\ref{lemma: functionalRelation}, it implies that $f^{2} = f$. We know that, for any $x \in \mathbb{R}$, $x^{2} = x$ iff $x \in \{0,1\}$, thus, for any $\gamma \in \sigma(C)$, $f^{2}(\gamma) = f(\gamma)$ iff $f(\gamma) \in \{0,1\}$. It means that $f(\sigma(C)) \subseteq \{0,1\}$, hence $f = \chi_{\Delta}(C)$, where $\Delta \equiv f^{-1}(\{1\})$. Then $E = \chi_{\Delta}(C)$, and therefore $E$ is a projection.  
\end{proof}

Another important consequence of the existence of algebraic operations between compatible observables is the existence of an orthogonality relation in the set of projections.

\begin{definition}[Orthogonality]\label{def: orthononality} Two projections $E,F$ are said to be orthogonal, denoted $E \perp F$, if they are compatible and satisfy $E \circ F = 0$.
\end{definition}

Note that every projection is orthogonal to the zero operator (definition~\ref{def: zeroAndUnit}). On the other hand, only the zero operator is orthogonal to the unit (definition~\ref{def: zeroAndUnit} again).

As usual, we write $A - B$ as a shorthand for $A + (-B)$, where $-B$ in turn is a shorthand for $(-1)\cdot B$. We can now introduce the following definition. 
\begin{definition}[Orthocomplement]\label{def: orthocomplement} The orthocomplement of a projection $E$ is the projection
    \begin{align}
        E^{\perp} \doteq \mathds{1}- E. 
    \end{align}
\end{definition}
Let $E$ be any projection. Then $E^{\perp} \circ E^{\perp} = (\mathds{1} - E)\circ (\mathds{1} - E) = \mathds{1}^{2} - E - E + E^{2} = \mathds{1} - E$ (see proposition~\ref{prop: operationsAndFunctions}), so $E^{\perp}$ is indeed a projection. In the following lemma, we summarize important properties that are satisfied by the orthocomplement.

\begin{lemma}[Orthocomplement]\label{lemma: orthocomplement} Let $\mathcal{P}$ be the set of projections of a system $\mathfrak{S}$, and let $\mathcal{P} \ni E \mapsto E^{\perp} \in \mathcal{P}$ be the mapping that assigns each projection to its orthocomplement. Then the following conditions are satisfied.
\begin{itemize}
    \item[(a)] $E$ and $E^{\perp}$ are compatible.
    \item[(b)] If $E=\chi_{\Delta}(A)$ for some observable $A$, $E^{\perp} = \chi_{\sigma(A)\backslash \Delta}(A)$.
    \item[(c)] $E$ and $E^{\perp}$ are orthogonal.
    \item[(d)] The mapping $\mathcal{P} \ni E \mapsto E^{\perp} \in \mathcal{P}$ is an involution, i.e., for any projection $E$, $(E^{\perp})^{\perp} =E$.
    \item[(e)] $\mathds{1}^{\perp} = 0$ and, consequently, $0^{\perp} = \mathds{1}$.
    \item[(f)] For any state $\rho$, $
        p_{\rho}(0,E) = p_{\rho}(1,E^{\perp}).$
    \item[(g)] For any state $\rho$, $\Exp{E^{\perp}} = 1-\Exp{E}.$
\end{itemize}
\end{lemma}
\begin{proof}
    Let $E$ be a projection. For any observable $A$ such that $E = \chi_{\Delta}(A)$, we have
    \begin{align*}
        E^{\perp} = \chi_{\sigma(A)}(A) - \chi_{\Delta}(A) = (\chi_{\sigma(A)} -\chi_{\Delta})(A) = \chi_{\sigma(A)\backslash \Delta}(A).
    \end{align*}
    It proves both items $(a)$ and $(b)$. Also, if $E = \chi_{\Delta}(A)$, then $E\circ E^{\perp} = \chi_{\Delta \cap (\sigma(A)\backslash\Delta)}(A) = \chi_{\emptyset}(A) = 0$, so item $(c)$ follows (equivalently, item $(c)$ follows from the fact that $E \circ E^{\perp} = E\circ (\mathds{1} - E) = E - E = 0$). For any projection $E$ we have  $(E^{\perp})^{\perp} = \mathds{1} - E^{\perp} = \mathds{1} - (\mathds{1} - E) = E$, thus item $(d)$ is satisfied. The validity of item $(e)$ is straightforward, since $\mathds{1}^{\perp} = \mathds{1} - \mathds{1} = 0$ and $0^{\perp} = 1-0 = 1$. Finally, items $(f)$ and $(g)$ follows from item $(b)$. In fact, let $E=\chi_{\Delta}(A)$ be any projection. Then, for any state $\rho$,
    \begin{align*}
        p_{\rho}(0;E) = 1-p_{\rho}(1;E) = 1-P_{\rho}(\Delta;A) = P_{\rho}(\sigma(A)\backslash \Delta;A) = p_{\rho}(1,E^{\perp}),
    \end{align*}
    which in turn is equivalent to saying that $\Exp{E^{\perp}} = 1-\Exp{E}$.
\end{proof}

If $E_{1},\dots,E_{m}$ are pairwise compatible projections, $\prod_{i=1}^{m}E_{i}$ is also a projection, and, for any cone $\chi_{\Delta_{i}}: A \rightarrow E_{i}$, $i=1,\dots,m$, of $E_{1},\dots,E_{m}$, we have
\begin{align}
    \prod_{i=1}^{m}E_{i} = \left(\prod_{i=1}^{m}\chi_{\Delta_{i}}\right)(A) = \chi_{\cap_{i=1}^{m}\Delta_{i}}(A).
\end{align}
Together with proposition~\ref{prop: jointPushforward}, this equality implies the following lemma.
\begin{lemma}\label{lemma: projectionJointProbability}
    Let $A_{1},\dots,A_{m}$ be pairwise compatible observables. Then, for any $\Delta_{1}\times \dots \times \Delta_{m} \subset \prod_{i=1}^{m} \sigma(A_{i})$ and any state $\rho$,
    \begin{align}
        P_{\rho}(\Delta_{1}\times\dots\times\Delta_{m};A_{1},\dots,A_{m})&= \Exp{\prod_{i=1}^{m}E_{\Delta_{i}}},
    \end{align}
    where $E_{\Delta_{i}} \equiv \chi_{\{\Delta_{i}\}}(A_{i})$ for every $i$.
\end{lemma}

The sum of two compatible projections $E$, $F$ is itself a projection if and only if $E$ and $F$ are orthogonal. In fact, for any cone $E \xleftarrow{\chi_{\Delta}} A \xrightarrow{\chi_{\Sigma}} F$ we have
\begin{align*}
    E + F = (\chi_{\Delta} + \chi_{\Sigma})(A) = \left(\chi_{\Delta \backslash\Sigma} + 2\chi_{\Delta \cap \Sigma} + \chi_{\Sigma \backslash \Delta}\right)(A).
\end{align*}
If $\Delta \cap \Sigma \neq \oldemptyset$, $2 \in \sigma(E+F)$, so $E + F$ is not a projection. On the other hand, if $\Delta \cap \Sigma = \oldemptyset$, we have $\Delta = \Delta \backslash\Sigma$, $\Sigma = \Sigma \backslash \Delta$ and
\begin{align*}
    E + F = \chi_{\Delta \cup \Sigma}(A).
\end{align*}
Similarly, if the projections $E_{1},\dots E_{m}$ are pairwise orthogonal, $\sum_{i=1}^{m}E_{i}$ is a projection. For any cone $\chi_{\Delta_{i}}: A \rightarrow E_{i}$, $i=1,\dots,m$, of $E_{1},\dots,E_{m}$ we have $\Delta_{i} \cap \Delta_{j} =\oldemptyset$ whenever $i\neq j$ and
\begin{align}
    \sum_{i=1}^{m} E_{i} =\label{eq: sumOfProjections} \left(\sum_{i=1}^{m}\chi_{\Delta_{i}}\right)(A) = \chi_{\cup_{i=1}^{m}\Delta_{i}}(A).
\end{align}

For this reason, states are finitely additive:

\begin{lemma}[Finite additivity]\label{lemma: finiteAdditivity} Let $\rho$ be any state. Then, for any set of pairwise orthogonal projections $E_{1},\dots,E_{m}$,
        \begin{align}
            \Exp{\sum_{i=1}^{m}E_{i}} = \sum_{i=1}^{m}\Exp{E_{i}}.
        \end{align}
    \end{lemma}
    \begin{proof}
    Let $E_{1},\dots,E_{m}$ be pairwise compatible projections, and let $\chi_{\Delta_{i}}:A \rightarrow E_{i}$, $i=1,\dots,m$, be a cone for them. We have seen that $\Delta_{i} \cap \Delta_{j} = \oldemptyset$ whenever $i \neq j$, thus it follows from equation~\ref{eq: sumOfProjections} and from the additivity of probability measures that, that for any state $\rho$,
    \begin{align*}
        \Exp{\sum_{i=1}^{m}E_{i}} &= P_{\rho}(\cup_{i=1}^{m} \Delta_{i};A) = \sum_{i=1}^{m}P_{\rho}(\Delta_{i};A) = \sum_{i=1}^{m}\Exp{E_{i}}.
    \end{align*}
\end{proof}

The exclusivity principle, mentioned in section \ref{sec: specker}, corresponds to the following result.

\begin{corollary}[Exclusivity principle]\label{cor: exclusivityPrinciple} Let $\rho$ be any state. Then, for any set of pairwise orthogonal projections $E_{1},\dots,E_{m}$,
\begin{align}
    \sum_{i=1}^{m}\Exp{E_{i}} \leq 1.
\end{align}
\end{corollary}

Recall that, according to definition~\ref{def: projectionOfEigenvalue}, the projection associated with the eigenvalue $\af$ of $A$ is the projection $E_{\af}$ associated with the event $(\af,A)$, namely  $\chi_{\{\af\}}(A)$, and it immediately follows from this definition that projections associated with distinct eigenvalues of an observable are pairwise orthogonal: if $\af,\af' \in \sigma(A)$ are different, we have $\{\af\} \cap \{\af'\} = \oldemptyset$, and therefore $E_{\af} \circ E_{\af'} = 0$. Among other things, this fact enables us to characterize the trace of a projection in a purely algebraic way. In fact, let $E$ be a projection, and let $A$ be any nondegenerate observable satisfying $E=\chi_{\Delta}(A)$ for some $\Delta \subset \sigma(A)$ (according to postulate~\ref{post: observables} and lemma~\ref{lemma: projection}, this nondegenerate observable always exists). Recall the trace of $E$  is the number $\Tra{E} \doteq \vert \Delta \vert$ (see definition~\ref{def: traceProjection}), and that this number does not depend on the choice of $A$, i.e., if $E=\chi_{\Sigma}(B)$ for some other nondegenerate observable $B$, then $\vert \Delta \vert = \vert \Sigma \vert$. Also, recall that the rank of a projection is its trace and that by a rank-$k$ projection we mean a projection whose trace is equal to $k$. We have seen that, in a $n$-dimensional system, the trace of a projection is a natural number between $1$ and $n$, and that the intrinsic expectation $\langle E \rangle_{\emptyset}$ of $E$, i.e., its expectation w.r.t. the completely mixed state (see definition~\ref{def: intrinsicProbability}), satisfies
\begin{align*}
    \langle E \rangle_{\emptyset} = \frac{\Tra{E}}{n}
\end{align*}
(see lemma~\ref{lemma: traceSimple}). It immediately follows from this equation that projections associated with eigenvalues of nondegenerate observables have rank $1$. On the other hand, let $E$ be a rank-$1$ projection, and let $A$ be any nondegenerate observable such that $E = \chi_{\Delta}(A)$. Then $1 = \vert \Delta \vert$, which implies that $\Delta=\{\af\}$ for some $\af \in \sigma(A)$, which in turn means that $E$ is the projection associated with the eigenvalue $\af$ of $A$. Hence, rank-$1$ projections and projections associated with eigenvalues of nondegenerate observables are equivalent concepts. It enables us to prove the following result: 

\begin{proposition}\label{prop: nondegenerateDecomposition}
A projection $E$ has rank $k$ if and only if there are $k$ pairwise orthogonal rank-$1$ projections $E_{1},\dots,E_{k}$ satisfying
\begin{align*}
        E = \sum_{i=1}^{k}E_{i}.
\end{align*}
\end{proposition}
\begin{proof}
Let $E$ be a rank-$k$ projection, and let $A$ be a nondegenerate observable such that $E = \chi_{\Delta}(A)$. We know that $\vert \Delta \vert =\Tra{E}= k$, so we can write $\Delta \equiv \{\af_{1},\dots,\af_{k}\}$, where $\af_{i} = \af_{j}$ iff $i=j$. For each $i=1,\dots,k$, let $E_{i}$ be the projection $\chi_{\{\af_{i}\}}(A)$. As we have seen, $E_{i}\perp E_{j}$ whenever $i \neq j$, so $E = \chi_{\Delta}(A) = \left(\sum_{i=1}^{k}\chi_{\{\af_{i}\}}\right)(A) =\sum_{1}^{k}\chi_{\{\af_{i}\}}(A) =\sum_{i=1}^{k}E_{i}$, which proves that any rank-$1$ projection can be written as a sum of $k$ pairwise orthogonal rank-$1$ projection. On the other hand, let $E$ be a projection, and suppose that $E = \sum_{i=1}^{k}E_{i}$ for pairwise orthogonal rank-$1$ projections  $E_{1},\dots,E_{k}$. The projections $E_{1},\dots,E_{k}$ are pairwise orthogonal and, in particular, pairwise compatible, so it follows from corollary~\ref{cor: nondegenerateSpecker} that there is a nondegenerate cone for them, i.e., a cone $C \xrightarrow{\chi_{\Delta_{i}}} E_{i}$, $i=1,\dots,m$, such that $C$ is nondegenerate. As we have discussed, we must have, for each $i$, $\Delta_{i}=\{\af_{i}\}$ for some $\af_{i} \in \sigma(A)$, because $E_{i}$ has rank $1$, and  $\af_{i} \neq \af_{j}$ whenever $i \neq j$, since $E_{i} \perp E_{j}$. Finally, $E = \sum_{i=1}^{k}E_{i} =\sum_{i=1}^{k}\chi_{\{\af_{i}\}}(A) = \left(\sum_{i}^{k} \chi_{\{\af_{i}\}}\right)(A) = \chi_{\Delta}(A)$, where $\Delta \doteq \{\af_{i}: i=1,\dots,k\}$, and since $A$ is nondegenerate, $\Tra{E} = \Tra{\chi_{\Delta}(A)} = \vert \Delta \vert = k$, which completes the proof.
\end{proof}

To conclude our discussion about projections, it is worth emphasizing the following result.

\begin{lemma}[Dimension of $\mathfrak{S}$]\label{lemma: dimension} Let $\mathfrak{S}$ be a $n$-dimensional system. Then any set of pairwise orthogonal projections in $\mathfrak{S}$ contains at most $n$ elements. Furthermore, there is at least one set $\mathcal{E}$ of pairwise orthogonal projections satisfying $\vert \mathcal{E}\vert = n$.    
\end{lemma}
\begin{proof}
    Let $\mathfrak{S}$ be a $n$-dimensional system, and let $C$ be any nondegenerate observable in $\mathfrak{S}$. The projections $E_{i} \equiv \chi_{\{i\}}(C)$, $i \in \sigma(C)$, are pairwise orthogonal by definition, and the fact that $C$ is nondegenerate implies that  $\vert \sigma(C)\vert = n$, therefore $\{E_{i}: i \in \sigma(C)\}$ is a set of $n$ pairwise orthogonal projections. On the other hand, let $E_{1},\dots,E_{m}$ be pairwise orthogonal projections in some $n$-dimensional system, and let $\chi_{\Delta_{i}}: C \rightarrow E_{i}$, $i=1,\dots,m$, be a nondegenerate cone for them (see corollary \ref{cor: nondegenerateSpecker}). We know that, for each pair $i,j$ of distinct coefficients, $\Delta_{i} \cap \Delta_{j} = \oldemptyset$, therefore $m \leq \vert \sigma(C)\vert =n$.
\end{proof}

Now let's turn our attention to observables in general. Let $A$ be an observable, and, for each $\af \in \sigma(A)$, write $E_{\af} \equiv \chi_{\{\af\}}(A)$, i.e., $E_{\af}$ is the projection associated with the eigenvalue $\af$ of $A$. Let $\id\equiv \id_{A}$ be the identity function on $\sigma(A)$. Then
\begin{align*}
    A &= \text{id}(A) = \left(\sum_{\af \in \sigma(A)} \af\chi_{\{\af\}}\right)(A) = \sum_{\af \in \sigma(A)}\af \chi_{\{\af\}}(A) =\sum_{\af \in \sigma(A)}\af E_{\af}.
\end{align*}
Furthermore, 
\begin{align*}
    \mathds{1} = \chi_{\sigma(A)}(A) = \left(\sum_{\af \in \sigma(A)} \chi_{\{\af\}}\right)(A) = \sum_{\af \in \sigma(A)}\chi_{\af}(A)=\sum_{\af \in \sigma(A)}E_{\af}.
\end{align*}
The set of projections $\{E_{\af}: \af \in \sigma(A)\}$ is what we call a partition of the unit:

\begin{definition}[Partition of the unit]\label{def: partitionOfUnit} A set $\{E_{1},\dots,E_{m}\}$ of nonzero  pairwise orthogonal projections is said to be a partition of the unit if it sums to one, i.e., if
    \begin{align*}
        \sum_{i=1}^{m}E_{i} = \mathds{1}.
\end{align*}
\end{definition}

The partition of the unit induced by the eigenvalues of $A$ determines its spectral decomposition:
\begin{definition}[Spectral decomposition]\label{def: spectralDecomposition} Let $A$ be an observable, and let $\{E_{\af}: \af \in \sigma(A)\}$ be the partition of the unit defined by $A$, i.e., $E_{\af} \equiv \chi_{\{\af\}}(A)$ for any $\af \in \sigma(A)$. The spectral decomposition of $A$ corresponds to the following equation.
    \begin{align}
        A = \sum_{\af \in \sigma(A)} \af E_{\af}.
    \end{align}
\end{definition}
Note that the spectral decomposition of a projection $0 \neq E \neq \mathds{1}$ is
    \begin{align}
        E = 1 \cdot E + 0 \cdot E^{\perp}.
    \end{align}
For $0$ and $\mathds{1}$ (definition~\ref{def: zeroAndUnit}) we have $0 = 0 \cdot \mathds{1}$ and $\mathds{1} = 1 \cdot \mathds{1}$.

The spectral decomposition of $A$ is the unique way of writing $A$ as a linear combination of pairwise orthogonal projections with distinct coefficients:

\begin{theorem}[Spectral theorem]\label{thm: spectralTheorem} Let $A$ be an observable, and suppose that, for some partition of the unit $\{F_{1},\dots,F_{m}\}$ and some set of pairwise distinct real numbers $\{\af_{1},\dots,\af_{m}\}$ we have
    \begin{align}
        A =\label{eq: decompositionTheorem} \sum_{i=1}^{m}\af_{i}F_{i}.
    \end{align}
Then equation~\ref{eq: decompositionTheorem} is the spectral decomposition of $A$, i.e., $\{\af_{1},\dots,\af_{m}\}=\sigma(A)$ and $\{F_{i},\dots,F_{m}\}$ is the partition of the unit defined by $A$.
\end{theorem}
\begin{proof}
    Let $\{F_{1},\dots,F_{m}\}$ be a partition of the unit, and let $\chi_{\Delta_{i}}: C \rightarrow F_{i}$, $i=1,\dots,m$, be a cone for $F_{1},\dots,F_{m}$. Let $\af_{1},\dots,\af_{m}$ be pairwise distinct real numbers. We have
    \begin{align}
        \sum_{i=1}^{m}\af_{i}F_{i} &= \sum_{i=1}^{m}\af_{i}\chi_{\Delta_{i}}(C) = \left(\sum_{i=1}^{m}\af_{i}\chi_{\Delta_{i}}\right)(C) = f(C),
    \end{align}
    where $f \equiv \sum_{i=1}^{m}\af_{i}\chi_{\Delta_{i}}$. The set $\{F_{1},\dots,F_{m}\}$ is a partition of the unit, so $\{\Delta_{1},\dots,\Delta_{m}\}$ is a partition of $\sigma(C)$, which means that $\Delta_{i} \cap \Delta_{j} = \oldemptyset$ if $i \neq j$ and $\cup_{i=1}^{m}\Delta_{i} = \sigma(C)$, thus $\sigma(f(C)) = f(\sigma(C)) = \{\af_{1},\dots,\af_{m}\}$. Also, we have $\Delta_{i} = f^{-1}(\af_{i})$ for each $i$. Now assume that $A = \sum_{i=1}^{m} \af_{i}F_{i}$, which is equivalent to saying that $A=f(C)$. Then $\sigma(A) = \sigma(f(C)) = \{\af_{1},\dots,\af_{m}\}$. Furthermore,  for any $\af_{i} \in \sigma(A)$ we have, according to lemma~\ref{lemma: KSdefinitionProjections},
    \begin{align}
        \chi_{\{\af_{i}\}}(A) = \chi_{\{\af_{i}\}}(f(C)) = \chi_{f^{-1}(\af_{i})}(C) = \chi_{\Delta_{i}}(C) = F_{i},
    \end{align}
    which completes the proof.
\end{proof}

The following result is important.
\begin{proposition}[Functional calculus]\label{prop: functionalCalculus} Let $A$ be an observable, and let $A = \sum_{\af \in \sigma(A)}\af E_{\af}$ be its spectral decomposition. Then, for any function $f(A)$ of $A$,
    \begin{align}
        f(A) = \sum_{\af \in \sigma(A)} f(\af) E_{\af}.
    \end{align}
    \end{proposition}
\begin{proof}
    We know that $\sigma(f(A)) = f(\sigma(A))$. Now, according to definition~\ref{def: spectralDecomposition} and lemma~\ref{lemma: KSdefinitionProjections},
    \begin{align}
        f(A) &= \sum_{\beta \in \sigma(f(A))} \beta \chi_{\{\beta\}}(f(A)) = \sum_{\beta \in \sigma(f(A))} \beta \chi_{f^{-1}(\beta)}(A) =  \sum_{\beta \in \sigma(f(A))} \sum_{\af \in f^{-1}(\beta)}\beta \chi_{\{\af\}}(A)
        \\
        &= \sum_{\beta \in \sigma(f(A))} \sum_{\af \in f^{-1}(\beta)}\beta E_{\af} = \sum_{\af \in \sigma(A)} f(\af) E_{\af}.
    \end{align}
\end{proof}

Let $A,B$ be compatible observables, and let $\conj{A,B}$ be their conjunction (definition \ref{def: binaryConjunction}). We have seen in lemma~\ref{lemma: conjunction} that, for any $\conj{\af,\beta} \in \sigma(\conj{A,B})$, $E_{\conj{\af,\beta}} \equiv \chi_{\{\conj{\af,\beta}\}}(\conj{A,B}) = E_{\af} \circ F_{\beta}$. Also, the same lemma says that $\sigma(\conj{A,B}) = \{\conj{\af,\beta}: (\af,\beta) \in \sigma(A) \times \sigma(B),E_{\af} \circ F_{\beta} \neq 0\}$, thus
\begin{align}
    \conj{A,B} &= \sum_{\af\in \sigma(A)}\sum_{\beta \in \sigma(B)} \conj{\af,\beta} E_{\af} \circ F_{\beta}. 
\end{align}
Hence, the following equalities easily follow from proposition~\ref{prop: functionalCalculus}:
\begin{align}
    A+B &= \sum_{\af\in \sigma(A)}\sum_{\beta \in \sigma(B)} (\af + \beta)E_{\af} \circ F_{\beta},
    \\
    A\circ B &= \sum_{\af\in \sigma(A)}\sum_{\beta \in \sigma(B)} (\af \cdot \beta)E_{\af} \circ F_{\beta}.
\end{align}
\subsection{States as functionals}\label{sec: StatesAsFunctionals}

According to the spectral theorem (theorem \ref{thm: spectralTheorem}), each observable $A$ can be written, in a unique way, as a linear combination of pairwise orthogonal projections with distinct coefficients. This result enables us to extend the state $\Exp{ \ \cdot \ }$ (more precisely, the expectation associated with $\rho$) from the set of projections $\mathcal{P}$ to the set of observables $\mathcal{O}$:

\begin{definition}[Expectation]\label{def: Expectation} Let $A$ be an observable, and let $A = \sum_{\af \in \sigma(A)}\af E_{\af}$ its spectral decomposition (definition~\ref{def: spectralDecomposition}). Given a state $\rho$, the expectation (or expected value) of $A$ with respect to $\rho$ is defined by
    \begin{align}
        \Exp{A} \doteq \sum_{\af \in \sigma(A)} \af \Exp{E_{\af}}.
    \end{align}
\end{definition}

It is easy to see that, if $A$ is a projection, then definition~\ref{def: Expectation} coincides with definition~\ref{def: expectationProjection}. In fact, we have seen that the spectral decomposition of a projection $0\neq E \neq \mathds{1}$ is $E = 1 \cdot E + 0 \cdot E^{\perp}$, whereas the spectral decompositions of $0$ and $\mathds{1}$ are, respectively, $0 = 0 \cdot \mathds{1}$ and $\mathds{1} = 1 \cdot \mathds{1}$. Denote the expectation given in definition~\ref{def: Expectation} by $\widetilde{\Exp{ \ \cdot \ }}$, for the time being. For any state $\rho$, $\widetilde{\Exp{0}} = 0 \Exp{\mathds{1}} = 0 \cdot 1 = \Exp{0} $, $\widetilde{\Exp{\mathds{1}}} = 1\cdot \Exp{\mathds{1}} = \Exp{\mathds{1}}$, whereas, for any projection $E$ other than $0,\mathds{1}$, $\widetilde{\Exp{E}} = 1 \cdot \Exp{E} + 0 \cdot \Exp{E^{\perp}} = \Exp{E}$, which completes the proof.

It easily follows from postulate~\ref{ax: subjectiveUpdate} and definition \ref{def: expectationProjection}  that, if $\Exp{ \ \cdot \ }$ is restricted to the set of projections $\mathcal{P}$, as in definition~\ref{def: expectationProjection}, then the mapping $\rho \mapsto \Exp{ \ \cdot \ }$ is convex. That is, if $\rho_{1},\dots,\rho_{m}$ are states and $a_{1},\dots,a_{m}$ are non-negative real numbers satisfying $\sum_{i=1}^{m}\af_{i} =1$, then, for any projection $E$,
\begin{align}
    \Exp{E} =\label{eq: preExpectationConvex} \sum_{i=1}^{m}a_{i}\langle E \rangle_{\rho_{i}},
\end{align}
where $\rho \doteq \sum_{i=1}^{m} a_{i} \rho_{i}$. The same result is valid for definition \ref{def: Expectation}:
\begin{proposition}\label{prop: expectationConvex} The mapping $\rho \mapsto \Exp{ \ \cdot \ }$  induced by definition~\ref{def: Expectation} is convex. That is, if $\rho=\sum_{i=1}^{m} a_{i}\rho_{i}$ is a convex decomposition of a state $\rho$ in the convex set of states $\mathcal{S}$, then, for any observable $B$,
\begin{align}
    \Exp{B} = \sum_{i=1}^{m}a_{i}\langle B \rangle_{\rho_{i}},
\end{align}
\end{proposition}
\begin{proof}
    Let $B$ be any observable, and let $B=\sum_{\beta \in \sigma(B)}\beta F_{\beta}$ be its spectral decomposition. Let $\rho=\sum_{i=1}^{m} a_{i}\rho_{i}$ be a convex decomposition of a state $\rho$. Then, according to equation~\ref{eq: preExpectationConvex} and definition~\ref{def: Expectation},
    \begin{align*}
        \Exp{B} &= \sum_{\beta \in \sigma(B)}\beta \Exp{F_{\beta}} = \sum_{\beta \in \sigma(B)}\beta \sum_{i=1}^{m}a_{i}\langle F_{\beta}\rangle_{\rho_{i}} = \sum_{i=1}^{m} a_{i} \sum_{\beta \in \sigma(B)} \beta \langle F_{\beta} \rangle_{\rho_{i}} = \sum_{i=1}^{m}a_{i} \langle B\rangle_{\rho_{i}}.
    \end{align*}
\end{proof}

It is important to note that we can evaluate the expected value of an observable $A$ using any decomposition of $A$ in terms of pairwise orthogonal projections:
\begin{proposition}[Expectation and decompositions]\label{prop: expectedValue}
    Let $A$ be an observable and $F_{1},\dots F_{m}$ be pairwise orthogonal projections such that $A = \sum_{i=1}^{m} \beta_{i}F_{i}$ for some real numbers $\beta_{i},\dots,\beta_{m}$. Then, for any state $\rho$,
    \begin{align}
        \Exp{A} = \sum_{i=1}^{m} \beta_{j} \Exp{F_{i}}.
    \end{align}
\end{proposition}
\begin{proof}
     Suppose, for the sake of simplicity, that $\{F_{1},\dots F_{m}\}$ is a partition of the unit (otherwise define $F_{0} \doteq \mathds{1} - \sum_{i=1}^{m}F_{i}$, $\beta_{0} \doteq 0$, and consider the partition of the unit $\{F_{0},\dots,F_{m}\}$, which satisfies $A = \sum_{i=0}^{m} \beta_{i}F_{i}$). Let $\chi_{\Delta_{i}}:C \rightarrow F_{i}$, $i=1,\dots,m$, be a cone for $F_{1},\dots,F_{m}$, and define $g \doteq \sum_{i=1}^{m} \beta_{i} \chi_{\Delta_{i}}$. We have $F_{i} = \chi_{\Delta_{i}}(C)$ for every $i$, and
     \begin{align*}
         A &= \sum_{i=1}^{m} \beta_{i}\chi_{\Delta_{i}}(C) =\left( \sum_{i=1}^{m}\beta_{i}\chi_{\Delta_{i}}\right)(C) = g(C).
     \end{align*}
     It implies that $\sigma(A)=g(\sigma(C))$, and since  $\{F_{1},\dots F_{m}\}$ is a partition of the unit, it follows that $g(\sigma(C)) = \{\beta_{1},\dots,\beta_{m}\}$, thus $\sigma(A)=\{\beta_{1},\dots,\beta_{m}\}$. According to lemma~\ref{lemma: KSdefinitionProjections}, for any $\af \in \sigma(A)$ we have
     \begin{align*}
         E_{\af} &\equiv \chi_{\{\af\}}(g(C)) = \chi_{g^{-1}(\af)}(C) = \sum_{\substack{i=1 \\ \beta_{i}=\af}}^{m} \chi_{\beta_{i}}(C) = \sum_{\substack{i=1 \\ \beta_{i}=\af}}^{m} F_{i},
     \end{align*}
     thus lemma~\ref{lemma: finiteAdditivity} ensures that, for any state $\rho$, $\Exp{E_{\af}} = \sum_{\substack{i=1 \\ \beta_{i}=\af}}^{m} \Exp{F_{i}}$. Finally,
     \begin{align*}
         \Exp{A} = \sum_{\af \in \sigma(A)} \af \Exp{E_{\af}} = \sum_{\af \in \sigma(A)}\sum_{\substack{i=1 \\ \beta_{i}=\af}}^{m} \af\Exp{F_{i}} = \sum_{i=1}^{m} \beta_{i} \Exp{F_{i}}.
     \end{align*}
\end{proof}

Let $X$ be any nonempty finite set, and let $P$ be a probability measure on its power set $\mathcal{P}(X)$. Let $f:X \rightarrow Y$ be any random variable (i.e., any function) on $(X,\mathcal{P}(X),P)$. In probability theory \cite{klenke2020probability}, the expected value of $f$ w.r.t. to $P$ is defined as
    \begin{align}
        \mathbb{E}_{P}(f) \doteq \sum_{x \in X} f(x) P(\{x\}) = \sum_{y \in f(X)} y P(f^{-1}(\{y\})).
    \end{align}
The expectation of an observable $A=f(C)$ coincides with expected value of the random variable $f$ in the probability space associated with $C$:
\begin{proposition}\label{prop: classicalExp} Let $C$ be any observable, and let $A=f(C)$ be any function of $C$. Then, for any state $\rho$,
\begin{align}
    \Exp{A} = \mathbb{E}_{\rho}(f),    
\end{align}
where $\mathbb{E}_{\rho}(f)$ is the expected value of the random variable $f:\sigma(C) \rightarrow \sigma(A)$ w.r.t. the probability measure $P_{\rho}^{C} \equiv P_{\rho}( \ \cdot \ ; C)$.
\end{proposition}
\begin{proof}
According to proposition~\ref{prop: functionalCalculus}, we have $A = \sum_{\gamma \in \sigma(C)} f(\gamma) F_{\gamma}$, where $F_{\gamma} = \chi_{\{\gamma\}}(C)$ for every $\gamma$. According to proposition~\ref{prop: expectedValue},
\begin{align*}
    \Exp{A} &= \sum_{\gamma \in \sigma(C)} f(\gamma)\Exp{F_{\gamma}} =  \sum_{\gamma \in \sigma(C)} f(\gamma)P_{\rho}^{C}(\{\gamma\}) = \mathbb{E}_{\rho}(f).
\end{align*}
\end{proof}

To conclude this section, let's show that states define partially linear mappings on $\mathcal{O}$:

\begin{proposition}[Partial linearity]\label{prop: expLinear}
    Let $A_{1},\dots, A_{m}$ be pairwise compatible observables, and let $a_{1},\dots,a_{m}$ be real numbers. Then, for any state $\rho$,
    \begin{align}
        \Exp{\sum_{i=1}^{m}a_{i}A_{i}} = \sum_{i=1}^{m}a_{i}\Exp{A_{i}}.
    \end{align}
    \end{proposition}
    \begin{proof}
    Let $A_{1},\dots A_{m}$ be pairwise compatible observables, and let $C \xrightarrow{f_{i}} A_{i}$, $i=1,\dots,m$, be a cone for them. According to proposition~\ref{prop: classicalExp}, for any state $\rho$ we have
    \begin{align*}
        \Exp{\sum_{i=1}^{m}a_{i}A_{i}} = \mathbb{E}_{\rho}(\sum_{i=1}^{m}a_{i}f_{i}) = \sum_{i=1}^{m} a_{i}\mathbb{E}_{\rho}(f_{i}) = \sum_{i=1}^{m}a_{i}\Exp{A_{i}}.
    \end{align*}
    \end{proof}
\subsection{Traces of observables and density operators}\label{sec: trace}

In lemma~\ref{lemma: traceSimple}, we demonstrated that the trace of a projection $E$ can be written as $\Tra{E} = n \langle E \rangle_{\emptyset}$, where $n$ is the dimension of the system and $\emptyset$ is the completely mixed state. We proved in section~\ref{sec: StatesAsFunctionals} that, for any state $\rho$, the mapping $\Exp{ \ \cdot \ }: \mathcal{P} \rightarrow [0,1]$ can be extended to a mapping $\Exp{ \ \cdot \ }: \mathcal{O} \rightarrow \mathbb{R}$, where $\mathcal{O}$ denotes the collection of all observables of the system. This extension enables us to define the trace of any observable:

\begin{definition}[Trace]\label{def: trace} Let $\mathfrak{S}$ be a $n$-dimensional system. Let $A$ be any observable of $\mathfrak{S}$, and let $\emptyset$ be the completely mixed state. We define the trace of $A$ by
\begin{align}
        \Tra{A} \doteq n \langle A \rangle_{\emptyset}.
\end{align}
\end{definition}

Lemma~\ref{lemma: traceSimple} ensures that definitions~\ref{def: traceProjection} and~\ref{def: trace} coincide if $A$ is a projection. Furthermore, proposition~\ref{prop: expectedValue} ensures that the trace of an observable is invariant under decompositions:

\begin{proposition}[Trace and decompositions]\label{prop: traceAndDecompositions}
    Let $A$ be an observable, and let  $F_{1},\dots F_{m}$ be pairwise orthogonal projections such that $A = \sum_{i=1}^{m} \beta_{i}F_{i}$ for some real numbers $\beta_{i},\dots,\beta_{m}$. Then
    \begin{align}
        \Tra{A} = \sum_{i=1}^{m} \beta_{j} \Tra{F_{i}}.
    \end{align}
\end{proposition}

Similarly, it easily follows from proposition~\ref{prop: expLinear} that the trace is a partially linear mapping on $\mathcal{O}$:
\begin{proposition}\label{prop: traceLinear} The mapping $\mathcal{O} \ni A \xmapsto{\text{tr}} \Tra{A} \in \mathbb{R}$ is partially linear, that is to say, if $A_{1},\dots,A_{m}$ are pairwise compatible observables and $a_{1},\dots,a_{m}$ are real numbers,
    \begin{align}
        \Tra{\sum_{i=1}^{m}a_{i}A_{i}} = \sum_{i=1}^{m} a_{i} \Tra{A_{i}}.
    \end{align}
\end{proposition}

We can now define density operators:
\begin{definition}[Density operator]\label{def: densityOperator} We say that an observable $A$ is a \textbf{positive operator} if $\sigma(A) \subset [0,\infty)$. A \textbf{density operator} is a positive operator $A$ that satisfies $\Tra{A} = 1$.
\end{definition}
We denote by $\mathcal{D}$ the set of all density operators.

Note that all rank-$1$ projections are density operators. More broadly, any projection canonically defines a density operator via normalization, that is, if $E$ is a projection, the observable
\begin{align*}
    \frac{E}{\Tra{E}} \equiv \frac{1}{\Tra{E}} \cdot E
\end{align*}
is a density operator. In fact, $\sigma(\frac{E}{\Tra{E}}) = \frac{1}{\Tra{E}}\sigma(E) \subset \{0,\frac{1}{\Tra{E}}\} \subset [0,\infty)$, so $E$ is positive, whereas $\Tra{\frac{E}{\Tra{E}}} = \frac{\Tra{E}}{\Tra{E}} = 1$.

Let $A$ be an observable, and let $\af$ be an eigenvalue of $A$. We define the \textbf{multiplicity} of $\af$ as the trace of the projection $E_{\af}$ associated with this eigenvalue (see definition~\ref{def: projectionOfEigenvalue}). A more precise way of putting it consists in defining the multiplicity of an objective  event $(\af,A)$ as the trace of the projection $E_{\af} \equiv \chi_{\{\af\}}(A)$ associated with it. Proposition~\ref{prop: nondegenerateDecomposition} ensures that the multiplicity of an eigenvalue $\af$ of $A$ tells us how many pairwise orthogonal rank-$1$ projections we must sum in order to obtain the projection $E_{\af}$ associated with it. More importantly, multiplicities enable us to relate the trace of $A$ with its spectrum:

\begin{proposition}[Trace and eigenvalues]\label{prop: traceSpectrum} Let $A$ be an observable in a $n$-dimensional system, and let $\sigma(A)$ be its spectrum. For each eigenvalue $\af \in \sigma(A)$, let $m_{\af}$ be its multiplicity, i.e., $m_{\af} \doteq \Tra{E_{\af}}$, where $E_{\af} \equiv \chi_{\{\af\}}$. Then
\begin{align}
    \Tra{A} = \sum_{\af \in \sigma(A)} m_{\af} \af.
\end{align}
Equivalently, if we define a sequence $\af_{1},\dots,\af_{n}$ satisfying $\{\af_{i}:i=1,\dots,n\} = \sigma(A)$ and, for every $\af \in \sigma(A)$, $\vert \{i: \af_{i} = \af\}\vert = m_{\af}$, we obtain
\begin{align}
    \Tra{A} = \sum_{i=1}^{n}\af_{i}.
\end{align}
\end{proposition}
\begin{proof}
Let $A = \sum_{\af \in \sigma(A)} \af E_{\af}$ be the spectral decomposition of $A$ (definition~\ref{def: spectralDecomposition}). According to propositions~\ref{prop: nondegenerateDecomposition} and~\ref{prop: traceLinear}, $\Tra{A} = \sum_{\af \in \sigma(A)} \af \Tra{E_{\af}} = \sum_{\af \in \sigma(A)}  m_{\af}\af$.
\end{proof}

\begin{corollary}\label{cor: DensitySumToONe} Let $A$ be a density operator in a $n$-dimensional system. Then
\begin{align}
    \sum_{\af \in \sigma(A)} m_{\af} \af = 1,
\end{align}
Equivalently, if we define a sequence $\af_{1},\dots,\af_{n}$, where $\{\af_{i}:i=1,\dots,n\} = \sigma(A)$ and, for every $\af \in \sigma(A)$, $\vert \{i: \af_{i} = \af\}\vert = m_{\af}$, we obtain
\begin{align}
    \sum_{i=1}^{n}\af_{i} = 1,
\end{align}
and consequently $\sigma(A) \subset [0,1]$.
\end{corollary}

The following result is important.
\begin{proposition}\label{prop: densityDecomposition} Let $A$ be an observable in a $n$-dimensional system. Then $A$ is a density operator if and only if it can be written as a convex combination of $n$ pairwise orthogonal rank-$1$ projections.    
\end{proposition}
\begin{proof}
    Suppose that $A$ is a density operator, and let $C$ be any nondegenerate observable such that $A=f(C)$. Let $C=\sum_{i=1}^{n} \gamma_{i} F_{i}$ be the spectral decomposition of $C$. We know that $\sigma(A) = f(\sigma(C)) = \{\af_{1},\dots,\af_{n}\}$, where $\af_{i} \equiv f(\gamma_{i})$. Corollary~\ref{cor: DensitySumToONe} ensures that $\af_{i} \in [0,1]$ for every $i$ and $\sum_{i=1}^{n}\af_{i} = 1$. Furthermore, according to proposition~\ref{prop: functionalCalculus}, $A = \sum_{i=1}^{n} \af_{i} F_{i}$, which proves that $A$ is a convex combination of the pairwise orthogonal rank-$1$ projections $F_{1},\dots,F_{n}$. On the other hand, let $E_{1},\dots,E_{n}$ be pairwise orthogonal rank-$1$ projections, and let $\af_{1},\dots,\af_{n}$ be non-negative numbers satisfying $\sum_{i=1}^{n}\af_{i} = 1$. Define $A \doteq \sum_{i=1}^{n} \af_{i} E_{i}$. Then $\sigma(A) = \{\af_{1},\dots,\af_{n}\} \subset [0,1]$, and, according to proposition~\ref{prop: traceAndDecompositions}, $\Tra{A} = \sum_{i=1}^{n} \af_{i} \Tra{F_{i}} = \sum_{i=1}^{n} \af_{i} =1$, which shows that $A$ is a density operator. 
\end{proof}
\subsection{A glimpse of the Born rule}\label{sec: BornRule}
 Let $E$ be a projection. Given any state $\rho$, denote by $\rho_{E}$ the state $T_{E}(\rho)$ (see definition~\ref{def: projectionUpdate}). According to lemmas~\ref{lemma: objectiveProjectiveUpdate} and~\ref{lemma: projectionJointProbability}, for any projection $F$ that is compatible with $E$ we have
\begin{align*}
    \langle F \rangle_{\rho_{E}} &= P_{T_{(1;E)}(\rho)}(1;F) = \frac{p_{\rho}(1,1;E,F)}{p_{\rho}(1;E)} = \frac{\Exp{E  F}}{\Exp{E}} = \Exp{\frac{E}{\Exp{E}}  F}.
\end{align*}
 More broadly, let $B$ be any observable compatible with $E$, and let $B=\sum_{\beta \in \sigma(B)} \beta F_{\beta}$ be its spectral decomposition. For every $\beta \in \sigma(B)$, $F_{\beta}$ is compatible with $E$, thus, according to propositions~\ref{prop: algebraicProperties} and~\ref{prop: expLinear},
\begin{align*}
    \langle B \rangle_{\rho_{E}} &= \sum_{\beta \in \sigma(B)} \beta \ \langle F_{\beta}\rangle_{\rho_{E}} = \sum_{\beta \in \sigma(B)} \beta \Exp{\frac{E}{\Exp{E}}  F_{\beta}} = \Exp{\frac{E}{\Exp{E}}B}.
\end{align*}
In the particular case where $\rho = \emptyset$, for any observable $B$ that is compatible with $E$ we obtain
\begin{align}
    \langle B \rangle_{\emptyset_{E}} &=\label{eq: preBorn} \frac{n}{n}\Tra{\frac{E}{\Tra{E}}  B} = \Tra{\frac{E}{\Tra{E}}  B}.
\end{align} 
The following definition is important.
\begin{definition}[Projective and pure states]\label{def: projectiveState} A state $\rho$ is said to be \textbf{projective} if $\rho = T_{E}(\emptyset)$ for some projection $E$. If $E$ is rank-$1$, we say that the projective state $\rho=T_{E}(\emptyset)$ is \textbf{pure}. We usually denote the projective state $T_{E}(\emptyset)$ by $\emptyset_{E}$.
\end{definition}

Equation~\ref{eq: preBorn} establishes a connection between the projective state $\emptyset_{E}$ and the density operator $\frac{E}{\Tra{E}}$ (see definition~\ref{def: densityOperator}), and it is easy to see that a similar result holds for any density operator of the system. To begin with, recall that, according to proposition~\ref{prop: densityDecomposition}, density operators and convex combinations of pairwise orthogonal rank-$1$ projections are equivalent concepts. Let $E_{1},\dots,E_{m}$ be pairwise orthogonal rank-$1$ projections, and, for each $i$, let  $\emptyset_{E_{i}}$ be the projective state associated with $E_{i}$. Let $\rho \doteq \sum_{i=1}^{m} \af_{i} \rho_{i}$ be any convex combination of $\emptyset_{E_{1}},\dots,\emptyset_{E_{m}}$ in $\mathcal{S}$, and let $A$ be the density operator $A \doteq \sum_{i=1}^{n} \af_{i} E_{i}$ (see proposition~\ref{prop: densityDecomposition}). For any observable $B$ that is compatible with $A$ (and consequently with $E_{i}$ for every $i$) we have, according to propositions~\ref{prop: expectationConvex},~\ref{prop: expLinear},~\ref{prop: algebraicProperties} and equation~\ref{eq: preBorn},
\begin{align*}
    \Exp{B} &= \sum_{i=1}^{n} \af_{i} \langle B \rangle_{\emptyset_{E_{i}}} =  \sum_{i=1}^{m} \af_{i} \Tra{\frac{E_{i}}{\Tra{E_{i}}}  B} = \sum_{i=1}^{m} \af_{i} \Tra{E_{i}  B} = \Tra{AB}.
\end{align*}
It proves the following proposition.

\begin{proposition}\label{prop: preBorn} Let $A$ be a density operator in some $n$-dimensional system and $A=\sum_{i=1}^{n} \af_{i} E_{i}$ be any convex decomposition of $A$ in terms for pairwise orthogonal rank-$1$ projections (see proposition~\ref{prop: densityDecomposition}). Let $\emptyset_{A}$ be the convex combination $\emptyset_{A} \doteq \sum_{i=1}^{n} \af_{i} \emptyset_{E_{i}}$, where, for each $i$, $\emptyset_{E_{i}}$ is the pure state associated with $E_{i}$. Then, for any observable $B$ that is compatible with $A$,
\begin{align}
    \langle B \rangle_{\emptyset_{A}} =\label{eq: propPreBorn} \Tra{AB}.
\end{align}
\end{proposition}
Note that, in principle, the state $\emptyset_{A}$ may depend on the decomposition we choose, so there is a slight abuse of notation in writing simply $\emptyset_{A}$. Fortunately, however, this ambiguity will disappear in the next section.

Rank-$1$ projections are density operators, so, for the sake of consistency, the pure state $\emptyset_{E}$ associated with the rank-$1$ projection $E$ should coincide with the state defined in proposition~\ref{prop: preBorn}, and it is easy to see that this is the case. In fact, let $E$ be a rank-$1$ projection, and let $E=\sum_{i=1}^{n} \af_{i} E_{i}$ any convex decomposition of $E$ in terms of pairwise orthogonal rank-$1$ projections. Let $C\xrightarrow{\chi_{\{\gamma_{i}\}}} E_{i}$, $i=1,\dots,n$, be any cone for $E_{1},\dots,E_{n}$, and define $f\doteq \sum_{i=1}^{n} \gamma_{i} \chi_{\{\gamma_{i}\}}$. Then $E = f(C)$, and consequently $\{0,1\}=\sigma(f(C))=f(\sigma(C)) = \{\af_{1},\dots,\af_{n}\}$. According to proposition~\ref{prop: nondegenerateDecomposition}, a projection has rank $k$ if and only if it can be written as a sum of $k$ pairwise orthogonal rank-$1$ projections, therefore we have $\af_{i} \neq 0$ for only one $i \in \{1,\dots,n\}$. Denoting by $i_{0}$ this index, we obtain $E =E_{i_{0}}$, and $\sum_{i=1}^{n} \af_{i} \emptyset_{E_{i}} = \emptyset_{E_{i_{0}}} = \emptyset_{E}$. Hence, any decomposition $E = \sum_{i=1}^{n} \af_{i} E_{i}$ of $E$ in terms of pairwise orthogonal rank-$1$ projections is trivial, i.e., we have $\af_{i} \neq 0$ for only one $i_{0} \in \{1,\dots,n\}$, and the state defined by this decomposition, as in proposition~\ref{prop: preBorn}, is the pure state associated with $E$ (definition~\ref{def: projectiveState}).

Now let $E$ be any projection, and let $D_{E}$ be the density operator $\frac{E}{\Tra{E}}$. Following the same line of thought of the previous paragraph, one can easily prove that, given any convex decomposition $E=\sum_{i=1}^{n} \af_{i} E_{i}$ of $E$ in terms of $n$ pairwise orthogonal rank-$1$ projections, we have $\af_{i} \in \{0,\frac{1}{\Tra{E}}\}$ for each $i$ and $\vert \{\af_{i}: i=\{1,\dots,n\}, \af_{i} \neq 0\} \vert = \Tra{E}$, which means that any decomposition of $D_{E}$ as a convex combination of pairwise orthogonal rank-$1$ projections is determined by pairwise orthogonal projection $E_{1},\dots,E_{k}$ satisfying $E=\sum_{i=1}^{k}E_{i}$, where $k \equiv \Tra{E}$. That is, the only way of decomposing $D_{E} \equiv \frac{E}{\Tra{A}}$ as a convex combination of pairwise orthogonal rank-$1$ projections is by writing
\begin{align*}
    \frac{E}{\Tra{E}} &= \frac{1}{\Tra{E}} \sum_{i=1}^{k} E_{i} = \sum_{i=1}^{k} \frac{1}{\Tra{E}} E_{i},
\end{align*}
where $k\equiv \Tra{E}$, for some set $E_{1},\dots,E_{k}$ of pairwise orthogonal rank-$1$  projections satisfying
\begin{align*}
    E=\sum_{i=1}^{k}E_{i}.
\end{align*}
Furthermore, given any cone $C \xrightarrow{\chi_{\{\gamma_{i}\}}} E_{i}$ for $E_{1},\dots,E_{k}$, the state $\emptyset_{D_{E}}$ determined by the decomposition $D_{E} = \sum_{i=1}^{k}\frac{E_{i}}{\Tra{E}}$ satisfies
\begin{align}
    \emptyset_{D_{E}}&=\label{eq: subjectiveNondegenerateMixedUpdate} \sum_{i=1}^{k}\frac{1}{\Tra{E}} \emptyset_{E_{i}} = \sum_{i=1}^{k}\frac{P_{\emptyset}(\gamma_{i};C)}{P_{\emptyset}(\Gamma;C)} \emptyset_{E_{i}} = \sum_{i=1}^{k}P_{\emptyset}^{C}(\{\gamma_{i}\}\vert \Gamma)T_{(\gamma_{i};C)}(\emptyset)
    \\
    &=T_{(\Gamma\vert C)}(\emptyset),
\end{align}
where $\Gamma \doteq \{\gamma_{i},\dots,\gamma_{k}\}$ and, consequently, $E = \chi_{\{\Gamma\}}(C)$. In the next section, we will show that equivalent  events update the \textit{completely mixed state}  in the same way, which implies that the projective state $\emptyset_{E}$ and the state $\emptyset_{D_{E}}$ are equal; this is shown in proposition~\ref{prop: updatingMixed}. In particular, it implies that all convex decompositions of a density operator define the same state, eliminating the aforementioned ambiguity that we find in the definition of $\emptyset_{A}$ presented in proposition~\ref{prop: preBorn}, and allowing us to introduce definition~\ref{def: densityState}.
\section{Connecting incompatible observables}\label{sec: connectingIncompatible}

To continue our derivation, we need to connect incompatible observables. In quantum mechanics, this connection is established by the inner product of the Hilbert space that represents the system, or equivalently by the Hilbert-Schmidt product of rank-$1$ projections. The distinctive way in which quantum mechanics connects incompatible observables is an essential part of the theory, and it has to be assimilated by our system for the quantum formalism to fully arise.
\subsection{Pure states and transition probabilities}\label{sec: transitionProbabilities}

In section~\ref{sec: basicFramework} we showed that, in a finite-dimensional system (definition~\ref{def: finiteSystem}), the spectrum of an observable $A$ is the set of all real numbers that can be obtained in a measurement of $A$, i.e., $\af \in\sigma(A)$ if and only if $p_{\rho}(\af;A) >0$ for some state $\rho$ (see equation~\ref{eq: preEigenvector}). In section~\ref{sec: observableEvents} we proved that, thanks to postulate~\ref{post: selfCompatibility}, the spectrum of an observable $A$ coincides with its point spectrum, i.e., with the set of all its eigenvalues. It means that, for any $\af \in \sigma(A)$, there exists a state $\rho_{\af}$ satisfying $p_{\rho_{\af}}(\af;A)=1$, as lemma~\ref{lemma: eigenvalues} shows. After introducing postulate~\ref{post: compatibility}, we learned how to explicitly construct the state $\rho_{\af}$: this state can be obtained by updating the completely mixed state with the observable event $(\af;A)$, i.e., $\rho_{\af} \doteq T_{(\af;A)}(\emptyset)$, which in turn corresponds to the projective state $\emptyset_{E_{\af}}$  associated with the projection $E_{\af} \equiv\chi_{\{\af\}}(A)$ (see definition~\ref{def: projectiveState}). Now let $C$ be any nondegenerate fine graining of $A$, i.e., $C$ is a nondegenerate observable and $A=f(C)$ for some function $f$. For each $i \in \sigma(C)$, let $F_{i}$ be the projection $\chi_{\{i\}}(C)$, and define $\Gamma_{\af} \doteq f^{-1}(\af)$ for every $\af \in \sigma(A)$. For each $i \in \sigma(C)$, let $\emptyset_{F_{i}}$ be the pure state associated with $F_{i}$, i.e., $\emptyset_{F_{i}} =T_{(1;F_{i})}(\emptyset) = T_{(i;C)}(\emptyset)$ (see definition~\ref{def: projectiveState}). Then, if $i \in \Gamma_{\af}$, we obtain
\begin{align}
    P_{\emptyset_{F_{i}}}(\af;A) =\label{eq: prePureState} \frac{P_{\emptyset}(\{i\} \cap \Gamma_{\af};C)}{P_{\emptyset}(\{i\};C)} = P_{\rho}^{C}(\gamma_{\af}\vert \{i\}) = 1.
\end{align}
We know that $\vert \Delta_{\af} \vert = \Tra{E_{\af}}$, so, for a given $\af \in \sigma(A)$, the nondegenerate fine graining $C$ of $A$ determines $k \equiv \Tra{E_{\af}}$ distinct pure states for which, in a measurement of $A$, the outcome $\af$ is necessarily obtained. Note also that, if $i \neq j$, the projection $F_{i}$ and $F_{j}$ are orthogonal, and that  $P_{\emptyset_{F_{i}}}(\af;A) = 0$ if $i \notin \Delta_{\af}$. Motivated by quantum theory, we introduce the following definition:

\begin{definition}[Eigenstate]\label{def: eigenstate} Let $A$ be an observable, and let $\af$ be any eigenvalue of $A$ (see lemma \ref{lemma: eigenvalues}). An eigenstate of $A$ corresponding to the eigenvalue $\af$ is a pure state $\emptyset_{F}$ satisfying $P_{\emptyset_{F}}(\af;A)=1$. Hence, an eigenstate of $A$ is a pure state that enables us to predict, with certainty, the outcome of a measurement of $A$.
\end{definition}
We have proved the following lemma.
\begin{lemma}\label{lemma: eigenstates} Let $A$ be an observable, and let $C$ be a nondegenerate fine graining of it, i.e., $C$ is a nondegenerate observable satisfying $A=f(C)$ for some function $f$. For each $i \in \sigma(C)$, let $F_{i}$ be the projection $\chi_{\{i\}}(C)$, and let $\emptyset_{F_{i}}$ be the pure state associated with it. Then $\emptyset_{F_{i}}$ is an eigenstate of $A$ corresponding to the eigenvalue $\af$ if and only if $i \in f^{-1}(\af)$.
\end{lemma}
Since nondegenerate fine grainings always exist (postulate~\ref{post: observables}), we have the following corollary.
\begin{corollary}\label{cor: eigenstates} Let $A$ be an observable, and let $k_{\af}$ the the trace of the projection associated with the eigenvalue $\af$ of $A$. Then there is a (not necessarily unique) set $\{\emptyset_{F_{1}},\dots,\emptyset_{F_{k_{\af}}}\}$   of $k_{\af}$ pairwise orthogonal eigenstates  of $A$ corresponding to $\af$, where by pairwise orthogonal we mean that the projections $F_{1},\dots,F_{k_{\af}}$ are pairwise orthogonal.    
\end{corollary}

The following result immediately follows from definition~\ref{def: eigenstate} and lemma~\ref{lemma: statisticalEquivalence}.
\begin{lemma}
    Let $A,B$ be observables satisfying $\chi_{\{\af\}}(A)=\chi_{\{\beta\}}(B)$ for some pair $(\af,\beta) \in \sigma(A)\times \sigma(B)$, and let $F$ be any rank-$1$ projection. Then $\emptyset_{F}$ is an eigenstate of $A$ corresponding to $\af$ if and only if $\emptyset_{F}$ is an eigenstate of $B$ corresponding to $\beta$.
\end{lemma}

Let $A$ be an observable, and let $C$ be a nondegenerate observable such that $A=f(C)$ for some function $f$. For each $i \in \sigma(C)$, let $\emptyset_{F_{i}}$ be the pure state associated with the projection $F_{i} \equiv \chi_{\{i\}}(C)$. The reason why  $\emptyset_{F_{i}}$ determines the outcome of $A$ is clear: for any $\af \in \sigma(A)$, the event $(\af;A)$ is statistically equivalent to the event $(f^{-1}(\af);C)$, which means that $P_{\rho}(\af;A)=P_{\rho}(f^{-1}(\af);C)$ for every state $\rho$, so, thanks to postulate~\ref{post: selfCompatibility}, the information encoded in the pure state $\emptyset_{F_{i}} = T_{(i;C)}(\emptyset)$ enables us to conclude, with certainty, that a measurement of $A$ will return the outcome $\af_{i} \in \sigma(A)$, where $\af_{i}$ denotes the unique eigenvalue $\af$ of $A$ for which $i \in f^{-1}(\af)$. Hence, we can predict with certainty the outcome of $A$ for a system in the state $\emptyset_{F_{i}}$ essentially because there is a cone $F_{i} \xleftarrow{\chi_{\{i\}}} C \xrightarrow{f} A$ for $F_{i}$ and $A$, which, thanks to postulate~\ref{post: selfCompatibility}, transforms the probability $P_{\emptyset_{F_{i}}}(\af;A)$ into the probability probability of the event $ \Gamma_{\af}\subset \sigma(C)$ occurring, under the evidence that the event $\{i\} \subset \sigma(C)$ has occur, in the probability space $(\sigma(C),\mathcal{P}(\sigma(C)),P^{C}_{\emptyset})$, as equation \ref{eq: prePureState} shows. In short, the cone  $F_{i} \xleftarrow{\chi_{\{i\}}} C \xrightarrow{f} A$ provides the informational link between the state $\emptyset_{F_{i}}$ and the observable event $(\af;A)$, allowing us to think about them in purely statistical terms. If a rank-$1$ projection $F$ does not commute with the observable $A$, i.e., if there is no cone for $F$ and $A$, then there is no spectrum  where the events $(1,F)$ and $(\af;A)$, $\af \in \sigma(A)$, can be embedded, and consequently, just by knowing that the state of the system is $\emptyset_{F}$ (equivalently, just by knowing that the event $(1,F)$ has occurred), we cannot predict with certainty the outcome of a measurement of $A$. It suggests that the first part of postulate~\ref{post: eigenstates} is a reasonable condition. 

An eigenstate of $A$ enables us to predict with certainty  the outcome of a measurement of $A$, but not all states satisfying this property are eigenstates of $A$. In fact, consider again a nondegenerate fine graining $C$ of $A$, where $A=f(C)$. For each $i \in \sigma(C)$, let $\emptyset_{F_{i}}$ be the pure state determined by the projection $F_{i} \equiv \chi_{\{i\}}(C)$, and  define, for any $\Gamma \subset \sigma(C)$,
\begin{align}
    \emptyset_{\Gamma} \doteq\label{eq: gammaConvex} \sum_{i \in \Gamma} \frac{1}{\Tra{F_{\Gamma}}}\emptyset_{F_{i}},
\end{align}
where $F_{\Gamma} \equiv \chi_{\Gamma}(C)$ and, consequently, $\Tra{F_{\Gamma}} = \vert \Gamma \vert$. Now let $\af$ be an eigenvalue of $\af$, and define $\Gamma_{\af} \doteq f^{-1}(\af)$. We saw in section~\ref{sec: BornRule} that, for any $\Gamma \subset \sigma(C)$, we have $\emptyset_{\Gamma} = T_{(\Gamma;C)}(\emptyset)$, so, given any $\af \in \sigma(A)$ and any $\Gamma \subset \Gamma_{\af}$, we obtain
\begin{align}
    P_{\emptyset_{\Gamma}}(\af;A) = P_{\emptyset}(\Gamma_{\af}\vert \Gamma) = 1.
\end{align}
The state $\emptyset_{G_{\Sigma}}$ enables us to predict with certainty the outcome of a measurement of $A$ for precisely the same reason the eigenstate $\emptyset_{F_{i}}$ enables us to do it: thanks to  postulate~\ref{post: selfCompatibility} and to the statistical equivalence between $(\af;A)$ and $(\Gamma_{\af};C)$, information encoded in the state $\emptyset_{\Gamma}$ ensures the occurrence of the event $(\af;A)$ in an eventual measurement of $A$. More specifically, as we discussed in section~\ref{sec: observableEvents}, knowing that the state of the system is $\emptyset_{\Gamma} =T_{(\Gamma;C)}(\emptyset)$ consists in knowing that the state is $\emptyset_{F_{i}}$, $i \in \Gamma$, with probability $\frac{1}{\vert \Gamma \vert}$. Since $\Gamma \subset \Gamma_{\af}$, all pure states $\emptyset_{F_{i}}$, $i \in \Gamma$, ensures the occurrence of $(\af;A)$ in a measurement of $A$, so knowing that the state is $\emptyset_{\Gamma} =T_{(\Gamma;C)}(\emptyset)$ is sufficient to assert, with certainty, that a measurement of $A$ will return the outcome $\af$. Now, consider the particular case where $A$ is nondegenerate. We know that, in this case, if $A=f(C)$ for some observable $C$ and some function $f$, then $A$ and $C$ are isomorphic objects in the category of observables, which means that the function $f$ is an isomorphism between $\sigma(C)$ and $\sigma(A)$. For each $i \in \sigma(C)$, write $\af_{i} \equiv f(i)$. According to lemma~\ref{lemma: KSdefinitionProjections}, for all $i \in \sigma(C)$ we have $E_{\af_{i}} \equiv \chi_{\{\af_{i}\}}(A) = \chi_{i}(C) \equiv F_{i}$, thus the eigenstate of $A$ defined by the projection $F_{i}$ coincides with the pure state $\emptyset_{E_{\af_{i}}}$. Hence, given any observable $D$ such that $A=g(D)$ for some function $g$, and any eigenvalue $\af \in \sigma(A)$, there is only one eigenstate corresponding to $\af$ defined by  $D$, which is the pure state $\emptyset_{E_{\af}}$, where $E_{\af} \equiv \chi_{\{\af\}}(A)$. Put differently, in order to be able to predict, with certainty, the outcome of a measurement of $A$, it is sufficient to prepare a (not necessarily pure)  state using an observable $C$ such that $A=f(C)$ and an event $(\Gamma;C)$ whose occurrence implies the occurrence of $(f^{-1}(\af);C)$ for some $\af \in \sigma(A)$, by which we mean that $\Gamma \subset \Gamma_{\af}$. If $A$ is nondegenerate, then $f^{-1}(\af)$ is  a singleton, so there is no room for considering convex combinations, as in equation~\ref{eq: gammaConvex}, or projective events that are not pure. Since we interpret states as epistemological entities, the second part of postulate~\ref{post: eigenstates} is also reasonable.

\begin{postulate}[Eigenstates]\label{post: eigenstates} If $\emptyset_{F}$ is an eigenstate of $A$, then $F$ and $A$ are compatible. Put differently, if a pure state $\emptyset_{F}$ enables us to predict with certainty the outcome of a measurement of $A$, then $F$ and $A$ are compatible. Furthermore, if $A$ is nondegenerate and $\rho \in \mathcal{S}$ satisfies $P_{\rho}(\af;A)=1$ for some $\af$, then $\rho$ is a pure state. It means that only pure states enable us to predict outcomes of measurements of nondegenerate observables with certainty.
\end{postulate}

Roughly speaking, postulate~\ref{post: eigenstates} says that the capacity to predict with certainty the outcome of some measurement is always a matter of gathering the right amount of information. 

The following lemma shows that there is a one-to-one correspondence between rank-$1$ projections and pure states. 
\begin{lemma}\label{lemma: rank1AndPureStates} Two rank-$1$ projections $E,F$ satisfy $\emptyset_{E}=\emptyset_{F}$ if and only if they are equal. 
\end{lemma}
\begin{proof}
     We already know that $E=F$ implies $\emptyset_{E}=\emptyset_{F}$, thus we just need to show that $E=F$ follows from $\emptyset_{E}=\emptyset_{F}$. So let $E,F$ be rank-$1$ projections satisfying $\emptyset_{E}=\emptyset_{F}$, and let $A$ be any nondegenerate observable such that $E=\chi_{\{\af\}}(C)$ for some $\af \in \sigma(A)$. Then $P_{\emptyset_{F}}(\af;A)=P_{\emptyset_{E}}(\af;A)=P^{A}_{\emptyset}(\{\af\}\vert \{\af\})=1$. According to postulate~\ref{post: eigenstates}, it implies that $F$ and $C$ commutes, and since $P_{\emptyset_{F}}(\af;A)=1$, we have $F=\chi_{\{\af\}}(A)=E$. 
\end{proof}

From now on, we will usually denote the expectation defined by a pure state $\emptyset_{E}$ by $\langle \ \cdot \ \rangle_{E}$ instead of $\langle \ \cdot \ \rangle_{\emptyset_{E}}$. Lemma~\ref{lemma: rank1AndPureStates} ensures that there is no ambiguity in doing so.

\begin{lemma}\label{lemma: uniquenessEigenvalue} Let $A$ be an observable such that $\Tra{E_{\af}} = 1$ for some $\af \in \sigma(A)$, where $E_{\af} \equiv \chi_{\{\af\}}(A)$. Let $\af$ be an eigenvalue of $A$ satisfying $\Tra{E_{\af}}=1$, and let $\rho$ be a state for which $P_{\rho}(\af;A) =1$. Then $\rho$ is the pure state $\emptyset_{E_{\af}}$   
\end{lemma}
\begin{proof}
    Let $A$ be an observable, and suppose that, for some $\af \in \sigma(A)$, we have $\Tra{E_{\af}}=1$. Fix an eigenvalue $\af_{0} \in \sigma(A)$ satisfying this condition, and let $\rho_{0}$ be a state for which $P_{\rho_{0}}(\af_{0};A) =1$. Since $\Tra{E_{\af_{0}}}=1$, there is a nondegenerate observable $C$ such that $E_{\af_{0}} = \chi_{\{i\}}(C)$ for some $i \in \sigma(C)$, and consequently the events $(\af_{0};A)$ and $(i;C)$ are statistically equivalent. Therefore,
    \begin{align*}
        P_{\rho_{0}}(i;C) = P_{\rho_{0}}(\af;A) = 1.
    \end{align*}
    According to postulate~\ref{post: eigenstates}, it implies that $\rho_{0}$ is a pure state, i.e., that $\rho_{0}=T_{G}(\emptyset) \equiv \emptyset_{G}$ for some rank-$1$ projection $G$, whereas definition \ref{def: eigenstate} ensures that $\emptyset_{G}$ is an eigenstate of $C$. Postulate~\ref{post: eigenstates} again implies that $G$ and $C$ are compatible, which in turn implies that $G=\chi_{\{i\}}(C)$, and therefore $\rho_{0}=T_{G}(\emptyset)=T_{\chi_{\{i\}}(C)}(\emptyset) =T_{E_{\af_{0}}}(\emptyset) =\emptyset_{E_{\af_{0}}}$. 
\end{proof}

\begin{corollary}\label{cor: pureStateUpdate} Let $E$ be a rank-$1$ projection, and let $\emptyset_{E}$ be the pure state associated with it. Then, for any state $\rho$, $T_{E}(\rho)=\emptyset_{E}$. 
\end{corollary}
\begin{proof}
    Let $E$ be a rank-$1$ projection, and let $A$ be any nondegenerate observable such that $E=\chi_{\{\af\}}(A)$ for some $\af \in \sigma(A)$ (which exists, according to postulate~\ref{post: observables} and lemma~\ref{lemma: projection}). For any state $\rho$, we have
    \begin{align*}
        P_{T_{E}(\rho)}(\af;A)&=P_{T_{(\af;A)}(\rho)}(\af;A)= P_{\rho}^{A}(\{\af\}\vert \{\af\}) = 1.
    \end{align*}
    It follows from lemma~\ref{lemma: uniquenessEigenvalue} that $T_{E}(\rho)=\emptyset_{E}$.
\end{proof}

Corollary~\ref{cor: pureStateUpdate} enables us to talk about transition probabilities between pure states. In fact, let $\emptyset_{E}$, $\emptyset_{F}$ be pure states, and let $A$ be an observable satisfying $\chi_{\{\af_{0}\}}(A)=F$ for some $\af_{0} \in \sigma(A)$. If we prepare the state $\emptyset_{E}$, measure $A$ and obtain $\af_{0}$, we end up with the state $T_{(\af_{0};A)}(\emptyset_{E})$, which, according to corollary~\ref{cor: pureStateUpdate}, corresponds to $\emptyset_{F}$, so we transition from the state $\emptyset_{E}$ to the state $\emptyset_{F}$. The probability of this transition consists in the probability of obtaining $\af_{0}$ in a measurement of $A$ for a system in the state $\emptyset_{F}$, which is given by $P_{\emptyset_{E}}(\af_{0};A)  = \langle F \rangle_{E}$. This probability depends only on the initial and final states $\emptyset_{E}$, $\emptyset_{F}$, or equivalently on the projections $E$ and $F$, and not on the measurement $A$. We can thus introduce the following definition.

\begin{definition}[Transition probability]\label{def: transitionProbability} Let $E,F$ be rank-$1$ projections. We define the ``probability of transitioning from $E$ to $F$'', denoted $P(E \rightarrow F)$, by
\begin{align}
    P(E \rightarrow F) \doteq \langle F \rangle_{E}.
\end{align}
This probability corresponds to the probability of transitioning from the pure state $\emptyset_{E}$ to the pure state $\emptyset_{F}$ by preparing  $\emptyset_{E}$ and performing a measurement $A$ such that $F=\chi_{\{\af\}}(A)$ for some $\af \in \sigma(A)$.
\end{definition}

Note that
\begin{align*}
    P(E \rightarrow F) = \langle F \rangle_{E} = P_{T_{(1,E)}(\emptyset)}(1,F) = \frac{P_{\emptyset}(1,1,E,F)}{P_{\emptyset}(1,E)} = n P_{\emptyset}(1,1,E,F),
\end{align*}
where $n \equiv \text{dim}(\mathfrak{S})$. 

Equality and orthogonality between rank-$1$ projections can by characterized via transition probabilities:

\begin{proposition}\label{prop: orthogonalityTransition}
Let $E,F$ be rank-$1$ projections. Then the following conditions are satisfied.
\begin{itemize}
    \item[(a)] $E$ and $F$ are orthogonal if and only if $P(E \rightarrow F) = 0$.
    \item[(b)] $E$ and $F$ are equal if and only if $P(E \rightarrow F) = 1$.
\end{itemize}
\end{proposition}
\begin{proof}
Let $E,F$ be rank-$1$ projections. We have $P(E \rightarrow F)=P_{\emptyset_{E}}(1,F)$, so, according to lemma~\ref{lemma: uniquenessEigenvalue}, $P(E \rightarrow F)=1$ if and only if $\emptyset_{E}=\emptyset_{\chi_{\{1\}}(F)}=\emptyset_{F}$, which in turn, thanks to lemma~\ref{lemma: rank1AndPureStates}, is equivalent to saying that $E=F$. It proves item $(a)$. Now let $E,F$ be orthogonal projections. According to lemma~\ref{lemma: projectionJointProbability},
\begin{align*}
    P(E\rightarrow F) = P_{\emptyset_{E}}(1;F) = \frac{P_{\emptyset}(1,1;E,F)}{P_{\emptyset}(1;E)} = \frac{\langle EF\rangle_{\emptyset}}{\langle E \rangle_{\emptyset}} = 0.
\end{align*}
Finally, let $E,F$ be rank-$1$ projections satisfying $P(E \rightarrow F)=0$. Then $P_{\emptyset_{E}}(0;F) = 1-P_{\emptyset_{E}}(1;F) = 1-\langle F \rangle_{E} = 1$, which implies that $\emptyset_{E}$ is an eigenstate of $F$, and therefore, according to postulate~\ref{post: eigenstates}, $E$ and $F$ are compatible. Hence,
\begin{align*}
    0=P_{\emptyset_{E}}(1;F)=\frac{\langle EF\rangle_{\emptyset}}{\langle E \rangle_{\emptyset}},
\end{align*}
which implies $EF=0$, i.e., $E \perp F$. 
\end{proof}

Note that, thanks to corollary~\ref{cor: pureStateUpdate}, it makes sense to assign probabilities for  sequences of transitions like $E_{1}\rightarrow E_{2} \rightarrow \dots \rightarrow E_{m}$ by defining
\begin{align*}
    P(E_{1}\rightarrow E_{2} \rightarrow \dots \rightarrow E_{m}) \doteq \prod_{i=1}^{m-1}P(E_{i} \rightarrow E_{i+1}).
\end{align*}
In fact, let $E_{1},\dots,E_{m}$ be rank-$1$ projections. For each $i\in \{1,\dots,m\}$, define $\rho_{i} \doteq (T_{E_{m}}\circ \dots \circ T_{E_{1}})(\emptyset)$. Since $E_{i}$ is rank-$1$ for all $i$, corollary~\ref{cor: pureStateUpdate} implies that $\rho_{i} = \emptyset_{E_{i}}$, and therefore, according to definition~\ref{def: sequentialMeasure},
\begin{align*}
    P(E_{1}\rightarrow E_{2} \rightarrow &\dots \rightarrow E_{m}) \doteq \prod_{i=1}^{m-1}P(E_{i} \rightarrow E_{i+1}) = \prod_{i=1}^{m} P_{\emptyset_{E_{i}}}(1,E_{i+1}) = \prod_{i=1}^{m} P_{\rho_{i}}(1,E_{i+1})
    \\
    &=P_{\emptyset_{E_{1}}}(1,\dots,1;E_{2},\dots,E_{m}) = n P_{\emptyset}(1,\dots,1;E_{1},\dots,E_{m}),
\end{align*}
where $n \equiv \text{dim}(\mathfrak{S})$. Let's emphasize this definition.
\begin{definition}\label{def: sequentialTransitionProbability} Let $E_{1},\dots,E_{m}$ be rank-$1$ projections. The number
\begin{align}
    P(E_{1} \ri \dots \ri E_{m}) \doteq \prod_{i=1}^{m}P(E_{i} \ri E_{m})
\end{align}
is said to be the probability associated with the sequence of transitions $E_{1},\dots,E_{m}$.    
\end{definition}

In the simplest case, the probability of transitioning between two pure states does not depend on the order of the transition, i.e., we have $P(E \rightarrow F)=P(F \rightarrow E)$ for any pair of rank-$1$ projections $E,F$. We introduce this condition in postulate~\ref{post: transitionProbability}, and we assume that it is valid from now on. Note that it easily follows from this condition that, if $E_{1},\dots,E_{m}$ are rank-$1$ projections, we have
\begin{align*}
    P(E_{1}\ri \dots E_{m}) = P(E_{m} \ri E_{m-1} \ri \dots E_{1}).
\end{align*}

If $P(E \ri F) = P(F \ri E)$, then there is a relation between the transition $E \rightarrow F$ and the ``closed cycle'', let's say, $E \rightarrow F \rightarrow E$. In fact, we have
\begin{align}
    P(E \rightarrow F)&= \sqrt{P(E \rightarrow F) P(E \rightarrow F)} = \sqrt{P(E \rightarrow F) P(F \rightarrow E)}
    \\
    &=\label{eq: preCycle}\sqrt{P(E \rightarrow F \rightarrow E)}
    \\
    &= \sqrt{P(F \ri E \ri F)}
    \\
    &= P(F \ri E).
\end{align}
Let's explore this relation. Let $E,F$ be rank-$1$ projections. Let $C$ be a nondegenerate observable, and, for each $i \in \sigma(C)$, let $G_{i}$ be the projection $\chi_{\{i\}}(C)$. A measurement of $C$ between the transition $E \rightarrow F$ alters the probability of this transition, whether or not the experimentalist takes the outcome of this measurement into account. That is, suppose that an experimentalist prepares the state $\emptyset_{E}$ and measures the observable $F$ many times, without knowing that, in each run of the experiment, another agent measures $C$ after the preparation of $\emptyset_{E}$ and before the measurement of $F$. According to our postulates, the relative frequency, denoted $T(E \rightarrow F)$, of the transition $E \rightarrow F$  the experimentalist will end up with will be correctly described by the probability measure $P_{T_{(\sigma(C);C)}(\emptyset_{E})}( \ \cdot \ ;F)$, and therefore, after many runs of the experiment, she will (ideally) obtain
\begin{align*}
    T(E \rightarrow F) &= p_{T_{(\sigma(C);C)}(\emptyset_{E})}(1;F) =\frac{P_{\emptyset_{E}}(\sigma(C) \times \{1\};C,F)}{P_{\emptyset_{E}}(\sigma(C);C)} = \sum_{i \in \sigma(C)} p_{\emptyset_{E}}(i,1;C,F)
    \\
    &= \sum_{i \in \sigma(C)} P(E \rightarrow G_{i})P(G_{i} \rightarrow F) = \sum_{i \in \sigma(C)}P(E \rightarrow G_{i} \rightarrow F).
\end{align*}
Now, according to equation~\ref{eq: preCycle},
\begin{align*}
    T(E \rightarrow F) &= \sum_{i \in \sigma(C)} \sqrt{P(E \rightarrow G_{i}) P(G_{i} \ri E)P(G_{i} \rightarrow F)P(F \ri G_{i})}
    \\
    &= \sum_{i \in \sigma(C)} \sqrt{P(E \rightarrow G_{i} \ri F \ri G_{i} \ri E)}.
\end{align*}
The sequence of transitions $E \rightarrow G_{i} \ri F \ri G_{i} \ri E$ can be depicted as follows.
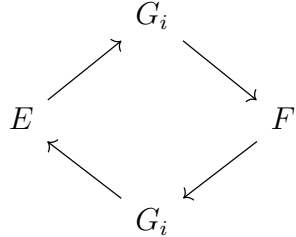
\begin{center}
    \begin{tikzcd}
        & G_{i}\arrow[dr] &
        \\
        E \arrow[ru] & & F\arrow[dl]
        \\
        & G_{i}\arrow[ul] & 
    \end{tikzcd}
\end{center}
Roughly speaking, we can say that the measurement $C$ ``breaks the symmetry'' between the transition $E \rightarrow F$ and the cycle $E \rightarrow F \rightarrow E$. In fact, thinking in terms of a measurement of $C$ between both transitions $E \rightarrow{F}$ and $F \rightarrow E$, the diagram presented above does not account for all possible transitions that can happen, because there is no reason for the outcome of $C$ in the transition $E \rightarrow F$ to be equal to its outcome in the transition $F \rightarrow E$. By summing over all possible outcomes, we do not end up with $T(E \rightarrow F)$, but actually with
\begin{align}
    \sum_{i,j \in \sigma(C)}  \sqrt{P(E \rightarrow G_{i} \ri F \ri G_{j} \ri E)} &=\sum_{i,j \in \sigma(C)}  \sqrt{P(E \rightarrow G_{i} \ri F)P(F \ri G_{j} \ri E)}
    \\
    &=\sum_{i,j \in \sigma(C)}  \sqrt{P(E \rightarrow G_{i} \ri F)P(E \ri G_{j} \ri F)}
    \\
    &=\label{eq: square}\left(\sum_{i \in \sigma(C)} \sqrt{P(E \rightarrow G_{i} \ri F)}\right)^{2}.
\end{align}
Note that 
\begin{align*}
    \sum_{i,j \in \sigma(C)}  \sqrt{P(E \rightarrow G_{i} \ri F \ri G_{j} \ri E)} &=\sum_{i,j \in \sigma(C)}\sqrt{p_{\emptyset}(1,i,1,j,1;E,C,F,C,E)}.
\end{align*}

If we had $T(E \rightarrow F) = P(E \ri F)$, it would follow from equation $P(E \ri F) = P(F \ri E)$ that $\sqrt{P(E \ri F \ri E)} = \sum_{i \in \sigma(C)} \sqrt{P(E \ri G_{i} \ri F \ri G_{i} \ri E)}$, which seems incoherent. In terms of marginalization, what seems reasonable to obtain is 
\begin{align*}
    \sqrt{P(E \ri F \ri E)} = \sum_{i,j \in \sigma(C)} \sqrt{P(E \ri G_{i} \ri F \ri G_{j} \ri E)},
\end{align*}
which in turn is equivalent to saying that $P(E \ri F) = \sum_{i,j \in \sigma(C)}  \sqrt{P(E \rightarrow G_{i} \ri F \ri G_{j} \ri E)}$. As we will see in section~\ref{sec: quantumTheory}, this equality, and consequently the distinction between $T(E \ri  F)$ and $P(E \ri F)$, is necessary for the emergence of \textit{interference terms} between non-orthogonal pure states. We thus have the following necessary conditions for the emergency of the quantum formalism.

\begin{postulate}[Transition probability]\label{post: transitionProbability} Let $E,F$ be rank-1 projections in a $n$-dimensional system. The probability of transitioning from $E$ to $F$ is equals to the probability of transitioning from $F$ to $E$, i.e.,
\begin{align}
    P(E \rightarrow F) &=\label{eq: transitioning} P(F \rightarrow E).
\end{align}
It enables us to talk about the ``transition probability between $E$ and $F$'' without  ambiguity. Furthermore, given any set $G_{1},\dots,G_{n}$ of pairwise orthogonal rank-$1$ projections, we have
\begin{align}
    \sqrt{P(E \ri F \ri E)} =\label{eq: interference} \sum_{i,j} \sqrt{P(E \ri G_{i}\ri F \ri G_{j} \ri E)}.
\end{align}
\end{postulate}
It is worth emphasizing the following result.
\begin{lemma}\label{lemma: transitionSquare} Let $E,F$ be rank-$1$ projections in a $n$-dimensional system $\mathfrak{S}$, and let $G_{1},\dots,G_{n}$ be pairwise orthogonal rank-$1$ projections in $\mathfrak{S}$. Then equation \ref{eq: interference} from postulate~\ref{post: transitionProbability} is satisfied if and only if
\begin{align}
    P(E \ri F) = \left(\sum_{i \in \sigma(C)} \sqrt{P(E \rightarrow G_{i} \ri F)}\right)^{2}.
\end{align}
\end{lemma}
\begin{proof}
    According to postulate~\ref{post: transitionProbability}, we have $P(E \ri F) = \sqrt{P(E \ri F \ri E)}$. On the other hand, it follows from equation~\ref{eq: square} that $\sum_{i,j} \sqrt{P(E \ri G_{i} F \ri G_{j} \ri E)}= \left(\sum_{i \in \sigma(C)} \sqrt{P(E \rightarrow G_{i} \ri F)}\right)^{2}$, so the proof is complete. 
\end{proof}

Now let's prove that, for any projection $E$, the projective state $\emptyset_{E}$ coincide with the state $\emptyset_{\frac{E}{\Tra{E}}}$ given by proposition~\ref{prop: preBorn}. It means that, if $E_{1},\dots,E_{k}$ are pairwise orthogonal rank-$1$ projections satisfying $E=\sum_{i=1}^{k} E_{i}$, then $\emptyset_{E} = \sum_{i=1}^{k} \frac{1}{\Tra{E}} \emptyset_{E_{i}}$. In order to do that, we need some preparatory results. 

\begin{lemma}\label{lemma: nondegenerateNonUpdate} Let $\emptyset$ be the completely mixed state. Then, for any nondegenerate observable $C$,
\begin{align}
    T_{(\sigma(C);C)}(\emptyset) &=\label{eq: nondegenerateNonUpdate} \emptyset.
\end{align}
\end{lemma}
\begin{proof}
Let $C$ be a nondegenerate observable in a $n$-dimensional system $\mathfrak{S}$. For each $i \in \sigma(C)$, let $E_{i}$ be the projection associated with $i$, i.e., $E_{i} \equiv \chi_{\{i\}}(C)$. According to proposition~\ref{prop: subjectiveUpdate} and postulate~\ref{post: transitionProbability}, given any rank-$1$ projection $F$,
\begin{align*}
    \langle F \rangle_{T_{(\sigma(C);C)}(\emptyset)} &= \sum_{i \in \sigma(C)} P_{\emptyset}^{C}(\{i\}\vert \sigma(C))\langle F \rangle_{T_{(i;C)}(\emptyset)} = \sum_{i \in \sigma(C)}\frac{1}{n} \langle F \rangle_{\emptyset_{E_{i}}} =  \frac{1}{n}  \sum_{i \in \sigma(C)}\langle E_{i} \rangle_{\emptyset_{F}} = \frac{1}{n} \langle \mathds{1}\rangle_{\emptyset_{F}}
    \\
    &= \frac{1}{n}.
\end{align*}
In implies that, for any nondegenerate observable $A$, $P_{T_{(\sigma(C);C)}(\emptyset)}^{A}$ is the uniform probability distribution on $\sigma(A)$, which is equivalent to saying that $P_{T_{(\sigma(C);C)}(\emptyset)}^{A} = P_{\emptyset}^{A}$. This in turn is equivalent to the fact that $P_{T_{(\sigma(C);C)}(\emptyset)}^{A} = P_{\emptyset}^{A}$ for any observable $A$, and therefore equation~\ref{eq: nondegenerateNonUpdate} is satisfied.
\end{proof}

\begin{lemma}\label{lemma: UpdatingMixed} Let $E$ be a projection, and let $C$ be  any nondegenerate observable satisfying $E=\chi_{\Gamma}(C)$ for some $\Gamma \subset \sigma(C)$. Then
\begin{align}
    T_{(1;E)}(\emptyset)=T_{(\Gamma;C)}(\emptyset).
\end{align}
\end{lemma}
\begin{proof}
Suppose that $E=\chi_{\Gamma}(C)$, where $E$ is a rank-$k$ projection, $C$ is a nondegenerate observable and $\Gamma \subset \sigma(C)$.  For each $i \in \sigma(C)$, let $E_{i}$ be the projection $\chi_{\{i\}}(C)$. Note that, by construction, $E$, $C$ and $E_{i}$ are compatible. Furthermore, note that $E \circ E_{i} = \chi_{\Gamma \cap \{i\}}(C)$, which implies that $E \circ E_{i} = E_{i}$ if $i \in \Gamma$ and $E \circ E_{i}=0$ otherwise, which implies that $T_{(1,E)}\circ T_{(i;C)} = T_{(i;C)}$ if $i \in \Gamma$.  According to proposition~\ref{prop: subjectiveSequentialUpdate} and lemma~\ref{lemma: nondegenerateNonUpdate},
\begin{align*}
    T_{(1;E)}(\emptyset) &= (T_{(1;E)} \circ T_{(\sigma(C);C)})(\emptyset) = \sum_{i \in \sigma(C)} P_{\emptyset}^{(C,E)}(\{(i,1)\}\vert \sigma(C) \times \{1\}) (T_{(1,E)} \circ T_{(i;C)})(\emptyset)
    \\
    &= \sum_{i \in \sigma(C)} \frac{ P_{\emptyset}^{(C,E)}(\{(i,1)\})}{P_{\emptyset}^{(C,E)}(\sigma(C) \times \{1\})} (T_{(1,E)} \circ T_{(i;C)})(\emptyset)
    \\
    &= \sum_{i \in \sigma(C)} \frac{P_{\emptyset}^{C}(\{i\} \cap \Gamma)}{P_{\emptyset}^{C}(\sigma(C) \cap \Gamma)} (T_{(1,E)} \circ T_{(i;C)})(\emptyset) = \sum_{i \in \Gamma} P^{C}_{\emptyset}(\{i\}\vert \Gamma)T_{(i;C)}(\emptyset)
    \\
    &=T_{(\Gamma;C)}(\emptyset).
\end{align*}
\end{proof}

\begin{proposition}\label{prop: updatingMixed} Let E be a projection, and let $(\Delta,A)$ be any observable event associated with $E$, i.e., $E=\chi_{\Delta}(A)$. Then
\begin{align}
    T_{(1;E)}(\emptyset) = T_{(\Delta;A)}(\emptyset).
\end{align}
In particular, the projective state $\emptyset_{E}$ associated with $E$ (see definition~\ref{def: projectiveState}) coincide with the state $\emptyset_{\frac{E}{\Tra{E}}}$ defined in proposition~\ref{prop: preBorn}.
\end{proposition}
\begin{proof}
    Let $E$ be a projection, and let $(\Delta;A)$ be an event satisfying $E=\chi_{\Delta}(A)$. Recall that $E = \sum_{\af \in \Delta}E_{\af}$, where $E_{\af} \equiv \chi_{\{\af\}}(A)$. Now let $C$ be any nondegenerate fine-graining of $A$, i.e., C is nondegenerate and there exists an arrow $C \xrightarrow{f} A$ in the category of observables. We have $E = (\chi_{\Delta} \circ f)(C) = \chi_{f^{-1}(\Delta)}(C) = \sum_{i \in f^{-1}(\Delta)} F_{i}$, where $F_{i} \equiv \chi_{\{i\}}(C)$ for each $i \in \sigma(C)$, and  $E_{\af} = \chi_{f^{-1}(\af)}(C) = \sum_{i \in f^{-1}(\af)} F_{i}$. Therefore, according to proposition~\ref{prop: subjectiveUpdate} and lemma~\ref{lemma: UpdatingMixed},
    \begin{align*}
        T_{(1;E)}(\emptyset) &= T_{(f^{-1}(\Delta);C)}(\emptyset) = \sum_{i \in f^{-1}(\Delta)} P_{\emptyset}^{C}(\{i\}\vert f^{-1}(\Delta)) T_{(i;C)}(\emptyset)
        \\
        &= \sum_{i \in f^{-1}(\Delta)} \frac{\langle F_{i} \rangle_{\emptyset}}{\langle E \rangle_{\emptyset}} T_{(i;C)}(\emptyset) = \sum_{\af \in \Delta} \frac{\langle E_{\af}\rangle_{\emptyset}}{\langle E \rangle_{\emptyset}}\sum_{i \in f^{-1}(\af)}\frac{\langle F_{i}\rangle_{\emptyset}}{\langle E_{\af}\rangle_{\emptyset}} T_{(i;C)}(\emptyset)
        \\
        &= \sum_{\af \in \Delta} \frac{\langle E_{\af}\rangle_{\emptyset}}{\langle E \rangle_{\emptyset}}\sum_{i \in f^{-1}(\af)}P_{\emptyset}^{C}(\{i\}\vert f^{-1}(\af)) T_{(i;C)}(\emptyset) = \sum_{\af \in \Delta} \frac{\langle E_{\af}\rangle_{\emptyset}}{\langle E \rangle_{\emptyset}}T_{(f^{-1}(\af);C)}(\emptyset)
        \\
        &= \sum_{\af \in \Delta} P_{\emptyset}^{A}(\{\af\}\vert \Delta)T_{(1;E_{\af})}(\emptyset) = T_{(\Delta;A)}(\emptyset).
    \end{align*}
\end{proof}

\begin{corollary}\label{cor: NoUpdate} Let $\emptyset$ be the completely mixed state. Then, for any observable $A$,
\begin{align}
    T_{(\sigma(A),A)}(\emptyset) = \emptyset.
\end{align}
\end{corollary}
\begin{proof}
    We know that, for any observable $A$, $\mathds{1}=\chi_{\sigma(A)}(A)$, thus it follows from proposition~\ref{prop: updatingMixed} that $T_{(1,\mathds{1})}(\emptyset) = T_{(\sigma(A);A)}(\emptyset)$. Lemma~\ref{lemma: unitUpdate} ensures that $T_{(1,\mathds{1})}(\emptyset) = \emptyset$, which completes the proof. 
\end{proof}

Let $A$ be any density operator, and let $A=\sum_{i=1}^{m}\af_{i}E_{i}$ be any convex decomposition of $A$ in terms of pairwise orthogonal rank-$1$ projections, i.e., $\sum_{i=1}^{m}\af_{i}=1$ and the rank-$1$ projections $E_{1},\dots,E_{m}$ are pairwise orthogonal. For any rank-$1$ projection $G$ we have
\begin{align*}
    \langle A \rangle_{\emptyset_{G}} &= \sum_{i=1}^{m}\af_{i}\langle E_{i} \rangle_{\emptyset_{G}} = \sum_{i=1}^{m}\af_{i}\langle G\rangle_{\emptyset_{ E_{i}}} = \langle G\rangle_{\sum_{i=1}^{m}\af_{i}\emptyset_{E_{i}}}.
\end{align*}
This is satisfied by any decomposition of $A$, thus, if $\sum_{i=1}^{m}\af_{i}E_{i}=A=\sum_{i=1}^{m'}\beta_{j}F_{j}$ are both convex decompositions of $A$ in terms of pairwise orthogonal rank-$1$ projections, we have
\begin{align*}
    \langle F\rangle_{\sum_{i=1}^{m}\af_{i}\emptyset_{E_{i}}} =  \langle F\rangle_{\sum_{i=1}^{m'}\beta_{j}\emptyset_{F_{j}}} 
\end{align*}
for any rank-$1$ projection $F$. Finally, it follows from propositions~\ref{prop: expLinear} and~\ref{prop: nondegenerateDecomposition} that this equation is valid for any projection $F$, thus lemma~\ref{lemma: expectationAsState} ensures that $\sum_{i=1}^{m}\af_{i}\emptyset_{E_{i}} = \sum_{i=1}^{m'}\beta_{j}\emptyset_{F_{j}}$. Hence, thanks to equation~\ref{eq: ordering}, the ambiguity in the definition of $\emptyset_{A}$ (see proposition~\ref{prop: preBorn}) mentioned in section~\ref{sec: BornRule} disappears, allowing us to introduce the following definition.

\begin{definition}[Density operators as states]\label{def: densityState} Let $A$ be a density operator in a $n$-dimensional system. We denote by $\emptyset_{A}$ the unique state satisfying
\begin{align}
    \emptyset_{A}&= \sum_{i=1}^{m}\af_{i}\emptyset_{E_{i}}
\end{align}
for every convex decomposition $A=\sum_{i=1}^{m}\af_{i}E_{i}$ of $A$ in terms of pairwise orthogonal rank-$1$ projections.    
\end{definition}

Let $E$ be any projection. Proposition~\ref{prop: updatingMixed} asserts that the projective state $\emptyset_{E}$ satisfies $\emptyset_{E} = T(\Delta;A)(\emptyset)$ for any event $\Delta;A)$ associated with $E$, and it easily follows from this result that, if $E_{1},\dots,E_{m}$ are pairwise orthogonal projections that sum to $E$, i.e., $E = \sum_{i=1}^{m}E_{i}$, then
\begin{align}
    \emptyset_{E} = \sum_{i=1}^{m}\frac{\Tra{E_{i}}}{\Tra{E}}\emptyset_{E_{i}},
\end{align}
where $\emptyset_{E_{i}}$ denotes the projective state associated with $E_{i}$. It is worth emphasizing this result.

\begin{lemma}[Projective states]\label{lemma: projectiveStates} Let $E$ be any projection, and let $\emptyset_{E} \equiv T_{E}(\emptyset)$ be the projective state associate with it (definition~\ref{def: projectiveState}). Let $\emptyset_{\frac{E}{\Tra{E}}}$ be the state defined by the density operator $\frac{E}{\Tra{E}}$, as in definition~\ref{def: densityState}. Then $\emptyset_{E} = \emptyset_{\frac{E}{\Tra{E}}}$. Furthermore, if $E_{1},\dots,E_{m}$ are pairwise orthogonal projections that sum to $E$, i.e.,  $E = \sum_{i=1}^{m}E_{i}$, then
\begin{align}
    \emptyset_{E} = \sum_{i=1}^{m}\frac{\Tra{E_{i}}}{\Tra{E}}\emptyset_{E_{i}},
\end{align}
where $\emptyset_{E_{i}}$ denotes the projective state associated with $E_{i}$.
\end{lemma}

All explicit examples of states we have presented thus far are given by definition~\ref{def: densityState}, and consequently they all can  be written as convex combinations of pure states. On the other hand, postulate~\ref{post: observables} asserts that any state is a convex combination of experimentally accessible states (definition~\ref{def: accessibleState}), so, if we show that experimentally accessible states are convex combinations of pure states, we can immediately conclude that every state of the system satisfies the same condition. To assure that experimentally accessible states are  convex combinations of pure states, it is sufficient to assume that,  as in quantum theory, a measurement can only turn a pure state into a non-pure one if we  fail to acquire all the information that this measurement provides, i.e., if we update the pure state using a subjective event. It means that the transition from pure and non-pure states via measurements always involves some sort of loss of information. This is the last key feature of quantum systems that we assume in order to embed our system in a Hilbert space:

\begin{postulate}[Pure states]\label{post: pureState} Let $\rho$ be a pure state, and let $(\af;A)$ be any objective event. Then $T_{(\af;A)}(\rho)$ is a pure state. Equivalently, if $\rho$ is a pure state and $E$ is a projection, $T_{E}(\rho)$ is a pure state. 
\end{postulate}

\begin{proposition}\label{prop: accessibleState}
    Any experimentally accessible state (definition~\ref{def: accessibleState}) is a convex combination of pure states.
\end{proposition}
\begin{proof}
    Let $\rho$ be an experimentally accessible state, and let $(\Delta_{i};A_{i})$, $i=1,\dots,m$ be a sequence of observable events such that
\begin{align*}
    \rho = (T_{(\Delta_{m},A_{m})} \circ \dots \circ T_{(\Delta_{1},A_{1})})(\emptyset).
\end{align*}
Write $\underline{A} \equiv (A_{1},\dots,A_{m})$, $\underline{\Delta} \equiv \Delta_{1} \times \dots \times \Delta_{m}$, and $T_{(\underline{\Delta};\underline{A})} \equiv T_{(\Delta_{m},A_{m})} \circ \dots \circ T_{(\Delta_{1},A_{1})}$. It follows from proposition~\ref{prop: subjectiveSequentialUpdate} that
\begin{align*}
    \rho &=\sum_{\underline{\af} \in \underline{\Delta}} P_{\emptyset}^{\underline{A}}(\{\underline{\af}\} \vert \underline{\Delta}) T_{(\underline{\af},\underline{A})}(\emptyset),
\end{align*}
where $\underline{\af} \equiv (\af_{1},\dots,\af_{m})$ and  $T_{(\underline{\af};A)} = T_{(\af_{m};A_{m})} \circ \dots \circ T_{(\af_{1};A_{1})}$. Now let $A_{0}$ be any nondegenerate observable. With a slight abuse of notation, let's write $\underline{\Delta}^{(0)} \equiv \sigma(A_{0}) \times \dots \times \Delta_{m}$, $\underline{A}^{(0)} \equiv (A_{0},\dots,A_{m})$ and $\underline{\af}^{(0)}\equiv (\af_{0},\dots,\af_{m}) \in \prod_{i=0}^{m}\sigma(A_{i})$. According to proposition~\ref{prop: subjectiveSequentialUpdate} and lemma~\ref{lemma: nondegenerateNonUpdate},
\begin{align*}
    \rho &=\sum_{\underline{\af} \in \underline{\Delta}} P_{\emptyset}^{\underline{A}}(\{\underline{\af}\} \vert \underline{\Delta}) T_{(\underline{\af},\underline{A})}(T_{(\sigma(A_{0});A_{0})}(\emptyset)) = \sum_{\underline{\af}^{(0)} \in \underline{\Delta}^{(0)}} P_{\emptyset}^{\underline{A}^{(0)}}(\{\underline{\af}^{(0)}\} \vert \underline{\Delta}^{(0)}) T_{(\underline{\af}^{(0)};\underline{A}^{(0)})}(\emptyset))
    \\
    &=\sum_{\underline{\af}^{(0)} \in \underline{\Delta}^{(0)}} P_{\emptyset}^{\underline{A}^{(0)}}(\{\underline{\af}^{(0)}\} \vert \underline{\Delta}^{(0)}) \rho_{\underline{\af}^{(0)}},
\end{align*}
where, for each $\underline{\af}^{(0)} \in \underline{\Delta}^{(0)}$, we define $\rho_{\underline{\af}^{(0)}} \doteq T_{(\underline{\af}^{(0)};\underline{A}^{(0)})}(\emptyset))$. We have $\rho_{\underline{\af}^{(0)}} = T_{(\underline{\af};\underline{A})}(\emptyset_{E_{\af_{0}}})$,
where $E_{\af_{0}}$ is the rank-$1$ projection $\chi_{\{\af_{0}\}}(A_{0})$, thus it follows by finite induction that  $\rho_{\underline{\af}^{(0)}}$ is a pure state, which completes the proof. 
\end{proof}

\begin{corollary}\label{cor: convexHull} Any state can be written as a convex combination of pure states.
\end{corollary}
\subsection{Local orthogonality and the partially ordered set of projections}\label{sec: latticeOfProjections}

The first part of postulate~\ref{post: transitionProbability}, which states that transition probabilities do not depend on the order of the transition, enables us to show that the set of projections is a partially ordered set whose ordering is defined by functional relations. The results we prove in this section will be explored at length in the next section.

Let $E,F$ be compatible projections, and let $E \xleftarrow{\chi_{\Delta}} A \xrightarrow{\chi_{\Sigma}} F$ be a cone for them. We have $\Delta \subset \Sigma$ if and only if $\Delta \cap \Sigma = \Delta$, which in turn is equivalent to saying that $\chi_{\Delta \cap \Sigma}(A) = \chi_{\Delta}(A)$, that is, $E \circ F = E$. Given any other cone  $E \xleftarrow{\chi_{\Delta'}} B \xrightarrow{\chi_{\Sigma'}} F$ for $E$ and $F$, we know that $E \circ F = \chi_{\Delta' \cap \Sigma'}(B)$ and $\chi_{\Delta'}(B) = E$, thus $\Delta \subset \Sigma$ if and only if $\Delta' \subset \Sigma'$. It shows that the following relation is well-defined.

\begin{definition}[Partial order]\label{def: partialOrder} Let $E,F$ be projections of the system $\mathfrak{S}$. We say that $E$ is smaller than or equals to $F$, denoted $E \leq F$, if and only if $E$ and $F$ are compatible and $E \circ F = E$, which is equivalent to saying that, given any cone  $E \xleftarrow{\chi_{\Delta}} A \xrightarrow{\chi_{\Sigma}} F$ for $E$ and $F$, we have $\Delta \subset \Sigma$. This definition canonically induces a relation $\leq$ in the set $\mathcal{P}$ of projections of the system.  
\end{definition}

For any projection $E$ we have $E \circ E = E$, thus $E \leq E$, which shows that the relation $\leq$ is reflexive. Similarly, if two projections $E,F$ satisfy $E \leq F$ and $F \leq E$, we have $E = E \circ F = F \circ E = F$, which proves that $\leq$ is antisymmetric. To show that $\leq$ is a partial order, as the title of definition~\ref{def: partialOrder} suggests, we have to prove that $\leq$ is transitive. In order to do that, we need some preliminary results. The first one is a straightforward generalization of proposition~\ref{prop: orthogonalityTransition}:

\begin{lemma}[Orthogonality]\label{lemma: orthogonality} Let $E,F$ be projections. The following conditions are equivalent.
\begin{itemize}
    \item[(a)] $E$ and $F$ are orthogonal (definition~\ref{def: orthononality})
    \item[(b)] The ``probability of transitioning between $E$ and $F$'' is zero, that is, $\langle F \rangle_{E}$.
\end{itemize}
\end{lemma}
\begin{proof}
If $E$ and $F$ are orthogonal, they are compatible, thus lemmas~\ref{lemma: traceSimple} and~\ref{lemma: projectionJointProbability} imply that
\begin{align*}
    \langle F\rangle_{E} = \frac{P_{\emptyset}(1,1;F,E)}{P_{\emptyset}(1,F)} = \frac{\Tra{FE}}{\Tra{F}} = \frac{\Tra{0}}{\Tra{F}} = 0,
\end{align*}
so item $(a)$ implies item $(b)$. On the other hand, assume that the projections $E$ and $F$ satisfy item $(b)$, and let $E=\sum_{i=1}^{m} E_{i}$ and $F=\sum_{i=1}^{k}F_{i}$ be decompositions of $E$ and $F$ in terms of pairwise orthogonal rank-$1$ projections (see proposition~\ref{prop: nondegenerateDecomposition}). According to lemma~\ref{lemma: projectiveStates}, we have
\begin{align*}
    \emptyset_{E} &= \sum_{i=1}^{m} \frac{\Tra{E_{i}}}{\Tra{E}}\emptyset_{E_{i}} = \sum_{i=1}^{m}\frac{1}{\Tra{E}}\emptyset_{E_{i}}
\end{align*}
and $\emptyset_{F} = \sum_{i=1}^{k} \frac{1}{\Tra{F}}\emptyset_{F_{i}}$ (recall that $m=\Tra{E}$ and $k=\Tra{F}$). Hence,
\begin{align*}
    \langle F \rangle_{E} &= \frac{1}{\Tra{E}}\sum_{j=1}^{k}\sum_{i=1}^{m} \langle F_{j}\rangle_{E_{i}}.
\end{align*}
We have $\langle F \rangle_{E}=0$ if and only if  $\langle F_{j}\rangle_{E_{i}}=0$ for every pair $i,j$, and, according to proposition~\ref{prop: orthogonalityTransition}, this is equivalent to saying that $E_{i}\perp F_{j}$. Using a cone for $E_{1},\dots,E_{m},F_{1},\dots,F_{k}$ (see corollary~\ref{cor: specker}), one can easily show that $E \perp F$, thus the proof is complete.
\end{proof}

Thanks to lemma~\ref{lemma: orthogonality}, we can use states to characterize orthogonality:

\begin{lemma}[Orthogonality via states]\label{lemma: stateOrthogonality} Let $E$, $F$ be projections. The following claims are equivalent.
\begin{itemize}
    \item[(a)] $E$ and $F$ are orthogonal (definition~\ref{def: orthononality}).
    \item[(b)] for every state $\rho$,
    \begin{align}\label{eq: pairwiseExclusivity}
        \Exp{E} + \Exp{F} \leq 1.
    \end{align}
\end{itemize}
\end{lemma}
\begin{proof}
The additivity of probability measures ensures that item $(a)$ implies item $(b)$. On the other hand, assume that item $(b)$ is satisfied, and consider the projective state $\emptyset_{E} \equiv T_{(1,E)}(\emptyset)$ (definition~\ref{def: projectiveState}). We have $\langle E \rangle_{E} = 1$, thus the validity of equation~\ref{eq: pairwiseExclusivity}  for the state $\rho=\emptyset_{E}$ implies that $\langle F \rangle_{E}= 0$, thus, according to lemma~\ref{lemma: orthogonality}, $E$ and $F$ are orthogonal.
\end{proof}

More importantly, we can characterize the relation $\leq$ (definition~\ref{def: partialOrder}) using states:

\begin{proposition}[Ordering]\label{prop: ordering}
    Let $E,F$ be projections. Then the following claims are equivalent.
    \begin{itemize}
        \item[(a)] $E \leq F$ (definition~\ref{def: partialOrder}).
        \item[(b)] For every state $\rho$,
        \begin{align}
            \Exp{E} \leq\label{eq: ordering} \Exp{F}.
        \end{align}
    \end{itemize}
\end{proposition}
\begin{proof}
    The monotonicity of probability measures \cite{klenke2020probability} ensures that item $(a)$ implies item $(b)$. Now let $E,F$ be projections satisfying item $(b)$, and let $\rho$ be any state. According to lemma~\ref{lemma: orthocomplement}, $\Exp{F} = 1-\Exp{F^{\perp}}$, thus we have $\Exp{E} \leq \Exp{F}$ if and only if $\Exp{E} + \Exp{F^{\perp}} \leq 1$, which, according to lemma~\ref{lemma: stateOrthogonality}, is equivalent to saying that $E$ and $F^{\perp}$ are orthogonal.Finally, $E \perp F^{\perp}$ if and only if, for each cone $E \xleftarrow{\chi_{\Delta}} A \xrightarrow{\chi_{\Sigma}} F^{\perp}$, we have $\Delta \cap \Sigma = \oldemptyset$, or equivalently $\Delta \subset \sigma(A) \backslash \Sigma$, and since $F = \chi_{\sigma(A)\backslash \Sigma}$ (see lemma~\ref{lemma: orthocomplement}), we have $E \leq F$, which completes the proof. 
\end{proof}

It immediately follows from proposition~\ref{prop: ordering} that $\leq$ is  transitive, so $\leq$ is indeed a partial order in $\mathcal{P}$. Consequently, we have a partially ordered set (poset) of projections. 

Let $E_{1},\dots E_{m}$ be pairwise compatible projections, and let $F$ be a projection satisfying $F \leq E_{i}$ for every $i$. According to Specker's principle (corollary~\ref{cor: specker}) and lemma~\ref{lemma: projection}, there is an observable $A$ and sets $\Delta_{1},\dots,\Delta_{m},\Sigma \subset \sigma(A)$ such that $E_{i}=\chi_{\Delta_{i}}(A)$ for all $i$ and $F = \chi_{\Sigma}(A)$. According to definition~\ref{def: partialOrder}, $\Sigma \subset \Delta_{i}$ for every $i$, thus $\Sigma \subset \cap_{i=1}^{m} \Delta_{i}$, which implies that $\chi_{\Sigma}(A) \leq \chi_{\cap_{i=1}^{m} \Delta_{i}}(A)$, or equivalently $F \leq \prod_{i=1}^{m}E_{i}$. It proves that any set of pairwise compatible projections $\{E_{1},\dots,E_{m}\}$ has an \textbf{infimum} $\wedge_{i=1}^{m} E_{i} \equiv \wedge\{E_{1},\dots,E_{m}\}$ in the poset $\mathcal{P}$, which is given by the product $\prod_{i=1}^{m}E_{i}$, i.e.,
\begin{align}
    \bigwedge_{i=1}^{m}E_{i} =\label{eq: infimum} \prod_{i=1}^{m}E_{i} = \chi_{\cap_{i=1}^{m} \Delta_{i}}(A).
\end{align}
It is analogous to show that any set of pairwise compatible projections $\{E_{1},\dots,E_{m}\}$ has a \textbf{supremum} $\vee_{i=1}^{m} E_{i} \equiv \vee\{E_{1},\dots,E_{m}\}$ in the poset $\mathcal{P}$, which satisfies
\begin{align}
    \bigvee_{i=1}^{m}E_{i} = \chi_{\cup_{i=1}^{m} \Delta_{i}}(A)
\end{align}
for any cone $A\xrightarrow{\chi_{\Delta_{i}}}E_{i}$, $i=1,\dots,m$, of $E_{1},\dots,E_{m}$. 

The unit $\mathds{1}$ and the zero $0$ are \textbf{top} and \textbf{bottom} elements in the poset $\mathcal{P}$  respectively, i.e., for each projection $E$, we have  $0 \leq E \leq \mathds{1}$.  

It easily follows from proposition~\ref{prop: ordering} and lemma~\ref{lemma: traceSimple} that the trace is order-preserving:
\begin{lemma}[Trace and order]\label{lemma: traceAndOrder} Let $E$, $F$ be projections. If $E \leq F$, then $\Tra{E} \leq \Tra{F}$.    
\end{lemma}

Although we proved the exclusivity principle (corollary~\ref{cor: exclusivityPrinciple}) back in section~\ref{sec: spectralTheory} --- and, as we mentioned, it could have been proved even earlier, in section~\ref{sec: specker} ---, only now we can show that our system satisfies the important principle of ``local orthogonality'' \cite{fritz2013local}. The reason is that, as far as we can tell, postulates \ref{ax: separability}-\ref{post: commutativity} do not imply that the equivalence we proved in proposition~\ref{prop: ordering} is valid, and therefore, if $\{E_{1},\dots,E_{m}\}$ and $\{F_{1},\dots,F_{k}\}$ are both sets of pairwise orthogonal projections and $E_{i} \perp F_{j}$ for some pair of indexes $i,j$, we cannot assure, based only on those postulates, that $\prod_{i'=1}^{m}E_{i'}$  and $\prod_{j'}^{k}F_{j'}$ (or equivalently $\wedge_{i'=1}^{m}E_{i'}$  and $\wedge_{j'}^{k}F_{j'}$) are orthogonal. Now that we have proposition~\ref{prop: ordering}, this result immediately follows from the following lemma.

\begin{lemma}\label{lemma: orderAndOrthogonality} Let $E,E',F,F'$ be projections satisfying $E' \leq E$ and $F' \leq F$. If $E \perp F$, then $E \perp E'$.
\end{lemma}
\begin{proof}
    Let $E,E',F,F'$ be projections satisfying $E' \leq E$,  $F' \leq F$ and $E \perp F$. According to proposition~\ref{prop: ordering}, $P(E')+P(F') \leq P(E) + P(F) \leq 1$, thus it follows from lemma~\ref{lemma: stateOrthogonality} that $E' \perp F'$.
\end{proof}

More broadly:
\begin{corollary} Let  $(E_{1},\dots,E_{m})$ and $(F_{1},\dots,F_{m})$ both be sequences of pairwise compatible projections, and suppose that $E_{i} \leq F_{i}$ for each $i$. If $F_{1},\dots,F_{m}$ are pairwise orthogonal, then $E_{1},\dots,E_{m}$ are pairwise orthogonal.
\end{corollary}

In our framework, local orthogonality corresponds to the following result.

\begin{proposition}[Local orthogonality] Let $(A_{1}^{(j)},\dots,A_{m}^{(j)})$ be a sequence of pairwise compatible observables for each $j \in \{1,\dots, k\}$, and let $(\af_{1}^{(j)},\dots,\af_{m}^{(j)}) \in \prod_{i=1}^{m}\sigma(A_{i})$. Suppose that, for each pair of distinct indexes $j,j' \in \{1,\dots,m\}$, we have $A_{i}^{(j)}=A_{i}^{(j')}$ and $\af_{i}^{(j)} \neq \af_{i}^{(j')}$ for some $i \in \{1,\dots,m\}$. Then, for every state $\rho$,
\begin{align}
    \sum_{j=1}^{k} p_{\rho}(\af_{1}^{(j)},\dots,\af_{m}^{(j)};A_{1}^{(j)},\dots,A_{m}^{(j)}) \leq 1.
\end{align}
\end{proposition}
\begin{proof}
    For each $j \in \{1,\dots,k\}$ and  $i \in \{1,\dots,m\}$, denote by $E^{(j)}_{i}$ the projection associated with the observable event $(\af_{i}^{(j)};A_{i}^{(j)})$. 
    According to equation~\ref{eq: infimum} and proposition \ref{prop: ordering}, for each $j \in \{1,\dots,k\}$ and $i \in \{1,\dots,m\}$ we have $\prod_{i'=1}^{m} E^{(j)}_{i'} \leq E^{(j)}_{i}$. For any pair of distinct indexes $j,j' \in \{1,\dots,m\}$, denote by $i^{(j,j')}$ some of the indexes $i \in \{1,\dots,m\}$ for which $A_{i}^{(j)}=A_{i}^{(j')}$ and $\af_{i}^{(j)} \neq \af_{i}^{(j')}$. We have $E^{(j)}_{i^{(j,j')}} \perp E^{(j')}_{i^{(j,j')}}$ whenever $j \neq j'$, and therefore the projections $\prod_{i=1}^{m} E^{(1)}_{i}, \dots, \prod_{i=1}^{m}E^{(k)}_{i}$ are pairwise orthogonal. Finally, lemma~\ref{lemma: projectionJointProbability} and the exclusivity principle (corollary~\ref{cor: exclusivityPrinciple}) complete the proof.
\end{proof}

We can finally show that any system satisfying our postulates is a quantum system, by which we mean that it can be embedded in a Hilbert space. More precisely, we show that, if $\mathfrak{S}$ is a $n$-dimensional system satisfying postulates \ref{ax: separability}-\ref{post: pureState}, then its set of observables $\mathcal{O}$ can be embedded in $\mathcal{B}(H)_{\text{sa}}$ for some $n$-dimensional Hilbert space $H$ in such a way that everything that we have defined and proved about $\mathcal{O}$ coincides with its counterpart in $\mathcal{B}(H)_{\text{sa}}$, and that the set of states of $\mathfrak{S}$ can be embedded in the  set of density operators of $H$ so that the expectation we defined corresponds to the Born rule. This is stated more precisely in theorem~\ref{thm: quantumEmbedding}. In corollary~\ref{cor: quantumCorrelations}, we just emphasized that we have what is called ``quantum correlations'' \cite{amaral2018graph, csw2014graph}.

\section{Quantum mechanics}\label{sec: quantumTheory}

In this section we will assume that the reader is familiar with the standard terminology of functional analysis, so we will use terms like ``operator'', ``projection operator'', ``trace'',  and ``basis'' without explaining what we mean by them. Since we used the same terminology to refer to the definitions we introduced throughout the paper, there will be some ambiguity in our discussion here. However, we believe that the context will make it clear whether we are referring to our definitions, which apply to the system $\mathfrak{S}$, or to their counterparts in functional analysis, which apply to a Hilbert space $H$. To emphasize the distinction between $\mathfrak{S}$ and $H$, we will denote by $\mathcal{O}(\mathfrak{S})$, $\mathcal{P}(\mathfrak{S})$ and $\mathcal{S}(\mathfrak{S})$ the sets of observables, projections and states of $\mathfrak{S}$ respectively, whereas $\mathcal{B}(H)_{\text{sa}}$, $\mathcal{P}(H)$ and $\mathcal{D}(H)$ denote the sets of selfadjoint operators (i.e., quantum observables), projections and density operators (i.e., quantum states) in $H$ respectively. 

Let $\mathcal{P}_{1}(\mathfrak{S})$ be the set of all rank-$1$ projections in a $n$-dimensional system $\mathfrak{S}$. Fix a set $\mathcal{E} \equiv \{E_{1},\dots,E_{n}\}$ of $n$ pairwise orthogonal rank-$1$ projections in $\mathfrak{S}$, which exists according to lemma~\ref{lemma: dimension}. Let $H$ be a $n$ dimensional Hilbert space, and fix an orthonormal basis $\psi_{1},\dots,\psi_{n}$ of $H$. For each $F \in \mathcal{P}_{1}(\mathfrak{S})$, associate a vector $\psi_{F} \in H$ satisfying $\vert \langle \psi_{i}\vert \psi_{F} \rangle \vert^{2} = P(F \rightarrow E_{i})$ for every $i$ (see definition~\ref{def: transitionProbability}) in such a way that the phase of $\langle \psi_{i}\vert \psi_{F}\rangle$ is a phase shift between $F$ and $E_{i}$. That is, fix a function $\mathcal{P}_{1}(\mathfrak{S}) \ni G \xmapsto{\theta} \theta_{G} \in [0,2\pi)$ and, for each $F \in \mathcal{P}_{1}(\mathfrak{S})$, let $\psi_{F} \in H$ be the unique vector such that
\begin{align}\label{eq: component}
    \langle \psi_{i}\vert \psi_{F}\rangle &= \sqrt{P(F \rightarrow E_{i})} e^{i(\theta_{E_{i}} - \theta_{F})}.
\end{align}
for all $i$. It defines a mapping
\begin{align}
    \mathcal{P}_{1}(\mathfrak{S}) \ni F \xmapsto{\Psi} \Psi(F) \equiv \psi_{F} \in H.
\end{align}
For each $F \in \mathcal{P}_{1}(\mathfrak{S})$, $\psi_{F}$ is by construction a unit vector, because $\Vert \psi_{F} \Vert^{2} = \sum_{i=1}^{n} \vert \braket{\psi_{i}}{\psi_{F}}\vert^{2} = \sum_{i=1}^{n}P(F \rightarrow E_{i}) = \sum_{i=1}^{n} \langle E_{i} \rangle_{F} = 1$. Furthermore, according to proposition~\ref{prop: orthogonalityTransition}, for any $j$ we have
\begin{align}
    \psi_{E_{j}} &= \sum_{i=1}^{n} \langle \psi_{i}\vert \psi_{E_{j}} \rangle \psi_{i}= \sum_{i=1}^{n}  \sqrt{P(E_{j} \rightarrow E_{i})} e^{i(\theta_{E_{i}} - \theta_{E_{j}})} \psi_{i} = \sum_{i=1}^{n} \delta_{i,j}e^{i(\theta_{E_{i}}-\theta_{E_{j}})} \psi_{i} = \psi_{j},
\end{align}
i.e., $\Psi$ assigns the rank-1 projections $E_{1},\dots,E_{n} \in \mathcal{E}$ to the basis of $H$ that we fixed. From now on, unless explicitly stated otherwise, by $\psi_{i}$ and $\theta_{i}$ we mean $\psi_{E_{i}}$ and $\theta_{E_{i}}$ respectively. We say that the mapping $\Psi$ is a \textbf{vector assignment defined by the set} $\mathcal{E}$. 

Let $F,G$ be rank-$1$ projections in $\mathfrak{S}$. Then
\begin{align}
    \braket{\psi_{F}}{\psi_{G}} &= \sum_{i}^{n} \braket{\psi_{F}}{\psi_{i}}   \braket{\psi_{i}}{\psi_{G}} = \sum_{i=1}^{n} \sqrt{P(F \rightarrow E_{i}) P(G \rightarrow E_{i})}e^{-i(\theta_{i} - \theta_{F})}e^{i(\theta_{i}-\theta_{G})}
    \\
    &=\label{eq: innerProduct} e^{i(\theta_{F} - \theta_{G})}\sum_{i=1}^{n}\sqrt{P(F \rightarrow E_{i} \rightarrow G)}
\end{align}
It leads us to the following lemma.

\begin{lemma}[Transition probability]\label{lemma: transition} Let $\mathfrak{S}$ be a $n$-dimensional system and $H$ a $n$-dimensional Hilbert space. Let $\Psi: \mathcal{P}_{1}(\mathfrak{S}) \rightarrow H$ be a vector assignment defined by a set $\mathcal{E} \equiv \{E_{1},\dots,E_{n}\} \subset \mathcal{P}_{1}(\mathfrak{S})$ of pairwise orthogonal rank-$1$ projections, as above. Then, for any $F,G \in \mathcal{P}_{1}(\mathfrak{S})$,
\begin{align}
    P(F \rightarrow G)&= \vert \langle \psi_{F} \vert \psi_{G} \rangle \vert^{2},
\end{align}
where $\psi_{J} \equiv  \Psi(J)$ for all $J \in \mathcal{P}_{1}(\mathfrak{S})$.
\end{lemma}
\begin{proof}
According to lemma~\ref{lemma: transitionSquare} and equations~\ref{eq: component},~\ref{eq: innerProduct},
\begin{align*}
    \vert \langle \psi_{F} \vert \psi_{G} \rangle \vert^{2} &=  \left \vert \sum_{i=1}^{n}\sqrt{P(F \rightarrow E_{i}\rightarrow G)} \right\vert^{2}
    = P(F \rightarrow G).
\end{align*}
\end{proof}

It immediately follows from proposition~\ref{prop: orthogonalityTransition} and lemma~\ref{lemma: transition} that $\psi_{F}$ and $\psi_{G}$ are orthogonal vectors in $H \equiv \mathbb{C}^{n}$ if and only if $F,G$ are orthogonal projections in $\mathfrak{S}$. Furthermore, the same results ensure that $\psi_{F}$, $\psi_{G}$ are linearly dependent if and only if $E = F$, which in turn implies that $\psi_{E}=\psi_{F}$ (because $\Psi$ is a function). Let's emphasize these results.

\begin{corollary}[Vector assignment]\label{cor: vectorAssignment} Let $\mathfrak{S}$ be a $n$-dimensional system and $H$ be a $n$-dimensional Hilbert space. Let $\mathcal{P}_{1}(\mathfrak{S}) \ni F \xmapsto{\Psi} \Psi(F) \equiv \psi_{F} \in H$ be a vector assignment defined by a set $\mathcal{E} \equiv \{E_{1},\dots,E_{n}\} \subset \mathcal{P}_{1}(\mathfrak{S})$ of pairwise orthogonal rank-$1$ projections. Then the following conditions are satisfied.
\begin{itemize}
    \item[(a)] $\Psi$ is injective, and, for any pair $F,G \in \mathcal{P}_{1}(\mathfrak{S})$, $\psi_{F}$ and $\psi_{G}$ are linear dependent if and only if $F=G$ (which in turn implies  $\psi_{F}=\psi_{G}$). 
    \item[(b)] $\Psi$ preserves  and reflects orthogonality, i.e., the vectors $\psi_{F},\psi_{G}$ are orthogonal in $H$ if and only if the projections $F,G$ are orthogonal in $\mathfrak{S}$ (definition~\ref{def: orthononality}). 
    \item[(c)] Let $F_{1},\dots,F_{m}$ be rank-$1$ projections in $\mathfrak{S}$. Then $\psi_{F_{1}},\dots,\psi_{F_{m}}$ are pairwise orthogonal vectors in $H$ if and only if $F_{1},\dots,F_{m}$ are pairwise orthogonal projections in $\mathfrak{S}$. In particular, $\psi_{F_{1}},\dots,\psi_{F_{m}}$ is a basis of $H$ if and and only if $m=n$ and $F_{1},\dots,F_{m} \in \mathcal{P}_{1}(\mathfrak{S})$ are pairwise orthogonal.
\end{itemize}
\end{corollary}

Let $\mathcal{P}_{1}(H)$ be the collection of rank-1 projections of $H$. Define a mapping $\pi_{1}: \mathcal{P}_{1}(\mathfrak{S}) \rightarrow \mathcal{P}_{1}(H)$ by
\begin{align}
    \pi_{1}(E) \doteq \ketbra{\Psi(E)} \equiv \ketbra{\psi_{E}},
\end{align}
where $\ketbra{\psi_{E}}$ denotes the projection onto the subspace spanned by $\Psi(E) \equiv \psi_{E}$. Recall that
\begin{align}
     \ketbra{\psi_{E}} = \sum_{i,j=1}^{n} \braket{\psi_{i}}{\psi_{E}} \braket{\psi_{E}}{\psi_{j}} \ketbra{\psi_{i}}{\psi_{j}},
\end{align}
where $\ketbra{\psi_{i}}{\psi_{j}}$ is the linear mapping given by $\forall_{\phi \in H}: \ketbra{\psi_{i}}{\psi_{j}} \phi \doteq \braket{\psi_{j}}{\phi} \psi_{i}$. According to corollary~\ref{cor: vectorAssignment}, if $E,F \in \mathcal{P}_{1}(\mathfrak{S})$ are distinct, $\psi_{E}$ and $\psi_{F}$ are linearly independent, which implies that the projections $\pi_{1}(E)$, $\pi_{1}(F)$ are different. Furthermore, if $E,F$ are orthogonal,  $\pi_{1}(E)$ and $\pi_{1}(F)$ are orthogonal projections, i.e., $\pi_{1}(E) \pi_{1}(F) = 0$, where, as usual, $\pi_{1}(E) \pi_{1}(F)$ denotes the product of the operators $\pi_{1}(E)$ and  $\pi_{1}(F)$. Hence, the mapping $\pi_{1}: \mathcal{P}_{1}(\mathfrak{S}) \rightarrow \mathcal{P}_{1}(H)$  is injective and preserves orthogonality. It is also important to note that, for any pair $F,G \in \mathcal{P}_{1}(\mathfrak{S})$,
\begin{align}
    \langle G \rangle_{F} = P(F \rightarrow G) = \Tra{\pi_{1}(F)\pi_{1}(G)},
\end{align}
where $\text{tr}$ denotes the trace in $H$. 

Let $G$ be any projection of $\mathfrak{S}$, and let $k$ be its trace (definition~\ref{def: traceProjection}). According to proposition~\ref{prop: nondegenerateDecomposition}, there are $k$ pairwise orthogonal projections $G_{1},\dots,G_{k} \in \mathcal{P}_{1}(\mathfrak{S})$ such that $G = \sum_{i=1}^{k}G_{i}$. Lemma~\ref{lemma: finiteAdditivity}  ensures that, for any rank-$1$ projection $F$,
\begin{align}
    \langle G \rangle_{F} &=\label{eq: extendingPi} \sum_{i=1}^{k} \langle G_{i} \rangle_{F} = \sum_{i=1}^{k} \Tra{\pi_{1}(F) \pi_{1}(G_{i})} =  \Tra{\pi_{1}(F)  \sum_{i=1}^{k}\pi_{1}(G_{i})} = \Tra{\pi_{1}(F) \pi(G)},
\end{align}
where $\pi(G) \equiv \sum_{i=1}^{k} \pi_{1}(G_{i})$. Now suppose that $G = \sum_{i=1}^{k'}G_{i}'$ for some other pairwise orthogonal projections $G_{1}',\dots,G_{k'}' \in \mathcal{P}_{1}(\mathfrak{S})$, and define $\pi(G)' \equiv \sum_{i=1}^{k'} \pi_{1}(G_{i}')$. According to proposition~\ref{prop: traceLinear}, $k=k'$, whereas equation~\ref{eq: extendingPi} implies that, for any $F \in \mathcal{P}_{1}(\mathfrak{S})$,
\begin{align}
    \Tra{\pi_{1}(F) \pi(G) } =\label{eq: extendingPi'}  \langle G \rangle_{F} = \Tra{\pi_{1}(F) \pi(G)'}.
\end{align}
In particular, for any $j \in \{1,\dots,k\}$,
\begin{align}
    1=\langle G \rangle_{G_{j}} &= \Tra{\pi_{1}(G_{j}) \pi(G)'} = \braket{\psi_{G_{j}}}{\pi(G)' \psi_{G_{j}}},
\end{align}
which means that $\pi(G_{j}) \leq \pi(G)'$, where $\leq $ denotes the standard order of projections\footnote{Recall that two (selfadjoint) projections $E,F$ in a separable Hilbert space $H$ satisfy $E \leq F$ if and only if $E(H) \subset F(H)$, which in turn is equivalent to saying that $EF=E$.}. Consequently, $\vee_{i=1}^{k}\pi(G_{j}) \leq \pi(G)'$, where $\vee_{i=1}^{k}\pi(G_{j})$ denotes the supremum of the set $\{\pi(G_{i}): i=1,\dots,k\}$ in the complete orthocomplemented lattice $\mathcal{P}(H)$, and since $\pi(G_{1}), \dots, \pi(G_{k})$ are pairwise orthogonal (see corollary~\ref{cor: vectorAssignment}), we have $\vee_{i=1}^{k}\pi(G_{i}) = \sum_{j=1}^{n} \pi(G_{j})$. We have shown that $\pi(G) \leq \pi(G)'$, and it is analogous to prove that $\pi(G)' \leq \pi(G)$, therefore $\pi(G) = \pi(G)'$. It enables us to extent $\pi_{1}$ to a mapping $\pi: \mathcal{P}(\mathfrak{S}) \rightarrow \mathcal{P}(H)$ given by
\begin{align}
    \pi(G) \doteq \sum_{i=1}^{k} \pi_{1}(G_{i})
\end{align}
if $G \in \mathcal{P}(\mathfrak{S})\backslash \{0\}$, where $\{G_{1},\dots,G_{k}\}$ is any set of pairwise orthogonal rank-$1$ projections in $\mathfrak{S}$ satisfying $G = \sum_{i=1}^{k} G_{i}$, and $\pi(0) \doteq 0$\footnote{The symbol $0$ at the left-hand side of this equation represents the null observable in $\mathfrak{S}$ (see section~\ref{sec: projections}), whereas at the right-hand side, it represents the projection $0 \in \mathcal{P}(H)$}. It leads us to the following definition.

\begin{definition}[Projection assignment]\label{def: projectionAssignment} The projection assignment induced by the vector assignment $\mathcal{P}_{1}(\mathfrak{S}) \ni F \xmapsto{\Psi} \Psi(F) \equiv \psi_{F} \in H$ consists in the unique extension $\pi: \mathcal{P}(\mathfrak{S}) \rightarrow \mathcal{P}(H)$ of $\mathcal{P}_{1}(\mathfrak{S}) \ni F \xmapsto{\pi_{1}} \pi(F) \equiv \ketbra{\psi_{F}} \in \mathcal{P}_{1}(H)$ such that, for any set $E_{1},\dots,E_{m} \in \mathcal{P}_{1}(\mathfrak{S})$ of pairwise orthogonal rank-$1$ projections, we have
\begin{align}
    \pi(\sum_{i=1}^{m}E_{i}) = \sum_{i=1}^{m}\pi(E_{i}).
\end{align}
\end{definition}
It is worth emphasizing the following results.
\begin{lemma}\label{lemma: additivityVectorAssignment} Let $\pi:\mathcal{P}(\mathfrak{S}) \rightarrow \mathcal{P}(H)$ be the projection assignment induced by the vector assignment $\Psi$. If $E_{1},\dots,E_{m} \in \mathcal{P}(\mathfrak{S})$ are pairwise orthogonal, we have
\begin{align}
    \pi(\sum_{i=1}^{m}E_{i}) = \sum_{i=1}^{m}\pi(E_{i}).
\end{align}
\end{lemma}
\begin{proof}
    For each $i$, let $k_{i}$ be the rank of $E_{i}$. According to proposition~\ref{prop: nondegenerateDecomposition}, for every $i$ we have $E_{i} = \sum_{j=1}^{k_{i}} F_{j}^{(i)}$, where $F_{1}^{(i)},\dots,F_{k_{i}}^{(i)}$ are pairwise orthogonal rank-$1$ projections. Also, it is easy to see that $F^{(i)}_{i} \perp F^{(i')_{j'}}$ whenever $i \neq i'$, thus
    \begin{align*}
         \pi(\sum_{i=1}^{m}E_{i}) &=\pi(\sum_{i=1}^{m} \sum_{j=1}^{k_{i}} F^{(i)}_{j}) = \sum_{i=1}^{m} \sum_{j=1}^{k_{i}} \pi(F^{(i)}_{j}) \sum_{i=1}^{m}\pi(E_{i}).
    \end{align*}
\end{proof}
\begin{lemma}\label{lemma: BornForProjections} Let $\pi:\mathcal{P}(\mathfrak{S}) \rightarrow \mathcal{P}(H)$ be the projection assignment induced by the vector assignment $\Psi$, and let $F \in \mathcal{P}(\mathfrak{S})$ be a rank-$1$ projection (equivalently, a pure state in $\mathfrak{S}$). Then, for any projection $E$,
\begin{align}
\langle E \rangle_{F} = \Tra{\pi(F) \pi(E)}.
\end{align}
\end{lemma}
\begin{proof}
    This result immediately follows from equation~\ref{eq: extendingPi'} and definition~\ref{def: projectionAssignment}.
\end{proof}

\begin{proposition}[Trace of projections]\label{prop: tracePreserving}  The projection assignment $\pi:\mathcal{P}(\mathfrak{S}) \rightarrow \mathcal{P}(H)$ induced by an vector assignment $\Psi$ is trace preserving, i.e., for any projection $E \in \mathcal{P}(\mathfrak{S})$ we have $\Tra{\pi(E)} = \Tra{E}$.    
\end{proposition}
\begin{proof}
Proposition~\ref{prop: nondegenerateDecomposition} asserts that a projection $E \in \mathcal{P}(\mathfrak{S})$ has trace $k$ iff it can be written as the sum of $k$ pairwise orthogonal rank-$1$ projections, and it is well known that the analogous result holds in the lattice of projections $\mathcal{P}(H)$ \cite{landsman2017foundations}. The proposition thus follows from definition~\ref{def: projectionAssignment} and from the fact that, if $F \in \mathcal{P}(\mathfrak{S})$ is rank-$1$ projection, $\pi(F)$ is a rank-$1$ projection in $\mathcal{P}(H)$.
\end{proof}

Finally, we can extend the projection assignment $\pi$ to an observable assignment  $\pi: \mathcal{O}(\mathfrak{S})\rightarrow \mathcal{O}(H)$, where $\mathcal{O}(\mathfrak{S})$ denotes the set of observables of $\mathfrak{S}$ and $\mathcal{O}(H) \equiv \mathcal{B}(H)_{\text{sa}}$ denotes the set of selfadjoint operators on $H$. In fact, let $A$ be an observable of $\mathfrak{S}$, and let $A=\sum_{\af \in \sigma(A)} \af E_{\af}$ be its spectral decomposition (see definition~\ref{def: spectralDecomposition}). Define
\begin{align}
    \pi(A) \doteq\label{eq: observableAssigned} \sum_{\af \in \sigma(A)} \af \pi(E_{\af}).
\end{align}
According to the spectral theorem (the real spectral theorem, which applies to Hilbert spaces \cite{kadison1997fundamentalsI, landsman2017foundations}), $\pi(A)$ is a selfadjoint operator, so the function $\pi: \mathcal{O}(\mathfrak{S})\rightarrow \mathcal{O}(H)$ is well defined. Furthermore, given any decomposition of $A$ in terms or pairwise orthogonal projections, i.e., for any decomposition $A=\sum_{i=1}^{m} \beta_{i}F_{i}$, where $F_{1},\dots,F_{m}$ are pairwise orthogonal projections and $\beta_{1},\dots,\beta_{m}$ are real numbers, we have
\begin{align}
    \pi(A) = \sum_{i=1}^{m} \beta_{i} \pi(F_{i}).
\end{align}
In fact, let $C \xrightarrow{\chi_{\Delta_{i}}} F_{i}$, $i=1,\dots,m$, be a cone for $F_{1},\dots,F_{m}$ (see corollary~\ref{cor: specker}), and define $f \doteq \sum_{i=1}^{m} \beta_{i}\chi_{\Delta_{i}}$. For the same reason we presented in the proof of proposition~\ref{prop: functionalCalculus}, we can assume, without loss of generality, that $\{F_{1},\dots,F_{m}\}$ is a partition of unit (definition~\ref{def: partitionOfUnit}). We have $A = \sum_{i=1}^{m}\beta_{i} \chi_{\Delta_{i}}(C) = \left( \sum_{i=1}^{m}\beta_{i} \chi_{\Delta_{i}}\right)(C)=f(C)$, thus $\sigma(A) = \{\beta_{1},\dots,\beta_{m}\}$. Furthermore, according to lemma~\ref{lemma: KSdefinitionProjections},  for any $\af \in \sigma(A)$ we obtain $E_{\af} \equiv \chi_{\{\af\}}(f(A)) = \chi_{f^{-1}(\af)}(C) = \sum_{\substack{i=1 \\ \beta_{i}=\af}}^{m} F_{i}$. Thus, according to lemma~\ref{lemma: additivityVectorAssignment},
\begin{align}
    \pi(A) = \sum_{\af \in \sigma(A)} \af \pi(E_{\af}) = \sum_{\af \in \sigma(A)} \af \sum_{\substack{i=1 \\ \beta_{i}=\af}}^{m} \pi(F_{i}) = \sum_{i=1}^{m}\beta_{i} \pi(F_{i}).
\end{align}
It leads us to the following definition.

\begin{definition}[Observable assignment]\label{def: observableAssignment} Let $\pi: \mathcal{P}(\mathfrak{S}) \rightarrow \mathcal{P}(H)$ be the projection assignment induced by the vector assignment $\Psi: \mathcal{S}(\mathfrak{S}) \rightarrow H$. The unique extension $\pi:\mathcal{O}(\mathfrak{S}) \rightarrow \mathcal{O}(H)$ of $\pi: \mathcal{P}(\mathfrak{S}) \rightarrow \mathcal{P}(H)$  satisfying
\begin{align}
    \pi(\sum_{i=1}^{m} \af_{i} E_{i}) = \sum_{i=1}^{m}\af_{i}\pi(E_{i})
\end{align}
for any set of pairwise orthogonal projections $E_{1},\dots,E_{m}$ and any set of real numbers $\af_{1},\dots,\af_{m}$ is said to be the observable assignment induced by $\Psi$.
\end{definition}

It is important to note that  $\pi(\mathds{1}) = \mathds{1}$ and $\pi(0) = 0$, where, with a slight abuse of notation, we denote both the unit of $\mathfrak{S}$ (definition \ref{def: zeroAndUnit}) and the identity operator of $H$ by $\mathds{1}$, and, similarly, we denote both the zero operator of $\mathfrak{S}$ (definition \ref{def: zeroAndUnit}) and the zero operator of $H$ by $0$.

The observable assignment  $\pi: \mathcal{O}(\mathfrak{S}) \rightarrow \mathcal{O}(H)$ preserves functional relations between observables. That is, given any arrow $A \xrightarrow{f} B$ in the category of observables (equivalently, if $B=f(A)$), we have $\pi(f(A)) = f(\pi(A))$, where $f(\pi(A))$ is defined by the (real) functional calculus \cite{kadison1997fundamentalsI, landsman2017foundations}. In fact, it follows from definition~\ref{def: observableAssignment}, proposition~\ref{prop: functionalCalculus} and from the functional calculus \cite{kadison1997fundamentalsI, landsman2017foundations} that
\begin{align*}
    \pi(f(A)) &= \pi (\sum_{\af \in \sigma(A)} f(\af) E_{\af}) = \sum_{\af \in \sigma(A)}f(\af) \pi(E_{\af}) = f(\pi(A)),
\end{align*}
where $A = \sum_{\af \in \sigma(A)} \af E_{\af}$ is the spectral decomposition of $A$. It leads us to the following theorem.

\begin{theorem}[Functional relations]\label{thm: functionalRelations} The observable assignment $\pi: \mathcal{O}(\mathfrak{S}) \rightarrow \mathcal{O}(H)$  induced by the vector assignment $\Psi$ preserves functional relations between observables. That is, for any observable $A$ and any real function $f$ on the spectrum of $A$,
    \begin{align}
        \pi(f(A)) =f(\pi(A)).
    \end{align}
\end{theorem}

The reader who is familiar with category theory will easily see that $\pi$ induces a \textbf{faithful covariant functor} from the category of observables to the category of selfadjoint operators on $H$, i.e., the category whose objects are selfadjoint operators and whose arrows are the functional relations between them, analogously to definition~\ref{def: categoryOfObservables}.

It is important to emphasize the following corollaries of theorem~\ref{thm: functionalRelations}.

\begin{corollary}[Spectral decompositions]\label{cor: spectralDecompositions} Let $\pi: \mathcal{O}(\mathfrak{S}) \rightarrow \mathcal{O}(H)$ be the observable assignment induced by the vector assignment $\Psi$. Let $A$ be any observable of $\mathfrak{S}$, and let $A=\sum_{\af \in \sigma(A)} \af E_{\af}$ be its spectral decomposition (see definition~\ref{def: spectralDecomposition}). Then the spectral decomposition of $\pi(A)$ in $H$ is given by $\pi(A) = \sum_{\af \in \sigma(A)} \af \pi(E_{\af})$, and consequently  we have $\sigma(\pi(A)) = \sigma(A)$.
\end{corollary}
\begin{proof}
    Let $A=\sum_{\af \in \sigma(A)} \af E_{\af}$ be the spectral decomposition of $A \in \mathcal{O}(\mathfrak{S})$. According to theorem~\ref{thm: functionalRelations} and to the functional calculus \cite{kadison1997fundamentalsI, landsman2017foundations}, if $\af,\af'$ are distinct eigenvalues of $A$ we obtain
    \begin{align*}
        \pi(E_{\af})\pi(E_{\af'}) &= \pi(\chi_{\{\af\}}(A))\pi(\chi_{\{\af'\}}(A)) =\chi_{\{\af\}}( \pi(A))\chi_{\{\af'\}}(\pi(A)) = (\chi_{\{\af\}} \cdot \chi_{\{\af'\}})(\pi(A))
        \\
        &=0
        \\
        &= \pi(E_{\af'})\pi(E_{\af}).
    \end{align*}
    It shows that $\{\pi(E_{\af}): \af \in \sigma(A)\}$ is a set of pairwise orthogonal projections in $H$, and therefore $\sum_{\af \in \sigma(A)} \pi(E_{\af})$ is a projection. Definition~\ref{def: partitionOfUnit} and proposition~\ref{prop: tracePreserving} ensure that
    \begin{align*}
        \Tra{\sum_{\af \in \sigma(A)} \pi(E_{\af})} = \sum_{\af \in \sigma(A)} \Tra{\pi(E_{\af})} =  \sum_{\af \in \sigma(A)} \Tra{E_{\af}} =n,
    \end{align*}
    where $n$ is the dimension of $\mathfrak{S}$ and $H$, which in turn implies that $\sum_{\af \in \sigma(A)} \pi(E_{\af}) = \mathds{1}$. Hence, $\{\pi(E_{\af}): \af \in \sigma(A)\}$ is a partition of the unit, and since $\pi(A) = \sum_{\af \in \sigma(A)} \af \pi(E_{\af})$, this equation has to be, by definition, the spectral decomposition of $\pi(A)$, and consequently $\sigma(\pi(A)) = \sigma(A)$.
\end{proof}

\begin{corollary}[Compatibility]\label{cor: compatibility}  Let $\pi: \mathcal{O}(\mathfrak{S}) \rightarrow \mathcal{O}(H)$ be the observable assignment induced by the vector assignment $\Psi$. If $A,B \in \mathcal{O}(\mathfrak{S})$ are compatible (definition \ref{def: compatibility}), $\pi(A)$ and $\pi(B)$ are compatible in $H$, i.e., they are commuting operators.    
\end{corollary}
\begin{proof}
    Suppose that $A$ and $B$ are compatible in $\mathfrak{S}$, and let $A \xleftarrow{f} C \xrightarrow{g} B$ be any cone for them (see definition \ref{def: compatibility}). According to theorem \ref{thm: functionalRelations}, we have $\pi(A) = f(\pi(C))$ and $\pi(B)=g(\pi(C))$, which implies that $\pi(A)$ and $\pi(B)$ commute.
\end{proof}

We can now easily prove the following proposition.

\begin{proposition}[Projection embedding]\label{prop: projectionEmbedding} The projection assignment  $\pi: \mathcal{P}(\mathfrak{S}) \rightarrow \mathcal{P}(H)$ induced by a vector assignment $\Psi: \mathfrak{S} \rightarrow H$ satisfies the following properties.
\begin{itemize}
    \item[(a)] $\pi$ is an order embedding, i.e., two projections $E,F \in \mathcal{P}(\mathfrak{S})$ satisfy $E \leq F$ (definition~\ref{def: partialOrder}) if and only if $\pi(E) \leq \pi(F)$. In particular, $\pi: \mathcal{P}(\mathfrak{S}) \ri \mathcal{P}(H)$ is injective.
    \item[(b)] $\pi$ preserves supremum and infimum of pairwise compatible projections, i.e., if $E_{1},\dots,E_{m} \in \mathcal{P}(\mathfrak{S})$ are pairwise compatible, then
    \begin{align}
        \pi(\wedge_{i=1}^{m}E_{i})= \wedge_{i=1}^{m}\pi(E_{i}), 
        \\
        \pi(\vee_{i=1}^{m}E_{i})= \vee_{i=1}^{m}\pi(E_{i}).
    \end{align}
\end{itemize}
\end{proposition}
\begin{proof}
Let $E,F$ be projections in  $\in \mathcal{P}(\mathfrak{S})$. Suppose that $E \leq F$, and let $E \xleftarrow{\chi_{\Delta}} C \xrightarrow{\chi_{\Sigma}} F$ be a cone for them. According to definition~\ref{def: partialOrder}, $\Delta \subset \Sigma$. Proposition~\ref{prop: observableAssignment} implies that $\pi(E) = \chi_{\Delta}(\pi(C))$ and $\pi(F) = \chi_{\Sigma}(\pi(C))$, thus the inclusion $\Delta \subset \Sigma$ ensures that $\pi(E) \leq \pi(F)$. On the other hand, assume that $\pi(E) \leq \pi(F)$, which is equivalent to saying that $\Tra{Q \pi(E)} \leq \Tra{Q \pi(F)}$ for every rank-$1$ projection $Q \in \mathcal{P}_{1}(H)$. Let $\rho$ be any state of $\mathfrak{S}$, and let $G_{1},\dots,G_{m} \in \mathcal{P}(\mathfrak{S})$ be rank-$1$ projections such that $\rho=\sum_{i=1}^{m}\af_{i} \emptyset_{G_{i}}$ for some sequence of non negative real numbers $\af_{1},\dots,\af_{m}$ that sum to one (see corollary~\ref{cor: convexHull}). According to proposition~\ref{prop: expectationConvex} and lemma~\ref{lemma: BornForProjections},
\begin{align*}
    \Exp{E} = \sum_{i=1}^{m} \af_{i} \langle E \rangle_{G_{i}} = \sum_{i=1}^{m}\af_{i} \Tra{\pi(G_{i})\pi(E) }  \leq \sum_{i=1}^{m}\af_{i} \Tra{\pi(G_{i})\pi(F) } =  \sum_{i=1}^{m} \af_{i} \langle F \rangle_{G_{i}} = \Exp{F}.
\end{align*}
Proposition~\ref{prop: ordering} ensures that this is equivalent to saying that $E \leq F$, thus the proof of item $(a)$ is complete. Next, let $E_{1},\dots,E_{m} \in \mathcal{P}(\mathfrak{S})$ be pairwise compatible projections, and let $C \xrightarrow{\chi_{\Delta_{i}}} E_{i}$, $i=1,\dots,m$, be a cone for them (see corollary~\ref{cor: specker}). In section~\ref{sec: spectralTheory}, we have seen that $\wedge_{i=1}^{m} E_{i} = \chi_{\cap_{i=1}^{m}\Delta_{i}}(C)$ and $\vee_{i=1}^{m} E_{i} = \chi_{\cup_{i=1}^{m}}(C)$, thus proposition~\ref{prop: projectionEmbedding} implies that $\pi(\wedge_{i=1}^{m} E_{i}) = \chi_{\cap_{i=1}^{m}\Delta_{i}}(\pi(C))$ and $\pi(\vee_{i=1}^{m} E_{i}) = \chi_{\cup_{i=1}^{m}}(\pi(C))$. Well known results from functional analysis ensure that $\chi_{\cap_{i=1}^{m}\Delta_{i}}(\pi(C)) = \wedge_{i=1}^{m}\chi_{\Delta_{i}}(\pi(C))$ and $\chi_{\cup_{i=1}^{m}}(\pi(C)) = \vee_{i=1}^{m} \chi_{\Delta_{i}}(\pi(C))$, and since $\chi_{\Delta_{i}}(\pi(C))=\pi(\chi_{\Delta_{i}}(C))=\pi(E_{i})$ for each $i$, item $(b)$ follows.
\end{proof}

The following proposition is important.

\begin{proposition}[Observable embedding]  The observable assignment $\pi: \mathcal{O}(\mathfrak{S}) \rightarrow \mathcal{O}(H)$  induced by $\Psi$ is injective, i.e., two observables $A,B \in \mathcal{O}(\mathfrak{S})$ satisfy $\pi(A)=\pi(B)$ if and only if they are equal.
\end{proposition}
\begin{proof}
    $\pi$ is a function, so we just need to prove that it is injective, i.e., that $\pi(A)=\pi(B)$ implies $A=B$. Let $A,B$ be observables in $\mathfrak{S}$ satisfying $\pi(A)=\pi(B)$, and let $A=\sum_{\af \in \sigma(A)} \af E_{\af}$, $B=\sum_{\beta \in \sigma(B)} \beta F_{\beta}$ be their spectral decompositions (definition~\ref{def: spectralDecomposition}). According to corollary~\ref{cor: spectralDecompositions}, the spectral decompositions of $\pi(A)$ and $\pi(B)$ are $\pi(A)=\sum_{\af \in \sigma(A)} \af \pi(E_{\af})$ and $\pi(B)=\sum_{\beta \in \sigma(B)} \beta \pi(F_{\beta})$, and since $\pi(A)=\pi(B)$, we have $\sigma(A)=\sigma(\pi(A))=\sigma(\pi(B))=\sigma(B)$ and $\{\pi(E_{\af}): \af \in \sigma(A)\}= \{\pi(F_{\beta}): \beta\in \sigma(B)\}$. It now follows from item $(a)$ of proposition~\ref{prop: projectionEmbedding} that $\{E_{\af}: \af \in \sigma(A)\}$ and  $\{F_{\beta}: \beta\in \sigma(B)\}$ are the same partition of the unit in $\mathfrak{S}$, thus theorem~\ref{thm: spectralTheorem} implies that $A=B$.
\end{proof}

We have seen in section~\ref{sec: latticeOfProjections} that the infimum $\wedge_{i=1}^{m}E_{i}$ of a set of pairwise compatible projections $E_{1},\dots,E_{m} \in \mathcal{P}(\mathfrak{S})$ is the product $\prod_{i=1}^{m}E_{i}$, and it is well known that the analogous result holds in $\mathcal{P}(H)$. Hence, it follows from proposition~\ref{prop: projectionEmbedding} that,  for any set of pairwise compatible projections $E_{1},\dots,E_{m} \in \mathcal{P}(\mathfrak{S})$, we have $\pi(\prod_{i=1}^{m} E_{i}) = \prod_{i=1}^{m}\pi(E_{i})$. Now let $A,B \in \mathcal{O}(\mathfrak{S})$ be compatible observables, and let $A=\sum_{\af \in \sigma(A)}\af E_{\af}$, $B=\sum_{\beta \in \sigma(B)} \beta F_{\beta}$ be their spectral decompositions (definition~\ref{def: spectralDecomposition}). We have seen in section~\ref{sec: spectralTheory} that $A + B = \sum_{\af \in \sigma(A)} (\af + \beta) E_{\af} F_{\beta}$ and $A B = \sum_{\af \in \sigma(A)} (\af \cdot \beta) E_{\af} F_{\beta}$, and it is easy to see that $\{E_{\af}F_{\beta} \in \mathcal{P}(\mathfrak{S}): (\af,\beta) \in \sigma(A) \times \sigma(B)\}$ is a set of pairwise orthogonal projections in $\mathfrak{S}$. Therefore, it follows from definition~\ref{def: observableAssignment} and proposition~\ref{prop: projectionEmbedding} that
\begin{align*}
    \pi(A+B) &= \pi(\sum_{\af \in \sigma(A)}\sum_{\beta \in \sigma(B)} (\af+ \beta) E_{\af} F_{\beta}) = \sum_{\af \in \sigma(A)}\sum_{\beta \in \sigma(B)} (\af+ \beta) \pi(E_{\af} F_{\beta})
    \\
    &= \sum_{\af \in \sigma(A)}\sum_{\beta \in \sigma(B)} (\af+ \beta) \pi(E_{\af}) \pi(F_{\beta}) = \left(\sum_{\af \in \sigma(A)} \af \pi(E_{\af})\right) + \left(\sum_{\beta \in \sigma(B)} \beta \pi(F_{\beta})\right)
    \\
    &=\pi(A) + \pi(B),
\end{align*}
and it is analogous to prove that 
\begin{align*}
    \pi(AB) &= \pi(A)  \pi(B).
\end{align*}
These results can easily be generalized from any set of pairwise compatible observables, i.e., if $A_{1},\dots,A_{m} \in \mathcal{O}(\mathfrak{S})$ are pairwise compatible, we have $\pi(\sum_{1}^{m}A_{i}) = \sum_{i=1}^{m} \pi(A_{i})$ and  $\pi(\prod_{1}^{m}A_{i}) = \prod_{i=1}^{m} \pi(A_{i})$. Finally, it is easy to see that, for any observable $A$ and any real number $\af$, we have $\pi(\af A) = \af \pi(A)$. 

In proposition~\ref{prop: traceSpectrum}, we saw that for any observable $A \in \mathcal{O}(\mathfrak{S})$ we have $\tr(A) = \sum_{\af \in \sigma(A)} m_{\af} \af$, where $m_{\af}$ is the multiplicity of the eigenvalue $\af$ of $A$, i.e., $m_{\af}$ is the trace of the projection $E_{\af} \equiv \chi_{\{\af\}}(A)$. It is well known that an analogous result holds in $\mathcal{O}(H) \equiv \mathcal{B}(H)_{\text{sa}}$, thus it follows from definition~\ref{def: observableAssignment} and theorem~\ref{thm: functionalRelations} that $\pi: \mathcal{O}(\mathfrak{S}) \rightarrow\mathcal{O}(H)$ is trace preserving. Consequently, if an observable $A \in \mathcal{O}(\mathfrak{S})$ is a density operator (definition~\ref{def: densityOperator}), $\pi(A) \in \mathcal{B}(H)_{\text{sa}}$ is a density operator, i.e., a positive operator of trace one (note that theorem~\ref{thm: functionalRelations} ensures that $\pi(A)$ is a positive operator, i.e., a selfadjoint operator whose spectrum contains only non-negative numbers \cite{kadison1997fundamentalsI, landsman2017foundations}). We summarize these results in the following proposition.

\begin{proposition}\label{prop: observableAssignment} The observable assignment $\pi: \mathcal{O}(\mathfrak{S}) \rightarrow \mathcal{O}(H)$  induced by $\Psi$ satisfies the following properties.
\begin{itemize}
    \item[(a)] $\pi$ is partially linear, i.e., if $A_{1},\dots,A_{m} \in \mathcal{O}(\mathfrak{S})$ are pairwise compatible observables and $\af_{1},\dots,\af_{m}$ are real numbers, we have
\begin{align}
    \pi(\sum_{i=1}^{m} \af_{i} A_{i}) &= \sum_{i=1}^{m} \af_{i} \pi(A_{i}).
\end{align}
\item[(b)] For any set $A_{1},\dots,A_{m} \in \mathcal{O}(\mathfrak{S})$ of pairwise compatible observables,
\begin{align}
    \pi(\prod_{i=1}^{m} A_{i}) &= \prod_{i=1}^{m} \pi(A_{i}).
\end{align}
\item[(c)] $\pi$ is trace preserving, i.e., for any observable $A$,
\begin{align}
    \tr(\pi(A)) = \tr(A).
\end{align}
\item[(d)] If $A$ is a density operator in $\mathfrak{S}$ (definition~\ref{def: densityOperator}), $\pi(A)$ is a density operator in $\mathcal{O}(H) \equiv \mathcal{B}(H)_{\text{sa}}$.
\end{itemize}
\end{proposition}

In section \ref{sec: transitionProbabilities}, we saw that a density operator $A \in \mathcal{O}(\mathfrak{S})$ defines a state $\emptyset_{A} \in \mathcal{S}(\mathfrak{S})$, which consists in the unique state satisfying
\begin{align}
    \emptyset_{A}&= \sum_{i=1}^{m}\af_{i}\emptyset_{E_{i}}
\end{align}
for every convex decomposition $A=\sum_{i=1}^{m}\af_{i}E_{i}$ of $A$ in terms of pairwise orthogonal rank-$1$ projections (see definition \ref{def: densityState}). Fix a convex decomposition $A=\sum_{i=1}^{m}\af_{i}E_{i}$ of $A$ in terms or pairwise orthogonal rank-$1$ projections, and let $B \in \mathcal{O}(\mathfrak{S})$ be any observable. According to proposition \ref{prop: observableAssignment},
\begin{align*}
    \langle B \rangle_{A}&= \sum_{i=1}^{m} \langle B \rangle_{E_{i}} = \sum_{i=1}^{m} \Tra{\pi(E_{i})\pi(B)} = \Tra{\pi(A)\pi(B)},
\end{align*}
where $\langle \cdot \rangle_{A} \equiv \langle \cdot \rangle_{\emptyset_{A}}$. Let's emphasize this result.

\begin{proposition}[Born Rule for density operators]\label{prop: BornForDensity} Let $A \in \mathcal{O}(\mathfrak{S})$ be a density operator (see definition \ref{def: densityOperator}), and let $\langle \ \cdot \ \rangle_{A}$ be the  state defined by it (definition \ref{def: densityState}). Then, for any observable $B \in \mathcal{O}(\mathfrak{S})$,
\begin{align}
    \langle B \rangle_{A} = \tr(\pi(A)\pi(B)). 
\end{align}
\end{proposition}

Now let $\rho$ be any state. According to theorem~\ref{cor: convexHull}, $\rho$ is a convex combination of pure states, so let $E_{1},\dots,E_{m}$ be rank-$1$ projections such that $\rho=\sum_{i=1}^{m}\af_{i} \emptyset_{E_{i}}$, where $\af_{1},\dots,\af_{m}$ are non negative numbers that sum to one. For any observable $B$,
\begin{align}
    \Exp{B} = \sum_{i=1}^{m} \af_{i}\langle B \rangle_{E_{i}} = \sum_{i=1}^{m} \af_{i}\Tra{\pi(E_{i})\pi(B)} = \Tra{\sum_{i=1}^{m}\af_{i}\pi(E_{i})\pi(B)} = \Tra{D_{\rho} \pi(B)},
\end{align}
where $D_{\rho} \doteq \sum_{i=1}^{m}\af_{i}\pi(E_{i})$. Note that $D_{\rho}$ is a density operator of $H$. 

It leads us to the following proposition.
\begin{proposition}[Born rule]\label{prop: BornRule} Let $\rho$ be any state of $\mathfrak{S}$, and let $\rho = \sum_{i=1}^{m}\af_{i} \emptyset_{E_{i}}$ be a convex decomposition of $\rho$ in terms of pure states  (see theorem~\ref{cor: convexHull}). Define  the density operator $D_{\rho} \doteq \sum_{i=1}^{m} \af_{i} \pi(E_{i})$ in $H$. Then, for any observable $B \in \mathcal{O}(\mathfrak{S})$,
\begin{align}
    \langle B \rangle_{\rho} = \Tra{D_{\rho} \pi(B)}. 
\end{align}
\end{proposition}

Together with the axiom of choice \cite{levy2002basic}, propositions \ref{prop: BornForDensity} and \ref{prop: BornRule} imply the following.

\begin{corollary}[State Assignment]\label{cor: stateAssignment}  Let $\pi: \mathcal{O}(\mathfrak{S}) \rightarrow \mathcal{O}(H)$ be the observable assignment induced by the vector assignment $\Psi$. Then there exists a mapping $\theta: \mathcal{S}(\mathfrak{S}) \rightarrow \mathcal{D}(H)$, which is said to be a \textbf{state assignment induced by $\pi$}, satisfying the following properties.
\begin{itemize}
    \item[(a)] If $\rho  \in \mathcal{S}(\mathfrak{S})$ is is defined by a density operator, i.e., if $\rho=\emptyset_{A}$ for some $A \in \mathcal{D}(\mathfrak{S})$ (see definition \ref{def: densityState}), then $\theta(\emptyset_{A})= \pi(A)$.
    \item[(b)] for any state $\rho$ and any observable $B \in \mathcal{O}(\mathfrak{S})$,
    \begin{align}
        \Exp{B} = \Tra{\theta(\rho)\pi(B)}.
    \end{align}
\end{itemize}
\end{corollary}

\begin{lemma}[State embedding]\label{lemma: stateEmbedding} Let $\theta: \mathcal{S}(\mathfrak{S}) \rightarrow \mathcal{D}(H)$ be a state assignment induced by some observable assignment $\pi: \mathcal{O}(\mathfrak{S}) \rightarrow \mathcal{O}(H)$. Then $\theta$ is injective.
\end{lemma}
\begin{proof}
    Let $\rho, \rho'$ be states satisfying $\theta(\rho)=\theta(\rho')$. For each projection $F \in \mathcal{P}(\mathfrak{S})$, we have
    \begin{align*}
        \langle F \rangle_{\rho} &= \Tra{\theta(\rho) \pi(F)} = \Tra{\theta(\rho') \pi(F)} = \langle F\rangle_{\rho'},
    \end{align*}
    thus it follows from lemma~\ref{lemma: expectationAsState} that $\rho=\rho'$.
\end{proof}

We have proved that a quantum embedding for $\mathfrak{S}$ exists:

\begin{theorem}[Quantum embedding]\label{thm: quantumEmbedding} Let $\mathfrak{S}$ be a $n$-dimensional system satisfying postulates \ref{ax: separability}-\ref{post: pureState}. Then there exists a $n$-dimensional Hilbert space $H$ and injective mappings $\pi: \mathcal{O}(\mathfrak{S}) \rightarrow \mathcal{B}(H)_{\text{sa}}$, $\theta: \mathcal{S}(\mathfrak{S}) \rightarrow \mathcal{D}(H)$ satisfying the following conditions. 
\begin{itemize}
    \item[(a)] $\pi$ preserves functional relations between observables. That is, for any observable $A \in \mathcal{O}(\mathfrak{S})$ and any real function $f$ on the spectrum of $A$,
    \begin{align}
        \pi(f(A)) =f(\pi(A)).
    \end{align}
    Consequently, if $A$ is a projection of $\mathfrak{S}$ (definition \ref{def: projection}), then $\pi(A)$ is a projection of $H$.
    \item[(b)] $\pi$ preserves the spectrum of all observables, i.e., for any observable $A \in \mathcal{O}(\mathfrak{S})$,
    \begin{align}
        \sigma(\pi(A)) = \sigma(A).
    \end{align}    
    \item[(c)] $\pi$ is trace-preserving, i.e., for any observable $A$, 
    \begin{align}
        \Tra{\pi(A)}=\Tra{A}.
    \end{align}
    \item[(d)] $\pi$ preserves algebraic operations, i.e., if $A_{1},\dots,A_{m} \in \mathcal{O}(\mathfrak{S})$ are pairwise compatible observables and $\af_{1},\dots,\af_{m}$ are real numbers, we have
    \begin{align}
        \pi(\sum_{i=1}^{m} \af_{i}A_{i}) &= \sum_{i=1}^{m} \af_{i} \pi(A_{i}),
        \\
        \pi(\prod_{i=1}^{m}A_{i}) &= \sum_{i=1}^{m} \af_{i} \pi(A_{i}).
    \end{align}
    \item[(e)] $\pi$ preserves positive and density operators. That is, if $A$ is positive in $\mathfrak{S}$ (see definition~\ref{def: densityOperator}), then $\pi(A)$ is positive in $H$, and analogously for density operators.
    \item[(f)] When restricted to the lattice of projections $\mathcal{P}(\mathfrak{S})$, $\pi$ is an order embedding that preserves supremum and infimum of pairwise compatible projections. That is, $\pi$ is order-preserving, order-reflecting, and for any set $E_{1},\dots,E_{m} \in \mathcal{P}(\mathfrak{S})$ of pairwise compatible projections we have
    \begin{align}
        \pi(\wedge_{i=1}^{m}E_{i})= \wedge_{i=1}^{m}\pi(E_{i}), 
        \\
        \pi(\vee_{i=1}^{m}E_{i})= \vee_{i=1}^{m}\pi(E_{i}).
    \end{align}    
    \item[(g)] $\theta$ preserves density operators. That is, if $\rho  \in \mathcal{S}(\mathfrak{S})$ is the state $\emptyset_{A}$ defined by a density operator $A \in \mathcal{D}(\mathfrak{S})$ (see definition \ref{def: densityState}), then $\theta(\emptyset_{A})= \pi(A)$.
    \item[(h)] Expectations in $\mathfrak{S}$ satisfy the Born rule, i.e., for any state $\rho$ and any observable $A \in \mathcal{O}(\mathfrak{S})$,
    \begin{align}
        \Exp{A} = \Tra{\theta(\rho)\pi(A)}.
    \end{align}
\end{itemize}
\end{theorem}

To conclude, let's emphasize that we have singled out the so-called set of quantum correlations:

\begin{corollary}[Joint measurements]\label{cor: quantumCorrelations} Let $\mathfrak{S}$ be a $n$-dimensional system satisfying postulates \ref{ax: separability}-\ref{post: transitionProbability}, and let $(H,\theta,\pi)$ be a quantum embedding for $\mathfrak{S}$, i.e., $H$ is a n-dimensional Hilbert space and $\pi: \mathcal{O}(\mathfrak{S}) \rightarrow \mathcal{B}(H)_{\text{sa}}$, $\theta: \mathcal{S}(\mathfrak{S}) \rightarrow \mathcal{D}(H)$ are the mappings whose existence has been proved in theorem \ref{thm: quantumEmbedding}. Let $A_{1},\dots,A_{m} \in \mathcal{O}(\mathfrak{S})$ be pairwise compatible observables, and, for each $i$, let $\Delta_{i} \subset \sigma(A_{i})$. Then, for any state $\rho$,
\begin{align}
    P_{\rho}(\Delta_{1}\times \dots \times \Delta_{m};A_{1},\dots,A_{m}) = \Tra{\theta(\rho)\prod_{i=1}^{m}\chi_{\Delta_{i}}(\pi(A_{i}))}.
\end{align}
\end{corollary}
\begin{proof}
According to lemma~\ref{lemma: projectionJointProbability} and theorem~\ref{thm: quantumEmbedding},
\begin{align*}
    P_{\rho}(\Delta_{1}\times \dots \times \Delta_{m};A_{1},\dots,A_{m}) &= \Exp{\prod_{i=1}^{m}\chi_{\Delta_{i}}(A_{i})} = \Tra{\theta(\rho)\pi(\prod_{i=1}^{m}\chi_{\Delta_{i}}(A_{i}))}
    \\
    &= \Tra{\theta(\rho)\prod_{i=1}^{m}\pi(\chi_{\Delta_{i}}(A_{i}))} = \Tra{\theta(\rho)\prod_{i=1}^{m}\chi_{\Delta_{i}}(\pi(A_{i}))}.
\end{align*}
\end{proof}

\section{Concluding remarks}\label{sec: conclusion}

We have identified key theory-independent features of the quantum formalism that are conjointly sufficient for the emergence of the algebraic structure of quantum mechanics in any system which includes states, real-valued observables and a state update, and we have gradually derived the entire formalism from them. Our step-by-step approach has enabled us to identify which aspects of this algebraic structure are incorporated by the system when each one of these features is postulated, helping us to understand why the formalism is the way it is. This approach has also enabled us to single out the mathematical elements of quantum theory that a physical system must assimilate in order to satisfy principles  known to be important for understanding quantum correlations, such as local orthogonality \cite{fritz2013local}, Specker's principle \cite{cabello2012specker, gonda2018almost} and  the exclusivity principle \cite{amaral2014exclusivity}, and it has enabled us to show that the distinctive way in which quantum systems connect incompatible observables plays a major role in shaping both the algebraic and the statistical facets of quantum mechanics. Besides, most of our postulates have a strong informational appeal, including those concerning states and the state update, so we believe that our work can shed light on the longstanding debate over the collapse (i.e., update) of the quantum state and over its ontological status \cite{halvorson2019realist, leifer2014real, spekkens2007toy, norsen2017foundations}. Consequently, it can shed light on the problem of interpreting quantum mechanics.

Functional relations are old-fashioned. They do not occupy a prominent place in quantum foundations nowadays, and there is no room for them in many attempts to explain the quantum formalism and the set of quantum correlations. Even in contextuality analysis, the field that arose from Kochen-Specker theorem \cite{kochen1967problem}, they have been left aside and replaced by contexts \cite{amaral2018graph, abramsky2011sheaf}. In our paper, on the contrary, functional relations play a pivotal role. We use them, directly or indirectly, to define compatibility, projections, algebraic operations, traces of observables, and so on. They also appear in essentially every single lemma, proposition, and theorem we prove throughout the paper, which in turn are well-known results in quantum mechanics. It is above all due to the existence of functional relations that, way before we are able to embed our system in a Hilbert space, essentially all the ``commutative part'' of the quantum formalism is already present in it, as we showed in the first part of the work (sections~\ref{sec: observableEvents}-\ref{sec: commutativePart}). We believe that a revival of Kochen and Specker's emphasis on functional relations --- which, as they point out, are naturally present in any system \cite{kochen1967problem} --- can shed light on many important issues in contemporary quantum foundations. An important example is contextuality, which we discuss in the appendix. Other instances have already been discussed throughout the paper, such as the distinction between local orthogonality and the exclusivity principle. 

In the first part of the paper (sections~\ref{sec: observableEvents}-\ref{sec: commutativePart}), we showed that, when functional relations and the state update are taken into account, basically all the ``commutative part'' of the quantum formalism is naturally present in a system where (1) nondegenerate observables are as simple as $n$-sided dices (postulate~\ref{post: observables}) (2) degenerate observables are coarse-grainings of nondegenerate ones (postulate~\ref{post: observables}), (3) a sequential measurement of two compatible observables is equivalent to a measurement of a single ``fine graining'' that they share in common (postulates~\ref{post: selfCompatibility} and \ref{post: compatibility}), and (4) commutativity requires compatibility (postulate~\ref{post: commutativity}). These restrictions, however, are not enough for the quantum formalism to arise fully. What is missing is a connection between incompatible observables, which we introduced in the second part of the paper (sections \ref{sec: connectingIncompatible}-\ref{sec: quantumTheory}). This connection is primarily comprised of postulate~\ref{post: transitionProbability}, which poses constraints on the probability of transitioning between two pure states (or, equivalently, two rank-$1$ projections). In quantum theory, transition probabilities between pure states are determined by the inner product of the Hilbert space representing the system, and the constraints we pose in postulate \ref{post: transitionProbability} are necessary to ensure that we can associate a unit vector to each pure state of our system in such a way that their inner products match their transition probabilities, as we showed in lemma~\ref{lemma: transition}. This result implies that our system is in agreement with the interference terms that appear in quantum mechanics when we evaluate the expectation of an observable with respect to some pure state by decomposing this state as a linear combination of vectors of a certain orthonormal basis. The appearance of interference terms is one of the most characteristic features of the quantum formalism \cite{tausk2018foundations, norsen2017foundations}, so sooner or later we would be forced to take it into account. Although in quantum mechanics we are usually interested only in sequential measurements of compatible observables, the connection between incompatible observables is naturally present in the formalism and plays a crucial role in the theory. It is the connection between incompatible observables that gathers everything (i.e., states and observables) together in a single mathematical object, namely the $C^{\ast}$ algebra of  operators associated with the Hilbert space representing the system. It is precisely this connection that makes a quantum system more than a collection of intersecting contexts, i.e., more than a collection of intersecting sets of pairwise compatible observables. Any system satisfying postulates \ref{ax: separability}-\ref{post: commutativity} is nondisturbing (see lemma~\ref{lemma: nondisturbance}) and satisfies the exclusivity principle (corollary~\ref{cor: exclusivityPrinciple}), but, despite all the similarities between this system and quantum systems that we proved in the first part of the paper, it is not a quantum system yet, and consequently its predictions are not ``quantum correlations'' \cite{amaral2018graph, csw2014graph}. To go from nondisturbing to quantum correlations in our formalism, we need to connect incompatible observables, as we showed in the second part of the paper. Much effort has been made recently to explain the set of quantum correlations, which proved to be a very difficult task \cite{gonda2018almost, navascues2015almost, elie2018geometry}. However, essentially all attempts have been based on frameworks where no connection between incompatible observables exists \cite{csw2014graph, navascues2015almost, amaral2014exclusivity, fritz2013local}. Our work shows that this difficulty disappears when this connection  is taken into account, but to do so we introduced a postulate that, despite being a key feature of quantum systems, has no clear explanatory power. The way in which quantum mechanics connects incompatible observables is an essential part of the theory and, as our paper shows, it is a powerful tool for deriving the quantum formalism and the set of quantum correlations. For these reasons, we  believe that a better understanding of this connection is important.

\section*{Acknowledgments} \label{sec:acknowledgements}
The author thanks Bárbara Amaral, Roberto Baldijão, and Giulio Halisson for reading and commenting on previous versions of this manuscript and Adán Cabello for insightful discussions and criticisms.

\appendix
\setcounter{secnumdepth}{0}
\section{Appendix: Kochen-Specker theorem, contextuality and order dependence}\label{sec: Contextuality}

In their remarkable paper on ``the problem of hidden variables in quantum mechanics'' \cite{kochen1967problem}, Kochen and Specker assert that one cannot consistently interpret all observables of a quantum system as representing properties simultaneously possessed by it \cite{kochen1967problem, isham1998topos, landsman2017foundations, hermens2010fromRealism}. According to them, a physically consistent assignment of values to observables of a physical system   ---  which must exist for the aforementioned interpretation to be possible --- should respect the functional relations that exist between observables (see definitions~\ref{def: functionalRelation} and \ref{def: categoryOfObservables})  and should also associate  an element of $\sigma(A)$ (see section~\ref{sec: basicFramework}) to each observable $A$, as in definition~\ref{def: valuation}. They then show that, as stated by theorem~\ref{thm: kochenSpecker}, such an assignment cannot exist in any quantum system described by a  Hilbert space of dimension larger than $2$.

\begin{definition}[Valuation function \cite{kochen1967problem, doring2005kochen}]\label{def: valuation} Let $\mathcal{O}$ be the set of observables of some physical system $\mathfrak{S}$. A function $V: \mathcal{O} \rightarrow \mathbb{R}$ is said to be a valuation on $\mathcal{O}$ if it satisfies the following conditions.
\begin{itemize}
    \item[(a)] \textit{Spectrum rule:} $V(A) \in \sigma(A)$ for every observable $A$
    \item[(b)] \textit{Functional composition principle:} If two observables $A,B$ satisfy $B=g(A)$ for some function $g$ on $\sigma(A)$, then $V(B) = g(V(A))$.
\end{itemize}
\end{definition}

\begin{theorem}[Kochen-Specker,  \cite{kochen1967problem, doring2005kochen}]\label{thm: kochenSpecker} Let $\mathcal{O}$ be the set of selfadjoint operators on a Hilbert space $H$. If $\text{dim}(H)>2$, there is no valuation function on $\mathcal{O}$. 
\end{theorem}

As Kochen and Specker point out \cite{kochen1967problem, doring2005kochen}, this theorem immediately implies that, if $\mathfrak{S}$ is a quantum system of dimension larger than $2$, the following type of classical \textit{realist} model \cite{doring2010thing, tezzin2021KCBS} cannot exist for $\mathfrak{S}$.\footnote{By evoking postulates \ref{ax: separability} - \ref{post: compatibility}, we implicitly restrict the definition of Kochen-Specker model to finite-dimensional systems. We do it, however, just for convenience. Kochen-Specker theorem applies to any separable Hilbert space \cite{kochen1967problem, isham1998topos}, and it can even be proved to any von Neumann algebra without summands of types $I_{1}$ and $I_{2}$ \cite{doring2005kochen}.}

\begin{definition}[Kochen-Specker model]\label{def: KSmodel} Let $\mathfrak{S}\equiv (\mathcal{O},\mathcal{S},P,T)$ be a physical system (definition~\ref{def: system}) satisfying postulates \ref{ax: separability}-\ref{post: compatibility}. A Kochen-Specker (KS) model for  $\mathfrak{S}$ consists of a measurable space $\boldsymbol{\Lambda} \equiv (\Lambda,\Sigma)$ and mappings $\Phi$, $\Xi$ satisfying the following conditions.
\begin{itemize}
    \item[(a)] $\Xi$ assigns, to each observable $A \in \mathcal{O}$, a measurable function $f_{A}: \Lambda \ri \sigma(A)$. Furthermore, given any arrow $A \xrightarrow{g} B$ in the category of observables, i.e., if $B=g(A)$ for some function $g$, we obtain
    \begin{align}
        f_{g(A)} = g \circ f_{A}.
    \end{align}
    \item[(b)]  $\Phi$ assigns, to each non null state $\rho$, a probability measure $\mu_{\rho}$ on $\boldsymbol{\Lambda}$, and it assigns the null measure $\mu_{0}$ to the null state. Furthermore, given any observable $A$ and any state $\rho$, the measure $P_{\rho}( \ \cdot \ ; A)$ (see definition~\ref{def: system}) is the pushforward of $\mu_{\rho}$ along $f_{A}$, which means that, for any $\Delta \subset \sigma(A)$,
    \begin{align}
        P_{\rho}(\Delta;A) = \mu_{\rho}(f_{A}^{-1}(\Delta)).
    \end{align}
\end{itemize}
\end{definition}

\begin{corollary}[Kochen-Specker, \cite{kochen1967problem, doring2005kochen}]\label{cor: kochenSpecker} Let $\mathfrak{S}$ be a quantum system represented by a Hilbert space $H$ of dimension larger than $2$. Then there is no KS-model for $\mathfrak{S}$.
\end{corollary}

Let $\mathfrak{M}\equiv (\boldsymbol{\Lambda},\Xi,\Phi)$ be a KS-model for a system $\mathfrak{S} \equiv (\mathcal{O},\mathcal{S},P,T)$. The elements of $\Lambda$ can be thought of as representing ``hidden states'' of the system under analysis, so the addition of $\Lambda$ ``completes'' the description provided by $\mathfrak{S}$: if an element of $\mathcal{S}$ represents the degree of knowledge of the experimentalist about the \textit{state of affairs} of the system, as we asserted in section~\ref{sec: categoryOfObservables}, the elements of $\Lambda$ represent the \textit{state of affairs} themselves. It means that, given a state $\rho \in \mathcal{S}$, $\mu_{\rho}(U)$ is the probability that the (hidden) state of the system (or equivalently its \textit{state of affairs}) lies in the measurable set $U$ \cite{kochen1967problem}. A hidden state $\lambda \in \Lambda$ assigns a definite value to each observable $A$, namely the value $V_{\lambda}(A) \doteq f_{A}(\lambda)$, and it immediately follows from definition~\ref{def: KSmodel} that the assignment $\mathcal{O} \ni A \xmapsto{V_{\lambda}} V_{\lambda}(A) \equiv f_{A}(\lambda) \in \mathbb{R}$ is a valuation function on $\mathcal{O}$. Hence, if a KS-model exists for a system $\mathfrak{S}$, there is no apparent inconsistency in interpreting observables of $\mathfrak{S}$ as representing properties simultaneously possessed by the physical system under description. A KS-model thus completes the description provided by $\mathfrak{S}$ in a way that is in agreement with classical realism \cite{doring2010thing, tezzin2021KCBS}. 

Kochen-Specker theorem obstructs the existence of hidden states for quantum systems of dimension larger than $2$ and, consequently, it rules out the possibility of constructing KS-models for them. To prove that these states cannot exist for a certain system $H$ one does not need to take all selfadjoint observables into account; a finite set of observables is usually sufficient, and different sets of observables provide different proofs of the theorem \cite{kochen1967problem, budroni2021review}. What all these proofs have in common, however, is the presence (possibly implicit) of intersecting \textit{contexts} and incompatible observables. In our formalism, contexts can be defined as follows.

\begin{definition}[Context]\label{def: context} Let $\mathfrak{S}$ be a system satisfying postulate \ref{ax: separability}-\ref{post: compatibility}. A context in $\mathfrak{S}$ consists of a nonempty set of pairwise compatible observables (see definition~\ref{def: compatibility}). Two contexts $\mathcal{C},\mathcal{D}$ are said to be compatible iff their union is also a context, and they are said to be incompatible otherwise.
\end{definition}

Compatibility is not a transitive relation, so, in principle, there exist distinct contexts in a system $\mathfrak{S}$. The trivial example of intersecting contexts whose union in itself is not a context is given by an observable $A$ that can be written as a function $A=g(C)=h(D)$ of incompatible observables $C,D$, as illustrated in the following diagram (which is defined in the category of observables).
\begin{center}
    \begin{tikzcd}
        C \arrow[r, "g"] & A &\arrow[l,swap, "h"] D
    \end{tikzcd}
\end{center}
Each arrow of this diagram is associated with a different context containing $A$: the arrow $A \xrightarrow{g} C$ is associated with the context $\{C,A\}$, and $D \xrightarrow{h} A$ is associated with $\{D,A\}$. On the other hand, if an observable $A$ belongs to two incompatible contexts $\mathcal{C}$, $\mathcal{D}$, then there are incompatible observables $C,D$ and arrows $C \xrightarrow{g} A \xleftarrow{h} D$. In fact, two contexts  $\mathcal{C}$ and $\mathcal{D}$ are incompatible if and only if there is a pair of observables $(C',D') \in \mathcal{C} \times \mathcal{D}$ such that $C'$ and $D'$ are incompatible. If $A \in \mathcal{C} \cap \mathcal{D}$, $A$ is compatible with $C'$ and $D'$, which means that the following cones exist in the category of observables.
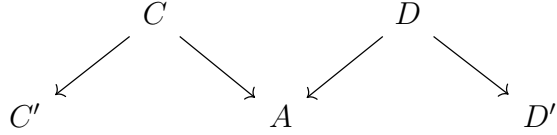
\begin{center}
    \begin{tikzcd}
        & C\arrow[dl]\arrow[dr] & & D\arrow[dl]\arrow[dr] & \\
        C' & & A & & D'
    \end{tikzcd}
\end{center}
$C$ and $D$ are incompatible, because otherwise $C'$ and $D'$ would be compatible, so the proof is complete. Therefore, the existence of intersecting incompatible contexts is equivalent to the existence of observables that can be written as a function of incompatible observables. Put differently, there are incompatible contexts in a system $\mathfrak{S}$ if and only if there exist incompatible observables in $\mathfrak{S}$ that have a coarse-graining in common. As we said, to prove Kochen-Specker theorem we need to evoke, in one way or another, incompatible contexts with a nonempty intersection, so Kochen-Specker theorem essentially follows from the existence of incompatible observables that share  coarse-grainings \cite{doring2005kochen, isham1998topos}. 

Hidden states in KS-models are \textit{noncontextual} in the sense that the values they assign to observables have no dependence at all on contexts. Kochen-Specker theorem tells us that, in order to assign values to \textit{all} observables of a quantum system without rejecting the spectrum rule (see definition~\ref{def: valuation}),  we need to reject the functional composition principle, which asserts --- to use our terminology --- that the value of an observable $B$ depends only in the object representing $B$ in the category of observables, and not on any specific arrow whose codomain is $B$. It follows from the discussion above that a ``context-dependent'' assignment of values to observables is equivalent to an assignment that depends not only on the observable $B$ but also on the arrows that have $B$ as codomain. Hence, it follows from Kochen-Specker theorem that any conceivable hidden state assigning values to all observables of a quantum system without violating the spectrum rule will be context-dependent \cite{budroni2021review, hermens2010fromRealism}. This is one of the reasons why quantum systems are said to be ``contextual'' \cite{deRonde2020unscrambling, isham1998topos, hermens2010fromRealism}. To be faithful to our formalism, we define contextual assignments using the category of observables:

\begin{definition}[Value assignment]\label{def: contextualAssignment} Let $\mathfrak{S}$ be a system satisfying postulates \ref{ax: separability}-\ref{post: compatibility}. Let $\mathcal{O}$ be the category of observables of $\mathfrak{S}$ (see definition~\ref{def: categoryOfObservables}), and let $\mathcal{O}_{1}$ be its set of arrows. A function $V: \mathcal{O}_{1} \rightarrow \mathbb{R}$ is said to be a value assignment on $\mathcal{O}$ if it satisfies the following conditions.
\begin{itemize}
    \item[(a)] For any arrow $g$ we have $V(g) \in \sigma(\text{cod}(g))= g(\sigma(\text{dom}(g)))$. More explicitly, given any arrow $A \xrightarrow{g} B$,
\begin{align}
    V(A \xrightarrow{g} B) \in \sigma(B).
\end{align}
\item[(b)] For each arrow  $A \xrightarrow{g} B$,
\begin{align}
    V(A \xrightarrow{g}B) = g(V(A \xrightarrow{\text{id}_{A}} A)).
\end{align}
\end{itemize}
A value assignment is said to be \textbf{noncontextual} if, for every arrow $A \xrightarrow{g} B$,
\begin{align}
    V(A \xrightarrow{g} B) = V(B \xrightarrow{\id_{B}} B),
\end{align}
which is equivalent to saying that $V(g)=V(h)$ whenever $\text{cod}(g)=\text{cod}(h)$. Otherwise, $V$ is said to be contextual.
\end{definition}

There is a clear one-to-one correspondence between noncontextual value assignments and valuation functions. In fact, given any value assignment $V: \mathcal{O}_{1} \ri \mathbb{R}$, we can define a valuation function $V_{0}: \mathcal{O}_{0} \rightarrow \mathbb{R}$\footnote{Recall that we usually denote the category of observables and the set of observables both by $\mathcal{O}$. Here, to distinguish the set of observables from the set of arrows, it is convenient to use the standard notation followed in category theory and to denote the set of observables by $\mathcal{O}_{0}$. This is why we are saying that a valuation function is defined in $\mathcal{O}_{0}$, whereas, in definition~\ref{def: valuation}, we say that it is defined in $\mathcal{O}$; in both cases, the domain of valuation functions is the set of observables.} by $V_{0}(A) \doteq V(A \rightarrow A)$ for each observable $A$ (recall that the category of observables is a thin category, so the only arrow from $A$ to $A$ is the identity arrow). Given any object $A$ and any arrow $A \xrightarrow{g} B$, we have $g(V_{0}(A)) = g(V(A \ri A)) = V(A \xrightarrow{g} B)$. If $V$ is noncontextual, then $ V(A \xrightarrow{g} B) = V(B \ri B) = V_{0}(B)$, and therefore $V_{0}(g(A)) = V_{0}(B) = g(V_{0}(A))$. It shows that $V_{0}$ satisfies the functional composition principle, and it is trivial to prove that $V_{0}$ satisfies the spectrum rule. On the other hand,  given any valuation function $V_{0}: \mathcal{O} \ri \mathbb{R}$, we can define a value assignment $V: \mathcal{O}_{1} \ri \mathbb{R}$ by $V(A \xrightarrow{g} B) \doteq V_{0}(B)$ for each arrow $A \xrightarrow{g} B$. Since $V_{0}$ satisfies the functional composition principle, we have $V(A\xrightarrow{g} B) = V_{0}(B) = V_{0}(g(A))=g(V_{0}(A))= g(V(A \ri A))$, which means that item $(b)$ from definition~\ref{def: contextualAssignment} is satisfied, and it easily follows from the spectral rule that item $(a)$ from definition~\ref{def: contextualAssignment} is also satisfied. Finally, $V$ is noncontextual by construction, and $V_{0}$ clearly corresponds to the valuation function induced by $V$. Hence, we have the following lemma. 

\begin{lemma} Let $\mathfrak{S}$ be a system satisfying postulates \ref{ax: separability}-\ref{post: compatibility}. Let $\mathcal{O}$ be its category of observables, and let $\mathcal{O}_{0}$, $\mathcal{O}_{1}$ be the sets of objects and arrows of $\mathcal{O}$ respectively. Any noncontextual value assignment $V: \mathcal{O}_{1} \ri \mathbb{R}$ canonically defines a valuation function $V_{0}$ given by $V_{0}(A) \doteq V(A \rightarrow A)$ for each $A \in \mathcal{O}_{0}$. Furthermore, any valuation function on $\mathcal{O}$ is induced by a noncontextual value assignment in this way.
\end{lemma}

We can finally make the aforementioned notion of contextuality precise. Motivated by Ref.~\cite{deRonde2020unscrambling}, we call it ontic contextuality.

\begin{definition}[Ontic contextuality]\label{def: onticContextuality} Let $\mathfrak{S}$ be a system satisfying postulates \ref{ax: separability}-\ref{post: compatibility}. We say that $\mathfrak{S}$ is ontic-noncontextual if there is a noncontextual value assignment on its category of observables. We say that $\mathfrak{S}$ is  ontic-contextual otherwise.
\end{definition}

Hidden states on KS-models define valuation functions (equivalently, noncontextual value assignments), so we have the following lemma.

\begin{lemma}[Ontic contextuality and KS-models] If a system $\mathfrak{S}$ is ontic-contextual, there is no KS-model for it.
\end{lemma}

Kochen-Specker theorem can be reformulated as follows.

\begin{proposition}[Ontic contextuality in quantum systems] Let $\mathfrak{S}$ be a quantum system of dimension larger than $2$. Then $\mathfrak{S}$ is ontic-contextual.
\end{proposition}

Another reason why quantum systems are said to be contextual is that some ``quantum predictions'' do not admit a ``global probability distribution'' \cite{budroni2021review, abramsky2011sheaf, amaral2018graph}. This probabilistic notion of contextuality has received much attention in recent years \cite{budroni2021review, amaral2018graph}, going so far as to be studied outside the realm of physics \cite{cervantes2018snow, jones2019causal, wang2021analysing}, and, just as ontic contextuality, it  implies that KS-models for quantum systems do not exist \cite{budroni2021review, amaral2018graph}. We will turn our attention to it from now on.

Let $\mathfrak{S}$ be a physical system satisfying postulates \ref{ax: separability}-\ref{post: compatibility}. As we saw in section~\ref{sec: compatibility}, if $A,B$ are compatible observables in $\mathfrak{S}$, then $A$ and $B$ commute (definition~\ref{def: commutativity}), which means that, for any $\Delta \subset \sigma(A)$, $\Sigma \subset \sigma(B)$ and any state $\rho$, we have
\begin{align*}
    P_{\rho}(\Delta \times \Sigma;A,B)=P_{\rho}(\Sigma \times \Delta;B,A),
\end{align*}
or equivalently 
\begin{align*}
    P_{\rho}(A^{\Delta})P_{\rho}(B^{\Sigma} \vert A^{\Delta}) = P_{\rho}(B^{\Sigma})P_{\rho}(A^{\Delta} \vert B^{\Sigma})  
\end{align*}
(the notation is explained in definition~\ref{def: updatedState}). Furthermore, lemma~\ref{lemma: updateCommute} ensure that
\begin{align*}
    T_{(\Sigma;B)} \circ T_{(\Delta;A)} = T_{(\Delta;A)} \circ T_{(\Sigma;B)}.
\end{align*}

Together with the well-known fact that any permutation can be written as a product of transpositions \cite{clark1984elements}, these results imply that the following lemma holds true.

\begin{lemma}\label{lemma: permutingProbabilityAppendix} Let $\mathfrak{S}$ be a physical system satisfying postulates \ref{ax: separability}-\ref{post: compatibility}. Let $A_{1},\dots,A_{m}$ be pairwise compatible observables in $\mathfrak{S}$, and let $\pi$ be any permutation of $\{1,\dots,m\}$ Then, for any state $\rho$ and any $\Delta_{1}\times \dots \times \Delta_{m} \subset \prod_{i=1}^{m}\sigma(A_{i})$,
\begin{align}
    P_{\rho}(\Delta_{1} \times \dots \times\Delta_{m};A_{1},\dots,A_{m}) &=P_{\rho}(\Delta_{\pi(1)}\times\dots\times\Delta_{\pi(m)};A_{\pi(1)},\dots,A_{\pi(m)}),\\
    T_{(\Delta_{1},\dots,\Delta_{m};A_{1},\dots,A_{m})} &=T_{(\Delta_{\pi(1)},\dots,\Delta_{\pi(m)};A_{\pi(1)},\dots,A_{\pi(m)})}.
\end{align}
\end{lemma}
Recall that $T_{(\Delta_{1},\dots,\Delta_{m};A_{1},\dots,A_{m})} \equiv T_{(\Delta_{m};A_{m})} \circ \dots \circ T_{(\Delta_{1};A_{1})}$, as in definition~\ref{def: sequentialEvent}. Note that this is exactly the statement of corollary~\ref{cor: permutingProbability}, which was proved for any system satisfying postulates \ref{ax: separability}-\ref{post: commutativity}.

Lemma~\ref{lemma: permutingProbabilityAppendix} implies that pairwise compatible observables are permutable:

\begin{definition}[Permutable sequence of observables]\label{def: permutableSet} Let $\mathfrak{S}$ be a physical system satisfying postulates \ref{ax: separability}-\ref{post: compatibility}, and let $A_{1},\dots,A_{m}$ be observables in $\mathfrak{S}$. Then $A_{1},\dots,A_{m}$ are said to be permutable --- more precisely, the sequence $(A_{1},\dots,A_{m})$ is said to be permutable --- if for any permutation $\pi$ of $\{1,\dots,m\}$, any state $\rho$ and any $(\af_{1}, \dots , \af_{m}) \in \prod_{i=1}^{m}\sigma(A_{i})$,
\begin{align}
    p_{\rho}(\af_{1}, \dots,\af_{m};A_{1},\dots,A_{m}) &= p_{\rho}(\af_{\pi(1)},\dots,\af_{\pi(m)};A_{\pi(1)},\dots,A_{\pi(m)}),\\
    T_{(\af_{1},\dots,\af_{m};A_{1},\dots,A_{m})} &=T_{(\af_{\pi(1)},\dots,\af_{\pi(m)};A_{\pi(1)},\dots,A_{\pi(m)})}.
\end{align}
\end{definition}
It is worth emphasizing that permutability extends to the probability measure $P_{\rho}$ induced by the distribution $p_{\rho}$:
\begin{lemma} Let $\mathfrak{S}$ be a physical system satisfying postulates \ref{ax: separability}-\ref{post: compatibility}, and let $A_{1},\dots,A_{m}$ be permutable observables in $\mathfrak{S}$. Let $\pi$ any permutation $\pi$ of $\{1,\dots,m\}$. Then, for any state $\rho$ and any $\Delta_{1}\times \dots \times\Delta_{m} \subset \prod_{i=1}^{m}\sigma(A_{i})$,
\begin{align}
    p_{\rho}(\Delta_{1} \times \dots \times\Delta_{m};A_{1},\dots,A_{m}) &=p_{\rho}(\Delta_{\pi(1)}\times\dots\times\Delta_{\pi(m)};A_{\pi(1)},\dots,A_{\pi(m)}),\\
    T_{(\Delta_{1},\dots,\Delta_{m};A_{1},\dots,A_{m})} &=T_{(\Delta_{\pi(1)},\dots,\Delta_{\pi(m)};A_{\pi(1)},\dots,A_{\pi(m)})}.
\end{align}
\end{lemma}

Permutable sequences satisfy the so-called nondisturbance condition, which, for convenience, we enunciate in a slightly different way than usual \cite{budroni2021review, abramsky2011sheaf, amaral2018graph}:

\begin{lemma}[Nondisturbance]\label{lemma: nondisturbance} Let $\mathfrak{S}$ be a system satisfying postulates \ref{ax: separability}-\ref{post: compatibility}. Let $A_{1},\dots,A_{m}$ be permutable observables in $\mathfrak{S}$, and let $A_{i_{1}},\dots,A_{i_{k}}$ be any subsequence of $A_{1},\dots,A_{m}$. For each $j \in \{1,\dots,k\}$, fix a set $\Delta_{i_{j}} \subset \sigma(A_{i_{j}})$. Define $\Delta_{i} \doteq \Delta_{i_{j}}$ if $i=i_{j}$ for some $j \in \{1,\dots,k\}$, and define $\Delta_{i} \doteq \sigma(A_{i})$ otherwise. Then, for any state $\rho$,
\begin{align}
    P_{\rho}(\Delta_{i_{1}}\times\dots\times \Delta_{i_{k}};A_{i_{1}},\dots,A_{i_{k}}) = P_{\rho}(\Delta_{1}\times\dots\times\Delta_{m};A_{1},\dots,A_{m}).
\end{align}
Equivalently, for any $(\af_{i_{1}},\dots,\af_{i_{k}}) \in \prod_{j=1}^{k} \sigma(A_{i_{j}})$,
\begin{align}
    p_{\rho}(\af_{i_{1}},\dots, \af_{i_{k}};A_{i_{1}},\dots,A_{i_{k}}) = \sum_{i \notin \{i_{1},\dots,i_{k}\}}\sum_{\af_{i} \in \sigma(A_{i})}p_{\rho}(\af_{1},\dots,\af_{m};A_{1},\dots,A_{m}).
\end{align}
\end{lemma}
\begin{proof}
Let $A_{1},\dots,A_{m}$ be permutable observables in a system $\mathfrak{S}$ satisfying postulates \ref{ax: separability}-\ref{post: compatibility}. Let $A_{i_{1}},\dots,A_{i_{k}}$ be any subsequence of $A_{1},\dots,A_{m}$, and fix some $(\af_{i_{1}},\dots,\af_{i_{k}}) \in \prod_{j=1}^{k}\sigma(A_{i_{k}})$. Let $\pi$ be any permutation of $\{1,\dots,m\}$ according to which $\pi(j) = i_{j}$ for all $j \in \{1,\dots,k\}$. Finally, let $\rho$ be any state, and define $\rho_{k} \doteq T_{(\af_{i_{1}},\dots,\af_{i_{k}};A_{i_{1}},\dots,A_{i_{k}})}(\rho) = T_{(\af_{\pi(1)},\dots,\af_{\pi(k)};A_{\pi(1)},\dots,A_{\pi(k)})}(\rho) $. It follows from definitions~\ref{def: sequentialMeasure} and \ref{def: permutableSet} that
\begin{align*}
    p_{\rho}(\af_{i_{1}},\dots, \af_{i_{k}};A_{i_{1}},\dots,A_{i_{k}})&= p_{\rho}(\af_{i_{1}},\dots, \af_{i_{k}};A_{i_{1}},\dots,A_{i_{k}})
    \\
    &\times P_{\rho_{k}}(\sigma(A_{\pi(k+1)})\times \dots \times \sigma(A_{\pi(m)});A_{\pi(k+1)},\dots,A_{\pi(m)})
    \\
    &= p_{\rho}(\af_{\pi(1)},\dots, \af_{\pi(k)};A_{\pi(1)},\dots,A_{\pi(k)})
    \\
    &\times\sum_{j=k+1}^{m} \sum_{\af_{\pi(j)} \in \sigma(A_{\pi(j)})}p_{\rho_{k}}(\af_{\pi(j)},\dots,\af_{\pi(m)};A_{\pi(k+1)},\dots,A_{\pi(m)})
    \\
    &=\sum_{j=k+1}^{m} \sum_{\af_{\pi(j)} \in \sigma(A_{\pi(j)})}p_{\rho}(\af_{\pi(1)},\dots,\af_{\pi(m)};A_{\pi(1)},\dots,A_{\pi(m)})
    \\
    &=\sum_{i \notin \{i_{1},\dots,i_{k}\}}\sum_{\af_{i} \in \sigma(A_{i})}p_{\rho}(\af_{1},\dots,\af_{m};A_{1},\dots,A_{m}).
\end{align*}
\end{proof}

It easily follows from this result that permutable observables commute pairwise:

\begin{corollary}\label{cor: permutableCommute} Let  $\mathfrak{S}$ be a physical system satisfying postulates \ref{ax: separability}-\ref{post: compatibility}, and let $A_{1},\dots,A_{m}$ be permutable observables in $\mathfrak{S}$. Then, for each pair $i,j \in \{1,\dots,m\}$, $A_{i}$ and $A_{j}$ commute.    
\end{corollary}

It follows from postulate~\ref{post: commutativity} that compatibility and commutativity are equivalent concepts, thus:

\begin{corollary}\label{cor: permutabilityAndCompatibility} If, in addition to postulates \ref{ax: separability}-\ref{post: compatibility}, a system $\mathfrak{S}$ satisfies postulate~\ref{post: commutativity}, then $A_{1},\dots,A_{m}$ are pairwise compatible observables in $\mathfrak{S}$ if and only if they are permutable.   
\end{corollary}

Together with the famous Kolmogorov extension theorem \cite{tao2011introduction},\footnote{The Kolmogorov extension theorem applies to infinite families of measurable spaces, so proposition~\ref{prop: permutabilityAndModels} can be extended to any context.} definition~\ref{def: permutableSet} and lemma~\ref{lemma: nondisturbance} ensure that the following proposition holds true --- note that the proposition can also easily be proved by hand.

\begin{proposition}\label{prop: permutabilityAndModels} Let $\mathfrak{S}$ be a system satisfying postulates \ref{ax: separability}-\ref{post: compatibility}, and let $A_{1},\dots,A_{m}$ be permutable observables in $\mathfrak{S}$. Let $\rho$ be a non-null state. Then there exists a probability space $\boldsymbol{\Lambda}_{\rho} \equiv (\Lambda,\Sigma,\mu_{\rho})$ and random variables (i.e., a measurable functions) $f_{i}: \Lambda \ri \sigma(A_{i})$, $i=1,\dots,m$, such that, for any subsequence $A_{i_{1}},\dots,A_{i_{k}}$ of $A_{1},\dots,A_{m}$, the sequential measure $P_{\rho}( \ \cdot \ ;A_{i_{1}},\dots,A_{i_{k}})$ (definition~\ref{def: sequentialMeasure}) is the pushforward of $\mu_{\rho}$ along $(f_{i_{1}},\dots,f_{i_{k}}): \Lambda \ri \prod_{j=1}^{k}\sigma(A_{i_{j}})$. It means that, for any $\Delta_{i_{1}}\times \dots \times \Delta_{i_{k}} \subset \prod_{j=1}^{k} \sigma(A_{i_{j}})$,
    \begin{align}
        P_{\rho}(\Delta_{i_{1}}\times \dots \times \Delta_{i_{k}};A_{i_{1}},\dots,A_{i_{k}})&=\mu_{\rho}(\cap_{j=1}^{m}f_{i_{j}}^{-1}(\Delta_{i_{j}})).
    \end{align}
\end{proposition} 

It proves that $\mathcal{C} \doteq \{A_{1},\dots,A_{m}\}$ is noncontextual if $A_{1},\dots,A_{m}$ are permutable:

\begin{definition}[Probabilistic contextuality]\label{def: probabilisticContextuality}  Let $\mathfrak{S}$ be a system satisfying postulates \ref{ax: separability}-\ref{post: compatibility}, and let $\mathcal{A}$ be any finite set of observables in $\mathfrak{S}$. We say that $\mathcal{A}$ is \textbf{noncontextual with respect to the state $\boldsymbol{\rho}$} if there exists a probability space $\boldsymbol{\Lambda}_{\rho} \equiv (\Lambda,\Sigma,\mu_{\rho})$ and an assignment of random variables $\mathcal{A} \ni A\mapsto f_{A}$, where $f_{A}: \Lambda \ri \sigma(A)$, such that, for any context $\mathcal{C} \equiv \{A_{1},\dots,A_{k}\} \subset \mathcal{A}$, the sequential measure $P_{\rho}( \ \cdot \ ;A_{1},\dots,A_{m})$  is the pushforward of $\mu_{\rho}$ along $(f_{A_{1}},\dots,f_{A_{m}}): \Lambda \ri \prod_{i=1}^{k}\sigma(A_{i_{j}})$. It means that, for any $\Delta_{1}\times \dots \times \Delta_{k} \subset \prod_{i=1}^{k} \sigma(A_{i})$,
\begin{align}
        P_{\rho}(\Delta_{1}\times \dots \times \Delta_{k};A_{1},\dots,A_{k})&=\mu_{\rho}(\cap_{i=1}^{m}f_{A_{i}}^{-1}(\Delta_{i})).
\end{align}  
Otherwise, we say that $\mathcal{A}$ is probabilistically contextual with respect to $\rho$. If $\mathcal{A}$ is noncontextual w.r.t. all states, we say that $\mathcal{A}$ is \textbf{probabilistically noncontextual}, and we say that it is probabilistically contextual otherwise. Finally, we say that $\mathfrak{S}$ is a \textbf{probabilistically noncontextual system} if all finite sets of observables in $\mathcal{S}$ are probabilistically noncontextual; otherwise, $\mathfrak{S}$ is said to be probabilistically contextual.     
\end{definition}

Again, we are using terminology that is slightly different than usual, but it is easy to see that this definition is equivalent to the standard definition of (probabilistic) contextuality \cite{budroni2021review, amaral2018graph, abramsky2011sheaf}. It means that a finite set of observables $\mathcal{A}$ is noncontextual w.r.t. a state $\rho$ if and only if there is a joint probability distribution accounting for the distributions defined by $\rho$ in each context included in $\mathcal{A}$ \cite{budroni2021review, abramsky2011sheaf}, or equivalently that there is a ``classical realization'' for the ``behavior'' defined by $\rho$ in the ``scenario'' $(\mathcal{A},\mathscr{C})$, where $\mathscr{C}$ denotes the set of all contexts included in $\mathcal{A}$ \cite{amaral2018graph, abramsky2011sheaf}.

Let $\mathfrak{S}$ be a system satisfying postulates \ref{ax: separability}-\ref{post: compatibility}. We know that compatible observables commute (definition~\ref{def: commutativity}) in $\mathfrak{S}$, and that finite contexts are permutable (lemma~\ref{lemma: permutingProbabilityAppendix}). It means that sequential measurements of pairwise compatible observables do not depend on ordering. However, as we discussed in section~\ref{sec: compatibility}, it is an open question whether incompatible observables commute, and consequently one cannot tell whether or not a set of observables is permutable when it is not a context. Probabilistic contextuality enables us to conclude that a certain set of observables depends on order (i.e., that is not permutable) just by analyzing sequential measurements of compatible measurements. Thus, for instance, one can argue that experimental tests of quantum contextuality \cite{budroni2021review, Zhang2019Experimental} provide experimental evidence that observables associated with noncommuting selfadjoint operators violate the Bayes rule (see definition~\ref{def: commutativity}), as quantum theory predicts, without in any way measuring these observables in sequence. On the other hand, if we are right in saying that it is problematic to represent observable events associated with incompatible observables as events in the same probability space, as we argued throughout the paper, then, from the perspective of our work, the existence of order dependence for incompatible observables and the non-existence of classical \textit{realist} models such as KS-models are apparently the only conclusions that we can immediately draw when probabilistic contextuality is witnessed. As we discussed in section~\ref{sec: categoryOfObservables}, we do not commit ourselves to the realist view on observables that lies behind these realist models (this is a lesson we learned from Kochen-Specker theorem), and, as discussed in section~\ref{sec: compatibility}, the equivalence between compatibility and commutative (i.e., independence of order) is for us a reasonable property, so, from our point of view, both these consequences of probabilistic contextuality seem perfectly reasonable. In short, once realism is left aside right from the beginning, learning that classical realist descriptions cannot exist ceases to be a surprise. In the particular case of Bell scenarios \cite{brunner2014Bell, Popescu2014beyond}, where probabilistic contextuality turns out to be equivalent to the violation of Bell inequalities, we don't see why it should be evident that these violations require some sort of nonlocal phenomena, as some authors assert \cite{Popescu2014beyond, brunner2014Bell, cabello2022howNonlocality}, given that, from the point of view of our work, the problem lies in the attempt of embedding events associated with incompatible observables in the same probability space, which is always a local problem. In any case, we have the following result.

\begin{proposition}\label{prop: contextualityNotPermutable} Let $\mathfrak{S}$ be a system satisfying postulates \ref{ax: separability}-\ref{post: compatibility}, and let $\mathcal{A}$ be any finite set of observables in $\mathfrak{S}$. If $\mathcal{A}$ is contextual w.r.t. some state, $\mathcal{A}$ is not a permutable set. 
\end{proposition}

To conclude this appendix, let's prove that, in the particular case of  systems that satisfy postulates \ref{ax: separability}-\ref{post: commutativity} (such as quantum systems), probabilistic contextuality obstructs the existence of KS-models. This result is well known in quantum systems --- this is actually the reason why probabilistic contextuality was defined in the first place \cite{amaral2018graph, budroni2021review} --- and it is a quite trivial result, but we prove it here for the sake of completeness. Also, note that, although probabilistic contextuality rules out KS-models,  showing that a system is probabilistically noncontextual does not seem to imply that a KS-model exists for it. This is because a KS-model consists of a probability space account for all observables and all states of the system at once, whereas probabilistic contextuality ensures the existence of a probability space for each pair $(\mathcal{A},\rho)$, where $\mathcal{A}$ is a finite set of observables and $\rho$ is a state.

To begin with, corollary~\ref{cor: permutabilityAndCompatibility} tells us that, if, in addition to postulates \ref{ax: separability}-\ref{post: compatibility}, $\mathfrak{S}$ satisfies postulate~\ref{post: commutativity}, then permutable finite sequences of observables and finite contexts are equivalent concepts. For this reason, we have the following corollary of proposition~\ref{prop: contextualityNotPermutable}.

\begin{corollary}\label{cor: contextualityNotCompatible}  Let $\mathfrak{S}$ be a system satisfying postulates \ref{ax: separability}-\ref{post: commutativity}, and let $\mathcal{A}$ be any finite set of observables in $\mathfrak{S}$. If $\mathcal{A}$ is contextual w.r.t. some state, $\mathcal{A}$ is not a context.
\end{corollary}

Finally: 

\begin{lemma}\label{lemma: KSprobabilisitc} Let $\mathfrak{S}$ be a system satisfying postulates \ref{ax: separability}-\ref{post: commutativity}. If $\mathfrak{S}$ is probabilistically contextual (definition~\ref{def: probabilisticContextuality}), there is no KS-model for $\mathfrak{S}$ (definition~\ref{def: KSmodel}).
\end{lemma}
\begin{proof}
    Let $\mathfrak{S}$ be a system satisfying postulates~\ref{ax: separability}-\ref{post: commutativity}. Let's show that, if a KS-model exists for $\mathfrak{S}$, then $\mathfrak{S}$ is probabilistically noncontextual. Needless to say, this is equivalent to proving that there is no KS-model for probabilistically contextual systems satisfying postulates~\ref{ax: separability}-\ref{post: commutativity}. So let $\mathfrak{M} \equiv (\boldsymbol{\Lambda},\Xi,\Phi)$  be a KS-model for $\mathfrak{S}$. According to corollary~\ref{cor: specker} (Specker's principle), if $A_{1},\dots,A_{m}$ are pairwise compatible observables, there is a cone $C \xrightarrow{g_{i}} A_{i}$, $i=1,\dots,m$, for them in the category of observables, i.e., there is an observable $C$ are real functions $g_{1},\dots,g_{m}$ on $\sigma(C)$ such that $A_{i}=g_{i}(C)$ for all $i$ (definition~\ref{def: functionalRelation}). Furthermore, it follows from proposition~\ref{prop: jointPushforward} and definition~\ref{def: KSmodel} that, for any $\Delta_{1} \times \dots \times \Delta_{m} \subset \prod_{i=1}^{m} \sigma(A_{i})$ and any state $\rho$,
    \begin{align*}
        P_{\rho}(\Delta_{1} \times \dots \times \Delta_{m};A_{1},\dots,A_{m})&= P_{\rho}(\cap_{i=1}^{m}g_{i}^{-1}(\Delta_{i});C) = \mu_{\rho}(f_{C}^{-1}(\cap_{i=1}^{m}g_{i}^{-1}(\Delta_{i})))
        \\
        &= \mu_{\rho}(\cap_{i=1}^{m}f_{C}^{-1}(g_{i}^{-1}(\Delta_{i}))) =  \mu_{\rho}(\cap_{i=1}^{m}(g_{i} \circ f_{C})^{-1}(\Delta_{i}))
        \\
        &=\mu_{\rho}(\cap_{i=1}^{m}f_{g_{i}(C)}^{-1}(\Delta_{i})) 
        \\
        &= \mu_{\rho}(\cap_{i=1}^{m}f_{A_{i}}^{-1}(\Delta_{i})) 
    \end{align*}
    where $\mu_{\rho} = \Phi(\rho)$ and, for any observable $D$, $f_{D} = \Xi(D)$ (see definition~\ref{def: KSmodel}). It easily follows from this result that $\mathfrak{S}$ is probabilistically noncontextual, which completes the proof.
\end{proof}

\bibliographystyle{IEEEtran}
\bibliography{bibliography}

\end{document}